\documentclass[10pt,a4paper]{article}
\usepackage[utf8]{inputenc}
\usepackage{amsmath, amsthm}
\usepackage{amsfonts}
\usepackage{amssymb}
\usepackage{bbm}
\usepackage{hyperref}
\usepackage{float}

\newcommand{\Cin}{\underline{C}_{in}}

\newcommand{\OmegaRN}{\Omega^2_{RN}}

\newcommand{\omer}{\omega_{res}}
\newcommand{\uhr}{u_{\mathcal{H}_R}}
\newcommand{\uhl}{u_{\mathcal{H}_L}}
\newcommand{\uchr}{u_{\mathcal{CH}_{\mathcal{R}}}}
\newcommand{\uchl}{u_{\mathcal{CH}_{\mathcal{L}}}}
\newcommand{\phiHL}{(\phi_{\mathcal{L}}')_{\mathcal{H}^+}}
\newcommand{\phil}{\phi_{\mathcal{L}}'}
\newcommand{\phiH}{\phi_{H}}

\newcommand{\uH}{u_{\mathcal{H}^+}}

\newcommand{\usi}{u_{\mathcal{S}_{i^+}}}

\newcommand{\ep}{\epsilon}
\newcommand{\uch}{u_{\CH}}

\newcommand{\A}{\mathcal{A}}
\newcommand{\T}{\mathcal{T}}
\newcommand{\R}{\mathcal{R}}

\newcommand{\blue}[1]{{\color{black} #1}}
\newcommand{\teal}[1]{{\color{black} #1}}
\newcommand{\magenta}[1]{{\color{black} #1}}

\newcommand{\CH}{\mathcal{CH}_{i^+}}

\newtheorem{theo}{Theorem}

\usepackage{slashed}
\usepackage[left=2cm,right=2cm,top=1.51cm,bottom=1.5cm]{geometry}

\usepackage{authblk}
\usepackage{enumerate}

\usepackage{enumerate}

\theoremstyle{plain}
\newtheorem{thm}{Theorem}[section]

\newtheorem{lemma}[thm]{Lemma}

\newtheorem{prop}[thm]{Proposition}

\newtheorem{cor}[thm]{Corollary}

\newtheorem{conjecture}[thm]{Conjecture}

\theoremstyle{remark}
\newtheorem{rmk}{Remark}[section]

\theoremstyle{definition}

\newtheorem{defn}{Definition}[section]

\newcommand{\RR}{\mathbb{R}}

\newcommand{\rd}{\partial}
\newcommand{\ls}{\lesssim}

\newcommand{\HH}{ \mathcal{H}^+}

\theoremstyle{plain}

\theoremstyle{remark}

\theoremstyle{definition}

\usepackage{hyperref, todonotes, mathtools}	

\mathtoolsset{showonlyrefs}

\usepackage{bm}

\usepackage{color}

\addtocounter{tocdepth}{-2}
\usepackage{graphicx}
\usepackage{authblk}
 
\numberwithin{equation}{section}

\title{Asymptotically flat black holes with \\ a singular Cauchy horizon and a spacelike singularity}

\author[1]{Maxime~Van~de~Moortel\thanks{maxime.vandemoortel@rutgers.edu}}
\affil[1]{\small  Department of Mathematics, Rutgers University, 
	Hill~Center,~New~Brunswick~NJ~08854,~United~States~of~America \vskip.1pc \  }

\date{\today}

\begin{document}
	\maketitle
	\thispagestyle{empty}
	\thispagestyle{empty}
	
	\begin{abstract}  
		
		In  our recent work  \emph{[Van de Moortel, The  coexistence of null and spacelike singularities inside spherically symmetric black holes]} \cite{bif}, we analyzed the transition between null and spacelike singularities in spherically symmetric dynamical black holes and demonstrated that the spacelike portion is described by a Kasner metric with positive varying exponents that degenerate to $(1, 0, 0)$ near the null-spacelike transition.

		In the present paper, we provide examples of global spacetimes satisfying  the local assumptions of \cite{bif} and apply its analysis to obtain a large class of  asymptotically flat (spherically symmetric) black hole spacetimes that exhibit coexisting null and spacelike singularities as described in \cite{bif}. Our main results include:

		\begin{enumerate}
			\item  The construction of \emph{one-ended asymptotically flat} black hole spacetimes  solving the Einstein–Maxwell-charged-scalar-field equations. The proof relies on a  new  spacelike-characteristic gluing method between any uncharged spherically symmetric solution to the event horizon of a    charged dynamical black hole.
			\item The construction of a large class of \emph{two-ended asymptotically flat} black hole spacetimes solving the Einstein–Maxwell-(uncharged)-scalar-field equations.

		\end{enumerate} In both cases, we show that the terminal boundary in the black hole interior only has  two distinct components: a weakly singular (null) Cauchy horizon $\CH$ where curvature blows up and a strong singularity $\mathcal{S}=\{r=0\}$. 
		
		\noindent Our construction provides  the first  examples of black holes with coexisting null and spacelike singularities.		These examples hold particular significance in the one-ended case as a model of gravitational collapse, where this phenomenon is conjecturally  generic for the Einstein-scalar-field model, even beyond spherical symmetry.

	\end{abstract}

	\section{Introduction}\label{intro.section}
	The nature of the singularity inside a realistic black hole formed from gravitational collapse remains one of the most profound open problems in General Relativity. While the influential Oppenheimer-Snyder spacetime \cite{OppenheimerSnyder} provided an  early model for a dynamical black hole with a purely spacelike singularity, this scenario has been shown  to be \emph{highly non-generic in gravitational collapse}.  Indeed, recent mathematical results for the Einstein vacuum equations establish that the black hole's terminal boundary must include a null component (a Cauchy horizon, $\CH$) \cite{KerrStab}, and therefore cannot be entirely spacelike. This progress in the understanding of dynamical black holes has thus led to the following conjecture (see \cite{review}).
	\begin{conjecture}[\cite{Dafermos:2004jp,MihalisICM,Kommemi,breakdown}] \label{spacelike.conj}
		The  black hole terminal boundary  in generic gravitational collapse consists of a (weak) null singularity -- the Cauchy horizon $\CH$ --  and a  spacelike singularity $\mathcal{S}$ (Figure~\ref{fig:spacelikeconj}).
	\end{conjecture}
	
	\noindent Yet, despite decades of progress, no black hole solution exhibiting the singularity structure of Conjecture~\ref{spacelike.conj} has ever been constructed. This paper provides the first such construction, which consists of  	spherically symmetric (one or two-ended) asymptotically flat black hole    solutions of the Einstein–Maxwell-charged-scalar-field system:
	\begin{equation} \label{1.1}   Ric_{\mu \nu}(g)- \frac{1}{2}R(g)g_{\mu \nu}= \mathbb{T}^{EM}_{\mu \nu}+  \mathbb{T}^{CSF}_{\mu \nu} ,    \end{equation} 
	\begin{equation} \label{2.1} \mathbb{T}^{EM}_{\mu \nu}=2\left(g^{\alpha \beta}F _{\alpha \nu}F_{\beta \mu }-\frac{1}{4}F^{\alpha \beta}F_{\alpha \beta}g_{\mu \nu}\right),\ \hskip 1 mm \mathbb{T}^{CSF}_{\mu \nu}= 2\left( \Re(D _{\mu}\phi \overline{D _{\nu}\phi})\ -\frac{1}{2}(g^{\alpha \beta} D _{\alpha}\phi \overline{D _{\beta}\phi}  )g_{\mu \nu} \right), \end{equation} \begin{equation} \label{4.1} \nabla^{\mu} F_{\mu \nu}= \frac{ q_{0} }{2}i (\phi \overline{D_{\nu}\phi} -\overline{\phi} D_{\nu}\phi),\ \hskip 1 mm F=dA,\ \hskip 1 mm D_{\mu} = \nabla_{\mu} + i q_0 A_{\mu},
	\end{equation} \begin{equation} \label{5.1} g^{\mu \nu} D_{\mu} D_{\nu}\phi = 0 .	\end{equation} These black holes have a  terminal boundary consisting of a weakly singular null Cauchy horizon ($\CH$) and a strong  singularity ($\mathcal{S}=\{r=0\}$) as depicted in Figure~\ref{fig:spacelikeconj}.
	Our construction offers a new and arguably more realistic global model of gravitational collapse, providing a contrast to the  Oppenheimer-Snyder scenario.

	\begin{figure}[H] \begin{center}\includegraphics[width=79 mm, height=45 mm]{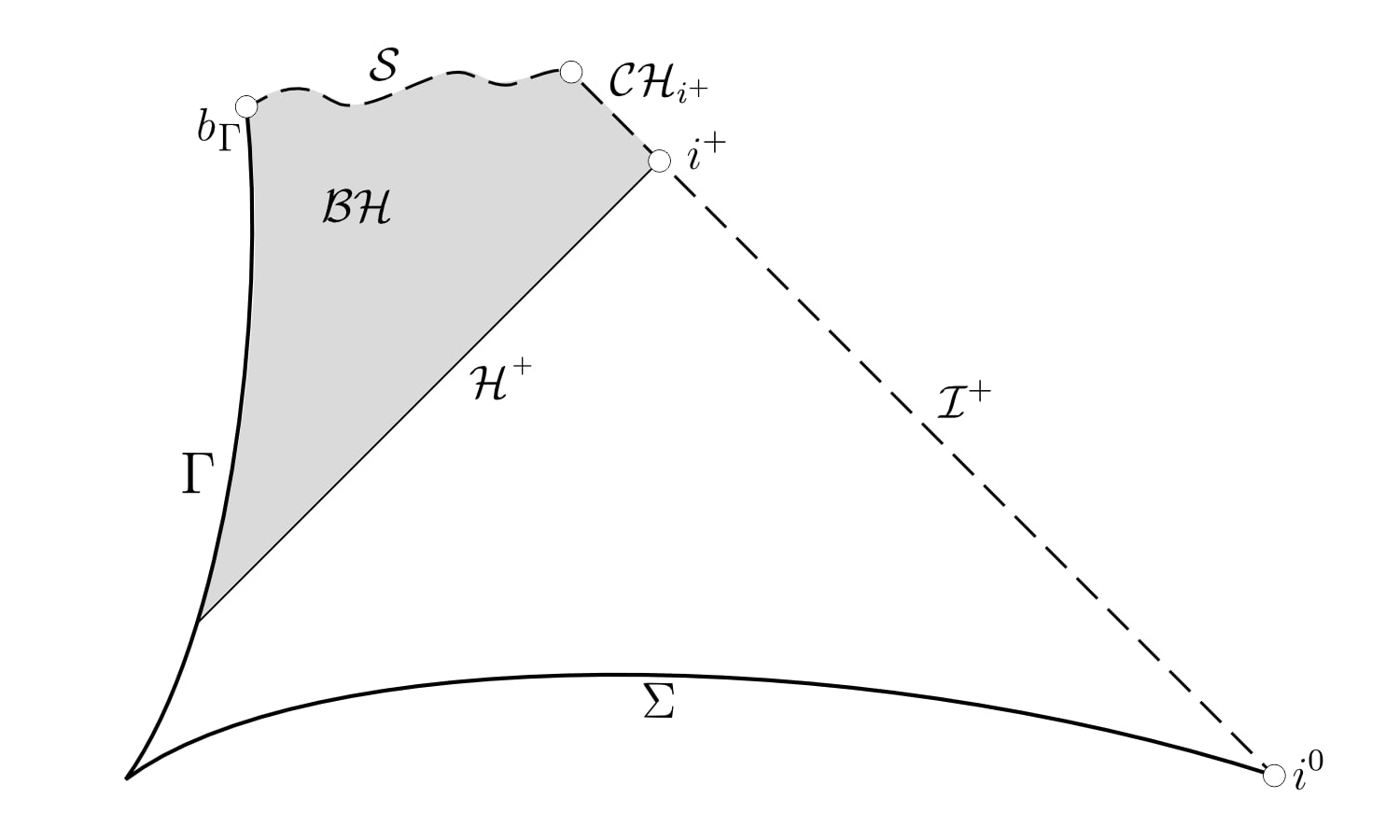}				
			\caption{Penrose diagram of the gravitational collapse (=one-ended) spacetimes obtained in Theorem~\ref{thm.II}. $\CH$ is a weakly singular Cauchy horizon and $\mathcal{S}=\{r=0\}$ is a strong singularity, spacelike near $\CH\cap\mathcal{S}$ and  $b_{\Gamma}$.}
			\label{fig:spacelikeconj}	\end{center}\end{figure}

	\vspace*{-1.5em} 
	In a broader context, the goal of the present paper is to apply the local analysis of our recent work \cite{bif} to a global, asymptotically flat setting, and obtain  black hole solutions of \eqref{1.1}--\eqref{5.1} possessing both a spacelike singularity and a weakly singular null Cauchy horizon whose transition is described by the analysis in \cite{bif}. 
	
	We recall the main result in \cite{bif} that precisely describes the  solution of \eqref{1.1}--\eqref{5.1} in this local region.
	\begin{theo} \label{thm.I}[Theorem I. in \cite{bif}]. Consider  local initial data in the interior of a black hole  consisting of an ingoing cone $\Cin$ and an outgoing cone $C_{out}$ terminating at the sphere of a weakly singular Cauchy horizon $\CH$, and denote  $\mathcal{B}$ the terminal boundary of the resulting solution of \eqref{1.1}--\eqref{5.1}.  Assume   a Cauchy horizon breakdown with no locally-naked singularity, i.e., $ \CH = \{v=+\infty\} \underset{\neq}{\subset} \mathcal{B}  $, 
		and  there exists $s>1$ such that the following  hold: \begin{equation}\label{decay.intro}
			v^{-s} \lesssim	|D_v \phi|_{|C_{out}}(v) \lesssim v^{-s},\  |\Im(\bar{\phi} D_v \phi)|_{|C_{out}}(v) \ll v^{-s},\  	|D^2_{vv} \phi|_{|C_{out}}(v) \lesssim v^{-s-1} \text{ as } v\rightarrow +\infty.
		\end{equation}		Then, $\mathcal{B}$ contains a  spacelike singularity $\mathcal{S} \neq \emptyset$ intersecting $\CH$ as depicted in Figure~\ref{fig:local} and the metric near $\CH \cap \mathcal{S}$ is approximated by a Kasner metric of $v$-dependent positive Kasner exponents  $(1-2p(u,v), p(u,v),p(u,v))$ \blue{which degenerate to $(1,0,0)$ at the following rate}
		
		\begin{equation}\label{p.est.intro}
			p(u,v) \approx \frac{1}{v} 	\text{ as } v \rightarrow +\infty.
		\end{equation} 
	\end{theo}\begin{figure}[H]
		\begin{center}
			
			\includegraphics[width=64 mm, height=40 mm]{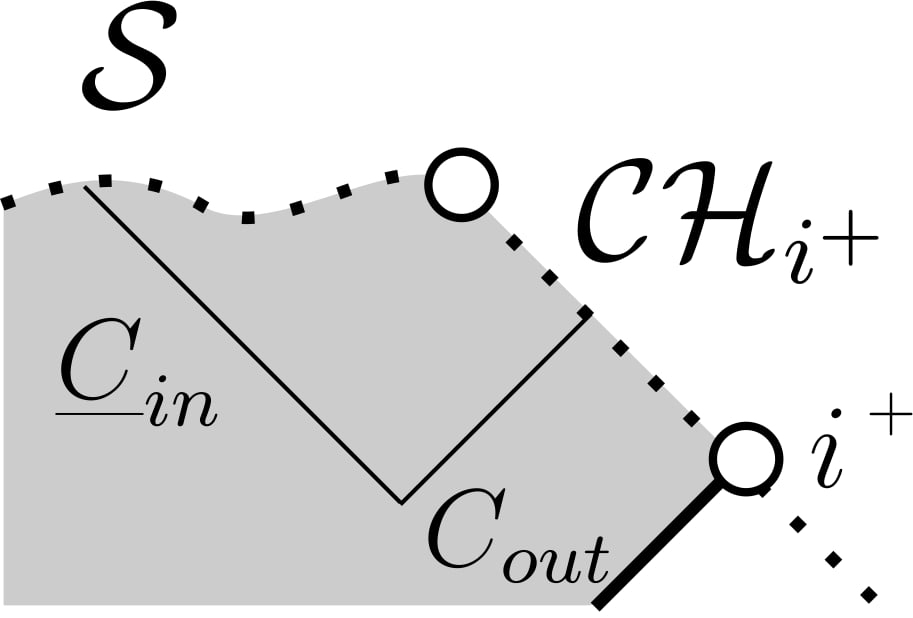}
			
		\end{center}
		\caption{Penrose diagram of the solution obtained in Theorem~\ref{thm.I} with bifurcate initial data $\Cin \cup C_{out}$.}
		\label{fig:local}

	\end{figure}	
	
	\paragraph{Examples of gravitational collapse black holes} We now present our main result for gravitational collapse: the construction of one-ended, asymptotically flat black holes realizing the spacelike/null singularity structure of Conjecture~\ref{spacelike.conj} while additionally analyzing the spacelike/null transition as an application of Theorem~\ref{thm.I}. While the scope of our rigorous analysis is restricted to the spherically symmetric model of equations \eqref{1.1}--\eqref{5.1}, we provide heuristic evidence in Section~\ref{collapse.section} that a generic rotating black hole will possess similar features. Theorem~\ref{thm.II} below is our main result, and a simplified version of Theorem~\ref{main.thm.global.ii} which can be found in Section~\ref{thm.section}.

	\begin{theo} \label{thm.II}
		\blue{Let $k \in \mathbb{N}$.}	There exists a large class of spherically symmetric one-ended asymptotically flat black hole \blue{$C^k$} solutions of \eqref{1.1}--\eqref{5.1} with $q_0\neq0$,  with  a regular center $\Gamma \neq \emptyset$ satisfying the following properties:
		
		\begin{itemize}

			\item The Penrose diagram  is given by Figure~\ref{fig:spacelikeconj}, namely the terminal boundary only has two non-empty components:    a weakly singular null Cauchy horizon $\CH  \neq~ \emptyset$, and  a crushing singularity $\mathcal{S}=\{r=0\}$.
			\item  	$\mathcal{S}$ is spacelike in a neighborhood of $\CH\cap \mathcal{S}$, and obeys the degenerating Kasner asymptotics of Theorem~\ref{thm.I}. 
			\item  $\mathcal{S}$ is spacelike in a neighborhood of $\Gamma$ and spatially-homogeneous, described by an asymptotically Kasner metric of exponents $(\frac{1}{3},\frac{1}{3},\frac{1}{3})$.
		\end{itemize}

	\end{theo}
	\begin{rmk}
		We already note that, to study spherically symmetric one-ended solutions of \eqref{1.1}--\eqref{5.1} with $F\neq 0$  (which is essential to have $\CH \neq \emptyset$ \cite{Christo1}), one must consider a charged scalar field, i.e., $q_0\neq 0$, see e.g. \cite{Kommemi,review}.
	\end{rmk}
	
	\paragraph{New spherically symmetric gluing results}

	Proving Theorem~\ref{thm.II} requires surmounting the challenge of constructing a \emph{global} asymptotically flat black hole solution,  while maintaining sufficiently precise control over the dynamics in order to satisfy the assumptions of Theorem~\ref{thm.I}  in the black hole interior. \\To achieve this, we introduce a novel gluing strategy allowing the construction of one-ended black holes with \emph{any spacelike singularity near the center} combined with \emph{any prescribed} event-horizon late-time tail  decaying at an inverse-polynomial rate. Theorem~\ref{construction.thm} below, a rough version of Theorem~\ref{EH.AF.thm}, describes this.

	\begin{thm}[Black hole  construction by gluing]\label{construction.thm}
		Let $k\in \mathbb{N}$ and		let $\mathcal{M}_L$ a one-ended spherically symmetric black hole \magenta{$C^k$} solution of \eqref{1.1}--\eqref{5.1} with $F\equiv 0$ and  $\phi_{H}(v)$ satisfying the decay assumptions, for some $s>\frac{3}{2}$ \begin{equation}
			|\phi_{H}|(v),\ |D_v\phi_{H}|(v) \lesssim [1+ |v|]^{-s} \text{ as }  v\rightarrow+\infty.
		\end{equation} 	Then, there exists $\mathcal{M}$  a one-ended asymptotically flat spherically symmetric black hole $C^k$ solution of \eqref{1.1}--\eqref{5.1} with $F\neq 0$, $q_0 \neq 0$ satisfying  the following properties:

		\begin{itemize}
			\item Near the center $\Gamma$, $\mathcal{M}$ coincides with $\mathcal{M}_L$.
			\item There exists $0<|e|<M$ such that the black hole relaxes to a Reissner--Nordstr\"{o}m solution of mass $M$ and charge $e$. Moreover, the metric and scalar field are $C^k$ regular across and on the event horizon $\mathcal{H}^+$.
			\item   In Eddington--Finkelstein  coordinate $v$, the scalar field coincides with $\phi_{H}(v)$ on $\mathcal{H}^+$, i.e.,  \begin{equation}
				\phi_{|\mathcal{H}^+}(v)=\phi_{H}(v).
			\end{equation} 
			
		\end{itemize}
	\end{thm} 
	\noindent	We then obtain Theorem~\ref{thm.II} as an application  of Theorem~\ref{construction.thm} where  $\mathcal{M}_L$ is a FLRW metric featuring a spacelike singularity, which, by design, is spatially-homogeneous \blue{with time-slices diffeomorphic to $\RR^3$} (see Section~\ref{gluing.intro}).\begin{rmk}\label{loc.naked.remark}
		One can also apply Theorem~\ref{construction.thm} where  $\mathcal{M}_L$  is an uncharged naked  singularity, as constructed by Christodoulou in \cite{Christo.existence}. The result is a black hole spacetime with a \blue{(charged)} locally naked singularity \blue{$\mathcal{CH}_{\Gamma}$} (see Section~\ref{qual.section}) \blue{and, a consequence from \cite{Moi}, this spacetime also has  a Cauchy horizon $\CH\neq \emptyset$  (see  Theorem~\ref{CH.thm.SS}).}
	\end{rmk}
	
	It turns out that Theorem~\ref{construction.thm} emerges as the outcome of a more general gluing strategy within spherical symmetry.  The main novelty of our approach is to exploit spherical symmetry and combine methods from spacelike and characteristic gluing to operate deep inside the trapped region of the black hole—a regime inaccessible to previous perturbative gluing methods, see Section~\ref{gluing.intro}.
	
	Theorem~\ref{construction.thm} will also be applied to construct analogues of the Oppenheimer--Snyder spacetime for the Einstein-scalar-field system \eqref{1.1}--\eqref{5.1}, both in the uncharged case with a Schwarzschild exterior (Figure~\ref{fig:uncharged}) and  in the charged case with a  Reissner--Nordstr\"{o}m exterior (Figure~\ref{fig:charged}), see \blue{already} Theorem~\ref{OS.thm.intro}. It should  be noted that while these Oppenheimer--Snyder analogues are of great historical interest, they are\blue{, however,} expected to be \emph{non-generic} under Conjecture~\ref{spacelike.conj} (see the discussion following Theorem~\ref{OS.thm.intro}).
	
	Finally, we remark 	that Theorem~\ref{construction.thm} is not yet sufficient to apply  Theorem~\ref{thm.I}. Our gluing method indeed allows us to \emph{construct a global spacetime} with the desired large-scale properties, but we must still \emph{verify} that this construction satisfies the precise local dynamics required by Theorem~\ref{thm.I}, specifically the assumptions \eqref{decay.intro}  on an  outgoing cone strictly to the future of the event horizon. To propagate estimates between the event horizon and the cone on which \eqref{decay.intro} should hold, we use (a slightly refined version of) the scattering theory developed by the author and Kehle  \cite{MoiChristoph} for \eqref{1.1}--\eqref{5.1} with $q_0\neq 0$. This will also be discussed in   Section~\ref{CH.section}.

	\paragraph{Examples of two-ended black holes} We now turn to global applications of Theorem~\ref{thm.I} for  \eqref{1.1}--\eqref{5.1} with $q_0= 0$ (uncharged scalar field). In this case, a spherically symmetric solution of \eqref{1.1}--\eqref{5.1} with $F\not \equiv 0$ is necessarily two-ended (see e.g.\ \cite{Kommemi}). In the two-ended case, however, there exist stable spherically symmetric solutions of \eqref{1.1}--\eqref{5.1} with no spacelike singularity as showed by Dafermos \cite{nospacelike} (left-most Penrose diagram in Figure~\ref{fig:nospacelike}). On the other hand, our result in the two-ended case only applies to  solutions in which there exists a non-empty singularity $\mathcal{S}=\{r=0\}$  (right-most Penrose diagram in Figure~\ref{fig:nospacelike}).

	\begin{theo} \label{thm.III}
		\blue{Let $k\in \mathbb{N}$.}	There exists a large class of spherically symmetric two-ended asymptotically flat black hole \blue{$C^k$} solutions of \eqref{1.1}--\eqref{5.1} for $q_0=0$   satisfying the following properties:
		
		\begin{itemize}

			\item The Penrose diagram  is given by the right panel of  Figure~\ref{fig:nospacelike}, namely the terminal boundary only has two non-empty components:    a weakly singular null Cauchy horizon $\CH  \neq~ \emptyset$, and  a crushing singularity $\mathcal{S}=\{r=0\}$. 
			\item  	$\mathcal{S}$ is spacelike in a neighborhood of $\CH\cap \mathcal{S}$, and obeys the Kasner asymptotics of Theorem~\ref{thm.I}. 
			
		\end{itemize}

	\end{theo}
	\begin{figure}[H]	\begin{center}
			\includegraphics[width=165 mm, height=50 mm]{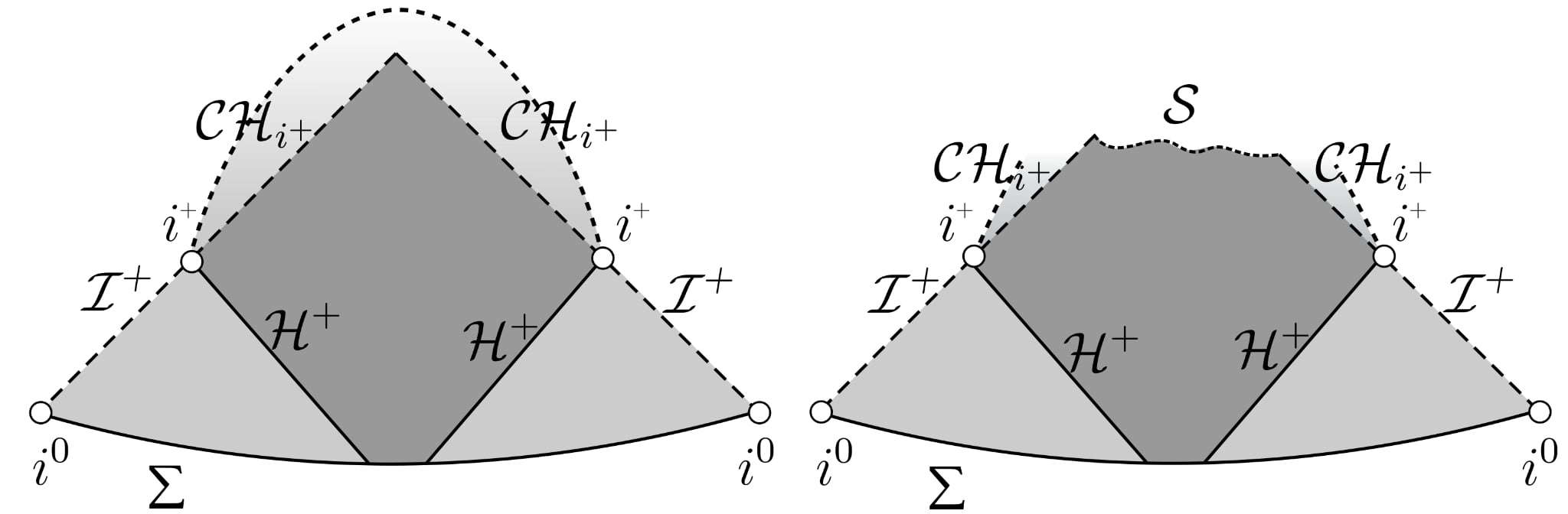}
		\end{center}
		\caption{Left: Two-ended black hole with a Cauchy horizon $\CH$ and no spacelike singularity. \\Right: Two-ended black hole with coexisting Cauchy horizon $\CH$ and   singularity $\mathcal{S}=\{r=0\}$ (Theorem~\ref{thm.III}).}\label{fig:nospacelike}\end{figure}Contrary to the charged case $q_0\neq 0$, the late-time tail theory of generic spherically symmetric solutions of \eqref{1.1}--\eqref{5.1} for $q_0=0$ is well-\blue{understood}, thanks to results by Dafermos--Rodnianski \cite{PriceLaw}, Luk--Oh \cite{JonathanStabExt,twotails} and most recently Gautam \cite{Gautam}. We leverage these results to  apply  the analysis of Theorem~\ref{thm.I} to the subclass of generic spherically symmetric two-ended solutions of \eqref{1.1}--\eqref{5.1} whose Penrose diagram is as in the right-most of Figure~\ref{fig:nospacelike}, and we construct the \emph{first  examples such that $\mathcal{S}\neq \emptyset$} for this model\footnote{The author, however, previously constructed two-ended black holes with a Cauchy horizon $\CH$ and a crushing singularity $\mathcal{S} =\{r=0\}$ for \eqref{1.1}--\eqref{5.1} with a massive scalar field \cite{violent} and with Li for a charged scalar field \cite{fluctuating}, i.e., \eqref{1.1}--\eqref{5.1} with $q_0 \neq 0$.}.
	Not only does Theorem~\ref{thm.III} construct the first black holes with the  right Penrose diagram  of  Figure~\ref{fig:nospacelike}, i.e., $\mathcal{S}\neq \emptyset$, but it also provides the first non-trivial quantitative estimates on the junction between the spacelike singularity $\mathcal{S}=\{r=0\}$ and the Cauchy horizon $\CH$ in a two-ended asymptotically flat black hole.

	\paragraph{Conditional global applications of Theorem~\ref{thm.I}}
	
	A secondary objective of the present manuscript is to find the most general  global conditions on a spherically symmetric black hole  so the  conclusions of Theorem~\ref{thm.I} (which is a local result) hold. In essence, we prove that this is the case as long as there exist one singular sphere with zero area-radius on the terminal boundary (or, equivalently, $\mathcal{S} \neq \emptyset$ in the language of Theorem~\ref{Kommemi.thm}). 
	
	\begin{theo}\label{thm.cond.intro}
		Consider a black hole spacetime with a Cauchy horizon $\CH \neq \emptyset$  and  assume that \eqref{decay.intro} holds on any outgoing cone $C_{out}$ under the Cauchy horizon $\CH$. Then, under the following respective conditions: \begin{enumerate} 
			\item (One-ended case). Assume the absence of a locally-naked singularity emanating from the center $\Gamma$.\label{IV.1}
			
			\item (Two-ended case). Assume that the spacetime is described by  the rightmost Penrose diagram of Figure~\ref{fig:nospacelike}, i.e., $\mathcal{S}\neq \emptyset$.\label{IV.2}
			
		\end{enumerate}
		
		Then, the  terminal boundary only has two components: $\CH$, which is weakly singular in the sense that the Hawking mass is infinite on $\CH$, and a crushing singularity $\mathcal{S}=\{r=0\}$, which is spacelike in a neighborhood of $\CH\cap \mathcal{S}$ and obeys the degenerating Kasner asymptotics of Theorem~\ref{thm.I}.
	\end{theo}
	
	\noindent Theorem~\ref{thm.cond.intro} corresponds\footnote{We note, however, that Theorem~\ref{main.thm.global.i} and Theorem~\ref{main.thm.2end.i} are slightly stronger results than Theorem~\ref{thm.cond.intro} as stated in that they only require \eqref{decay.intro} to hold on \emph{one} outgoing cone $C_{out}$ under $\CH$ so that the conclusion of Theorem~\ref{thm.I} applies.} to Theorem~\ref{main.thm.global.i} (one-ended case) and Theorem~\ref{main.thm.2end.i} (two-ended case).
	\begin{rmk}
		As noted in Remark~\ref{loc.naked.remark}, there exist one-ended black holes with both a Cauchy horizon $\CH$ and a locally naked singularity $\mathcal{CH}_{\Gamma}$. However, we note that locally naked singularities are conjecturally non-generic; this is related to  Weak Cosmic Censorship  for the model \eqref{1.1}--\eqref{5.1}, see \cite{Kommemi,review} or \cite{bif}, Section 1.3 for  details. 
	\end{rmk}
	
	As part of the proof of Theorem~\ref{thm.cond.intro}, we obtain an intermediate result of independent interest, excluding the presence of  certain terminal boundary components that were a priori possible, see Figure~\ref{fig:one-ended} and Section~\ref{qual.section}.\begin{itemize}
		\item  in the one-ended case, the Cauchy horizon $\CH$ cannot be the only terminal boundary component, as in the Penrose diagram of Figure~\ref{fig:disproof} (this result was already obtained in \cite{breakdown} where  we proved the   breakdown of the Cauchy horizon; however, the new proof presented here is shorter and  more constructive).
		\item There is no  ingoing collapsed cone $\mathcal{S}_{i^+}$ on which $r\equiv 0$ that extends  the Cauchy horizon $\CH$ to the future, as depicted in Figure~\ref{fig:one-ended} (more precisely, we prove  $\mathcal{S}_{i^+}=\emptyset$ under the assumptions of Theorem~\ref{thm.I}).
	\end{itemize} \begin{rmk}
		It is also of interest to consider the Einstein--Maxwell equations coupled to a massive and charged scalar field. The main result of \cite{bif}, Theorem~\ref{thm.I}, is also valid in this case, as are the conditional results presented in Theorem~\ref{thm.cond.intro}. However, the construction of one or two-ended asympotically flat black holes provided in Theorem~\ref{thm.II} or Theorem~\ref{thm.III} do not work in the massive case, where late-time tails are more subtle, see \cite{KGSchw1,KGSchw2}.
	\end{rmk}
	
	\subsection{Models of gravitational collapse: Oppenheimer--Snyder's solution and beyond}\label{collapse.section}
	
	\subsubsection{The Oppenheimer--Snyder spacetime}  The rigorous theoretical modeling of gravitational collapse that predicted the  dynamical formation of a black hole began with the celebrated Oppenheimer-Snyder \cite{OppenheimerSnyder} spacetime in 1939. Their construction is a spherically symmetric solution of the Einstein equations coupled with a homogeneous ball of dust  (the density of dust is constant  inside  a ball of fixed radius $R_b$ and zero outside). Its Penrose diagram  is as in Figure~\ref{fig:OS}: a black hole whose terminal boundary only contains the spacelike singularity $\mathcal{S}$. We make some observations on the Oppenheimer--Snyder spacetime: \begin{itemize}
		\item Outside of the ball of fixed radius $R_b$, the Oppenheimer--Snyder spacetime is exactly Schwarzschild, thus $\mathcal{S}$ coincides with Schwarzschild's spacelike singularity, with bounded Jacobi fields in the orthoradial direction but an unbounded Jacobi field in the radial one (also known as ``spaghettification'', see \cite{Hawking,gravitation,ONeill}).
		\item Inside of the ball of fixed radius $R_b$, $\mathcal{S}=\{r=0\}$ is isotropic, coinciding with a FLRW singularity.
		\item The  Oppenheimer--Snyder  spacetime is the MGHD of one-ended asymptotically flat  $W^{1,\infty} \cap C^0$ initial data. As such, it only solves the Einstein-dust equations in a weak (distributional) sense.
	\end{itemize}
	
	\begin{figure}[H]
		\begin{center}
			
			\includegraphics[width=75 mm, height=50 mm]{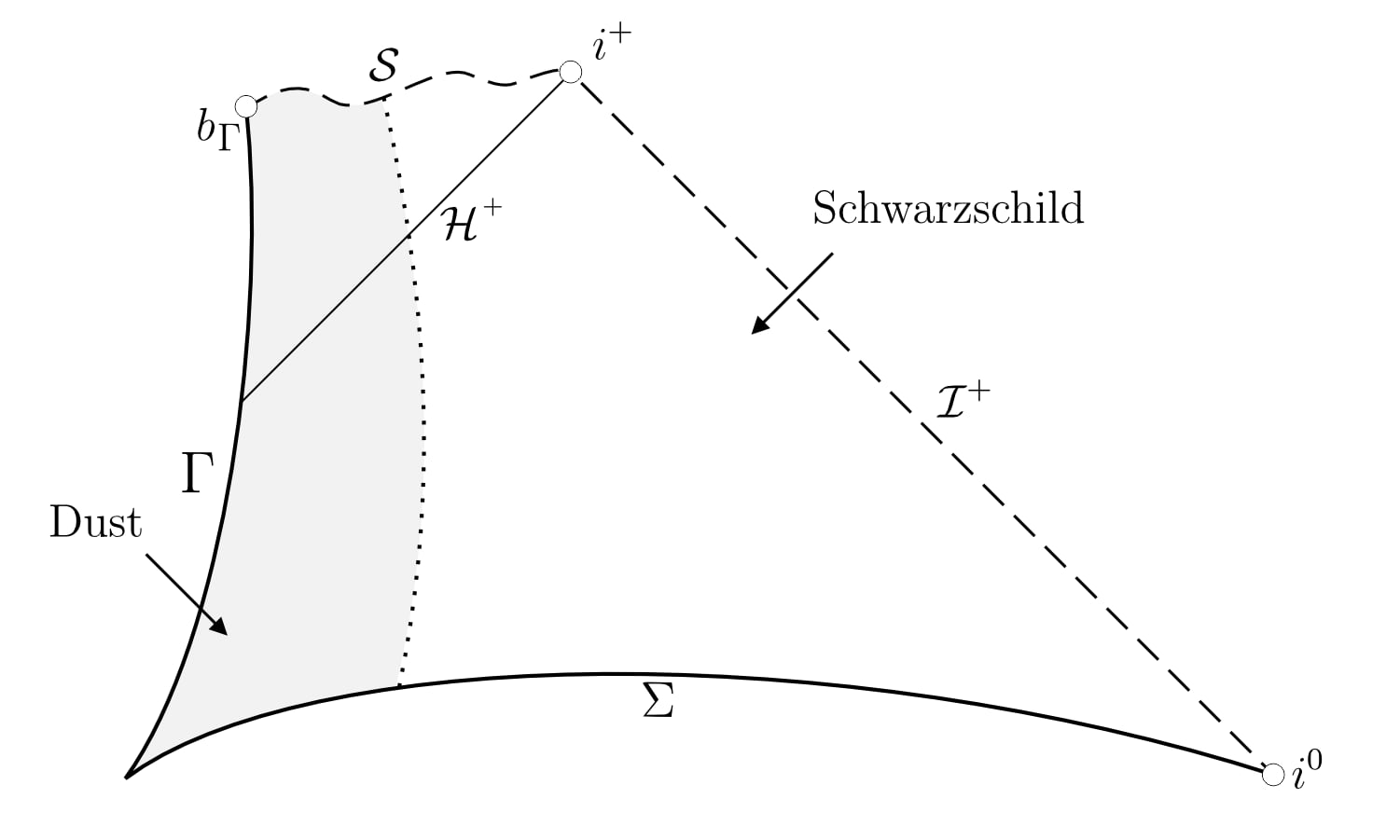}
			
		\end{center}
		\caption{Penrose diagram of the Oppenheimer--Snyder solution of the Einstein-dust equations.}
		\label{fig:OS}
		
	\end{figure} Despite its historical importance, the Oppenheimer--Snyder spacetime is not generic, even among spherically symmetric solutions of the Einstein-dust model: Christodoulou \cite{Christo4} indeed showed that even a slightly inhomogeneous ball of dust collapses into a naked singularity instead of a black hole, due to dust shells crossing.
	
	\subsubsection{Einstein-scalar-field gravitational collapse in spherical symmetry} \label{uncharged.collapse.section}In a series of papers \cite{Christo1,Christo2,Christo.existence,Christo3}, Christodoulou famously studied spherically symmetric solutions of \eqref{1.1}--\eqref{5.1} with $F\equiv 0$, a model known as Einstein-scalar-field, and which is free from shell-crossing singularities. He considered the gravitational collapse case (one-ended asymptotically flat) and showed the following results: \begin{itemize}
		\item Solutions of \eqref{1.1}--\eqref{5.1} (for rough initial data) with a naked (or locally naked) singularity exist \cite{Christo.existence}.
		\item Such (locally)-naked singularity solutions, however, are unstable \cite{Christo3} within a class of rough solutions.
		\item For generic rough initial data, the terminal boundary of the black hole  is a spacelike singularity $\mathcal{S}$ \cite{Christo1}.
	\end{itemize}
	Moreover, it was later shown that (smooth) black holes for this model converge to Schwarzschild \cite{DejanAn,PriceLaw,twotails} in some sense towards timelike infinity $i^+$; in particular, $\mathcal{S}$ also exhibits the above-mentioned ``spaghettification'' near $i^+$. Away from $i^+$, it was shown generally that $\mathcal{S}$ locally adopts a Kasner form \cite{Warren1} with varying exponents.
	
	We also mention the uncharged Kehle--Unger spherically symmetric gluing methods \cite{KehleUnger}, allowing to construct a one-ended black hole whose event horizon is  exactly  Schwarszschild at late times  (see also Section~\ref{gluing.intro}).

	Nonetheless, to the best of the author's knowledge, there is no  known analogue to the Oppenheimer--Snyder solution for \eqref{1.1}--\eqref{5.1}, in the sense of a one-ended asymptotically flat black hole whose terminal boundary is a spacelike singularity coinciding with Schwarzschild's in a open neighborhood of $i^+$, as depicted in Figure~\ref{fig:uncharged}. Theorem~\ref{construction.thm}, however, allows for such a construction (see Theorem~\ref{OS.thm.intro}), as we will explain in Section~\ref{charged.collapse.section}.
	
	\begin{rmk}
		Even though Christodoulou proved that generic (among rough data solutions) black holes only have a spacelike singularity \cite{Christo3}, it is not known how to construct any $C^2$ black hole solution of  \eqref{1.1}-\eqref{5.1} with this feature. This is because the strategy used in \cite{Christo3} to generate (generic) examples of such black holes is based on a perturbation argument of a (locally)naked singularity. The gluing strategy employed in Theorem~\ref{construction.thm}, on the other hand, offers a  method  to construct more concrete examples with arbitrarily high $C^k$-regularity.
	\end{rmk}

	\begin{figure}[H]
		\begin{center}
			
			\includegraphics[width=75 mm, height=50 mm]{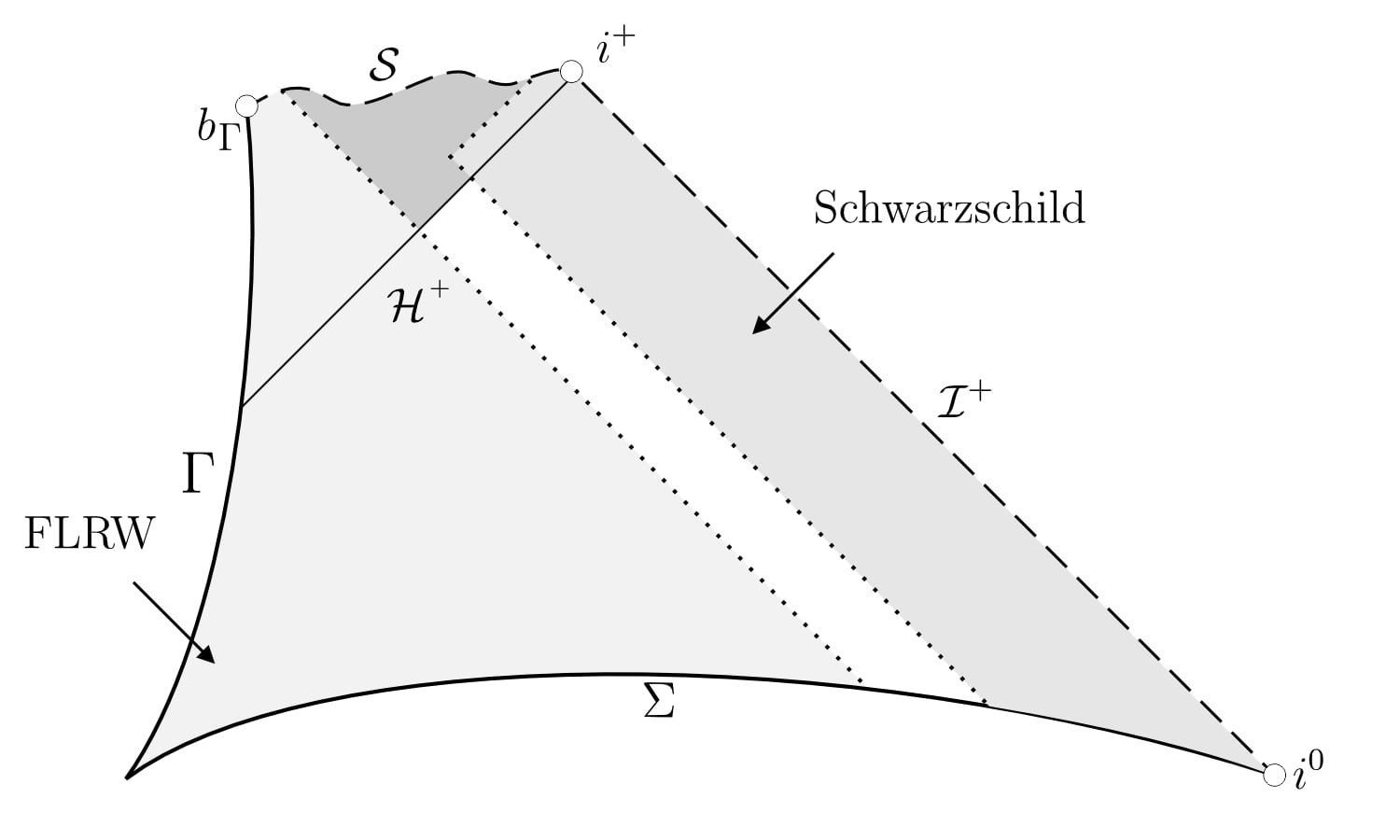}
			
		\end{center}
		\caption{Penrose diagram of the Oppenheimer--Snyder \blue{scalar field analogues solving}  \eqref{1.1}-\eqref{5.1} in Theorem~\ref{OS.thm.intro} for $q=0$.  $\mathcal{S}$ is exactly Schwarzschild's singularity near $i^+$, and exactly the FLRW singularity near $b_{\Gamma}$. Note the absence of a Cauchy horizon ($\CH=\emptyset$), in contrast to the examples of Theorem~\ref{thm.II} (depicted in Figure~\ref{fig:spacelikeconj}).}
		\label{fig:uncharged}
		
	\end{figure}
	
	\subsubsection{Einstein--Maxwell-charged-scalar-field gravitational collapse in spherical symmetry}\label{charged.collapse.section}
	The study of \eqref{1.1}--\eqref{5.1} in the charged case $F\neq 0$ is in parts motivated by its resemblance with   the Einstein equations in vacuum,	 see the discussions in \cite{bif,review}. In particular, as we will elaborate  in Section~\ref{rotating.collapse.section}, dynamical rotating black holes admit a Cauchy horizon from infinity $\CH$, in the same way dynamical charged black holes do \cite{Mihalis1,Moi}. The author  extensively studied  the black hole interior for the  \eqref{1.1}--\eqref{5.1} in spherical symmetry \cite{MoiChristoph,Moi,Moi4,breakdown,bif} and has shown the following results for dynamical black holes converging to Reissner--Nordstr\"{o}m (see Theorem~\ref{CH.thm.SS}  below and Section~\ref{qual.section} for more precise versions of the statements): \begin{itemize}
		\item There exists a Cauchy horizon $\CH$, which is $C^0$-regular under physical assumptions on the event horizon.
		\item $\CH$ is $C^2$-singular under physical assumptions on the event horizon (\emph{weak singularity}).
		\item Due to the weak singularity, $\CH$ breaks down in finite retarded-time. Therefore, there is another terminal boundary component, which could be a crushing singularity $\mathcal{S}$, or a locally naked singularity $\mathcal{CH}_{\Gamma}$.
		\item If we assume  $\mathcal{CH}_{\Gamma}=\emptyset$, then $\mathcal{S}$ is spacelike   and tidally contracting near $\mathcal{S}\cap \CH$, in the sense that all its Jacobi fields tend to zero (or equivalently, its generalized Kasner exponents are positive).
		\item As a conclusion, if $\mathcal{CH}_{\Gamma}=\emptyset$, the only terminal boundary components of the black hole are $\mathcal{S}$ and $\CH$.
	\end{itemize}
	\noindent	The attentive reader will have recognized a rephrasing of the first statement of Theorem~\ref{thm.cond.intro} in \blue{the last two bullet points}. We point out that Theorem~\ref{thm.II} precisely provides (one-ended asymptotically flat) examples such that  $\mathcal{CH}_{\Gamma}=\emptyset$, and  described as above; in addition, the crushing singularity $\mathcal{S}$ is spatially homogeneous near the center $\Gamma$, similarly to the Oppenheimer--Snyder case. However, these examples provide a \emph{model of gravitational collapse competing} with Oppenheimer--Snyder's  with two essential differences that we emphasize:
	\begin{itemize}
		\item The   Cauchy horizon from infinity $\CH$, leading to the Kasner exponents of $\mathcal{S}$  degenerating at $\CH\cap \mathcal{S}$.
		\item The absence of unbounded Jacobi fields (a.k.a spaghettification, or negative Kasner exponents)  near $i^+$.
	\end{itemize}
	
	\noindent	As we discuss in Section~\ref{qual.section}, it is expected that $\mathcal{CH}_{\Gamma}=\emptyset$ for generic solutions, leading to the following conjecture\footnote{The conjectured tidally contracting character of the spacelike singularity goes back to the celebrated BKL heuristics on the stability of spacelike singularities \cite{BKL1,BKL2}, see the discussion in \cite{bif}, Section 1.6.}, which is more  specific than our introductory Conjecture~\ref{spacelike.conj}. \begin{conjecture}\label{charged.conj}
		Generic (regular) spherically symmetric one-ended asymptotically flat black hole solutions for \eqref{1.1}--\eqref{5.1} with $q_0\neq 0$ have only two terminal boundary components: \begin{itemize}
			\item the Cauchy horizon from infinity $\CH$, a weak null singularity.
			\item the crushing singularity $\mathcal{S}=\{r=0\}$, which is spacelike and tidally contracting.
		\end{itemize}
	\end{conjecture} \noindent In other words, it is conjectured that the examples of Theorem~\ref{thm.II} reflect the properties of generic black holes. Proving the validity of Conjecture~\ref{charged.conj}  requires, among other things, a resolution of Weak Cosmic Censorship for charged scalar field (the analogue of Christodoulou's work in the uncharged case \cite{Christo1,Christo.existence,Christo3}),   open at present. \\
	
	We now return to the construction of  Oppenheimer--Snyder like solutions already mentioned in Section~\ref{uncharged.collapse.section}. Our goal is to produce a spatially-homogeneous spacelike singularity near the center $\Gamma$, and exactly isometric to Schwarzschild (or Reissner--Nordstr\"{o}m in the charged case) in an open neighborhood of $i^+$ that includes part of the trapped region.
	The following theorem, a simplified version of Corollary~\ref{charging.cor} (see also Corollary~\ref{uncharged.cor} for the uncharged version), carries out such a construction for $C^{k}$ solutions, where $k$ is arbitrarily large.
	\begin{thm} (Oppenheimer--Snyder \blue{spacetime} analogues with a scalar field).\label{OS.thm.intro}
		For any regularity index $k \in \mathbb{N}$ and   charge to mass ratio $q \in [0,1)$, there exist $C^k$ spherically symmetric one-ended asymptotically flat black hole solutions of \eqref{1.1}--\eqref{5.1} with $q_0\neq0$,  with  a regular center $\Gamma \neq \emptyset$ and satisfying the following:
		
		\begin{itemize}
			\item There exists an open neighborhood of $i^+$ (timelike infinity), $\mathcal{I}^+$ (null infinity) and $i^0$ (spacelike infinity) in which the spacetime coincides with a Reissner--Nordstr\"{o}m metric of charge ratio $\pm q M_f$, for some $M_f>0$ (Schwarzschild if $q=0$).  Thus, at late enough times, the event horizon $\mathcal{H}^+$ is exactly Reissner--Nordstr\"{o}m.
			\item If $q=0$, \blue{the construction is such that its} Penrose diagram  is given by Figure~\ref{fig:uncharged}, namely the terminal boundary only has one non-empty component:    a spacelike singularity $\mathcal{S}=\{r=0\}$. 
			
			\item  If $q\in (0,1)$, the Penrose diagram  is given by Figure~\ref{fig:charged}, namely the terminal boundary only has two non-empty components:    a null Cauchy horizon $\CH  \neq~ \emptyset$, and  a crushing singularity $\mathcal{S}=\{r=0\}$.

			\item  $\mathcal{S}$ is spacelike in a neighborhood of $\Gamma$ and spatially-homogeneous, described by an asymptotically Kasner metric of exponents $(\frac{1}{3},\frac{1}{3},\frac{1}{3})$.
		\end{itemize}

	\end{thm}
	
	\noindent The examples presented in Theorem~\ref{OS.thm.intro} are known to be non-generic \cite{review} and violate  Conjecture~\ref{spacelike.conj}, since $\CH =\emptyset$ in the case $q=0$, and $\CH\neq \emptyset$ is \emph{non-singular} near $i^+$ (since it coincides with Reissner--Nordstr\"{o}m) in the case $q\neq 0$, in contrast to the spacetimes of Theorem~\ref{thm.II}. This distinction prevents the application of Theorem~\ref{thm.I} and leaves the quantitative coexistence of the spacelike and null singularities unresolved \blue{in these examples}. 	Despite the fact that the examples of Theorem~\ref{OS.thm.intro} violate Conjecture~\ref{spacelike.conj}, our proof of Theorem~\ref{thm.II} \emph{builds up} on that of Theorem~\ref{OS.thm.intro}, which already illustrates the heart of our gluing strategy (see Section~\ref{gluing.intro}).
	
	\begin{rmk}
		In Theorem~\ref{OS.thm.intro}, we allow for an arbitrary choice of $q \in [0,1)$ and the resulting spacetime is initially free of trapped surfaces, consistently with our modelization of gravitational collapse (see Section~\ref{qual.section}). However, if we want to construct a spacetime such that the event horizon $\mathcal{H}^+$ lies strictly to the future of the \blue{asymptotically flat} initial data \blue{$\Sigma$} as depicted in Figure~\ref{fig:uncharged} and Figure~\ref{fig:charged}, we must \blue{additionally} restrict $q$ to a \blue{sufficiently} small value\footnote{\blue{More precisely,  the final black hole mass $M$ must be assumed to be  small, and  moreover $0<|q| =O(M^3)$, where $0<M\ll 1$.}}, see Corollary~\ref{charging.cor}.
	\end{rmk}
	
	\begin{figure}[H]
		\begin{center}
			
			\includegraphics[width=90 mm, height= 60 mm]{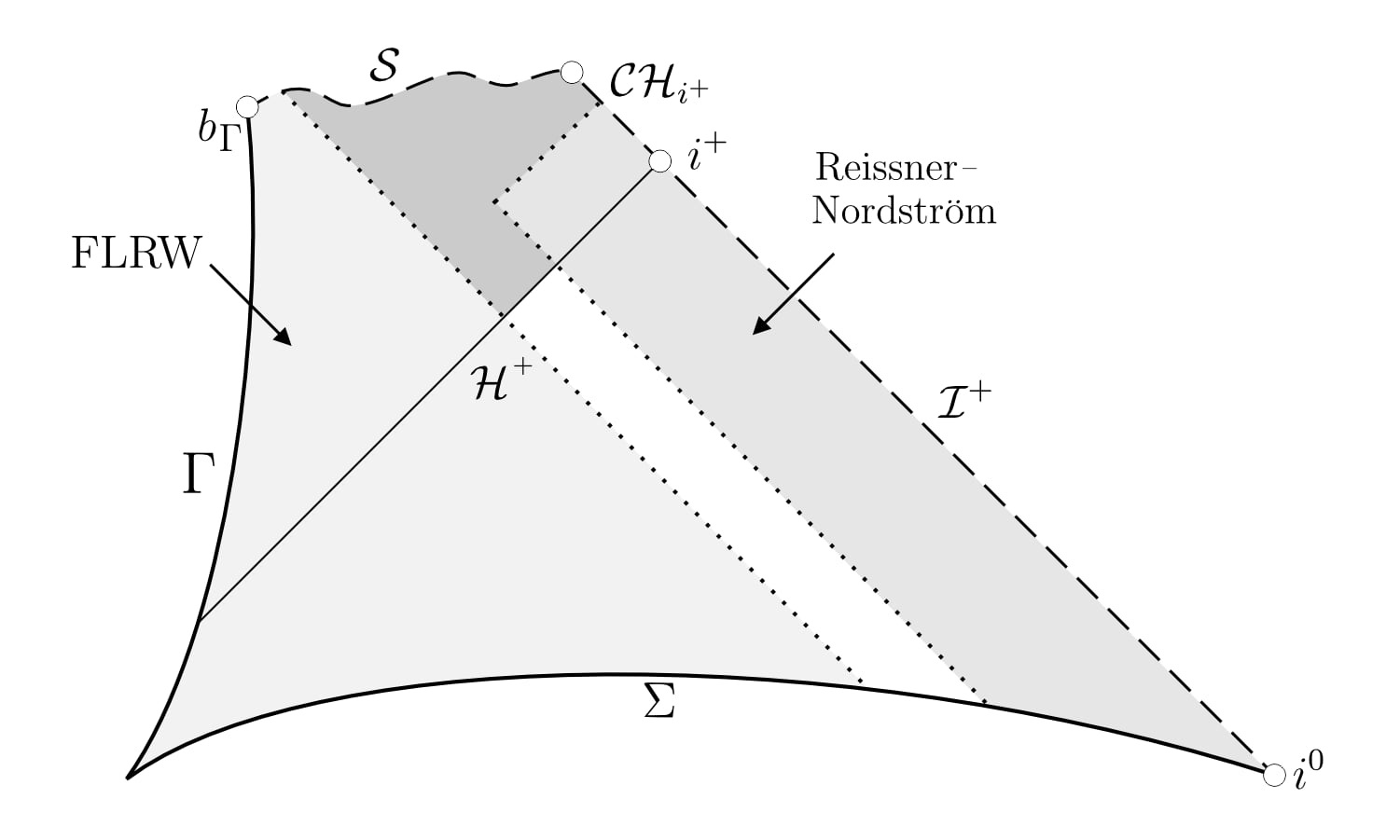}
			
		\end{center}
		\caption{Penrose diagram of the Oppenheimer--Snyder \blue{spacetime charged scalar field analogue solving}   of \eqref{1.1}-\eqref{5.1} obtained in Theorem~\ref{OS.thm.intro} for $q\neq 0$.  $\CH$ is a Reissner--Nordstr\"{o}m Cauchy horizon near $i^+$, and exactly  FLRW  near $b_{\Gamma}$. Note that $\CH$ is not (everywhere) singular, in contrast to the examples of Theorem~\ref{thm.II}.}
		\label{fig:charged}
		
	\end{figure}
	\subsubsection{Einstein-scalar-field gravitational collapse outside of spherical symmetry} \label{rotating.collapse.section}

	The dynamics in the black hole interior is less understood outside of spherical symmetry.  A major result by Dafermos--Luk \cite{KerrStab} established that small perturbations of Kerr spacetimes develop a non-empty Cauchy horizon $(\CH)$, and the work of Luk \cite{JonathanWeakNull} provided the first examples of such Cauchy horizons containing weak null singularities. 			While these findings resolve the nature of the terminal boundary near $i^+$, the global interior structure—particularly the presence of a  spacelike singularity—remains an open problem central  to an eventual resolution of the Strong Cosmic Censorship Conjecture (see Section~\ref{SCC.section}).  This question is conjectured to be particularly delicate \emph{in vacuum}, as no dynamically-stable spacelike singularities are known for the Einstein equations, except if we restrict the dynamics to non-generic symmetry classes \cite{Fournoxakis,FournodavlosRodnianskiSpeck}. This is consistent with the BKL heuristics, which predict chaotic dynamics rather than a stable singular terminal boundary \cite{BKL1,BKL2}.
	
	The result of Dafermos--Luk \cite{KerrStab} notably also extends\footnote{While this is not explicitly stated in \cite{KerrStab}, it can be inferred from the estimates of \cite{KerrStab}, together the methods to handle the scalar field on the Kerr black hole interior, see \cite{Franzen2}.} to the Einstein-scalar-field model (i.e., \eqref{1.1}-\eqref{5.1} with $F\equiv 0$). Contrary to the vacuum case, however, spacelike singularities are better understood in the presence of a scalar field. In this setting, the BKL heuristics \cite{BKL1,BKL2} identify certain spatially-homogeneous solutions (a sub-class of Kasner spacetimes \cite{Kasner}, solutions of  Einstein-scalar-field) as stable, and Fournodavlos, Rodnianski, and Speck indeed proved this \cite{FournodavlosRodnianskiSpeck}. The tidally contracting spacelike singularity (positive Kasner exponents) present in these Kasner models thus provides a plausible model for the structure of generic spacelike singularities.

	In the spherically symmetric case, the spacelike singularity  $\mathcal{S}$   addressed by Theorem~\ref{thm.I} is  indeed tidally contractive (positive Kasner exponents), consistently with the \blue{BKL expectations}. However, its construction cannot not be addressed by the methods of \cite{FournodavlosRodnianskiSpeck} due to a crucial technical distinction. The stability analysis in \cite{FournodavlosRodnianskiSpeck}, as well as other known constructions \cite{FournodavlosLuk}, requires non-degenerating Kasner exponents. In contrast, the singularity $\mathcal{S}$ in Theorem~\ref{thm.I} has Kasner exponents that degenerate \blue{to $(1,0,0)$} towards the null/spacelike transition sphere $\CH\cap \mathcal{S}$, a phenomenon that requires specific new estimates in the proof of Theorem~\ref{thm.I}.
	
	We conjecture that this degeneration is not an artifact of spherical symmetry but  instead a key feature of the general problem. This is analogous to the study of the Cauchy horizon, where the full dynamical treatment for rotating black holes in vacuum  \cite{KerrStab} built upon foundational insights from the spherically symmetric case \cite{MihalisPHD,Mihalis1,JonathanStab,Moi}. Motivated by this, we propose a generalization of Conjecture~\ref{charged.conj} to the Einstein-scalar-field system (i.e, \eqref{1.1}--\eqref{5.1} with $F\equiv 0$) outside of spherical symmetry. This constitutes a strengthened version of Conjecture~\ref{spacelike.conj}, and we hope to return to this and related problems in future work.

	\subsection{Applications to the Strong Cosmic Censorship Conjecture}  \label{SCC.section}
	The celebrated Strong Cosmic Censorship Conjecture of  Penrose \cite{penrose1974gravitational} is a central motivation to study the interior of dynamical black holes, see the related discussion in our previous work \cite{bif} or in the review \cite{review}.	Here, we focus on a simplified version of Strong Cosmic Censorship within spherically symmetric solutions of \eqref{1.1}--\eqref{5.1} (see \cite{ChristoCQG,nospacelike,Kommemi,review}), formulated in the context of gravitational collapse following Penrose \cite{PenroseSCC,penrose1974gravitational,Penrose1979}.
	\begin{conjecture}[$C^2$ Strong Cosmic Censorship] \label{SCC.conj}
		The Maximal Globally Hyperbolic Development for \eqref{1.1}--\eqref{5.1} of generic, one-ended asymptotically
		flat complete spherically symmetric initial data  is
		$C^2$-(future)-inextendible.
	\end{conjecture}
	Conjecture~\ref{SCC.conj} remains open at present. The charged spacetimes of Theorem~\ref{OS.thm.intro} are smoothly extendible across an open subset of $\CH$ and thus violate the statement of Conjecture~\ref{SCC.conj}, but as we explained, they are non-generic. It is instructive to notice, however, that  the spacetimes  constructed in Theorem~\ref{thm.II} and Theorem~\ref{thm.III} (more generally, those of Theorem~\ref{thm.cond.intro}) are $C^2$-inextendible, and thus respect the paradigm of Conjecture~\ref{SCC.conj}.
	
	\begin{thm}\label{inext.thm}
		The one-ended or two-ended asymptotically flat spacetimes obeying the assumptions of  Theorem~\ref{thm.cond.intro} are 	$C^2$-(future)-inextendible. In particular, the one-ended examples of Theorem~\ref{thm.II} and the two-ended examples of Theorem~\ref{thm.III} are 	$C^2$-(future)-inextendible.
	\end{thm} Therefore, the final step in proving Conjecture~\ref{SCC.conj} is to show that generic one-ended spherically symmetric solutions of \eqref{1.1}--\eqref{5.1} uphold the assumptions of Theorem~\ref{thm.cond.intro}, which is indeed conjectured (recall Section~\ref{collapse.section}).

	\subsection{Qualitative results on the terminal boundary and Cauchy horizon breakdown}\label{qual.section}
	\paragraph{Mathematical setting of gravitational collapse} Gravitational collapse  (see \cite{review}, Section 5 or \cite{bif}, Section 1.3-1.4) is modeled by the MGHD of asymptotically flat initial data $(\Sigma,g)$ with one-end, assuming that $\Sigma$ is diffeomorphic to $\mathbb{R}^3$ and thus has a center $\Gamma$ corresponding to the origin of  $\mathbb{R}^3$. We also assume that $\Sigma$ is free from anti-trapped surfaces or trapped surfaces. 
	Note that the spacetimes constructed in Theorem~\ref{thm.II} are of gravitational collapse type, while those of Theorem~\ref{thm.III}, which are two-ended asymptotically flat, are not.	\begin{figure}\begin{center}	\includegraphics[width=105 mm, height=60 mm]{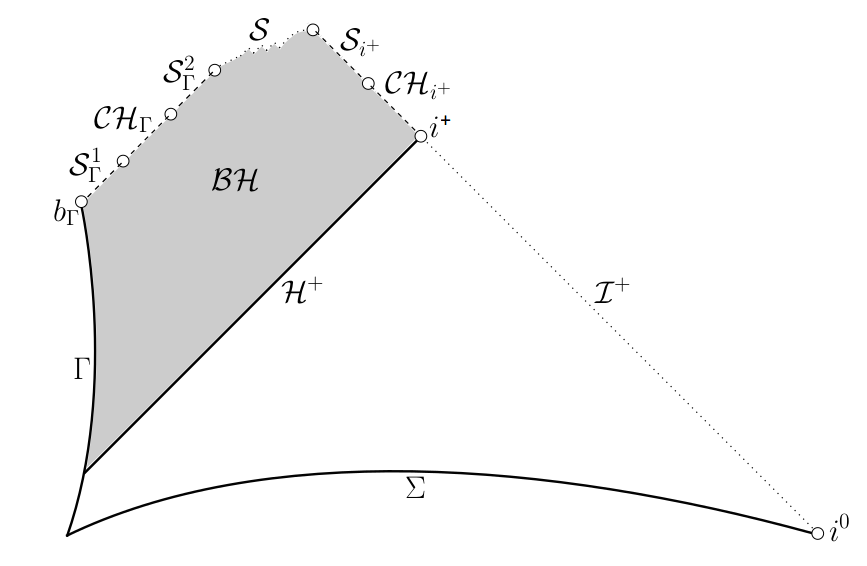}\caption{General Penrose diagram of a one-ended spherically symmetric black hole solution after Theorem~\ref{Kommemi.thm}.}\label{fig:one-ended}\end{center}\end{figure}	\paragraph{A priori characterization of the black hole terminal boundary}
	For spherically symmetric black hole solutions of \eqref{1.1}-\eqref{5.1}, it is possible to characterize the future boundary of the MGHD (also called terminal boundary) \cite{Mihalis1,Kommemi}: $\mathcal{B}'$ is a component of the terminal boundary if either of the following conditions is satisfied: \begin{enumerate}[A.]
		\item (\textbf{Crushing singularity}). The area-radius $r$ extends to $0$ on  $\mathcal{B}'$ (corresponds to $\mathcal{S}$,  $\mathcal{S}^{1}_{\Gamma}$,  $\mathcal{S}^{2}_{\Gamma}$ or  $\mathcal{S}_{i^+}$ in Figure~\ref{fig:one-ended}).
		\item (\textbf{Cauchy horizon from infinity}). $\mathcal{B}'$ is a null segment emanating from timelike infinity $i^+$ (corresponds to $\CH$ or  $\mathcal{S}_{i^+}$ in Figure~\ref{fig:one-ended}).
		\item\label{sing.C} (\textbf{Locally naked singularity}) [Only if $\Gamma\neq \emptyset$]. $\mathcal{B}'$ is a null segment emanating from the center $\Gamma$ (corresponds to $\mathcal{CH}_{\Gamma}$,  $\mathcal{S}^{1}_{\Gamma}$ or $\mathcal{S}^{2}_{\Gamma}$  in Figure~\ref{fig:one-ended}).
	\end{enumerate}
	This is formalized by Theorem~\ref{Kommemi.thm} (originally from \cite{Kommemi})  in the one-ended case, and Theorem \ref{Dafermos.thm} (originally from \cite{Mihalis1})  in the two-ended case. Note that, in the two-ended case, there is no center (i.e., $\Gamma=\emptyset$) and therefore Case~\ref{sing.C} is impossible.  In the one-ended case, however, Case~\ref{sing.C} is possible (Remark~\ref{loc.naked.remark}) although we recall Christodoulou proved  that Case~\ref{sing.C} is non-generic \cite{Christo3}  for \eqref{1.1}-\eqref{5.1} with $F\equiv 0$ in spherical symmetry.
	
	\paragraph{Breakdown of the Cauchy horizon: old and new} A gravitational collapse scenario in which a Cauchy horizon from infinity $\CH$  is the only terminal boundary component implies in particular that $\mathcal{S}=\emptyset$ in Figure~\ref{fig:one-ended}, corresponding to the absence of any spacelike singularity. Such examples where the Cauchy horizon closes-off the spacetime, as depicted in Figure~\ref{fig:disproof}, can be rigorously constructed if the Cauchy horizon $\CH$ is non-singular \cite{KehleUnger}. However, the author proved in \cite{Moi,Moi4} that $\CH$ is (weakly) singular under generic assumptions on the event horizon and subsequently in \cite{breakdown} that, due to this singularity, the Cauchy horizon cannot close-off the spacetime, i.e., $\mathcal{S}\cup \mathcal{S}^{1}_{\Gamma} \cup \mathcal{CH}_{\Gamma}\neq \emptyset$ in Figure~\ref{fig:one-ended}. This result--\emph{the breakdown of weak null singularities}--is the starting point of the investigation of the transition between null and spacelike singularities initiated in our   work \cite{bif}.

	\begin{figure}[H]	\begin{center}
			\includegraphics[width=62.5 mm, height=50 mm]{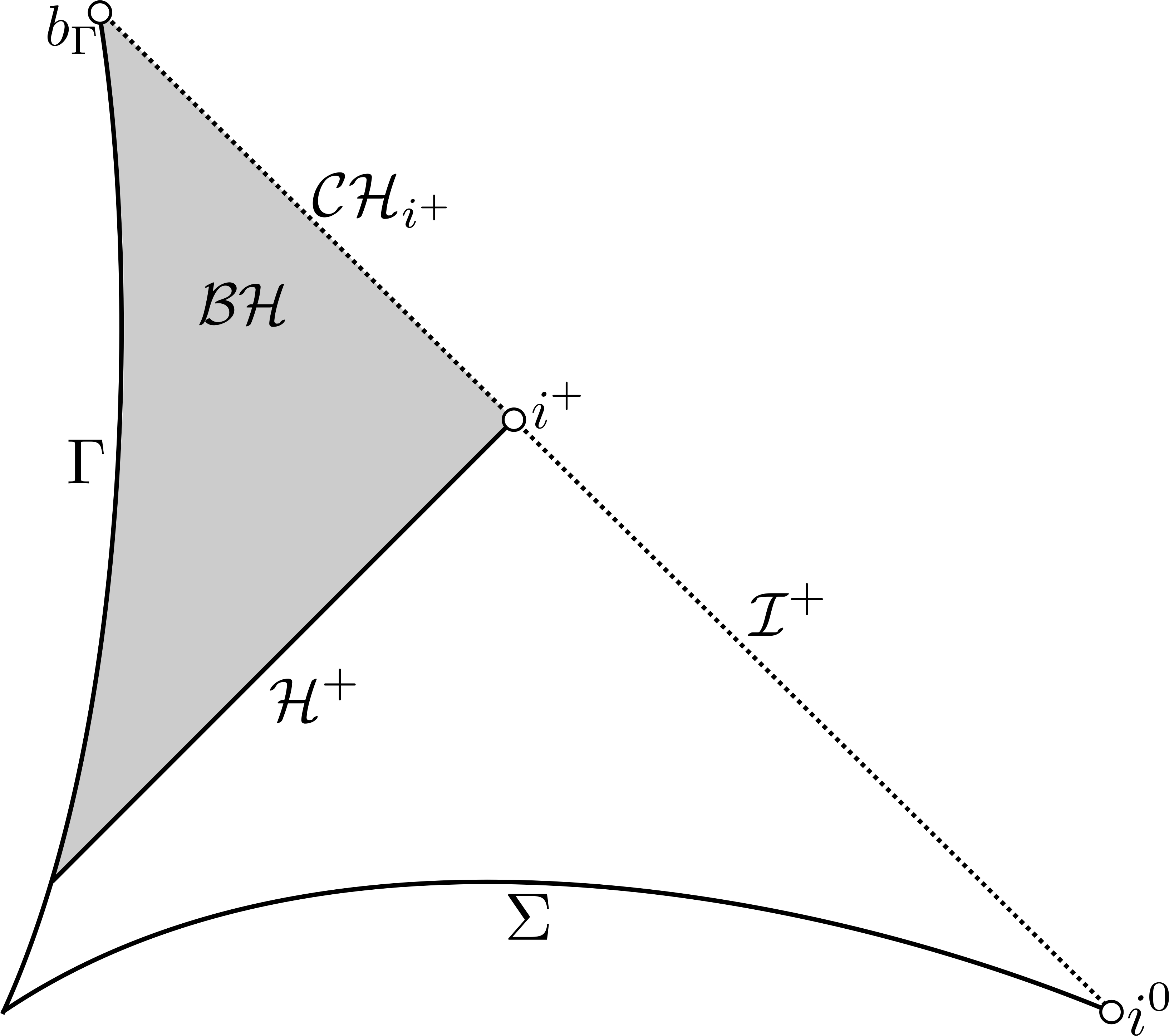}
		\end{center}
		\caption{The impossible Penrose diagram if $\CH$ is weakly singular, as a consequence of the result of \cite{breakdown}.}\label{fig:disproof}\end{figure}
	
	In \cite{breakdown}, the breakdown of the Cauchy horizon was proved using an argument by contradiction which did not provide any  quantitative estimates on the spacelike singularity $\mathcal{S}$.  It can be summarized as follows: \begin{itemize}
		\item The Cauchy horizon $\CH$ is surrounded by the trapped region, due to the blow-up of the Hawking mass. 
		\item Assuming that $\CH$ closes off the spacetime, i.e., that the Penrose diagram is given by Figure~\ref{fig:disproof}, there exists a sequence of apparent horizon spheres $(a_n)_{n\in \mathbb{N}}$  converging to $b_{\Gamma}$, the  endpoint of the center $\Gamma$. This is proven by interpolation between the center $\Gamma$,  within the regular region, and $\CH$, which is trapped.
		\item  A charge (denoted $Q$) to mass (denoted $\varpi$) ratio inequality of the form $|Q|\ll \varpi$ holds at the apparent horizon $\mathcal{A}$.  The proof of this inequality uses  the ingoing Raychaudhuri equation (see \cite{breakdown}, Section 4.1.2).
		\item The charge to mass ratio inequality shows that the ingoing future of the apparent horizon $\mathcal{A}$ is immediately inside the trapped region, which is a contradiction if $\mathcal{A}$ is chosen to be ``the left-most apparent horizon''.
	\end{itemize}

	In 
	our previous work \cite{bif},  we have  developed quantitative  estimates under the Cauchy horizon to handle the transition with a spacelike singularity, in order to prove Theorem~\ref{thm.I}. In the present manuscript, we will use  these estimates in the proof of Theorem~\ref{thm.cond.intro} and also offers a shorter, more constructive proof of the breakdown of the Cauchy horizon originally proven in \cite{breakdown}. Essentially, we will prove the following (see Proposition~\ref{apriori.prop1}): \begin{thm}\label{breakdown.thm.new}
		Let $\CH$ be a Cauchy horizon from infinity such that, in Eddington--Finkelstein coordinates, the Hawking mass blows up at the following exponential rate, where $K_-<0$ and $s>\frac{1}{2}$: there exists $u$ such that \begin{equation}\label{mass.inflation.intro}
			\varpi(u,v) \gtrsim_u e^{ 2|K_-| v} v^{-2s}.
		\end{equation} Then, there exists an outgoing  cone $C_0$ intersecting $\CH$, such that the entire causal rectangle generated by $C_0$ and $\CH$ is within the trapped region. In particular, this causal rectangle does not intersect the center $\Gamma$.
	\end{thm}	\begin{figure}[H] \begin{center}\includegraphics[width=50 mm, height=50 mm]{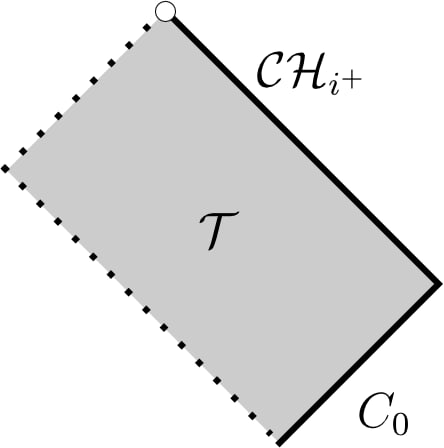}				
			\caption{The trapped region $\mathcal{T}$ under the Cauchy horizon $\CH$, extending to its future end-sphere.}
			\label{fig:CH}	\end{center}\end{figure}
	The major difference between Theorem~\ref{breakdown.thm.new} and the older strategy of \cite{breakdown} is that the size of the trapped region does not shrink towards the endpoint of the Cauchy horizon $\CH$, which is remarkable and immediately disallows a spacetime geometry where $\CH$ closes off the spacetime as depicted in Figure~\ref{fig:disproof}. The required assumption \eqref{mass.inflation.intro} is, however, more quantitative and demanding than the qualitative Hawking mass blow-up from \cite{breakdown}. Theorem~\ref{breakdown.thm.new} will be applied to the one-ended global setting as  part of Theorem~\ref{main.thm.global.i}.
	
	Unexpectedly, Theorem~\ref{breakdown.thm.new} will also play a crucial in constructing examples of \emph{two-ended spacetimes} with a spacelike singularity $\mathcal{S}\neq \emptyset$, an essential step in the proof of Theorem~\ref{thm.III}, see Section~\ref{section.global.uncond2}.

	\subsection{A novel gluing theorem for spacetimes relaxing  to Reissner--Nordstr\"{o}m}\label{gluing.intro}
	
	Gluing methods are routinely used in General Relativity to construct initial data on a spacelike hypersurface that satisfy  the constraints equations imposed by \eqref{1.1}, see e.g., \cite{BieriChruscielYau,Chrusciel.gluing} or the reviews 	\cite{Carlotto.review,Corvino.review} for further details and references. In the black hole setting, we mention the notable results 	 \cite{ChruscielDelay,Corvino,CorvinoSchoen} allowing to glue a general solution to the Kerr black hole outside of a compact region and recent progress on gluing two black hole initial data sets \cite{Hintzgluing1}. See also \cite{LC3,LC2,LC1} for a novel class of spacelike gluing results relying on the flexibility of the linearized theory.
	In other recent developments, Aretakis, Czimek and Rodnianski  designed a novel so-called \emph{characteric gluing approach} \cite{ACR3,ACR2,ACR1} allowing to construct solutions to the null constraint equations imposed by \eqref{1.1} on a null cone in the near-Minkowski perturbative regime, see also the more recent work of Czimek and Rodnianski on obstruction-free characteristic gluing \cite{CR}. Even more recently, the characteristic gluing approach was used by Kehle and Unger \cite{KehleUnger3,KehleUnger} in the context of black holes, who obtained event horizon gluing results.
	
	\noindent These	previous gluing results, however, are either limited to the perturbative near-Minkowski regime or, in the black hole setting, to the black hole exterior region near the asymptotically flat end, with the exception of \cite{KehleUnger3,KehleUnger} that address the event horizon. In the present manuscript, we require in contrast a gluing result deep inside the trapped region of the black hole, in the interior region. The main novelty of our approach is to rely on spherical symmetry and combine methods from spacelike and characteristic gluing to glue \emph{any uncharged regular sphere to a dynamical black hole event horizon converging to Reissner--Nordstr\"{o}m at any prescribed rate}.\\

	\noindent	The main application of our gluing theorems is to construct a one-ended asymptotically black hole such that: 
	\begin{itemize}
		\item The terminal boundary near $\Gamma$ is $\mathcal{S}$ spacelike. In particular, there is no locally naked singularity: $\mathcal{CH}_{\Gamma} = \emptyset$. In our case, we choose $\mathcal{S}$ modeled after the FLRW spacelike singularity, which is spatially-homogeneous.
		\item The event horizon  is transversally 
		\blue{$C^k$} regular \blue{for arbitrarily large $k \in \mathbb{N}$}, and converges to a  Reissner--Nordstr\"{o}m event horizon at a \emph{freely prescribed rate}, governed by a freely prescribed scalar field.
		\item The assumptions of Theorem~\ref{thm.I}, most importantly \eqref{decay.intro}, are satisfied under the Cauchy horizon $\CH$.
	\end{itemize}\noindent To fix the terminology, let us recall the FLRW metric, a spatially-homogeneous solution of \eqref{1.1}--\eqref{5.1} with $F\equiv 0$ and  initial data diffeomorphic to $\RR^3$ (note, however, that it is obviously not asymptotically flat): \begin{equation}\label{FLRW.intro}
		g=  -dt^2 + a^2(t) \left( d\rho^2 + \rho^2 d\sigma_{\mathbb{S}^2}\right).
	\end{equation} Spacetimes of the form \eqref{FLRW.intro} are analyzed in Proposition~\ref{FRLW.prop}, where it is shown that a spacelike singularity is formed in finite time $t=T_S$.
	Then, we recall the Reissner--Nordstr\"{o}m metric, a stationary solution of \eqref{1.1}--\eqref{5.1} in spherical symmetry   \begin{equation}\label{RN}
		g_{RN} = -(1-\frac{2M}{r}+ \frac{e^2}{r^2}) dt^2+ (1-\frac{2M}{r}+  \frac{e^2}{r^2})^{-1} dr^2+ r^2 ( d\theta^2+ \sin^2(\theta) d\varphi^2),
	\end{equation} corresponding to $\phi \equiv 0 $, and $F= \frac{e}{r^2} dt \wedge dr$. In the sub-extremal case $0<|e|<M$, \eqref{RN} is a two-ended black hole solution whose MGHD terminates a smooth Cauchy horizon $\CH=\{r= M-\sqrt{M^2-e^2}\}$.
	
	We will now describe our new gluing strategy and sketch its main steps below, comparing it to the existing literature. The reader can also consult  a more detailed outline  of the construction in Section~\ref{outline.section}.

	\subsubsection{Spatial gluing of a regular uncharged sphere to an apparent horizon}\label{intro.glue1}
	
	The first step is to start with a regular, uncharged sphere $\mathbb{S}_R$, which we will eventually choose to be FLRW from \eqref{FLRW.intro}. To get into the trapped region, we first want to glue it to an uncharged apparent horizon using spacelike gluing techniques. Note that the spacelike constraints induced by \eqref{1.1}--\eqref{5.1} in spherical symmetry form a system of ODEs, detailed in Section~\ref{gluing.section}. We show the following results (see already Proposition~\ref{spacelike.gluing.thm}): \begin{itemize}
		\item For any $k \in \mathbb{N}$, one can glue $\mathbb{S}_R$ spatially to any apparent horizon sphere $\mathbb{S}_A$ as a $C^{k}$ solution providing \begin{equation}\label{gluing.condition.intro}
			1<	[	R_A |\rd_v \phi|_{|\mathbb{S}_A}]^2 < 1+\frac{1}{k-1},
		\end{equation} where $R_A$ is the area-radius of $\mathbb{S}_A$, and $v$ is chosen so that $\mathbb{S}_A$  is lapse-normalized (see Section~\ref{gluing.section}). Moreover, this can be done through a spacelike hypersurface $\blue{\Sigma_G}$, where $\blue{\Sigma_G}-\mathbb{S}_A$ is in the regular region.
		
	\end{itemize}
	
	\noindent Therefore, the above  gluing procedure is flexible, since the apparent horizon sphere $\mathbb{S}_A$ is essentially free, except for the condition \eqref{gluing.condition.intro}. In particular, the scalar field ingoing derivatives can be prescribed arbitrarily at $\mathbb{S}_A$.

	\subsubsection{Characteristic gluing of an apparent horizon to a Schwarzschild trapped surface}\label{intro.gluing2}
	
	As a next step, we glue $\mathbb{S}_A$ to a Schwarzschild trapped sphere $\mathbb{S}_S^{\T}$. However, this is not possible for any arbitrary choice of apparent horizon  sphere $\mathbb{S}_A$.  We prove the following result (see already Proposition~\ref{uncharged.null.gluing.prop}):
	\begin{itemize}
		\item We can \emph{choose} $\mathbb{S}_A$ so that it can be characteristically glued to a Schwarzschild trapped sphere $\mathbb{S}_S^{\T}$.
	\end{itemize}
	We note that, while it is  difficult to glue two arbitrary spheres in a characteristic manner, it is much easier if one has flexibility on one of the two spheres, which is what we exploit here. For comparison,  Kehle--Unger \cite{KehleUnger} achieve gluing from a Minkowski regular sphere to the Schwarzschild event horizon, which already requires the use of the (difficult, albeit classical) Borsuk--Ulam Theorem. Our strategy, on the other hand, allows to \blue{deform the spacelike-characteristic gluing  into purely spacelike gluing and} prove the following (see Theorem~\ref{uncharged.gluing.thm}):\begin{itemize}
		\item  one can glue $\mathbb{S}_R$ spatially to a Schwarzschild \emph{trapped surface} $\mathbb{S}_S^{\T}$ as a $C^{k}$ solution of the constraints.
	\end{itemize}
	\noindent This  immediately proves Theorem~\ref{OS.thm.intro} in the uncharged case ($q=0$) and serves as foundation for Theorem~\ref{construction.thm}.

	\subsubsection{Characteristic gluing of Schwarzschild to a Reissner--Nordstr\"{o}m trapped surface} Then, we need to ``charge'' the spacetime, which we do via characteristic gluing (see already Theorem~\ref{charged.gluing.thm}):\begin{itemize}
		\item  one can glue $\blue{\mathbb{S}_S^{\T}}$ characteristically to a Reissner--Nordstr\"{o}m trapped sphere $\mathbb{S}_{RN}^{\T}$ as a $C^{k}$ solution.
	\end{itemize}\noindent In the above, one needs to assume that $\mathbb{S}_{RN}^{\T}$ is a \emph{sub-extremal}  Reissner--Nordstr\"{o}m \emph{trapped} sphere. This is to be compared to the charged gluing result of Kehle--Unger \cite{KehleUnger}, achieving to glue a Schwarzschild regular sphere to a (possibly extremal) Reissner--Nordstr\"{o}m event horizon of mass $M_f$, charge ratio $q\in (0,1]$ under the condition: \begin{equation}
		\frac{	|q_0| M_f}{q} \gg 1.
	\end{equation} In Theorem~\ref{charged.gluing.thm}, we also require an analogous condition, but we remain away from extremality ($q=1$) as such\begin{equation}
		\frac{	|q_0| M_f (1-q)}{q} \gg 1,
	\end{equation} since we carry out \emph{trapped spheres gluing}. This step of the proof is  inspired from the methods and formalism of \cite{KehleUnger}, in particular the use of the Borsuk--Ulam Theorem on the sphere, that we already mentioned previously.
	Combining this with the earlier results completes the proofs of Theorem~\ref{OS.thm.intro} (case $q\neq 0$).
	
	\subsubsection{Gluing a Reissner--Nordstr\"{o}m trapped surface to any dynamical event horizon} To prove Theorem~\ref{construction.thm}, we must go beyond spacetimes that are stationary in a neighborhood of $i^+$ like those of Theorem~\ref{OS.thm.intro}. To do this, we impose tangential event horizon data, which we trivially glue to the Reissner--Nordstr\"{o}m trapped surface, providing it is only weakly trapped. The following summarizes Proposition~\ref{EH.AF.1st.prop}: \begin{itemize}
		\item Let $\Phi_H(v)$ an arbitrary decaying profile  as $v\rightarrow+\infty$. One can glue $\mathbb{S}_{RN}^{\T}$ to an event horizon $\mathcal{H}^+$ which is asymptotically Reissner--Nordstr\"{o}m and such that the scalar field coincides with $\Phi_H$ for  large $v$, i.e.,  $$ \phi_{|\mathcal{H}^+}(v) = \Phi_H(v).$$
	\end{itemize}
	We have now constructed the black hole spacetime in Theorem~\ref{construction.thm} to the future of its event horizon $\mathcal{H}^+$.
	\subsubsection{Connecting the event horizon to an asymptotically flat end} To complete the proof of Theorem~\ref{construction.thm}, one must construct an asymptotically flat end. For this, we rely on spherical symmetry to solve \eqref{1.1}--\eqref{5.1} from left to right. The following result corresponds to  Theorem~\ref{EH.AF.thm}: \begin{itemize}
		\item For any event horizon that relaxes to  Reissner--Nordstr\"{o}m sufficiently fast, and with sufficiently small charge, one can construct a corresponding complete null infinity $\mathcal{I}^+$ and an asymptotically flat end $i^0$.
	\end{itemize}
	We	note, however, that the left to right evolution is only relevant in a neighborhood of $i^+$, while the region near spacelike infinity $i^0$ is treated by backwards-in-time evolution of \eqref{1.1}--\eqref{5.1} with prescribed scattering data at $\mathcal{I}^+$, see Section~\ref{right.section} . The proof relies on an application of the $r^p$ method inspired by the author's work \cite{Moi2}.
	
	\subsection{Scattering theory for charged scalar fields in the black hole interior}  \label{CH.section}

	Construction/gluing results such as Theorem~\ref{construction.thm} described in Section~\ref{gluing.intro} mark an important step towards the proof of Theorem~\ref{thm.II}, in constructing a spherically symmetric black hole spacetime: \begin{itemize}
		\item which is one-ended asymptotically flat (with no anti-trapped surface, and no trapped surface initially).
		\item which relaxes to a sub-extremal Reissner--Nordstr\"{o}m black hole.
		\item  which  coincides with given $\mathcal{M}_L$ near $\Gamma$. In particular, if $\mathcal{M}_L$ is FLRW, it has  a spacelike singularity $\mathcal{S}$.
	\end{itemize}

	\noindent However, it is not yet clear that the constructed spacetime even admits a non-empty null Cauchy horizon from infinity, which is crucial to apply Theorem~\ref{thm.I}.  The author previously studied  black interior dynamics for \eqref{1.1}--\eqref{5.1} in \cite{Moi}  in the vicinity of $i^+$, interpreted as timelike infinity. To study this problem, one  poses characteristic data of a future affine-complete outgoing cone $\mathcal{H}^+$ interpreted as the black hole event horizon and a regular ingoing cone $\Cin$ penetrating $\mathcal{H}^+$ as depicted in Figure~\ref{fig:interior} and assume that $\mathcal{H}^+$ asymptotically converges to a sub-extremal Reissner--Nordstr\"{o}m  black hole and that the following decay assumption holds for the scalar field	\begin{equation}\label{decay}
		|\phi_{|\mathcal{H}^+}|(v),\  |D_v\phi_{|\mathcal{H}^+}|(v) \ls [1+|v|]^{-s},
	\end{equation}\blue{for some $s>0$,} in an Eddington--Finkelstein advanced-time coordinate $v$ on $\mathcal{H}^+$ as $v\rightarrow +\infty$. Then, the author proved in \cite{Moi} the existence of a non-empty Cauchy horizon from infinity $\CH$ as depicted in Figure~\ref{fig:interior}.

	\begin{thm}[\cite{Moi}, Theorem 3.2]\label{CH.thm.SS}

		Consider spherically symmetric characteristic initial data  for \eqref{1.1}--\eqref{5.1} 		on the event horizon $\mathcal{H}^+$ converging to a sub-extremal  Reissner--Nordstr\"{o}m black hole and on a $C^1$-regular ingoing cone $\Cin$. Assume \eqref{decay}		holds as an upper bound on $\mathcal{H}^{+}=[v_0,+\infty)$ for some  decay rate $s>\frac{1}{2}$.\\					Then, the spacetime is bound to the future by an ingoing null boundary $\CH \neq \emptyset$  (the Cauchy horizon) foliated by spheres of  positive radius and emanating from $i^+$, and the Penrose diagram is given by the dark gray region in Figure~\ref{fig:interior}. Moreover, if $s>1$, then $\phi$ is uniformly bounded and $g$ is continuously-extendible.	
	\end{thm}
	
	\begin{figure}	\begin{center}\includegraphics[width=0.4\linewidth]{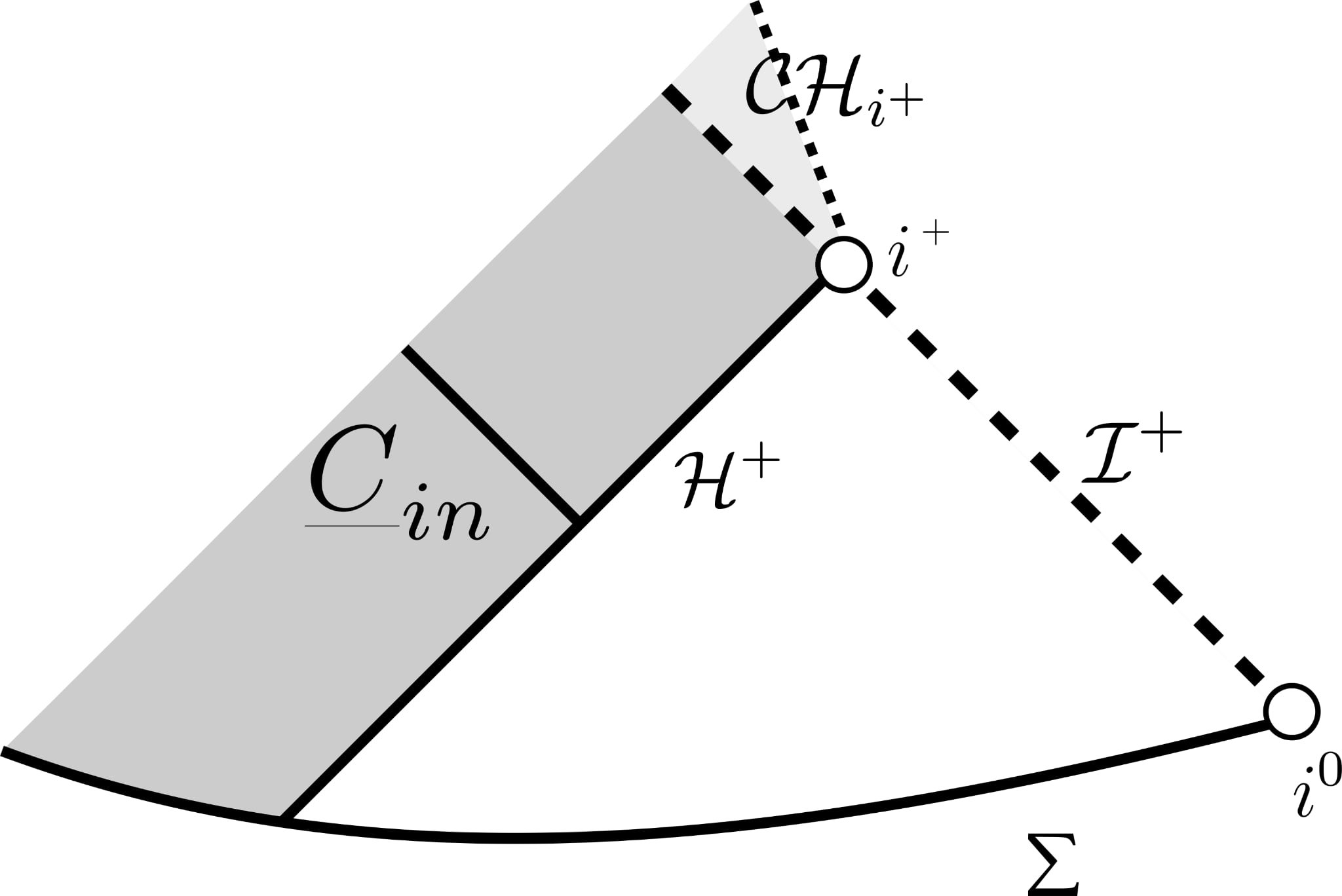}			\caption{\small Local structure of the black hole terminal boundary near $i^+$ for a spherical solution of \eqref{1.1}--\eqref{5.1}.}		\label{fig:interior}\end{center}	\end{figure}
	
	\noindent Theorem~\ref{CH.thm.SS} also shows  scalar field upper bounds consistent with \eqref{decay.intro}, but those are insufficient to obtain the more detailed estimates required by \eqref{decay.intro} as part of the assumptions of Theorem~\ref{thm.I}. Therefore, to prove that \eqref{decay.intro} holds, we appeal to the nonlinear scattering theory developed by the author and Kehle in \cite{MoiChristoph}. While the  event horizon to Cauchy horizon scattering map is quite complicated for general frequencies,  we identify a resonating frequency $\omer \in \RR-\{0\}$ (depending on the Reissner--Nordstr\"{o}m parameters $(M,e)$) at which it takes a simpler form; informally, \blue{a key feature of the scattering map is that} $$ \text{If } \phi_{|\mathcal{H}^+} \text{ oscillates  at frequency } \omega=-\omer, \text{ then }  	D_v \phi(u,v) \approx \phi_{|\mathcal{H}^+}(v) e^{i q_0 \omer v} \text{ near } \CH.$$ 
	In other words, if $\phi$  oscillates  at frequency  $\omega=-\omer$  at $\mathcal{H}^+$, then  	$D_v \phi$ does \underline{not} oscillate under $\CH$.

	Nonlinear estimates were also employed in \cite{MoiChristoph} to show that some aspects of the linear theory subsist in the nonlinear setting but, \blue{as it turns out}, the estimates in \cite{MoiChristoph} are insufficient to verify whether  \eqref{decay.intro} holds or fails for a given event horizon profile. Thus, building up on the methods of \cite{MoiChristoph}, we establish a new nonlinear scattering result (Theorem~\ref{nonlinearscat.thm}), whose rough version we give below.

	\begin{thm}\blue{[Nonlinear scattering theorem]}.\label{scat.thm.intro} Let  $\Phi_H(v)$ satisfying \eqref{decay} with $s>1$, and such that $$ |\rd_v^{2} \Phi_H|(v) \ll  |\rd_v \Phi_H|(v)\ll| \Phi_H|(v) \text{ as }  v\rightarrow+\infty.$$ Assume that \begin{equation}\label{PhiHintro}
			\phi_{|\mathcal{H}^+}(v) = \Phi_H(v) e^{- i q_0 \omer v},
		\end{equation} and that the event horizon $\mathcal{H}^+$ is regular as in the assumptions \blue{of} Theorem~\ref{CH.thm.SS}. Then, for every fixed $u$, and as $v\rightarrow+\infty$:
		\begin{equation}\begin{split}&
				|D_v \phi|(u,v) \approx_u |\Phi_H|(v),\\ &  |D_v^2 \phi|(u,v) \approx_u |\rd_v \Phi_H|(v).
			\end{split}
		\end{equation}
	\end{thm}
	
	\noindent Therefore, to satisfy \eqref{decay.intro}, it is clear that \emph{one must choose} 	$\phi_{|\mathcal{H}^+}(v)$ to oscillate at the frequency $-\omer$ (or equivalently choose $\Phi_H$ in \eqref{PhiHintro} that does not oscillate) so that $D_v \phi$ does not oscillate towards the Cauchy horizon $\CH$, which is an essential requirement of \eqref{decay.intro}.
	
	We conclude by emphasizing that it is conjectured that $D_v \phi$ is oscillatory in the black hole interior for generic solutions of \eqref{1.1}--\eqref{5.1}(see, e.g., \cite{review}) which is why the solutions of Theorem~\ref{thm.II} are non-generic\footnote{However, as explained in Section~\ref{collapse.section}, the solutions of Theorem~\ref{thm.II} possess the main  features  of conjecturally generic solutions.}. Therefore, to require $D_v \phi$ to be non-oscillatory in the black hole interior, as demanded by \eqref{decay.intro}, it is crucial for our gluing strategy discussed in Section~\ref{gluing.intro} to allow for the \emph{prescription of (possibly non-generic) tangential event horizon data},  to ensure that the assumptions of Theorem~\ref{scat.thm.intro} are satisfied. This strategy is at the heart of our approach.

	\section*{Outline of the paper}
	
	In Section~\ref{geometricframework}, we introduce the necessary geometric preliminaries, together with relevant gauge choices and the Einstein--Maxwell--Klein--Gordon equations \eqref{1.1}--\eqref{5.1} in $(u,v)$-coordinates. In Section~\ref{thm.section}, we provide precise statements of our main results. In Section~\ref{oldresults.section}, we recall some essential results of \cite{bif} which we will be using in subsequent sections.	In Section~\ref{section.global.cond}, we show our conditional statements, i.e., we prove Theorem~\ref{thm.cond.intro}. In Section~\ref{section.global.uncond1}, by far the most involved of the manuscript, is dedicated to the construction of one-ended black holes to which Theorem~\ref{thm.I} applies, i.e., the proof Theorem~\ref{thm.II}.  Section~\ref{section.global.uncond1} will also contain the proofs of Theorem~\ref{construction.thm}, Theorem~\ref{OS.thm.intro} and Theorem~\ref{scat.thm.intro}.		Finally, in Section~\ref{section.global.uncond2}, we carry out the construction of two-ended black holes to which Theorem~\ref{thm.I} applies, i.e., we prove Theorem~\ref{thm.III}.
	
	\section*{Acknowledgement} The author warmly thanks Jan Sbierski for helpful conversations regarding FLRW spacetimes and their properties, which inspired the discussion on scalar field analogues of the Oppenheimer--Snyder spacetime. We also credit Haydee Pacheco for  making the figures in this manuscript.   Finally, we  gratefully
	acknowledge the support from the NSF Grant DMS-2247376.
	
	\section{Geometric preliminaries} \label{geometricframework}

	The purpose of this section is to provide the precise setup, together with the definition of various geometric quantities, the coordinates and the equations that we will use throughout the paper.

	\subsection{Spherically symmetric solutions} \label{preliminary}
	We consider $(M,g,\phi,F)$, a regular solution of the system \eqref{1.1}, \eqref{2.1}, \eqref{4.1}, \eqref{5.1}, where $(M,g)$ is a Lorentzian manifold of dimension $3+1$, $\phi$ is a complex-valued function on $M$ and $F$ is a real-valued 2-form on $M$. 	$(M,g,\phi,F)$ is related to a quadruplet of scalar functions $ \{\Omega^2(u,v), r(u,v), \phi(u,v), Q(u,v)\}$, with $(u,v) \in \mathcal{Q}^+ \subset \RR^{1+1}$ by \begin{equation} \label{gdef}
		g= g_{\mathcal{Q}^+}+ r^{2} \cdot  (d\theta^2 + \sin(\theta)^2 d\varphi^2)= -\Omega^2(u,v) du dv+ r^{2}(u,v) \cdot  (d\theta^2 + \sin(\theta)^2 d\varphi^2),
	\end{equation} 
	\begin{equation}\label{Q.def}
		F(u,v)= \frac{Q(u,v)}{2r^{2}(u,v)}  \Omega^{2}(u,v) du \wedge dv.
	\end{equation}  
	One can now formulate the Einstein equations \eqref{1.1}, \eqref{2.1},  \eqref{4.1}, \eqref{5.1} as a system of non-linear PDEs on $\Omega^2$, $r$, $\phi$ and $Q$ expressed in the double null coordinate system $(u,v) \in \mathcal{Q}^+  $: \begin{equation}\label{Omega}
		\partial_{u}\partial_{v} \log(\Omega^2)=-2\Re(D_{u} \phi \overline{D_{v}\phi})+\frac{ \Omega^{2}}{2r^{2}}+\frac{2\partial_{u}r\partial_{v}r}{r^{2}}- \frac{\Omega^{2}}{r^{4}} Q^2,
	\end{equation} \begin{equation}\label{Radius}\partial_{u}\partial_{v}r =\frac{- \Omega^{2}}{4r}-\frac{\partial_{u}r\partial_{v}r}{r}
		+\frac{ \Omega^{2}}{4r^{3}} Q^2, 
	\end{equation} 	\begin{equation}\label{Field}
		D_{u} D_{v} \phi =-\frac{ \partial_{v}r \cdot D_{u}\phi}{r}-\frac{\partial_{u}r \cdot  D_{v}\phi}{r} +\frac{ iq_{0} Q \Omega^{2}}{4r^{2}} \phi,
	\end{equation} 	\begin{equation} \label{chargeUEinstein}
		\partial_u Q = -q_0 r^2 \Im( \phi \overline{ D_u \phi}),
	\end{equation}	\begin{equation} \label{ChargeVEinstein}
		\partial_v Q = q_0 r^2 \Im( \phi \overline{D_v \phi}),
	\end{equation}
	\begin{equation} \label{RaychU}\partial_{u}(\frac {\partial_{u}r}{\Omega^{2}})=\frac {-r}{\Omega^{2}}|  D_{u} \color{black}\phi|^{2}, \end{equation} 
	\begin{equation} \label{RaychV}\partial_{v}(\frac {\partial_{v}r}{\Omega^{2}})=\frac {-r}{\Omega^{2}}|D_{v}\color{black}\phi|^{2},\end{equation} where the gauge derivative is defined by $D_{\mu}:= \partial_\mu+iq_0 A_{\mu}$, and the electromagnetic potential $A_{\mu}=A_u du + A_v dv$ satisfies 	\begin{equation}\label{Maxwell} \partial_u A_v - \partial_v A_u = \frac{Q \Omega^2}{2r^2}.\end{equation}	 Note that, under our electromagnetic gauge choice $A_v \equiv 0$ (see \eqref{A.gauge}), \eqref{Maxwell} and  \eqref{Field} can also be written as
	\begin{equation}\label{dv.Au}
		\partial_v A_u = -\frac{Q \Omega^2}{2r^2},
	\end{equation}	\begin{equation}\label{Field4}
		\partial_{u}\partial_{v} \phi =-\frac{\partial_{u}\phi\partial_{v}r}{r}-\frac{\partial_{u}r \partial_{v}\phi}{r} +\frac{ q_{0}i \Omega^{2}}{4r^{2}}Q \phi
		- i q_{0} A_{u}\frac{\phi \partial_{v}r}{r}-i q_0 A_{u}\partial_{v}\phi.\end{equation}
	In terms of the radiation field $\psi=r\phi$, we have \begin{equation}\label{Field.psi}
		D_u\partial_{v} \psi = \left( \frac{- \Omega^{2}}{4r^2}-\frac{\partial_{u}r\partial_{v}r}{r^2}
		+\frac{ \Omega^{2}}{4r^{4}} Q^2 +\frac{  iq_{0}Q\Omega^{2}}{4r^2}\right) \psi.  
	\end{equation}

	Subsequently, we define the Lorentzian gradient of $r$, and introduce the mass ratio $\mu$ by the formula $$ 1-\mu:=g_{\mathcal{Q}^+}(\nabla r,\nabla r),$$ where we recall that $g_{\mathcal{Q}^+}$ was the spherically symmetric part of $g$ defined in \eqref{gdef}. We can also define the Hawking mass:	$$ \blue{\mathfrak{m}} := \frac{\mu \cdot r}{2} =\frac{r}{2} \cdot(1- g_{\mathcal{Q}^+} (\nabla r, \nabla r )).$$	
	
	Notice that the $(u,v)$ coordinate system, we have $g_{\mathcal{Q}^+} (\nabla r, \nabla r )= \frac{-4 \partial_u r \cdot \partial_v r}{\Omega^2}$. Now we introduce the modified mass $\varpi$ which involves the charge $Q$:
	
	\begin{equation} \label{electromass}
		\varpi := \blue{\mathfrak{m}}  + \frac{Q^2}{2r}= \frac{\mu r}{2} + \frac{Q^2}{2r} ,
	\end{equation}obeying \begin{equation} \label{massUEinstein}
		\partial_u \varpi = \frac{-2r^2 \Omega^2}{\partial_v r}| D_u \phi |^2 
		- q_0 Q r \Im(\phi \overline{D_u \phi}),
	\end{equation}
	\begin{equation} \label{massVEinstein} 
		\partial_v \varpi = \frac{r^2}{2\kappa}| D_v \phi |^2
		+ q_0 Q r \Im(\phi \overline{D_v\phi}).
	\end{equation}
	An elementary computation relates the previously quantities : 	\begin{equation} \label{murelation}
		1-\frac{2\blue{\mathfrak{m}} }{r} = 1-\frac{2\varpi}{r}+\frac{Q^2}{r^2}=\frac{-4 \partial_u r \cdot \partial_v r}{\Omega^2}.
	\end{equation} We also define \begin{equation}
		2K = \frac{2}{r^2} \left( \varpi - \frac{Q^2}{r}\right).
	\end{equation}

	Now we can reformulate our former equations to put them in a form that is more convenient to use. For instance, the Klein-Gordon wave equation \eqref{Field} can be expressed in different ways, using the commutation relation $[D_u,D_v]=\frac{ iq_{0} Q \Omega^{2}}{2r^{2}}$ 
	and  under our electromagnetic gauge choice $A_v \equiv 0$ (see \eqref{A.gauge}): 
	\begin{equation}\label{Field2}
		\rd_u \theta =  -\frac{\rd_v r}{r} \cdot\xi +\frac{ \Omega^{2} \cdot \phi}{4r} \cdot  i q_{0} Q- i q_0 A_u  r\rd_v \phi -i q_0 \partial_{v}r \cdot A_u \phi,
	\end{equation}  \begin{equation}\label{Field3}
		\rd_v \xi =  -\frac{\partial_{u}r }{r}\cdot \theta +\frac{ \Omega^{2} \cdot \phi}{4r} \cdot  i q_{0} Q- i q_0 A_u  r\rd_v \phi -i q_0 \partial_{v}r \cdot A_u \phi,
	\end{equation} where we introduced the notations $\theta=r\rd_v \phi$ and $\xi=r\rd_u \phi$.

	Introducing the following notations \begin{equation}\begin{split}
			& \lambda=\rd_v r,\  \nu=\rd_u r,\\ & \iota = \frac{\Omega^2}{4\lambda },\ \kappa = -\frac{\Omega^2}{4\nu },
		\end{split}
	\end{equation}we can also re-write  \eqref{Omega} and \eqref{Radius},	\begin{equation}\label{Omega3}
		\partial_{u}\partial_{v} \log(r\Omega^2)=  \frac{ \Omega^{2}}{4r^{2}} \cdot \left(1-   \frac{3Q^2}{r^{2}}  - 8 r^2\Re( \frac{D_{u}\phi}{\Omega^2} \cdot  D_{v}\bar{\phi})  \right),
	\end{equation}	\begin{equation} \label{Radius3}
		-\partial_u  \partial_v (\frac{r^2}{2})	=\partial_u (-r \lambda) =	-\partial_v (r\nu ) = \frac{\Omega^2}{4}\cdot (1- \frac{Q^2}{r^2}),
	\end{equation} 
	\begin{equation} \label{Omega2}
		\partial_{u}\partial_{v} \log(\Omega^2)=\kappa \partial_u (2K)-2\Re(D_{u} \phi D_v\bar{\phi})- \frac{2\kappa}{r^2}( \partial_u \varpi- \frac{\partial_u Q^2}{r})= \\
		\iota \partial_v (2K) -2\Re(D_{u} \phi D_v{v}\bar{\phi})-  \frac{2\iota}{r^2}( \partial_v \varpi- \frac{\partial_v Q^2}{r}).
	\end{equation}

	Moreover, the following $\rd_v$-commuted equation will also be useful: 
	\begin{equation}\label{Field.dvpsi}\begin{split}
			&		D_u\partial_{v}^2 \psi = \left( \frac{ \rd_v\log(\Omega^{2}) \Omega^2 }{4r^2}[-1+\frac{Q^2}{r^2} + i q_0 Q]+\frac{3\nu \lambda^2}{r^3}-\frac{\nu \rd_v \lambda}{r^2}+ \frac{\Omega^2\lambda}{4r^3}[1-iq_0 -\frac{5Q^2}{r^2}]+ q_0r^2 \Im(\phi \overline{D_v\phi})\frac{\Omega^2}{4r^2}[2Qr^{-2}+ i q_0]\right) \psi\\ & + \left[- \frac{ \Omega^{2}}{4r^2}-\frac{\partial_{u}r\partial_{v}r}{r^2}
			+\frac{ \Omega^{2}}{4r^{4}} Q^2 -\frac{  iq_{0}Q\Omega^{2}}{4r^2}\right] \rd_v \psi.  \end{split}
	\end{equation}

	\subsection{Double null coordinate and gauge choices}\label{gauge.section}

	It is well-known that the above system of equations are invariant under the following gauge transformations: \begin{equation}
		\tilde{u} = f(u),
	\end{equation}
	\begin{equation}
		\tilde{v} = g(v),
	\end{equation} where $f$ and $g$ are increasing $C^1$ functions. Note that the gauge transform $(u,v) \rightarrow (\tilde{u},\tilde{v})$ transforms the lapse $\Omega^2$ in the following way \begin{equation}
		\tilde{\Omega}^2 = \frac{\Omega^2}{f'(u) g'(v)}.
	\end{equation}

	The system of equations \eqref{1.1}, \eqref{2.1},  \eqref{4.1}, \eqref{5.1} is also invariant under the  electromagnetic gauge transformation: $$ \phi \rightarrow  e^{-i q_0 F } \phi ,$$
	$$ A \rightarrow  A+ d F. $$
	where $F$ is a smooth real-valued function.    $|\phi|$ and $|D_{\mu}\phi|$, on the other hand, are gauge invariant.

	We will use several gauge choices throughout the paper, always for convenience of the proof; none of them will truly matter, as all the theorem statements are formulated independently of the gauge. 
	
	For the convenience of the reader, we regroup some of the gauge choices we will use in the sequel, and refer to this section when each choice if elected. \begin{enumerate}[G.i]
		\item \label{gauge.spacelike} Let $\Sigma$ be a spacelike hypersurface on which we prescribe $\rd_v r>0$ and $\rd_u r<0$. This fixes the choice of \begin{equation}
			u_{|\Sigma_0} \text{ and }  	v_{|\Sigma_0},
		\end{equation} and we then fix $(u,v)$ in spacetime by requiring $u$ and $v$ to solve the standard eikonal equation $$ g^{\mu \nu} \rd_{\mu }\Psi \rd_{\nu} \Psi =0$$ with the above initial conditions. This determines $(u,v)$ in a neighborhood of the hypersurface $\Sigma$. This also provides a parametrization of the hypersurface $\Sigma$ by \begin{equation}\label{rho.def}
			\rho=v-u.
		\end{equation}
		\item \label{gauge.unit.lapse} Let $C_0$ be an outgoing, spherically symmetric  cone. We can fix an unitary lapse on $C_0$, i.e., \begin{equation}
			\Omega^2_{|C_0} =1,
		\end{equation} which fixes, up to constant, the coordinate $v$, not only on $C_0$, but also in the strip $\mathcal{D}_0$ obtained shooting ingoing (radial) null geodesics with starting spheres on $C_0$.
		
		\item \label{gauge.EH.v} Let $\mathcal{H}^+$ be the event horizon of a black hole spacetime, which is affine complete towards the future. One can consider a future subset  $\mathcal{H}^+_{v_0}\subset \mathcal{H}^+$, which is affine complete towards the future and with a past sphere $\mathcal{S}_{v_0}$.  We then impose the gauge \begin{equation}
			\frac{-4\rd_u r}{\Omega^2}_{|\mathcal{H}^+_{v_0}}=1,
		\end{equation}\begin{equation}
			v_{|\mathcal{S}_{v_0}} = v_0.
		\end{equation} As in the previous case, this determines $v$ in the causal future of $\mathcal{S}_{v_0}$, and more broadly on a strip $\mathcal{D}_{v_0}$ emanating from $\mathcal{H}^+_{v_0}$ in the ingoing direction.
		
		\item \label{gauge.EH.U}  In the same setting as before, let us assume additionally that $\mathcal{H}^+$ is asymptotically Reissner--Nordstr\"{o}m with mass $M$ and charge $e$, $0\leq |e| \leq M$. Let $C_{v_0}$ an ingoing cone emanating from $\mathcal{S}_{v_0}$ towards the future or the past. We can fix the $U$-gauge \begin{equation}
			\rd_U r(U,v_0) = - e^{2K_+(M,e) v_0}.
		\end{equation}

		\item \label{gauge.U.future} Let $\mathcal{M}$ globally hyperbolic, $C_{u_0}$ an outgoing cone and denote $\mathcal{B}_{u_0}=  \mathcal{B}\cap J^{+}[C_{u_0}]$. We can fix the $u$-gauge in $J^{+}(C_{u_0})$ by imposing \begin{equation}
			r\rd_u r_{|\mathcal{B}_{u_0}} = -1.
		\end{equation}
		
		\item \label{gauge.U.past} Let $\mathcal{H}^+$ be the event horizon of a black hole spacetime, which we normalize to be $\{U=0,\ v_0 \leq v<+\infty\}$ (this is not a $v$-gauge choice). We place ourselves in the black hole exterior region  $\{U<0,\ v_0 \leq v<+\infty\}$.	Then, one can define the $U$-gauge as $U=-[2K_+]^{-1} e^{-2K_+ u}$, and $u\in \RR$ is determined by \begin{equation}
			\lim_{v\rightarrow+\infty} \frac{4\rd_v r(u,v) }{\Omega^2(u,v)} =1.
		\end{equation} 
		
	\end{enumerate}

	As for fixing the electromagnetic gauge freedom, we impose throughout the paper \begin{equation}\label{A.gauge}
		A_v \equiv 0.
	\end{equation} (except, perhaps in notational passages such as Section~\ref{gluing.section} where convenience leads us to impose $A_v = A_u$. However, when using these notions later in the paper, we will work in an uncharged spacetime region where $A\equiv 0$ anyway, so this difference in gauge will be inconsequential.)
	
	We also have the freedom to fix \begin{equation}\label{Au.gauge}
		A_u(\cdot,v_0) = 0,
	\end{equation} on some ingoing cone $\underline{C}_{v_0}$, which we will invoke when convenient.
	
	\subsection{Reissner--Nordstr\"{o}m solution}
	Recall that the Reissner--Nordstr\"{o}m solution is a spherically symmetric and stationary solution of \eqref{1.1}--\eqref{5.1} with $\phi\equiv 0$ and $F= \frac{e}{r^2} dt \wedge dr$:
	\begin{equation}\label{RN2}
		g_{RN} = -(1-\frac{2M}{r}+ \frac{e^2}{r^2}) dt^2+ (1-\frac{2M}{r}+  \frac{e^2}{r^2})^{-1} dr^2+ r^2 ( d\theta^2+ \sin^2(\theta) d\varphi^2),
	\end{equation} with charge $e$, mass $M$ and in the sub-extremal case $0<|e|<M$. Note that for \eqref{RN2}, $\varpi \equiv M$ and $Q \equiv e$. 
	$r_{\pm}(M,e)$ correspond to the area-radius of sections of the event horizon $\mathcal{H}^+$ and the Cauchy horizon $\mathcal{CH}_{i^+}$, defined by \begin{equation}\begin{split}
			&  r_{\pm}(M,e) := M \pm \sqrt{M^2-e^2}>0,\\ & 2K(\varpi=M,Q=e,r=r_{\pm}(M,e))=2K_{\pm}= \frac{2}{r_{\pm}^2(M,e)} [M -\frac{e^2}{r_{\pm}(M,e)}]\neq 0,
		\end{split}
	\end{equation} and $2 K_{\pm}$ are respectively the surface gravities of $\mathcal{H}^+$ and $\CH$. Note that, in the notations of Section~\ref{gauge.section}, and the Eddington--Finkelstein gauges \eqref{gauge.EH.v}, \eqref{gauge.EH.U}: \begin{equation}\begin{split}
			&\Omega^2= 4(1-\frac{2M}{r}+ \frac{e^2}{r^2}),\\ & \rd_v r= -\rd_u r = (1-\frac{2M}{r}+ \frac{e^2}{r^2}). \end{split}
	\end{equation}
	\subsection{Trapped region and apparent horizon}

	We define the trapped region $\mathcal{T}$, the regular region $\R$ and the apparent horizon $\A$ as 
	\begin{enumerate}
		\item  \label{characttrapped}$(u,v) \in \T$ if and only if $\partial_v r(u,v)<0$ if and only if $1-
		\frac{2\blue{\mathfrak{m}} (u,v)}{r(u,v)}<0$,
		\item $(u,v) \in \R$ if and only if $\partial_v r(u,v)>0$ if and only if $1-\frac{2\blue{\mathfrak{m}} (u,v)}{r(u,v)}>0$,
		\item $(u,v) \in \A$ if and only if $\partial_v r(u,v)=0$ if and only if $1-\frac{2\blue{\mathfrak{m}} (u,v)}{r(u,v)}=0$.
	\end{enumerate}	
	
	Note that, if $\mathcal{M}$ is the MGHD of initial data on $\Sigma$ which has no anti-trapped sphere, i.e., $\rd_u r <0$ on $\Sigma$, then $\rd_u r <0$ on $\mathcal{M}$, as a consequence of \eqref{RaychU}. We will only consider spacetimes free of anti-trapped sphere in this manner.
	
	Note also that, if $\mathbb{S}_1 \in \mathcal{T}$ (respectively $\mathcal{A}\cup \mathcal{T}$) and $\mathbb{S}_2$ is located on an outgoing  cone emanating from $\mathbb{S}_1$, then  $\mathbb{S}_2 \in \mathcal{T}$ (respectively $\mathcal{A}\cup \mathcal{T}$), as a consequence of \eqref{RaychV}. This well-known property will be used repetitively throughout the paper with no further reference.
	
	\subsection{Asymptotically flat one-ended initial data}\label{initial.data.section}
	We prescribe initial data $(r,f,h,l,Q,\phi,\dot{\phi})$ on an hypersurface $\Sigma_0$ diffeomorphic to $\RR^3$ and denote $\Gamma$ the center (diffeomorphic image of the origin of $\RR^3$), where $r \in C^2(\Sigma_0)$, $f \in C^{1}(\Sigma_0)$, $Q \in C^{1}(\Sigma_0)$, $\phi \in C^{1}(\Sigma_0)$,  $\dot{\phi} \in C^{0}(\Sigma_0)$. The induced metric 	$\hat{g}$,  second fundamental form $\hat{k}$, Maxwell field  $F$ on $\Sigma_0$, are given by
	\begin{equation}\label{hat.g.def}
		\hat{g}= f^2(\rho) d\rho^2 + r^2(\rho) d\sigma_{\mathbb{S}^2},
	\end{equation}
	\begin{equation}\label{hat.k.def}
		\hat{k}= h(\rho) d\rho^2 + l(\rho) d\sigma_{\mathbb{S}^2},
	\end{equation}
	\begin{equation}\label{dot.phi.def}
		(\phi, D_{n} \phi)_{|\Sigma_0}=(\phi,\dot{\phi}),
	\end{equation}
	\begin{equation}\label{F.def}
		F(n,\rd_{\rho})_{|\Sigma_0}=\frac{Q f}{r^2},
	\end{equation} where $n$ is the unit future-directed normal to $\Sigma_0$. Moreover, we need to specify an electromagnetic potential $A$ such that $dA=F$ (electromagnetic gauge choice). For simplicity (note that this is different from the gauge choice \eqref{A.gauge}, but this will not impact the subsequent results.), we can choose $A_{\rho}=0$, in which case \eqref{Maxwell} gives \begin{equation}
		\rd_{\rho} A_n = \frac{Q f}{r^2}.
	\end{equation}

	\noindent	Furthermore,	\eqref{1.1}-\eqref{5.1} imposes the following constraint equations for $(r,f,h,l,Q,\phi,\dot{\phi})$ on $\Sigma_0$: \begin{equation}\label{constraint1}
		R(\hat{g})- \hat{g}^{\alpha \mu} \hat{g}^{\beta \nu}\hat{k}_{\alpha \beta} \hat{k}_{\mu \nu}+(\hat{g}^{\alpha \beta}\hat{k}_{\alpha \beta})^2= 2 |\dot{\phi}|^2+ \frac{2}{f^2}|D_{\rho}\phi|^2+ \frac{2Q^2}{r^4},
	\end{equation}
	\begin{equation}\label{constraint2}
		\nabla_{\hat{g}}^{\alpha} \hat{k}_{\alpha \rho} -\rd_{\rho} (\hat{g}^{\alpha \beta}\hat{k}_{\alpha \beta})= 2\Re(\dot{\phi}D_\rho \phi),
	\end{equation}
	\begin{equation}\label{constraintQ}
		\rd_{\rho} Q = -q_0 r^2\Im(\bar{\phi} \dot{\phi}).
	\end{equation}
	The following definition of asymptotic flatness is inspired from \cite{JonathanStabExt}, Definition 3.1.
	
	\begin{defn}\label{AF.def}
		We say $(r,f,h,l,Q,\phi,\dot{\phi})\in  (C^{2}(\Sigma_0),C^{1}(\Sigma_0),C^{0}(\Sigma_0),C^{1}(\Sigma_0),C^{1}(\Sigma_0),C^{1}(\Sigma_0),C^{0}(\Sigma_0))$ is an asymptotically flat set of initial data for \eqref{1.1}--\eqref{5.1}  if  $(r,f,h,l,Q,\phi,\dot{\phi})$ satisfy \eqref{constraint1}, \eqref{constraint2}, \eqref{constraintQ} and there exists $\ep>0$ such that  the following estimates hold as $\rho \rightarrow +\infty$: \begin{equation}\begin{split} &| r(\rho) - \rho|\ls\log(\rho),\ |\rd_{\rho}r(\rho)|\ls \rho^{-1},\ |\rd_{\rho}^2 r(\rho)|\ls \rho^{-2},  \\ 
				&|f(\rho)-1|\ls \rho^{-1},\  |\rd_{\rho}f(\rho)|\ls \rho^{-2},\\ & |h(\rho)|\ls \rho^{-2},\\ &|l(\rho)|\ls 1,\  |\rd_{\rho}l(\rho)|\ls \rho^{-1},\\   & |Q|(\rho)\ls 1,\ |\rd_{\rho} Q(\rho)|\ls \rho^{-1-\ep},\\  & |\phi(\rho)| \ls \rho^{-1-\ep},\ |\rd_{\rho}\phi(\rho)|,\  |\dot{\phi}(\rho)| \ls \rho^{-2-\ep}.
			\end{split}
		\end{equation} 
	\end{defn}
	
	\begin{rmk}
		Note that $(r,f,h,l,Q,\phi,\dot{\phi})\in  (C^{2}(\Sigma_0),C^{1}(\Sigma_0),C^{0}(\Sigma_0),C^{1}(\Sigma_0),C^{1}(\Sigma_0),C^{1}(\Sigma_0),C^{0}(\Sigma_0))$ is sufficient to make sense of   \eqref{constraint1}, \eqref{constraint2}, \eqref{constraintQ} in the sense of distributions. In fact, we will work with a $C^k$ solution for any arbitrarily large $k\in \mathbb{N}$ and thus    \eqref{constraint1}, \eqref{constraint2}, \eqref{constraintQ}  are satisfied classically as long as $k\geq 2$. If desired, it  is possible to introduce a stronger notion of asymptotic flatness (involving the spatial decay of more derivatives) and prove that this stronger notion is also satisfied for the class of spacetimes we work with. 
	\end{rmk}

	\subsection{Sphere gluing in spherical symmetry}\label{gluing.section}
	We recall some relevant concepts  regarding the gluing of spheres for \eqref{1.1}--\eqref{5.1}, which can be found in \cite{KehleUnger}. In the notations of Definition~\ref{lapse.normalized}, the reader should associate $(\varrho,\omega,\varphi,q,a)$ to $(r,\Omega,\phi,Q,A)$, $\varrho_u^{j}$ to $\rd_u^{j} r$, etc...
	\begin{defn}[$C^k$ sphere data set, Definition 2.3 in \cite{KehleUnger}]\label{lapse.normalized}
		A $C^k$ sphere data set for  \eqref{1.1}--\eqref{5.1} is the following list of numbers \begin{enumerate}
			\item $(2k+3)$ real numbers $(\varrho,\varrho_u^{1},...\varrho_u^{k+1},\varrho_v^{1},...\varrho_v^{k+1})$ with $\rho>0$, $\rho_u^{1}<0$.
			
			\item $(2k+1)$ real numbers $(\omega,\omega_u^{1},...\omega_u^{k},\omega_v^{1},...\omega_v^{k})$ with $\omega>0$.
			\item $(2k+1)$ complex numbers $(\varphi,\varphi_u^{1},...,\varphi_u^{k},\varphi_v^{1},...\varphi_v^{k})$.
			\item $(2k+1)$ real numbers $(q,q_u^{1},...q_u^{k},q_v^{1},...q_v^{k})$.
			
			\item $(2k+2)$ real numbers $(a,a_u^{1},...a_u^{k},a_v^{1},...a_v^{k+1})$.

		\end{enumerate} satisfying the following conditions: \begin{itemize}
			
			\item $\varrho_u^{i+2}$ can be expressed in terms of $\varrho_u^{j+1}$, $\omega_u^{j+1}$, $\varphi_u^{j+1}$, and $a_u^j$ for $0\le j\le i$ by formally differentiating \eqref{RaychU},
			\item $\varrho_v^{i+2}$ can be expressed in terms of $\varrho_v^{j+1}$, $\omega_v^{j+1}$, and $\varphi_v^{j+1}$ for $0\le j\le i$ by formally differentiating \eqref{RaychV},
			\item $q_u^{i+1}$ can be expressed in terms of $\varrho_u^j$, $\varphi_u^j$, and $a_u^j$ for $0\le j\le i$ by formally differentiating \eqref{chargeUEinstein}.
			\item $q_v^{i+1}$ can be expressed in terms of $\varrho_u^j$, and $\varphi_u^j$ for $0\le j\le i$ by formally differentiating \eqref{ChargeVEinstein}, and 
			\item $a_v^{i+1}$ can be expressed in terms of $\varrho_v^j$, $\omega_v^j$, and $q_v^j$ for $0\le j\le i$ by formally differentiating \eqref{dv.Au}.
		\end{itemize} 
		
		We say the above $C^k$ sphere data is $C^k$ lapse-normalized if $(\omega,\omega_u^{1},...\omega_u^{k},\omega_v^{1},...\omega_v^{k})=(1,0,...0)$.
		
	\end{defn}
	\begin{rmk}
		We note that, contrary to \cite{KehleUnger}, we make the assumption that $\varrho_u^{1}<0$ in Definition~\ref{lapse.normalized}, so that the relevant sphere is not anti-trapped
	\end{rmk}
	
	We note that every  $C^k$ sphere data set is gauge equivalent to a $C^k$ lapse-normalized sphere data set \cite{KehleUnger} (recall Section~\ref{gauge.section} for a discussion of gauge choices). \begin{defn}
		A $C^{k}$ sphere data set is uncharged if $(q,q_u^{1},..., q_u^{k},q_v^{1},...,q_v^{k})=0$ and $(a,a_u^{1},..., a_u^{k},a_v^{1},...,a_v^{k+1})=0$. If $D$ is an uncharged $C^k$  sphere data set and $D'$ is a $C^k$ sphere data set that is gauge-equivalent to $D$, then $D'$ is uncharged too.
	\end{defn}
	
	We now turn to the definition of spacelike gluing; although we will only use it in the uncharged case, we formulate it the general charged case below. Note that \eqref{d.rho.A} in Definition~\ref{spacelike.gluing.def} below is formulated in the gauge $A_{\rho} = 0$ for convenience; however, it is easy to make a more general definition not imposing this gauge.
	
	\begin{defn}\label{spacelike.gluing.def} Let $D_1 $ and $D_2$, two $C^k$ sphere data sets. We say that $D_1$ can be spatially glued to $D_2$ if there exists  $D(s)$, $s_1 \leq s \leq s_2$ a collection of $C^k$ sphere data sets such that $D(s_1)$ is gauge-equivalent to $D_1$, $D(s_2)$ is gauge-equivalent to $D_2$  and, defining \begin{equation}\begin{split} & r(s)=\varrho(s),\\ &   \lambda(s)=\varrho_v^{1}(s),\  \nu(s)=\varrho_u^{1}(s),\\ & Q(s) = q(s), \\ & A_u(s) =a_u^{1}(s),\ A_v(s) = a_v^{1}(s),\\ & \Omega^2(s) = \omega^2(s),\\  & \varpi(s) = \frac{r(s)}{2 } \left[ 1 + \frac{Q^2(s)}{r^2(s)} + \frac{\lambda(s) \nu(s)}{\Omega^2(s)} \right],\\ &  \phi(s) =  \varphi(s),\\ & \theta(s) = r(s) [\varphi_v^{1}(s)+ i q_0 A_v(s)],\\ & \xi(s) = r(s) [\varphi_u^{1}(s)+ i q_0 A_u(s)],
			\end{split}
		\end{equation}  such that   $r$ is $C^{k+2}$, $\lambda$, $\nu$, $\Omega^2$, $\varpi$, $Q$, $\phi$ are $C^{k+1}$, $A_u$ and $A_v$ are $C^{k}$, and they respectively satisfy \begin{equation}\label{d.rho.r}
			\rd_{\rho} r(s) = \lambda(s) -\nu(s),
		\end{equation}
		\begin{equation}\label{d.rho.}
			\rd_{\rho} \varpi(s) =\frac{1}{2} [1-\frac{2\varpi(s)}{r(s)}+ \frac{Q^2(s)}{r^2(s)}] \left( \frac{|\theta|^2(s) }{\lambda(s)}+ \frac{|\xi|^2(s) }{|\nu|(s)} \right) + q_0 Q(s) \Im( \phi(s) [\overline{\theta}(s)+\overline{\xi}(s)] ),
		\end{equation}
		\begin{equation}\label{d.rho.phi}
			\rd_{\rho} \phi(s) =r^{-1}(s) [\theta(s) - \xi(s)] - i q_0 [ A_v(s)- A_u(s)] \phi(s), 
		\end{equation}
		\begin{equation}\label{d.rho.Q}
			\rd_{\rho} Q(s) = - q_0 r(s) \Im( \bar{\phi}(s)  [\theta(s) + \xi(s)]),
		\end{equation} 
		\begin{equation}\label{d.rho.A}
			\rd_{\rho} [A_u(s) + A_v(s)] = \frac{ \Omega^2(s) Q^2(s)}{r^2(s)},
		\end{equation}and for all $1\leq j\leq k+1$, $\rd_{\rho}^{j+1} r(s)$, $\rd_{\rho}^{j} \varpi(s)$, $\rd_{\rho}^{j} \phi(s)$, $\rd_{\rho}^{j} Q(s)$   satisfy the respective equations obtained by formally differentiating \eqref{d.rho.r}, \eqref{d.rho.}, \eqref{d.rho.phi}, \eqref{d.rho.Q}, \eqref{d.rho.A}, replacing $\rd_{\rho}$ by $\rd_v - \rd_u$, $\rd_v \lambda$ by $\varrho_v^2$, etc...
		
		If $D_1 \in \mathcal{R}$, $D_2\in \mathcal{R}\cup \mathcal{A}$ and $D(s) \in \mathcal{R}$ for all $s_1 \leq s <s_2$, we say $D_1$ and $D_2$ can be spatially glued within the regular region.

		If $D(s)$ is an uncharged $C^k$ sphere data set for all $s_1\leq s\leq s_2$, we say $D_1$ and $D_2$ are spatially glued in an uncharged way.
	\end{defn}
	
	\section{Precise statements of the main theorem}\label{thm.section}

	\subsection{The local result of our previous paper \cite{bif}}
	We first recall the main result of \cite{bif}, which is formulated for local initial data on bifurcate hypersurfaces $\Cin \cup C_{out} = [u_0,u_F)\times \{v_0\} \cup \{u_0\}\times [v_0,+\infty)$ as depicted in Figure~\ref{fig:local}. In the discussion below, the $u$-gauge is fixed by gauge~\eqref{gauge.U.future}, while the $v$-gauge is inherited from an Eddington--Finkelstein gauge \eqref{gauge.EH.v}.

	\begin{thm}\label{main.thm}[Theorem 3.1 in \cite{bif}].
		Let $u_0 < u_F$ and $v_0 \in \RR$.  Let $\uch= \sup\{u \in [u_0,u_F],\ \underset{v\rightarrow +\infty}\lim r(u,v) >0 \} $. We denote $\underset{v\rightarrow +\infty}\lim r(u,v)= r_{CH}(u)$ for all $u\in [u_0,\uch]$. Assume  the following estimates hold on $C_{out}=\{u_0\}\times [v_0,+\infty)$:
		
		\begin{equation}	\label{hyp1}\begin{split}  &r_{CH}(u_0) >0,\ \lim_{v \rightarrow +\infty} \rd_u r(u_0,v) <0, \\
				& \Omega^2(u_0,v) \leq D \cdot e^{ 2K_- v + C v^{1-\eta}},\\	 &  (2+\eta)K_- \leq -\rd_v \log(\Omega^2)(u_-,v) \leq (2-\eta)K_-<0\\  &  				L_-\ v^{-2s}\leq -r\partial_v r(u_0,v) \leq L_+\ v^{-2s}, \\ & |\phi|^2(u_0,v),\  |Q|(u_0,v) \leq D,
			\end{split} 
		\end{equation}

		where $s>\frac{1}{2}$,  $\eta\in (0,1)$. Then, assuming $v_0$ is large enough with respect to  the constants involved in \eqref{hyp1}, \begin{enumerate}[i.]
			\item \label{main.thm.I} $\uch \in (u_0,u_F]$ and $[u_0,\uch] \times [v_0,+\infty) \subset \T$. Moreover, there exists $D_->0$, $D_+>0$, such that for all  $u \in [u_0,\uch]$, $v\geq v_0$ \begin{equation} \label{quant1}\begin{split} &   D_-\ v^{-2s}\leq -r\rd_v r(u,v) \leq D_+\ v^{-2s},\\ &
					r^2_{CH}(u) + \frac{2D_- }{2s-1}\cdot v^{1-2s} \leq	r^2(u,v)\leq   r^2_{CH}(u) +  \frac{2D_+ }{2s-1}\cdot v^{1-2s}.\end{split}
			\end{equation} 
			If, moreover, there exists $s>1$ such that the initial data satisfy \eqref{hyp1} and the additional estimate: \begin{equation}\label{hyp2}
				\begin{split}
					|D_v \phi|(u_0,v) \leq \tilde{D} \cdot v^{-s}
				\end{split}
			\end{equation} for some $\tilde{D}>0$, then the following spacetime estimates are satisfied: for all $u \in [u_0,\uch]$, $v\geq v_0$: 
			\begin{equation}\label{quant2}
				\begin{split}&	
					|\phi|(u,v) \lesssim r^{-\frac{1}{2s-1}}(u,v),\\ 
					&
					r |D_u\phi|(u,v) \lesssim r^{-\frac{s}{s-\frac{1}{2}}}(u,v),
					\\ &
					r| D_v \phi|(u,v)\lesssim v^{-s},\\ & |Q|(u,v) \leq \check{D},
				\end{split}
			\end{equation} for some $\check{D}>0$.
			
			\item \label{main.thm.II}	If  $\uch< u_F$ (breakdown assumption), then for all $ \uch <u \leq u_F$, there exists $v_{\mathcal{S}}(u)<+\infty$ such that  \begin{equation}
				\lim_{v \rightarrow v_{\mathcal{S}}(u)} r(u,v) =0 \text{, } \lim_{v\rightarrow +\infty} r(\uch,v)=0 \text{ and } \lim_{u \rightarrow \uch} r_{CH}(u)=0.\end{equation}
			
			If, moreover, there exists $s>1$ such that the initial data satisfy \eqref{hyp1}, \eqref{hyp2} and the following additional estimate holds: there exists $D_L>0$,  $D_C>0$, $\delta>0$, $\alpha_{\infty} \in \RR$ such that for all $v\geq v_0$:   \begin{equation}\label{hyp4}
				\begin{split}& 	|D_{v v}^2 \phi|(u_0,v) \leq D_{C} \cdot v^{-s-1},
			\end{split}	\end{equation} \begin{equation}\label{hyp5}
				|\Im( e^{i q_0 \int_{v_0}^{v}  A_v(u_0,v') dv'} e^{-i\alpha_{\infty}}D_v \phi(u_0,v))| \leq D_{C} \cdot|\phi|(u_0,v) \cdot v^{-s-\delta},
			\end{equation}
			\begin{equation}\label{hyp3}
				\begin{split}&  |D_v \phi|(u_0,v) \geq  D_L \cdot  v^{-s}.
				\end{split}
			\end{equation}

			Then, defining  $\mathcal{S}=\{(u,v_{\mathcal{S}}(u)),\ u\in (\uch,u_F)\}$, there exists $0<\ep< u_F -\uch$  such that \\$\mathcal{S} \cap (\uch,\uch+\ep)$ is spacelike with the following estimate for all $\uch<u \leq \uch +\ep$: \begin{equation}\label{K1}
				v_{\mathcal{S}}(u) \approx \left( u- \uch\right)^{-\frac{1}{2s-1}},\ v'_{\mathcal{S}}(u) \approx -\left( u- \uch\right)^{-\frac{2s}{2s-1}}.
			\end{equation}
			Moreover, the metric takes the following approximate Kasner form:  there exists coordinates $(x,\tau)$ so that $\mathcal{S}=\{\tau=0\}$, $\mathcal{S}\cap \CH=\{\tau=0,\ x=0\}$  and   $x_0\geq 0$ small enough so that for all 
			$0 \leq x \leq x_0$, $\tau\geq0$: 
			\begin{equation}\label{K2}\begin{split}
					&	g= - (1+ \mathcal{E}_T(\tau,x)) d \tau^2 + \tau^{2 (1-2p(\tau,x))}(1+ \mathcal{E}_X(\tau,x)) dx^2 + \tau^{2p(\tau,x)} ( d\theta^2+\sin^2(\theta)d\varphi^2 ) ,\\   &  \bigl| p(\tau,x)-p(0,x)\bigr|\ls \frac{|\log|(x)}{|\log|(\tau)},\ p(0,x) \approx x^{\frac{1}{2(s-1)}},\ |\rd_x p|(\tau,x) \ls x^{-1+\frac{1}{2(s-1)}},\\ & \phi(\tau,x) = p_{\phi}(x) \left( \log(\frac{x^{\frac{2s-1}{2(s-1)}}}{\tau})+ \mathcal{E}_\phi(\tau,x)\right)+\varXi_{\mathcal{S}}(x),\\ & |p_{\phi}|(x) \approx x^{\frac{1}{4(s-1)}},\ |\varXi_{\mathcal{S}}|(x)  \ls x^{\frac{1}{4(s-1)}},  \\ & |\mathcal{E}_T|(\tau,x),\  |\mathcal{E}_X|(\tau,x),\ |\mathcal{E}_\phi|(\tau,x) \lesssim \frac{\tau^{2p(\tau,x)}}{x^{\frac{2s-1}{s-1}}}\left[1+\log^2(\frac{\tau^{2p(\tau,x)}}{x^{\frac{2s-1}{s-1}}})\right], \end{split}
			\end{equation}  where $p_\phi(x) \in \mathbb{C}$ satisfies the usual Kasner relations \begin{equation}\begin{split}\label{K3}
					&  p_1^2(\tau,x)+ p_2^2(\tau,x)+ p_3^2(\tau,x)+2 |p_\phi|^2(x)=1,\\ & p_1(\tau,x) = 1-2p(\tau,x),\ p_2(\tau,x) = p_3(\tau,x) = p(\tau,x), \text{ i.e., }  p_1(\tau,x)+ p_2(\tau,x)+ p_3(\tau,x)=1. 
				\end{split}	
			\end{equation}
			in the coordinate system $(\tau,x,\theta,\varphi)$, which relates to  $(u,v,\theta,\varphi)$ in the following way, with $(u,v)=(\uch,+\infty)$ corresponding to $(\tau,x)=(0,0)$ and $\mathcal{S}=\{r=0\}=\{\tau=0\}$ and defining $ x_{\mathcal{S}}(v):=  \underset{u \rightarrow u_{\mathcal{S}}(v)}{\lim}x(u,v)$: \begin{equation}\begin{split}\label{K4}
					&\tau(u,v) =[ r(u,v)]^{p^{-1}(u,v)},\\ & \bigl| x(u,v) - x_{\mathcal{S}}(v)\big| \lesssim \frac{r^2(u,v)}{r_0^2(v)}\left[1+\log^2(\frac{r^2(u,v)}{r_0^2(v)})\right],\  x_{\mathcal{S}}(v) \approx v^{2(1-s)}.
				\end{split}
			\end{equation}
			
			In particular, the Kasner exponents and scalar field  obey the following estimates in $(u,v)$ coordinates: \begin{equation}\begin{split}\label{K5}
					&	p(u,v) \approx v^{-1},\ |p_{\phi}|(v) \approx v^{\frac{1}{2}},\ \phi(u,v) = p_{\phi}(v) \log(\frac{r_0(v)}{r(u,v)})+ \tilde{\varXi}_{\mathcal{S}}(v),\  |\tilde{\varXi}_{\mathcal{S}}|(v)\ls v^{\frac{1}{2}},\\ & |\mathcal{E}_T|(u,v),\  |\mathcal{E}_X|(u,v),\ |\mathcal{E}_\phi|(u,v) \lesssim \frac{r^2(u,v)}{r_0^2(v)}\left[1+\log^2(\frac{r^2(u,v)}{r_0^2(v)})\right], \end{split}
			\end{equation} where $r_0(v):= r(\uch,v) \approx v^{\frac{1}{2}-s}$ as a consequence of \eqref{quant1}.

			\item  \label{main.thm.III}If we assume that $[u_0,u_F) \times\{v_0\} \subset \T\cup \A$ and $\underset{ u \rightarrow u_F}{\lim} r(u,v_0)=0$, then $\uch < u_F$, so Statement~\ref{main.thm.II} holds.
			
			Moreover, there exist (a large class of) initial data $\underline{C}_{in} \cup C_{out}=[u_0,u_F) \times [v_0,+\infty)$ such that the assumptions \eqref{hyp1}, \eqref{hyp2}, \eqref{hyp3} hold on   $C_{out}$ and $\Cin \subset \T$ with $\underset{ u \rightarrow u_F}{\lim} r(u,v_0)=0$.    So, for such initial data,  the conclusion of Statement~\ref{main.thm.II} holds. These initial data are constructed as such: starting from the gauge choice    \begin{equation}\begin{split}
					& A_u(\cdot,v_0)=0,\\ & -r\partial_u r(\cdot,v_0) \equiv-1\ \& \lim_{u\rightarrow u_F} r(u,v_0)=0, \text{ or equivalently } u_F - u = \frac{r^2(u,v_0)}{2} ,\label{Kasner.data.construction1}\end{split}
			\end{equation}   we then assume Kasner-like scalar field asymptotics, in  that there exists a constant $|\psi_0|>1$ such that 
			\begin{equation} \label{Kasner.data.construction2}\begin{split}
					& |\phi|(u,v_0) \lesssim \log(r^{-1})(u,v_0),\ \frac{ |D_u \phi|}{|\rd_u r|}(u,v_0)  \ls r^{-1}(u,v_0),\\ & \liminf_{u\rightarrow u_F}\frac{  r|D_u \phi|(u,v_0)}{|\rd_u r|(u,v_0)} \geq |\psi_0|.\end{split}
			\end{equation}
			
			Under these assumptions, the following  quantities are well-defined: \begin{equation}\begin{split}\label{norms.def}
					&I(\phi)=  \sup_{u_0 \leq u < u_F } \left(- r^2(u,v_0) +  [Q(u,v_0)+q_0\int_{u_0}^{u} r^2 \Im(\bar{\phi}D_u\phi)(u',v_0) du' ]^2\right), \\ & N_0(\phi)= \int_{u_0}^{u_F} r^{-2}(u,v_0) 
					\exp(-\mathcal{F}(u)) |\rd_u r|(u,v_0)  du,
			\end{split}	\end{equation}  
			where we have introduced the notation $\mathcal{F}(u)=\int_{u_0}^{u} \frac{r|D_u \phi|^2(u',v_0)}{|\rd_u r|(u',v_0)}  du'$. Then, we make the additional quantitative assumption that 
			
			\begin{equation}\label{construction.main.assumption}
				2 \mathfrak{m}(u_0,v_0) -r(u_0,v_0)  >   I[\phi]N_0(\phi),
			\end{equation}
			which is in particular  satisfied if $2\mathfrak{m}(u_0,v_0) -r(u_0,v_0)$ is large with respect to  $|Q|(u_0,v_0)$ and $\phi(\cdot,v_0)$. Moreover, \eqref{construction.main.assumption} is also satisfied for a  class of large scalar field initial data, specifically with Kasner asymptotics of the form \begin{equation}\label{kasner.data.intro}
				\phi(u,v_0) = \Psi_0 \log(r^{-1}(u_0,v)) + \tilde{\phi}(u,v_0),
			\end{equation} where $\tilde{\phi}(u,v_0)$  is bounded and $|\Psi_0|$ is sufficiently large.

		\end{enumerate}

	\end{thm}
	
	\begin{rmk}
		Note that, instead of \eqref{hyp5}, Theorem 3.1 in \cite{bif} required the (stronger) assumption that \begin{equation}\label{hyp5'}
			|\Im(\bar{\phi} D_v \phi)|(u_0,v) \lesssim v^{-s-\delta}.
		\end{equation} We proved in Proposition 6.5 in \cite{bif} that \eqref{hyp5'} implies \eqref{hyp5}, and subsequently we have only used \eqref{hyp5} (and not \eqref{hyp5'}) in the proof of Theorem 3.1.  Although \eqref{hyp5'} is a ``cleaner'' assumption to state, it turns out that   \eqref{hyp5}  is simpler to achieve in practice and thus, in this manuscript we have decided to reformulate Theorem 3.1 as Theorem~\ref{main.thm} with  \eqref{hyp5} instead of \eqref{hyp5'}.
	\end{rmk}

	\subsection{Applications of the local result to one-ended black hole spacetimes}
	Based on geometric extension principles, it is possible to isolate all the possibilities for spacetime boundary components for one-ended spherically symmetric solutions of \eqref{1.1}--\eqref{5.1}: this result was obtained in \cite{Kommemi}, finding only a finite number of possibilities as seen in the following theorem.

	\begin{thm}[Theorem 1.1\ of \cite{Kommemi}\color{black}] \label{Kommemi.thm}
		
		We consider the maximal development $(M=\mathcal{Q}^+ \times_r \mathcal{S}^2,g_{\mu \nu}, \phi,F_{\mu \nu})$ of smooth, spherically symmetric, containing no anti-trapped surface, one-ended asymptotically flat \color{black} initial data satisfying the Einstein--Maxwell--Klein--Gordon system, where $r: \mathcal{Q}^+\rightarrow [0,+\infty)$ is the area-radius function. Then the Penrose diagram of $\mathcal{Q}^+$ is given by Figure~\ref{fig:one-ended}, with boundary $\Sigma \cup \Gamma$ in the sense of manifold-with-boundary --- where $\Sigma$ is space-like, and $\Gamma$, the center of symmetry, is time-like with $r_{|\Gamma}=0$ --- and boundary $\mathcal{B}^+$ induced by the manifold ambient $\RR^{1+1}$: $$ \mathcal{B}^+ = b_{\Gamma} \cup \mathcal{S}^1_{\Gamma} \cup \mathcal{CH}_{\Gamma} \cup  \mathcal{S}^2_{\Gamma} \cup  \mathcal{S}  \cup \mathcal{S}_{i^+}  \cup \CH \cup i^{+} \cup \mathcal{I}^+ \cup i^0,$$ where $i^0$ is space-like infinity, $\mathcal{I}^+$ is null infinity, $i^{+}$ is time-like infinity and 	  \begin{enumerate}
			\item $\CH$ is a connected (possibly empty) half-open null ingoing segment emanating from $i^{+}$. The area-radius function $r$ extends as a strictly positive function on $\CH$, except maybe at its future endpoint.
			\item $\mathcal{S}_{i^+}$ is a connected (possibly empty) half-open null ingoing segment emanating (but not including) from the end-point of $\CH \cup i^{+}$. $r$ extends continuously to zero on $\mathcal{S}_{i^+}$.
			\item $b_{\Gamma}$ is the center end-point i.e.\ the unique future limit point of $\Gamma$ in $\overline{\mathcal{Q}^+}-\mathcal{Q}^+$.
			\item $ \mathcal{S}^1_{\Gamma} $ is a connected (possibly empty) half-open null outgoing segment emanating from $b_{\Gamma}$.\\ $r$ extends continuously to zero on $ \mathcal{S}^1_{\Gamma} $.
			\item $\mathcal{CH}_{\Gamma}$ is a connected (possibly empty) half-open null outgoing segment emanating from the future end-point of $b_{\Gamma} \cup  \mathcal{S}^1_{\Gamma} $. $r$ extends as a strictly positive function on $\mathcal{CH}_{\Gamma}$, except maybe at its future endpoint.
			\item $\mathcal{S}^2_{\Gamma}$ is a connected (possibly empty) half-open null outgoing segment emanating from the future end-point of $\mathcal{CH}_{\Gamma}$.  $r$ extends continuously to zero on $\mathcal{S}^2_{\Gamma}$. 	\item $\mathcal{S} $ is a connected (possibly empty) achronal curve that does not intersect null rays emanating from $b_{\Gamma}$ or $i^+$. $r$ extends continuously to zero on $\mathcal{S}$.
		\end{enumerate} 
		We also define the black hole region $ \mathcal{BH}:= \mathcal{Q}^+ \backslash J^{-}(\mathcal{I}^+)$, and the event horizon $\mathcal{H}^+ = \overline{J^{-}(\mathcal{I}^+)} \backslash J^{-}(\mathcal{I}^+) \subset~ \mathcal{Q}^+$.
	\end{thm} We moreover record a useful definition from \cite{Kommemi,breakdown}: \begin{defn}\label{first.sing.def}
		We say $b_{\Gamma}$ is a (central) first singularity if  $\mathcal{S} \cup \mathcal{S}^1_{\Gamma}\cup \mathcal{S}^2_{\Gamma}\cup \mathcal{CH}_{\Gamma}\neq \emptyset$.
	\end{defn}
	The main open problem in spherically symmetric black hole dynamics is then to determine which of these spacetime boundaries components $\CH$, $\mathcal{S}_{i^+}$, $b_{\Gamma}$, $\mathcal{S}^1_{\Gamma}$,  $\mathcal{CH}_{\Gamma}$, $\mathcal{S}^2_{\Gamma}$,  $\mathcal{S}$, if any, is empty, and provide quantitative estimates on the spacetime metric and matter fields near the non-empty components.
	
	A partial resolution of this open problem is provided in   the main result of \cite{breakdown} reflecting the breakdown of weak null singularities, i.e., the fact that $\CH$ cannot be the only non-empty boundary component in Figure~\ref{fig:spacelikeconj}, providing $\CH$ is weakly singular (the weak singularity of $\CH$ is, in turn, obtained in \cite{Moi,Moi4}).
	
	\begin{thm}[Theorem 3.1 of \cite{breakdown}] \label{breakdown.thm}
		For a one-ended spacetime as in Theorem~\ref{Kommemi.thm}, assume there exists an outgoing future cone emanating from $(u_1,v_1) \in \T$ and reaching $\CH$ on which $\phi$ and $Q$ obey the following upper bounds: for all $v \geq v_1$, \begin{equation} \label{corollaryassumption}
			|\phi|(u_1,v)+ |Q|(u_1,v) \leq C \cdot |\log(\mathfrak{m})|,
		\end{equation} for some $C>0$, and the Hawking mass $\mathfrak{m}$ blows up \begin{equation} \label{corollaryassumption2}
			\lim_{v \rightarrow+\infty} \mathfrak{m}(u_1,v) =+\infty.
		\end{equation} Then the Cauchy horizon $\CH$ breaks down, namely $$\mathcal{S}^1_{\Gamma} \cup \mathcal{CH}_{\Gamma} \cup \mathcal{S}^2_{\Gamma} \cup  \mathcal{S}  \neq \emptyset.$$
	\end{thm}
	
	We now  give applications of the quantitative Theorem~\ref{main.thm} to the global one-ended setting in the following Theorem~\ref{main.thm.global.i} below, which also includes a new, simpler proof of Theorem~\ref{breakdown.thm}. The novel result of Theorem~\ref{main.thm.global.i} is that   $\mathcal{S}_{i^{+}} =\emptyset$ (unconditionally) and that locally-naked singularities on which $r>0$ (defined as $ \mathcal{CH}_{\Gamma}$) are the only obstruction to the presence of spacelike singularities described by the quantitative estimates of Theorem~\ref{main.thm}. Theorem~\ref{main.thm.global.i} below corresponds to the first statement of Theorem~\ref{thm.cond.intro} from the introduction.

	\begin{thm}\label{main.thm.global.i} 
		\item We consider a one-ended black hole interior in the sense of Theorem~\ref{Kommemi.thm} and we assume that the estimates \eqref{hyp1} are satisfied on some outgoing cone $C_{out}= \{u_0\}\times[v_0,+\infty)$  with $(u_0,v=+\infty)\in \CH$. Then  	 \begin{enumerate}[A.]
			\item \label{A} There is no collapsed light cone emanating from $i^+$, namely:  $\mathcal{S}_{i^{+}} =\emptyset$. 
			\item \label{B} The Cauchy horizon $\CH$ breaks down, namely: $\mathcal{S}^1_{\Gamma} \cup \mathcal{CH}_{\Gamma} \cup \mathcal{S}^2_{\Gamma} \cup  \mathcal{S}  \neq \emptyset.$

			\noindent	Now we make the initial data assumptions \eqref{hyp1}, \eqref{hyp2}, \eqref{hyp4}-\eqref{hyp3} on $C_{out}$ and we moreover assume that there is no locally-naked singularity on which $r>0$, i.e., $ \mathcal{CH}_{\Gamma}=\emptyset$. In particular, this assumption is satisfied if $\A$ has an endpoint at $b_{\Gamma}$. Then,

			\item  There is no locally-naked singularity on which $r=0$, namely: $ \mathcal{S}_{\Gamma}^{1}= \mathcal{S}_{\Gamma}^{2}=\emptyset$.
			\label{C.stat}
			
			\item There  exists a non-empty singularity not emanating from $i^+$ or $b_{\Gamma}$ and on which $r$ extends to $0$: namely $\mathcal{S} \neq \emptyset$ and for $\epsilon>0$ small enough, $\mathcal{S} \cap  (\uch,\uch+\epsilon] \times [v_0,+\infty)$ is a spacelike boundary in $\overline{\mathcal{Q}^+}$. Moreover, the quantitative estimates \eqref{quant1}, \eqref{quant2} and the Kasner asymptotics \eqref{K1}--\eqref{K5} of Theorem~\ref{main.thm} hold.\label{D.stat} 
			
			\noindent Finally, assume that there exists $u_s \in \RR$ such that  for all $u_0 \leq u_s$ with $(u_0,v=+\infty) \in \CH$, there exists $v_0(u_0)$ such that	  \eqref{hyp1}, \eqref{hyp2}, \eqref{hyp4}-\eqref{hyp3} hold on $C_{out}=\{u_0\} \times [v_0,+\infty)$. \item Then, $\CH$ is a weakly singular Cauchy horizon, in the sense of mass inflation, i.e., for all $u\leq \uch$: \begin{equation}
				\lim_{v\rightarrow +\infty} \varpi(u,v)=			\lim_{v\rightarrow +\infty} \mathfrak{m}(u,v)=+\infty.
			\end{equation}\label{E.stat} 
			
		\end{enumerate}
	\end{thm}
	
	\begin{rmk}Notice that, comparing to Theorem~\ref{main.thm}, Theorem~\ref{main.thm.global.i} does not require a breakdown condition of the Cauchy horizon. However, Statement~\ref{C.stat} and \ref{D.stat} require the absence of a locally-naked singularity $\mathcal{CH}_{\Gamma}$ emanating from the center $\Gamma$, an assumption which only depends on the local structure of $\Gamma$ and therefore cannot be excluded generically without a resolution of  Weak Cosmic Censorship as we explain in the Section~\ref{qual.section}.
	\end{rmk}
	\begin{rmk}
		It is useful to note that, by	  \eqref{murelation}, \eqref{hyp1} implies that $\underset{v\rightarrow+\infty}{\lim} \mathfrak{m}(u_0,v)=\infty$ (Hawking mass blow-up) so Statement~\ref{B} of Theorem~\ref{main.thm.global.i} is a slightly weaker result than Theorem~\ref{breakdown.thm}. However, the proof of Statement~\ref{B} is much simpler and shorter, as we will see in Section~\ref{breakdown.reproof}.
	\end{rmk}

	Then, we turn to the construction of  one-ended asymptotically examples to which Theorem~\ref{main.thm} applies. Theorem~\ref{main.thm.global.ii} below corresponds to  Theorem~\ref{thm.II} from the introduction, and the main result of the manuscript.
	
	\begin{thm} \label{main.thm.global.ii}
		\blue{Let $k\in \mathbb{N}$.}	There exists a large class of spherically symmetric black hole \blue{$C^k$} solutions of \eqref{1.1}--\eqref{5.1} with $q_0 \neq 0$ with the following properties \begin{itemize}
			\item Their induced initial data on $\Sigma$, a one-ended spacelike hypersurface, are asymptotically flat, free of trapped and anti-trapped spheres. Therefore, the MGHD of $\Sigma$ is free of  anti-trapped spheres.
			
			\item Their event horizon $\mathcal{H}^+$ is future-complete and in the strict causal future of $\Sigma$. Null infinity $\mathcal{I}^+$ is also future-complete.
			\item  	The Penrose diagram  is given by Figure~\ref{fig:spacelikeconj}, i.e.,  $\CH \neq \emptyset$,  $\mathcal{S}_{i^+}= \mathcal{S}_{\Gamma}^{1}=\mathcal{S}_{\Gamma}^{2}=\emptyset$,  $\mathcal{S} \neq \emptyset$.
			\item $\CH$ is a weakly singular Cauchy horizon, in the sense of mass inflation, i.e., for all $u\leq \uch$: \begin{equation}
				\lim_{v \rightarrow +\infty} \varpi(u,v) =  	\lim_{v \rightarrow +\infty} \mathfrak{m}(u,v) =+\infty.
			\end{equation}
			\item For $\epsilon>0$ small enough, $\mathcal{S} \cap  (\uch,\uch+\epsilon] \times [v_0^R,+\infty)$ is a spacelike boundary in $\overline{\mathcal{Q}^+}$. Moreover,  the quantitative estimates \eqref{quant1}, \eqref{quant2} and the Kasner asymptotics \eqref{K1}--\eqref{K5} of Theorem~\ref{main.thm} hold.
			\item For $\epsilon>0$ small enough, $\mathcal{S} \cap   (v_{\Gamma},v_{\Gamma}+\ep)$ is a spacelike boundary in $\overline{\mathcal{Q}^+}$. Moreover, $\mathcal{S} \cap   (v_{\Gamma},v_{\Gamma}+\ep)$ is spatially-homogeneous and described by an asymptotically Kasner metric of exponents $(\frac{1}{3},\frac{1}{3},\frac{1}{3})$.
		\end{itemize} 
		
	\end{thm}

	\begin{rmk}
		The construction obtained in  Theorem~\ref{main.thm.global.ii} relies on three ingredients. The first is the novel spacelike-characteristic gluing strategy (see Theorem~\ref{main.gluing.thm}) in order to glue any uncharged regular sphere to the Reissner--Nordstr\"{o}m trapped region. The second are propagation estimates in the black hole exterior for small charge, using spherical symmetry to propagate from left to right (i.e., from the event horizon towards the asymptotically flat end) leading to Theorem~\ref{EH.AF.thm} and Corollary~\ref{EH.AF.cor}. The last ingredient is a refinement of the scattering estimates in the black hole interior first obtained by the author and Kehle in \cite{MoiChristoph} (Theorem~\ref{nonlinearscat.thm}), which allows to verify the assumptions of Theorem~\ref{main.thm} under the Cauchy horizon from infinity $\CH$.
	\end{rmk}

	\subsection{Applications of the local result to  two-ended black hole spacetimes}
	While the main motivation to  Theorem~\ref{main.thm} is on its applications to the one-ended case modeling the physical process gravitational collapse, it turns out that Theorem~\ref{main.thm} also applies  to the two-ended case, which has been extensively studied as discussed in Section~\ref{section.global.uncond2}.
	
	A simplifying feature of the two-ended case is the absence of a center, i.e.\ $\Gamma=\emptyset$, so in particular locally-naked singularities $\mathcal{CH}_{\Gamma}$ cannot arise. On the other hand, for an open set of initial data, no breakdown of the Cauchy horizon  occurs \cite{nospacelike} and in this case the spacelike component of the boundary is empty, namely $\mathcal{S}=\emptyset$, as in the rightmost diagram of Figure~\ref{fig:nospacelike}. We summarize these results, together with the analogue of Theorem~\ref{Kommemi.thm} below. It applies to the maximal development of  future-admissible initial data $\Sigma$, namely under the assumption that $\Sigma = \Sigma_A \cup \Sigma_B$,  where  $\Sigma_A $ satisfies $\rd_u r<0$, and $\Sigma_B $ satisfies $\rd_v r<0$ and $\Sigma_A $, $\Sigma_B $  are connected (see \cite{nospacelike}).

	\begin{thm}[\cite{Mihalis1,nospacelike,Kommemi}] \label{Dafermos.thm}

		We consider the maximal development $(M=\mathcal{Q}^+ \times_r \mathcal{S}^2,g_{\mu \nu}, \phi,F_{\mu \nu})$ of smooth, spherically symmetric, future-admissible  two-ended asymptotically flat \color{black} initial data satisfying the Einstein--Maxwell--Klein--Gordon system, where $r: \mathcal{Q}^+ \rightarrow [0,+\infty)$ is the area-radius function. 
		\begin{itemize}
			\item (\cite{nospacelike}, Proposition 4.1, after \cite{Mihalis1} and \cite{Kommemi}). Then, the Penrose diagram of $\mathcal{Q}^+$ is given by Figure \ref{fig:twoended}, with  boundary $\mathcal{B}^+$ induced by the manifold ambient $\RR^{1+1}$: $ \mathcal{B}^+ =  \mathcal{S}_{i^+}  \cup \CH \cup i^{+} \cup \mathcal{I}^+ \cup i^0$.
			
			\item (\cite{nospacelike}). Assume  that the event horizon $\mathcal{H}^+$ settles quantitatively towards a sub-extremal Reissner--Nordstr\"{o}m event
			horizon and that the scalar field is \blue{initially} small, then $\mathcal{S}= \mathcal{S}_{i^+}=\emptyset$; in other words,  $ \mathcal{B}^+ =   \CH \cup i^{+} \cup \mathcal{I}^+ \cup i^0$ and the Penrose diagram is given by the rightmost diagram in Figure~\ref{fig:nospacelike}.
		\end{itemize}

	\end{thm}
	\begin{figure}[H]	\includegraphics[width=0.8\linewidth]{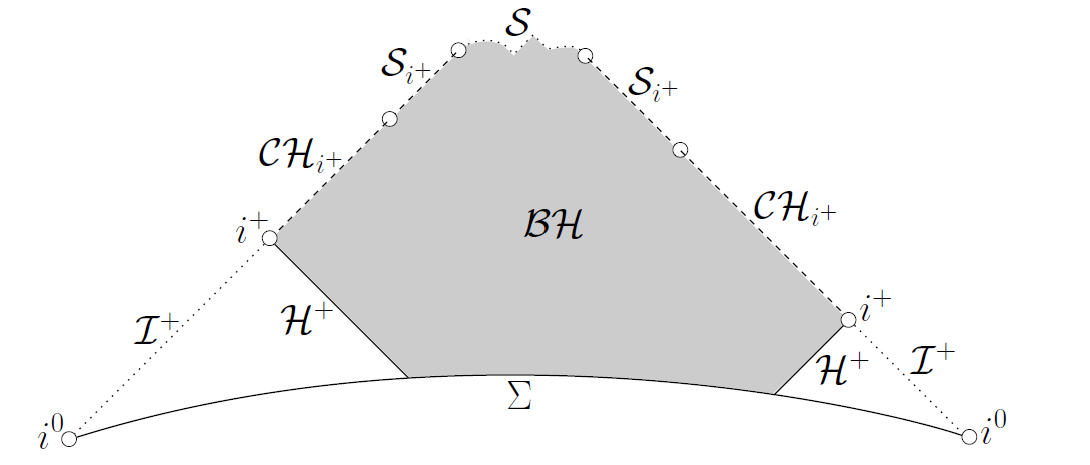}			
		\caption{General Penrose diagram of a two-ended future-admissible black hole, figure from \cite{Kommemi}.}	\label{fig:twoended}	\end{figure}
	In particular, the second statement of Theorem~\ref{Dafermos.thm} shows that there is no analogue of Theorem~\ref{breakdown.thm} in the two-ended case, since the Cauchy horizon $\CH$ can close-off a two-ended spacetime without breaking down. 
	We will prove,  however,  that there exist two-ended black holes such that $\CH \neq \emptyset$ and $\mathcal{S}\neq \emptyset$ (recall Theorem~\ref{thm.III} and see Theorem~\ref{main.thm.2end.ii} below). The following result is an application of Theorem~\ref{main.thm} to two-ended spacetimes such that  $\CH \neq \emptyset$  and $\mathcal{S}\neq \emptyset$ (in particular,  by Theorem~\ref{Dafermos.thm}, they  must be large perturbations of the Reissner--Nordstr\"{o}m black hole). Theorem~\ref{main.thm.2end.i} provides  a rather general  result, conditional on certain estimates being satisfied, \blue{as} the two-ended analogue of \blue{Theorem~\ref{main.thm.global.i}}. Theorem~\ref{main.thm.2end.i} below corresponds to the second statement of Theorem~\ref{thm.cond.intro} from the introduction; note that any value $q_0 \in \RR$ is allowed in  Theorem~\ref{main.thm.2end.i}. \phantom{ \eqref{K2},\eqref{K3},\eqref{K4}}

	\begin{thm}\label{main.thm.2end.i}
		
		\item  	\label{Itwo} We consider a two-ended black hole interior as in Theorem~\ref{Dafermos.thm} and we assume that there exists $C_{out}^{R}=\{u_0^{R}\}\times [v_0^{R},+\infty)$ with $(u_0^{R},v=+\infty)\in \CH$ such that \eqref{hyp1} is satisfied on  $C_{out}^{R}$ and  $C_{out}^{L}= [u_0^{L},+\infty)\times \{v_0^{L}\}$ with $(u=+\infty, v_0^{L})\in \CH$ such that the analogue of \eqref{hyp1} is satisfied on  $C_{out}^{L}$. Then
		\begin{enumerate}[A.]
			\item  \label{Atwo}There are no collapsed light cones emanating from $i^{+}$, namely: $\mathcal{S}_{i^+}=\emptyset$.

			Now we make the initial data assumptions  \eqref{hyp1}, \eqref{hyp2}, \eqref{hyp4}-\eqref{hyp3} on  $C_{out}^{R}$ and their analogue of $C_{out}^{L}$. 	 Moreover, we assume that $\CH$ does not close-off the spacetime, namely: $ \mathcal{S}\neq \emptyset$. Then,

			\item \label{Btwo} For $\epsilon>0$ small enough, $\mathcal{S} \cap  (\uch^{R},\uch^{R}+\epsilon] \times [v_0^R,+\infty)$ (and analogously  $\mathcal{S} \cap   [u_0^L,+\infty) \times (\uch^{L},\uch^{L}+\epsilon]  $) is a spacelike boundary in $\overline{\mathcal{Q}^+}$. Moreover, the quantitative estimates \eqref{quant1}, \eqref{quant2} and the Kasner asymptotics \eqref{K1}--\eqref{K5} of Theorem~\ref{main.thm} hold.
			
			\noindent Finally, assume that there exists $u_s^{R} \in \RR$ such that  for all $u_0^{R} \leq u_s^{R}$ with $(u_0^{R},v=+\infty) \in \CH$, there exists $v_0^{R}(u_0^{R})$ such that	  \eqref{hyp1}, \eqref{hyp2}, \eqref{hyp4}-\eqref{hyp3} hold on $C_{out}=\{u_0^{R}\} \times [v_0^{R},+\infty)$. Also assume that the analogous condition is satisfied on the left.

			\item \label{Ctwo}  Then, $\CH$ is a weakly singular Cauchy horizon, in the sense of mass inflation.

		\end{enumerate}

	\end{thm}
	Finally,  we also provide  an unconditional construction that, however, only applies to the case $q_0=0$ \blue{(uncharged scalar field)} previously studied by Luk--Oh \cite{JonathanStab,JonathanStabExt}. We note that this construction is very different from that of Theorem~\ref{main.thm.global.ii} and allows to construct a much larger class\footnote{The constructions in Theorem~\ref{main.thm.2end.ii} start from Cauchy data and thus the late-time tail behavior of the scalar field is the generic one, using results of \cite{Gautam,JonathanStab,twotails}. This is to be contrasted with  the examples of Theorem~\ref{main.thm.global.ii} which require solving the equations \eqref{1.1}--\eqref{5.1}  backwards in time and a scattering argument resulting in non-generic late-time tails (as first discussed in Section~\ref{CH.section}).} of two-ended black hole spacetimes than in the one-ended case. Theorem~\ref{main.thm.2end.ii} below corresponds to Theorem~\ref{thm.III} in the introduction.
	
	\begin{thm}\label{main.thm.2end.ii}
		
		In the case $q_0=0$, there exists a large class $\mathcal{G}'$ of two-ended asymptotically flat spherically symmetric black hole solutions of \eqref{1.1}--\eqref{5.1} with $\CH \neq \emptyset$, a weakly singular Cauchy horizon in the sense of mass inflation, $\mathcal{S}_{i^+}=\emptyset$ and  $\mathcal{S} \neq \emptyset$ and for $\epsilon>0$ small enough, $\mathcal{S} \cap  (\uch^{R},\uch^{R}+\epsilon] \times [v_0^R,+\infty)$ (and analogously  $\mathcal{S} \cap   [u_0^L,+\infty) \times (\uch^{L},\uch^{L}+\epsilon]  $) is a spacelike boundary in $\overline{\mathcal{Q}^+}$. Moreover, the quantitative estimates \eqref{quant1}, \eqref{quant2} and the Kasner asymptotics \eqref{K1}--\eqref{K5} of Theorem~\ref{main.thm} hold.
		
	\end{thm}

	\subsection{Spacelike gluing theorem}\label{gluing.thm.section}
	
	Theorem~\ref{main.thm.global.ii} relies on \blue{a} new gluing result, which we now present. Theorem~\ref{main.gluing.thm} will be split into two sub-gluing statements in Section~\ref{section.global.uncond1}: Theorem~\ref{uncharged.gluing.thm} and Theorem~\ref{charged.gluing.thm}, the underlying strategy of which consists in a spacelike-characteristic gluing approach. 
	
	\begin{thm}\label{main.gluing.thm}
		Let $D\in \mathcal{R}$ an uncharged $C^{k}$ sphere data of area-radius $R>0$, $0 \leq q <1$, $\varsigma=\pm 1$, and \begin{equation}\label{small.R.condition}
			\frac{R}{1+\sqrt{1-q^2}}<M_f < \frac{R}{1-\sqrt{1-q^2}}
		\end{equation} 
		
		Assume also that $\frac{|q_0| M_f(1-q)}{q}$ is sufficiently large. Then, there exists $R_f$ satisfying the bound $$ M_f [1- \sqrt{1-q^2}] < R_f < M_f [1+ \sqrt{1-q^2}],$$ such that $D$ can be spatially glued to a trapped Reissner--Nordstr\"{o}m sphere of Hawking mass $M_f$, charge $\varsigma q  M_f$ and area-radius $R_f$. 
	\end{thm}
	
	Note that, in the case $q=0$, Theorem~\ref{main.gluing.thm} permits spatial gluing to Schwarzschild. We apply Theorem~\ref{main.gluing.thm} below to construct  scalar field analogues to the celebrated Oppenheimer--Snyder solution (corresponding to Theorem~\ref{OS.thm.intro} from the introduction) in providing  spacetimes which are spatially-homogeneous near the center $\Gamma$, and Schwarzschild/Reissner--Nordstr\"{o}m near timelike infinity $i^+$: this is the object of Theorem~\ref{uncharged.thm} and Theorem~\ref{charging.thm}, respectively. Note that in Theorem~\ref{main.gluing.thm} below, we are allowed to prescribe the area-radius  of the regular sphere $R>0$ arbitrarily, which then constraints $M_f$ to be comparable to $R$; however, in subsequent applications of Theorem~\ref{main.gluing.thm}, such freedom will not be allowed, as we will require a sufficiently small $R$ to evolve the equations towards the past (see already the proofs of Theorem~\ref{uncharged.thm}, Theorem~\ref{charging.thm} and Theorem~\ref{EH.AF.thm}). A contrario, we will also see in the proof of Theorem~\ref{main.gluing.thm} that $R_f$ can be chosen more flexibly.

	\section{Intermediate results proven in our previous work \cite{bif}}\label{oldresults.section}

	In this section, we summarize some of the results proven in our previous work \cite{bif}.
	
	\subsection{A priori estimates}\label{apriori.est.section}

	We now place ourselves in the setting of Theorem~\ref{main.thm}, Statement~\ref{main.thm.I} and consider the Cauchy development of data on $C_{out} = \{u_0\} \times [v_0,+\infty)$ and $\Cin =[u_0,u_F) \times  \{v_0\}$.
	
	\begin{prop}\label{apriori.prop1}[Proposition 4.2 in \cite{bif}]. We  assume that the estimates \eqref{hyp1} on $C_{out}$ are true. 	 Then, $\uch$ satisfies the estimate $$u_0 \leq\uch \leq u_F,$$ and there exists $v_0(D)>0$ large enough so that $[u_0, \uch]\times [v_0,+\infty) \subset \T$. Furthermore, for all $ (u,v) \in [u_0, \uch]\times [v_0,+\infty)$, the following estimates hold\blue{, denoting $r_{|\CH}(u):= r_{CH}(u)$} \begin{align}
			&0.9 L_-\ v^{-2s}\leq -r\rd_v r(u,v) \leq  1.1 L_+\ v^{-2s}	\label{lambda.est.prelim},\\  &  \label{r.est.prelim}  r^2_{CH}(u) + \frac{0.9 L_-}{2s-1}\ v^{1-2s}\leq r^2(u,v) \leq r^2_{CH}(u) + \frac{1.1 L_+}{2s-1}\ v^{1-2s},\\ &  \bigl| r|\rd_u r|(u,v) -r|\rd_u r|_{CH}(u) \bigr| \leq E(D)\ e^{1.9 K_- v} \label{nu.est.prelim}.
		\end{align}

	\end{prop}
	
	\noindent 
	Note that Proposition~\ref{apriori.prop1} \blue{immediately} proves Theorem~\ref{breakdown.thm.new} from the Section~\ref{intro.section}.
	
	From Proposition~\ref{apriori.prop1}, it is clear that the gauge~\eqref{gauge.U.future} can be imposed, which we will do in the rest of the \blue{section}. We will also fix the gauge freedom so that $\uch=0$ (see Section~\ref{geometricframework}). Therefore, we have for all $u_0 \leq u \leq 0$: \begin{equation*}\begin{split}
			&(r\rd_u r)_{CH}(u) \equiv -1,\\ &  r_{CH}^2(u) =2|u|. 
		\end{split}
	\end{equation*}

	\subsection{Sufficient condition for  a local breakdown of the Cauchy horizon} The following proposition is important in the new proof of the breakdown of weak null singularities (Section~\ref{breakdown.reproof}).
	\begin{prop}\label{localbreak.prop} [Proposition 4.3 in \cite{bif}]. Assume that the estimates \eqref{hyp1} on $C_{out}$ are satisfied. Moreover, assume that $[u_0,u_F) \times\{v_0\} \subset \T$ and $\underset{ u \rightarrow u_F}{\lim} r(u,v_0)=0$. Then $\uch < u_F$ and for all $\uch < u \leq u_F$, there exists $v_{\mathcal{S}}(u)<+\infty$ such that \begin{equation*}
			\lim_{ v \rightarrow v_{\mathcal{S}}(u)} r(u,v) =0.
		\end{equation*}
	\end{prop}

	\section{Conditional applications to one or two-ended spacetimes}\label{section.global.cond}

	In this section, we address the proof of the conditional results, Theorem~\ref{main.thm.global.i} (one-ended case),  Theorem~\ref{main.thm.2end.i} (two-ended case) and  Theorem~\ref{inext.thm} ($C^2$-inextendibility result), and we relegate the proof of the (more difficult to prove) unconditional results    Theorem~\ref{main.thm.global.ii} and  Theorem~\ref{main.thm.2end.ii} to Section~\ref{section.global.uncond1} and Section~\ref{section.global.uncond2} respectively.
	
	\paragraph{Discussion of the assumptions}
	
	We emphasize that the main conditions in our conditional results are the assumptions \eqref{hyp1}, \eqref{hyp2}, \eqref{hyp4}--\eqref{hyp3} on an outgoing cone inside the black hole. In the two-ended case, these assumptions are \emph{known} to hold for generic asymptotically flat spherically symmetric black hole solutions of \eqref{1.1}--\eqref{5.1} for $q_0=0$ by the work  of Luk--Oh \cite{JonathanStab,JonathanStabExt,twotails} and\footnote{Indeed, \cite{twotails} is a conditional result on the satisfaction of weak decay assumptions on the black hole exterior. These assumptions have been proved to hold, however, in the recent work of Gautam \cite{Gautam}, see the discussion below.} the recent work \cite{Gautam} (see also  Lemma~\ref{propagation.lemma} in Section~\ref{section.global.uncond2}). In particular, imposing the  assumptions \eqref{hyp1}, \eqref{hyp2}, \eqref{hyp4}--\eqref{hyp3} is only \blue{needed} in the case $q_0 \neq 0$ (both in the one and two-ended cases) not covered by \cite{Gautam,JonathanStab,JonathanStabExt,twotails}, which is the only case allowing to study one-ended black holes, as we explained in the introduction. To  provide  unconditional analogues of our results for generic Cauchy data  would require to understand the dynamics in the exterior of the black hole for \eqref{1.1}--\eqref{5.1} in the $q_0 \neq 0$ case, which is an open problem (see, however \cite{Moi2,MoiDejan} for partial progress).
	\paragraph{Strategy of the proof in the one-ended case}
	In the one-ended case, our conditional result -- Theorem~\ref{main.thm.global.i} -- additionally requires the absence of a locally-naked singularity on which $r>0$, i.e., $\mathcal{CH}_{\Gamma} =\emptyset$, recalling the terminology of Theorem~\ref{Kommemi.thm}. In view of Theorem~\ref{breakdown.thm} (which we reprove as part of   Theorem~\ref{main.thm.global.i}), assuming $\mathcal{CH}_{\Gamma} =\emptyset$ is sufficient to obtain a non-empty spacelike singularity $\mathcal{S}$ described by  the estimates of Theorem~\ref{main.thm}.
	
	\paragraph{Strategy of the proof in the two-ended case}
	In the two-ended case, locally-naked singularities do not arise, since $\Gamma=\emptyset$ (recalling the first part of Theorem~\ref{Dafermos.thm}). However, we recall that the two-ended case presents the additional complication that a black hole spacetime may not contain a spacelike singularity, i.e., $\mathcal{S}=\emptyset$, as Theorem~\ref{Dafermos.thm} demonstrates.  Our conditional result, of course, only applies to black holes such that  $\mathcal{S}\neq\emptyset$. In Section~\ref{section.global.uncond2}, we will show that the class of such black holes is non-empty as part of our unconditional  results in the two-ended case (Theorem~\ref{main.thm.2end.ii}), relying on the arguments leading to the proof of Theorem~\ref{breakdown.thm.new}.

	\subsection{A priori characterization of the spacetime boundary}
	
	In what follows, we prove a result of independent interest, which is that the Cauchy horizon $\CH$ (which obeys estimates given by \eqref{hyp1} that are essentially equivalent to mass inflation) cannot be followed by an ingoing light cone $\mathcal{S}_{i^+}$ on which $r$ extends to $0$. We also show that, if the Cauchy horizon breaks down, then $r$ tends to $0$ towards its endpoint. \blue{In what follows, we set initial data on $C_{out} \cup \underline{C}_{in}$, where $C_{out}=\{u_0\} \times [v_0,+\infty)$ and $\underline{C}_{in}=[u_0,u_F] \times \{v_0\}$. We also recall the notation $r_{|\CH}(u) = r_{CH}(u)$.}
	
	\begin{prop}\label{apriori.prop2}
		We  assume that the estimates \eqref{hyp1} on $C_{out}$ are true, and that $\uch<u_F$. Then \begin{equation}\begin{split}
				\lim_{v\rightarrow+\infty} r(\uch,v) = \lim_{u \rightarrow \uch,\ u< \uch} r_{CH}(u)=0.		 	\end{split}
		\end{equation} Moreover, $\mathcal{S}_{i^+}=\{ \uch < u \leq u_F, \underset{v\rightarrow+\infty}{\lim} r(u,v) = 0 \}=\emptyset$ and for all $\uch<u \leq u_F$, there exists $v_{\mathcal{S}}(u)<+\infty$ such that \begin{equation*}
			\lim_{ v \rightarrow v_{\mathcal{S}}(u)} r(u,v) =0.
		\end{equation*}
	\end{prop}
	
	\begin{proof} Note that by Proposition~\ref{apriori.prop1}, $(\uch,v_0) \in \T$. Since $\uch< u_F$, and $\T$ is open, there exists $\epsilon>0$ such that $[\uch,\uch+\ep] \times \{v_0\} \subset \T$ and by the monotonicity of \eqref{RaychV} we deduce that  $([\uch,\uch+\ep] \times [v_0,+\infty))\cap \mathcal{Q}^+\subset \T$ (recalling from Section~\ref{preliminary} that  $\mathcal{Q}^+$ is the notation used for the Penrose diagram).
		
		Since $\uch< u_F$ by definition, then it means that for every $\uch < u< u_F$, there exists $v_{\mathcal{B}}(u) \in \RR \cup \{+\infty\}$ such that (as a straightforward consequence of the extension principle of  \cite{Kommemi}) $$ \lim_{v\rightarrow v_{\mathcal{B}}(u)} r(u,v)=0.$$ 
		
		\noindent Let $V< \underset{u\in[\uch,\uch+\ep]}{\inf}\ v_{\mathcal{B}}(u)$, a sufficiently large constant.	We then integrate $-\rd_u \rd_v (r^2)$ on $\{ u\in [\uch-\ep,\uch+\ep],\ V \leq v\leq v_{\mathcal{B}}(u) \} $ (adopting the convention that  $v_{\mathcal{B}}(u)=+\infty$ for $u\leq \uch$) using \eqref{Radius} and obtain, exploiting the monotonicity of \eqref{RaychV} and the fact that $ \{ u\in [\uch-\ep,\uch+\ep], V \leq v\leq v_{\mathcal{B}}(u) \} \subset \T $ (namely we use  $\int_{v_1}^{v_2} \Omega^2(u,v) dv \leq \frac{\Omega^2}{|\rd_v r|}(u,v_1) r(u,v_1)$):    \begin{equation*}
			r^2(\uch+\ep,V)- r^2(\uch-\ep,V) + r^2_{CH}(\uch-\ep) \leq \int_{\uch-\ep}^{\uch+\ep} \int_{V}^{v_{\mathcal{B}}(u)} \Omega^2(u,v) dv du  \lesssim \ep.\end{equation*}
		Taking $\ep \rightarrow 0$ gives (using the continuity of $u \rightarrow r^2(u,V)$  \blue{for fixed $V$})  $$ \lim_{u \rightarrow \uch,\ u< \uch} r_{CH}(u)=0. $$  Then, using \eqref{r.est.prelim} we obtain $$\lim_{v\rightarrow+\infty} r(\uch,v)=0.$$

		Let us prove that $\mathcal{S}_{i^+}=\{\uch < u \leq u_F, \underset{v\rightarrow+\infty}{\lim} r(u,v) = 0 \}:=(\uch,\usi] = \emptyset$ by contradiction. Assuming that $\mathcal{S}_{i^+}=\{\uch< u \leq u_F, \underset{v\rightarrow+\infty}{\lim} r(u,v) = 0 \} \neq \emptyset$, we revisit the proof of Proposition~\ref{apriori.prop1} to show \eqref{lambda.est.prelim}, \eqref{r.est.prelim}, \eqref{nu.est.prelim} are still  valid on the larger rectangle $ (u,v) \in [u_0, \usi]\times [v_0,+\infty) \subset \T$, and we have $r_{|\blue{\mathcal{S}_{i^+}}}(u)=0$ for all $\uch \leq u \leq \usi$, hence $[r\rd_u r]_{|\blue{\mathcal{S}_{i^+}}}(u)=0$ for all $\uch \leq u \leq \usi$ in a regular $u$-gauge. To choose the $u$-gauge for $\uch< u \leq \usi$, we first choose $v_0$ sufficiently large so that $$ |r\rd_u r|(\uch,v_0) \geq \frac{1}{2},$$ then choose the $u$-gauge to be for all  $\uch< u \leq \usi$  \begin{equation}\label{gauge.new}
			r\rd_u r(u,\blue{v_0})= 	r\rd_u r(\uch,\blue{v_0}).
		\end{equation}	 But, then we note that integrating \eqref{nu.est.prelim} on the ingoing cone $\{v=v_0\}$ under the gauge~\eqref{gauge.new} gives $$ \frac{u-\uch}{2}   \ls e^{1.9 K_- \blue{v_0}} ,$$\blue{where the implicit constants are independent of $v_0$: thus}, as $\blue{v_0}\rightarrow +\infty$, \blue{we get} $u-\uch=0$, which is obviously a contradiction.
	\end{proof}

	\subsection{Conditional applications to one/two-ended black holes}
	
	In this section, we provide the proof of  Theorem~\ref{main.thm.global.i} (one-ended case) and of Theorem~\ref{main.thm.2end.i} (two-ended case). We will also address  Theorem~\ref{inext.thm}.
	
	\subsubsection{Absence of an ingoing collapsed null cone $\mathcal{S}_{i^+}$}
	Consider a one-ended spherically symmetric spacetime in the sense of Section~\ref{preliminary} and under the assumptions of Theorem~\ref{main.thm.global.i}. Then $\mathcal{S}_{i^+} =\emptyset$ follows directly from Proposition~\ref{apriori.prop2}, which proves Statement~\ref{A} of Theorem~\ref{main.thm.global.i}. A similar result is obtained in the two-ended case (Statement~\ref{Atwo} of Theorem~\ref{main.thm.2end.i}).
	
	\subsubsection{New proof of the breakdown of the Cauchy horizon}\label{breakdown.reproof}
	
	Statement~\ref{B} of Theorem~\ref{main.thm.global.i}  follows immediately from Proposition~\ref{apriori.prop1}. Indeed, consider  a one-ended spherically symmetric spacetime in the sense of Section~\ref{preliminary} and denote $\uch$ the $u$-value of the future end-point of $\CH$, and $u_{\Gamma}$ the $u$-value of $b_{\Gamma}$ in the notations of Theorem~\ref{Kommemi.thm}. Since we are assuming \eqref{hyp1}, by Proposition~\ref{apriori.prop1}, there exists $v_0$ such that $\{u\}\times [v_0,+\infty) \subset \mathcal{T}\subset \mathcal{Q}$ for all $u\leq \uch$ so we must have $\uch< u_{\Gamma}$, and therefore  $\mathcal{S}^1_{\Gamma} \cup \mathcal{CH}_{\Gamma} \cup \mathcal{S}^2_{\Gamma} \cup  \mathcal{S}  \neq \emptyset$.

	\subsubsection{Non-emptiness of $\mathcal{S}$ and Kasner estimates}

	We turn to the proof of Statements~\ref{C.stat}--\ref{D.stat} of Theorem~\ref{main.thm.global.i}, for which we assume $\mathcal{CH}_{\Gamma} = \emptyset$. By Statement~\ref{A} of Theorem~\ref{main.thm.global.i}, $\mathcal{S}_{i^+}=\emptyset $. But $\mathcal{S} \cup \mathcal{S}_1^{\Gamma} \cup \mathcal{S}_2^{\Gamma} \neq \emptyset$ by Statement~\ref{B} of Theorem~\ref{main.thm.global.i} (breakdown of the Cauchy horizon), and there exists  $v_T\in \RR$ such that $[\uch-\epsilon,\uch) \times [v_T,+\infty)\subset \T\subset \mathcal{Q}$ and $C>0$, $p>0$ such that for all $(u,v) \in~ [\uch-~\epsilon,\uch) \times~ [v_T,+\infty)$  \begin{equation}\label{lower}
		r(u,v) \geq C v^{-p}>0.
	\end{equation}

	Next, we show by contradiction that $\mathcal{S}\neq \emptyset$. Assume not, i.e., that $\mathcal{S}= \emptyset$.  Since $\mathcal{S} \cup \mathcal{S}_1^{\Gamma} \cup \mathcal{S}_2^{\Gamma} \neq \emptyset$, then it must mean that $\mathcal{S}_1^{\Gamma} \cup \mathcal{S}_2^{\Gamma} \neq \emptyset$. Then the endpoint of $\mathcal{S}_1^{\Gamma}\cup \mathcal{S}_2^{\Gamma}$  intersects the endpoint of $\CH$, which contradicts  \eqref{lower}. Therefore,  $\mathcal{S}\neq \emptyset$.

	This means that the endpoint of $\mathcal{S}$ intersects the endpoint of $\CH$. Therefore, since $[\uch-\epsilon,\uch) \times [v_T,+\infty)\subset \T$, one can construct a $\Cin= [u_0,u_F) \times \{ v_T\}\subset \T$ such that $r(u,v_T) \rightarrow 0 $ as $u\rightarrow u_F$ and apply Statement~\ref{main.thm.III} of Theorem~\ref{main.thm}. This means that $\mathcal{S} \cap \{ u \leq \uch+\epsilon\}$  is spacelike for small enough $\epsilon>0$ and described by the Kasner asymptotics of Theorem~\ref{main.thm}. The proof of Statement~\ref{Btwo} of Theorem~\ref{main.thm.2end.i} is obtained analogously. 
	
	\subsubsection{Mass inflation}
	
	Finally, the proof of Statement~\ref{E.stat} of  Theorem~\ref{main.thm.global.i} follows from the additional  assumption that there exists $u_s\in \RR$ such that \eqref{hyp1}, \eqref{hyp2}, \eqref{hyp4}-\eqref{hyp3} hold on $\{u_0\} \times [v_0(u_0),+\infty)$ for all $\uH < u_0 \leq u_s$,  which immediately leads to mass inflation \blue{on} this interval, i.e., for all $\blue{\uH<}u\leq u_s$: $$ \lim_{v \rightarrow+\infty} \varpi(u,v)= \lim_{v \rightarrow+\infty} \mathfrak{m}(u,v) =+\infty.$$
	Then, by  the propagation of the Hawking mass blow-up to the future (see, e.g., \cite{breakdown}, Lemma 4.9), this extends to the whole Cauchy horizon $\CH$, i.e., for all $\uH<u\leq \uch$
	\begin{equation}\label{mass.inflation}
		\lim_{v \rightarrow+\infty} \varpi(u,v)= \lim_{v \rightarrow+\infty} \mathfrak{m}(u,v) =+\infty.
	\end{equation}

	\noindent The proof of Statement~\ref{Ctwo} of Theorem~\ref{main.thm.2end.i} is obtained analogously. This concludes both the proofs of   Theorem~\ref{main.thm.global.i} and Theorem~\ref{main.thm.2end.i}. 
	\subsubsection{Proof of Theorem~\ref{inext.thm}}
	
	We conclude this section with the proof of   Theorem~\ref{inext.thm}, i.e., the fact that the \blue{above spacetimes} are $C^2$-future-inextendible.  Note indeed  that the Kretschmann scalar $K= R_{\alpha \beta \gamma \delta} R^{\alpha \beta \gamma \delta}$ is infinite on $\mathcal{S}=\{r=0\}$, hence there can be no $C^2$-extension through $\mathcal{S}$ \cite{Kommemi,JonathanStab}. Furthermore,  by mass inflation \eqref{mass.inflation},  there can be no $C^2$-extension either through $\CH$ (see \cite{Moi4}), therefore the spacetime is $C^2$-future-inextendible. Thus,  Theorem~\ref{inext.thm} is proved.

	\section{Unconditional constructions of one-ended spacetimes}\label{section.global.uncond1}
	
	In this section, we address  the proof of the unconditional results in the one-ended case (Theorem~\ref{main.thm.global.ii}), building up on some of the soft arguments of Section~\ref{section.global.cond} while providing a new strategy  to construct examples of spacetimes satisfying  \eqref{hyp1}, \eqref{hyp2}, \eqref{hyp4}--\eqref{hyp3} in Theorem~\ref{main.thm}. This is, by far, the most involved section of the manuscript.

	In the one-ended case, it is not possible to study \eqref{1.1}--\eqref{5.1} in spherical symmetry with $q_0=0$  (except if $F\not \equiv 0$), which is the only case where  \eqref{hyp1}, \eqref{hyp2}, \eqref{hyp4}--\eqref{hyp3} are known to hold: this well-known phenomenon is due to the presence a regular center $\Gamma$, which is incompatible with a non-trivial $F$ in the $q_0=0$ case (see \cite{review}[Section 5] for a   discussion). Therefore, we have to study charged scalar fields, which have more complicated dynamics than uncharged ones (see, e.g., the discussion in \cite{review}). To circumvent this difficulty, we make use a gluing/scattering procedure in several steps to obtain an asymptotically flat solution of \eqref{1.1}--\eqref{5.1} with $q_0\neq 0$, for which we show that \eqref{hyp1}, \eqref{hyp2}, \eqref{hyp4}--\eqref{hyp3} are satisfied on an outgoing cone inside the black hole, which we summarize below. A more detailed outline of the construction is also  available in Section~\ref{outline.section}.
	
	\begin{itemize}
		
		\item  As  a first step in Section~\ref{uncharged.section}, we construct uncharged black hole solutions of  \eqref{1.1}--\eqref{5.1} with $q_0 \neq 0$, but $F\equiv 0$ and a real-valued scalar field $\phi$ (Theorem~\ref{uncharged.thm}). Near the center, the singularity is spacelike and modeled after a FLRW metric (Corollary~\ref{uncharged.cor}), while near $i^+$ (including part of the trapped region) and $\mathcal{I}^+$, the metric is exactly that of Schwarzschild spacetime.
		The initial data are regular, localized, posed on $\RR^3$, free of trapped surfaces and asymptotically flat, non-intersecting with the black hole region. The construction relies on a new spacelike-characteristic gluing strategy in the uncharged case (Theorem~\ref{uncharged.gluing.thm}).  This step also allows to complete the proof of Theorem~\ref{OS.thm.intro} in the uncharged case $(q=0)$.
		
		\item In Section~\ref{charging.section}, we want to ``charge'' the previously constructed spacetime (Theorem~\ref{charging.thm}) and replace the Schwarzschild event horizon by that of a Reissner--Nordstr\"{o}m black hole. To do this, we glue the  trapped Schwarzschild sphere  $(u,v) = (\blue{\uH+}\delta,v_{Q=0})$ to a Reissner--Nordstr\"{o}m trapped sphere $(u,v) = (\blue{\uH+\delta},v_0)$, \blue{where $v_0 > v_{Q=0}$}, invoking a new characteristic gluing \blue{result} inspired by the Kehle--Unger strategy \cite{KehleUnger}  (Theorem~\ref{charged.gluing.thm}). 
		
		As in the previous step, the singularity is FLRW spacelike near the center  (Corollary~\ref{charging.cor}), while near $i^+$ (including part of the trapped region and a Cauchy horizon $\CH$) and $\mathcal{I}^+$, the metric is exactly (sub-extremal) Reissner--Nordstr\"{o}m. \blue{At the end of this step, the proof of Theorem~\ref{OS.thm.intro} is achieved}.

		\item We then need to generalize the charged black hole construction to an arbitrary, dynamical event horizon that converges to  Reissner--Nordstr\"{o}m but is not exactly  Reissner--Nordstr\"{o}m  in  Section~\ref{right.section} (Theorem~\ref{EH.AF.thm} and Corollary~\ref{EH.AF.cor}). Contrary to the   exact  Reissner--Nordstr\"{o}m case, the construction of the asymptotically flat end is then non-trivial and involves quantitative estimates in the black hole exterior.

		First, we  choose our event horizon late-time behavior by picking $\phi_{|\mathcal{H}^+}$ to be a given  (smooth) function $\Phi_H(v)$  (in Eddington--Finkelstein $v$ from the gauge choice \eqref{gauge.EH.v}) such that  for some $s>\frac{3}{2}$ \begin{equation}\label{decay.s}
			|\Phi_H|(v),\  |\rd_v\Phi_H|(v) \ls v^{-s}.
		\end{equation} 
		We then solve for \eqref{1.1}--\eqref{5.1} with initial data on the bicharacteristic hypersurface (in the gauge \eqref{gauge.EH.U}) \begin{equation}\begin{split}
				& \Cin\cup [\mathcal{H}^+\cap [v_0,+\infty)],\\ & \Cin= [\blue{\uH},\blue{\uH}+\delta]\times \{v=v_0\},\\ & \mathcal{H}^+\cap [v_0,+\infty)=\{u=\blue{\uH}\}\times [v_0,+\infty).
			\end{split}
		\end{equation} where we want the ingoing data on $\Cin$ to be exactly Reissner--Nordstr\"{o}m. For compatibility at the intersection sphere $(u,v)=(\blue{\uH},v_0)$, we want to cut $\phi_{|\mathcal{H}^+}$ off: let $\chi_0$, a smooth, compactly supported function such that $\chi(x) =1$ for $0 \leq x\leq 1$, and $\chi(x)=0$ for $x\geq 2$: we impose the initial data  \begin{equation}\label{data.in}
			\phi(u,v_0) =0, \text{ for } \blue{\uH}\leq u\leq\blue{\uH}+\delta,
		\end{equation}
		\begin{equation}\label{data.EH.intro}
			\phi(\blue{\uH},v) = \Phi_H(v) \left[1-\chi(v-v_0)\right], \text{ for } v\geq v_0.
		\end{equation}
		Assuming $\delta$ small enough, Theorem~\ref{CH.thm.SS} allows to solve  \eqref{1.1}--\eqref{5.1} in the spacetime rectangle $(u,v)\in [\blue{\uH},\blue{\uH}+\delta] \times [v_0,+\infty)$ and shows the presence of a non-empty Cauchy horizon $\CH$. 
		
		We can then glue the resulting spacetime region to that of Theorem~\ref{uncharged.thm} and obtain a spacetime ``to the left of its event horizon'' (Proposition~\ref{EH.AF.1st.prop}).
		
		To complete the proof of Theorem~\ref{main.thm.global.ii}, we still need to construct the right-part of the Penrose diagram\blue{, i.e., the asymptotically flat end}. 	To construct the asymptotically flat end, we use a ``scattering'' argument in the black hole exterior starting from event horizon outgoing data and regular ingoing data and propagating from left to right, taking advantage of  spherical symmetry (Propositions~\ref{RS.prop}--\ref{nullinf.prop}, Corollary~\ref{nullinf.cor}, Lemma~\ref{i0.Delta.lemma} and Proposition~\ref{i0.prop}). This step is strongly inspired from the author's previous work \cite{Moi2,Soffer70}  and takes advantage of the smallness of the black hole final charge.

		\item In Section~\ref{scattering.section}, we turn to the ``middle-part of the Penrose diagram''; \blue{namely, the region in the black hole interior near the event horizon.} This spacetime region has already been constructed in the previous steps, but we seek precise information of the metric there (Theorem~\ref{nonlinearscat.thm} and Corollary~\ref{nonlinearscat.cor}). The goal is to	 show  that  the assumptions \eqref{hyp1}, \eqref{hyp2}, \eqref{hyp4}--\eqref{hyp3} of Theorem~\ref{main.thm} are satisfied inside the black hole. The proof relies on   estimates from \cite{Moi} (for \eqref{hyp1}, \eqref{hyp2}), and a refinement of the scattering estimates in \cite{MoiChristoph} for  \eqref{hyp4}--\eqref{hyp3}.
		
		\item   In Section~\ref{final.section},  we complete the construction by applying Theorem~\ref{main.thm} to obtain the coexistence of null and spacelike singularities in the interior of our one-ended black hole. In addition to the null boundary $\CH$ (Cauchy horizon), we note that the only remaining boundary component is $\mathcal{S}=\{r=0\}$ (in the terminology of Theorem~\ref{Kommemi.thm}). While $\mathcal{S}=\{r=0\}$ need not be entirely spacelike (it might have null segments), it has by construction two distinct spacelike subsets which are both tidally contractive: \begin{itemize}
			\item An isotropic FLRW-like spacelike singularity near the center 
			$\Gamma$, with  $(\frac{1}{3},\frac{1}{3},\frac{1}{3})$ Kasner exponents.

			\item A  spacelike singularity with positive Kasner exponents degenerating to $(1,0,0)$ as $v\rightarrow +\infty$, as a direct application of Theorem~\ref{main.thm}.
		\end{itemize}
	\end{itemize}  	
	
	\subsection{The detailed outline of the construction}\label{outline.section}
	
	Now we start  the proof of  Theorem~\ref{main.thm.global.ii}. We describe a first construction  of a spherically symmetric spacetime $(\mathcal{M},g,\phi,Q)$ solving \eqref{1.1}--\eqref{5.1}  satisfying the following properties: for any arbitrarily large $k\geq 2$: \begin{enumerate}
		\item\label{ST1} $(\mathcal{M},g,\phi,Q)$ is the MGHD of one-ended asymptotically flat $C^{k}$ spherically symmetric initial data on some spacelike  $\Sigma\approx \RR^3$, with no  trapped or anti-trapped surfaces. $b_{\Gamma} = (u_{\Gamma},v_{\Gamma})$ denotes the endpoint of  $\Gamma\neq \emptyset$.
		\item \label{ST2} 
		There exists a  $v_{Q=0}>v_{\Gamma}$ such that  $F\equiv 0$ on $\mathcal{M} \cap \{v \leq v_{Q=0}\}$ (Einstein-scalar-field solution). 
		\item \label{ST3} There exists  $v_{\Gamma}<\blue{v_{L}}<v_{Q=0}$ such that  $g$  is spatially-homogeneous on $\mathcal{M} \cap \{v \leq \blue{v_{L}}\}$ (FLRW metric). Moreover, denoting $\mathcal{B}$, the future boundary of $\mathcal{M}$, $\mathcal{B} \cap \{v_{\Gamma}\leq v \leq \blue{v_{L}}\}$ is a spacelike singularity $\blue{\mathcal{S}_L} \subset \mathcal{S}  \subset \{r=0\}$.

		\item \label{ST5} The following estimates are satisfied on the event horizon  $\{u=u_{\mathcal{H}^+},\  v \geq \blue{v_{\Gamma}(\uH)}\}$, in the  gauge choices \eqref{gauge.EH.v}, \eqref{A.gauge}, for some $s>\frac{3}{2}$ and for a constant $\omer \in \RR-\{0\}$ to be fixed later, depending on the black hole parameters and $\delta>0$\blue{: for $v$ sufficiently large} \begin{equation}\label{EH.bound2}\begin{split}
				&	| \phi|(	u_{\mathcal{H}^+},v) \approx v^{-s},\ |\rd_v( \phi e^{i q_0 \omer v})|(	u_{\mathcal{H}^+},v) \approx v^{-s-1},\ |\rd_v^{2}( \phi e^{i q_0 \omer v})|(	u_{\mathcal{H}^+},v)\ls v^{-s-2},\\ &  |\rd_v^{3}( \phi e^{i q_0 \omer v})|(	u_{\mathcal{H}^+},v)\ls v^{-s-3},\ |\Im( \phi e^{i q_0 \omer v})|(	u_{\mathcal{H}^+},v)\ls v^{-s-\delta}.\end{split}
		\end{equation} 
		
		\item \label{ST6} The event horizon is  transversely $C^k$-regular, namely \blue{for all $v\geq v_{\Gamma}(\uH)$,} there exist $\ep>0$, $D(v)>0$  such that for all $u\in [u_{\mathcal{H}^+}-\ep,u_{\mathcal{H}^+}+\ep]$, $1\leq i \leq k$: \begin{equation}\label{EH.boundU}
			| \phi|(	u,v),\  |D_u^{i} \phi|(	u,v) \lesssim  D(v),
		\end{equation} and moreover $\mathcal{H}^+$ is located in the strict causal future of $\Sigma$.
		
	\end{enumerate}

	From Theorem~\ref{CH.thm.SS}, we know that $\CH \neq \emptyset$ and that \eqref{hyp2} is satisfied. \eqref{hyp1} is partially satisfied (however, we have not yet proven a lower bound on $|\rd_v r|$ or $|\rd_u r|$). We will come back to arranging that \eqref{hyp1}, \eqref{hyp4}--\eqref{hyp3} are satisfied later using the scattering theory in the black hole interior.

	Once  \eqref{hyp1}, \eqref{hyp2}, \eqref{hyp4}--\eqref{hyp3}  are shown to hold, and after an application of  Theorem~\ref{main.thm.global.i}, the proof of Theorem~\ref{main.thm.global.ii}  reduces to the  construction of a spacetime $(\mathcal{M},g,\phi,Q)$ solving \eqref{1.1}--\eqref{5.1} with $q_0\neq 0$,  and  satisfying \ref{ST1}--\blue{\ref{ST6}}. To carry out this construction, we proceed as follows (see Figure~\ref{fig:final_bif}). \begin{enumerate}[A.]
		\item\label{A.} Start with a spatially-homogeneous  solution of \eqref{1.1}--\eqref{5.1} with $F\equiv 0$ (but $q_0\neq 0$), representing a FLRW metric with $\RR^3$ topology (Proposition~\ref{FRLW.prop}). Show that a spacelike singularity $\blue{\mathcal{S}_L}\subset\{r=0\}$  forms in finite time. Truncate this solution $\mathcal{M}$ to only keep  the region $\mathcal{M}\cap \{ v \leq \blue{v_L}\}$ for some $\blue{v_L}> v_{\Gamma}$.  Then, denoting $\mathcal{B}$ the terminal boundary of $\mathcal{M}$, $\mathcal{B}\cap \{ v \leq \blue{v_L}\}$ is a spacelike singularity $\blue{\mathcal{S}_L}$ and $\mathcal{M}\cap \{ v \leq \blue{v_L}\}$ contains a spacelike hypersurface $\blue{\Sigma_L}$ with  no trapped surface. This truncated solution contains a non-empty regular center $\Gamma$ with $\blue{\Sigma_L} \cap \Gamma \neq \emptyset$\blue{; we denote $\mathbb{S}_L= (u_L,v_L)= \blue{\Sigma_L} \cap \{v= v_L\} \in \R$}.
		
		\item\label{B.} We use the new uncharged spacelike-characteristic gluing result (Theorem~\ref{uncharged.gluing.thm}) to glue \blue{the} FLRW sphere 
		$\blue{\mathbb{S}_{L}}= (u_H,\blue{v_L})$ to  a Schwarzschild trapped sphere  $\mathbb{S}_S^{\T}=(u_S^{\T},v_S^{\T})$. The first step is to glue $\blue{\mathbb{S}_{L}}$ to	
		a sphere $\mathbb{S}_{\mathcal{A}}$ belonging to the apparent horizon in a spacelike-fashion. We then glue characteristically $\mathbb{S}_{\mathcal{A}}=(u_A,v_A)$ to a Schwarzschild trapped sphere  $\mathbb{S}_S^{\T}=(u_S^{\T},v_S^{\T})$. This construction can be completed to give an uncharged black hole with a Schwarzschild event horizon and a Schwarzschild asymptotically flat end (Theorem~\ref{uncharged.thm}/Corollary~\ref{uncharged.cor}). By Theorem~\ref{Kommemi.thm}, the terminal boundary only consists of a spacelike singularity $\mathcal{S}\subset \{r=0\}$.

		\item \label{C.}  Let $\teal{v_0} > v_S^{\T}$.  We will glue the Schwarzschild trapped sphere  $\mathbb{S}_S^{\T}=(u_S^{\T},v_S^{\T})$ to a Reissner--Nordstr\"{o}m trapped sphere  $\mathbb{S}_{RN}^{\T}=(u_{S}^{\T}=u_{RN}^{\T},\blue{v_{0}})$ via a new characteristic gluing result (Theorem~\ref{charged.gluing.thm})  inspired from \cite{KehleUnger}. As in \cite{KehleUnger}, this step imposes restrictions on the Reissner--Nordstr\"{o}m parameters \blue{$(M,e)$} (small charge \blue{$e$ in our case}). This construction can be completed to produce a charged black hole with sub-extremal Reissner--Nordstr\"{o}m event horizon, a Cauchy horizon $\CH$, a spacelike singularity $\mathcal{S}$ and a Reissner--Nordstr\"{o}m asymptotically flat end (Theorem~\ref{charging.thm}/Corollary~\ref{charging.cor}).

		\item \label{D.} We trivially glue the Reissner--Nordstr\"{o}m trapped sphere  $\mathbb{S}_{RN}^{\T}=(u_{RN}^{\T},v_{0})$ towards a Reissner--Nordstr\"{o}m regular sphere  $\blue{\mathbb{S}_{RN}^{\R}}=(u_{\mathcal{H}^+},v_{0})$ in the ingoing past direction\blue{, with $u_{\mathcal{H}^+}<u_{RN}^{\T}$}. Then, we impose event horizon   data as in  \eqref{data.EH.intro} satisfying \eqref{EH.bound2} on the outgoing cone $\mathcal{H}^+=\{u=u_{\mathcal{H}^+}\}\times [v_{0},+\infty)$, which will turn out to be the black hole's event horizon once the construction is completed.

		\item \label{E.} We pose regular ingoing data on  $[u_{\mathcal{H}^+}-\ep,u_{\mathcal{H}^+}]\times \{\blue{v_0}\}$  smoothly connecting to \eqref{data.in} at the sphere $(u,v) = (\blue{\uH},v_0)$, and we solve sideways (taking advantage of spherical symmetry) to obtain a solution in the spacetime rectangle  $[u_{\mathcal{H}^+}-\epsilon,u_{\mathcal{H}^+}]\times [v_0,+\infty) $ comprising a portion of null infinity $\mathcal{I}^+\cap \{ u_{\mathcal{H}^+}-\epsilon \leq u \leq u_{\mathcal{H}^+},\ v=+\infty\}$. For this step, we  use the smallness of $\ep$ and of the black hole charge $e$.
		
		\item \label{F.} Impose compactly supported (or more generally decaying at a rate $|u|^{-q}$ for $q\gg1$ as $u\rightarrow -\infty$) data on  $\mathcal{I}^+\cap \{ -\infty< u \leq u_{\mathcal{H}^+}-\ep\}$ and combined with the outgoing data on  $\{u=u_{\mathcal{H}^+}-\ep\} \times[v_0,+\infty)$, and solve backwards up to the asymptotically flat end $i^0$. We have  constructed the right part of the Penrose diagram.
		
		\item\label{G.} We invoke Theorem~\ref{CH.thm.SS} to show the existence of $\CH$, a null Cauchy horizon in the black hole interior and estimates in a spacetime rectangle  $[u_{\mathcal{H}^+},u_{\mathcal{H}^+}+\ep]\times [v_{s},+\infty) $, where $v_s$ is a large constant.
		
		\item \label{H.} We use a refinement of the nonlinear scattering arguments based on \cite{MoiChristoph} to prove that \eqref{hyp1}, \eqref{hyp2}, \eqref{hyp4}--\eqref{hyp3} are satisfied inside the black hole for profiles satisfying  \eqref{EH.bound2} (a much more refined assumption than \eqref{decay.s}).

		\item \label{I.} We note the presence of an ingoing trapped cone $\Cin\subset \{ v=v_s\}$ towards which $\{r=0\}$, which is part of the assumptions of Theorem~\ref{main.thm}.

	\end{enumerate}

	\noindent Carrying out these  Steps~\ref{A.}--\ref{I.} will conclude the proof as an application of  Theorem~\ref{main.thm.global.i}. 
	
	We note that Steps~\ref{B.}--\ref{G.} provides a general result (Theorem~\ref{EH.AF.thm}) allowing to glue any uncharged spacetime region (on the left) to an asymptotically flat  charged black hole converging to Reissner--Nordstr\"{o}m of small charge but   with  arbitrary event horizon late-time tails prescribed by \eqref{data.EH.intro} (modulo the mild decay assumption \eqref{decay.s}, which must be satisfied with $s>\frac{3}{2}$). It is precisely this freedom which allows to satisfy the assumptions of Theorem~\ref{main.thm}, after we specify a well-chosen profile $\Phi_H(v)$ from \eqref{data.EH.intro} such that \eqref{EH.bound2} is satisfied.

	\subsection{Construction of uncharged one-ended black hole spacetimes}\label{uncharged.section}
	
	In this section, we consider real-valued solutions of \eqref{1.1}--\eqref{5.1} with $q_0\neq 0$, but $F\equiv 0$, $A\equiv 0$ (Einstein-scalar-field model). We start with the following uncharged gluing result, allowing to glue any spherically symmetric spacetime to a Schwarzschild black hole interior. 
	We recall the definition of a first singularity from Definition~\ref{first.sing.def}. 
	
	\begin{thm}\label{uncharged.thm} Let $k \in \mathbb{N}$, $k\geq 2$ and  $(\mathcal{M}_L,g_L,\phi_L)$, a subset of the MGHD of \blue{$C^k$}
		spherically symmetric asymptotically flat initial data on a hypersurface $\Sigma_L$ with one end for \eqref{1.1}--\eqref{5.1}  containing no anti-trapped spheres  and no trapped spheres  and such that $b_{\Gamma}$ is a first singularity.

		Then, there exists  $C^k$ solutions $(\mathcal{M},g,\phi)$ of \eqref{1.1}--\eqref{5.1} with $F\equiv 0$ with the following properties: \begin{itemize}
			\item $(\mathcal{M},g,\phi)$ is the MGHD of smooth
			spherically symmetric asymptotically flat initial data on a spacelike hypersurface $\magenta{\Sigma_0}$ with one end for \eqref{1.1}--\eqref{5.1}  containing no anti-trapped spheres and no trapped spheres.
			\item The black hole region of $(\mathcal{M},g,\phi)$  is non-empty with an event horizon $\mathcal{H}^+$ and, moreover, $\mathcal{H}^+$ does not intersect $\magenta{\Sigma_0}$, i.e., it is located in the strict causal future of $\magenta{\Sigma_0}$.

			\item  There exists an incoming null  cone $\underline{C}_S$ and an outgoing null  $C_{S}$ intersecting at a trapped sphere $\mathbb{S}_S^{\T}$\blue{--which is  the future endpoint of  $\underline{C}_S$ and the past endpoint of $C_S$--}such that $\mathcal{M} \cap J^{+}(\underline{C}_{S}) \cap J^{-}(C_{S}) $ is isometric to a Schwarzschild metric with some mass $M>0$. In particular,   $\mathbb{S}_S^{\T}$ is a Schwarzschild trapped sphere, and $\mathcal{H}^+ \cap J^{+}(\underline{C}_{S})$ is coincides with a \blue{future} affine complete portion of  a Schwarzschild event horizon. 
			\item There exists an incoming null  cone $\underline{C}_{v_L}$  such that $\mathcal{M} \cap J^{-}(\underline{C}_{v_L})$ coincides with $\mathcal{M}_L \cap J^{-}(\underline{C}_{v_L})$.  Moreover, $\underline{C}_{v_L}$ can be chosen to be in the complement of the causal past of $b_{\Gamma}$.
		\end{itemize}
	\end{thm}
	
	\begin{rmk}\label{mass.prescribed.rmk}
		In Theorem~\ref{uncharged.thm}, as stated, we are not allowed to fix the final black hole mass $M>0$ (although the proof reveals a finite permissible range). This is only due to the construction of an initial hypersurface $\magenta{\Sigma_0}$ that does not intersect the event horizon. Should this condition be relaxed to  merely obtaining an initial hypersurface $\magenta{\Sigma_0}$ free of trapped or anti-trapped surfaces, one can then choose any mass above a minimal value determined by  $(\mathcal{M}_L,g_L,\phi_L)$, as the proof of Theorem~\ref{uncharged.gluing.thm} shows.
	\end{rmk} \blue{Theorem~\ref{uncharged.thm} will be proven in Section~\ref{global.uncharged.subsec}.}
	Next, we apply Theorem~\ref{uncharged.thm} to the case where $(\mathcal{M}_L,g_L,\phi_L)$ is a FLRW spacetime near the center $\Gamma$ (Corollary~\ref{uncharged.cor}). This results in the construction of scalar field analogues of the celebrated Oppenheimer--Snyder solution corresponding to Theorem~\ref{OS.thm.intro} in the case $q=0$.
	
	The uncharged gluing procedure in Theorem~\ref{uncharged.thm}, and its application to the construction of a class of spacetimes in Corollary~\ref{uncharged.cor}    are depicted in Figure~\ref{fig:unchargedgluing}.
	
	\begin{cor}\label{uncharged.cor} Let $k \in \mathbb{N}$, $k\geq 2$.
		There exists $C^k$  solutions $(\mathcal{M},g,\phi)$ of \eqref{1.1}--\eqref{5.1} with $F\equiv 0$ with the following properties: \begin{itemize}
			\item $(\mathcal{M},g,\phi)$ is the MGHD of smooth
			spherically symmetric asymptotically flat initial data on a spacelike hypersurface $\magenta{\Sigma_0}$ with one end for \eqref{1.1}--\eqref{5.1}  containing no anti-trapped spheres  and no trapped spheres.
			
			\item The black hole region of $(\mathcal{M},g,\phi)$  is non-empty with an event horizon $\mathcal{H}^+$ and, moreover, $\mathcal{H}^+$ does not intersect $\magenta{\Sigma_0}$, i.e., it is located in the strict causal future of $\magenta{\Sigma_0}$.
			
			\item The MGHD terminal boundary of   $(\mathcal{M},g,\phi)$ is \begin{equation}
				\mathcal{S}=\{r=0\},
			\end{equation} a spacelike singularity.
			
			\item  There exists an incoming null  cone $\underline{C}_S$ and an outgoing null  $C_{S}^{\T}$ intersecting at a trapped sphere $\blue{\mathbb{S}_S^{\T}}$\blue{--which is  the future endpoint of  $\underline{C}_S$ and the past endpoint of $C_S$--}such that $\mathcal{M} \cap J^{+}(\underline{C}_{S}) \cap J^{-}(C_{S}) $ is isometric to a Schwarzschild metric with some mass $M>0$. In particular,   $\mathbb{S}_S^{\T}$ is a Schwarzschild trapped sphere and $\mathcal{H}^+ \cap J^{+}(\underline{C}_{S})$  coincides with a future affine complete portion of  a Schwarzschild event horizon. 
			
			\item There exists an incoming null  cone $\underline{C}_{v_L}$  in the complement of the causal past of $b_{\Gamma}$ such that $\mathcal{M} \cap J^{-}(\underline{C}_{v_L})$ is spatially homogeneous. Moreover, $\mathcal{S}_L:=  \mathcal{S}\cap J^{-}(\underline{C}_{v_L})$ is spacelike and coincides with the \blue{terminal} singularity \blue{$\{T=T_S\}$} of a FLRW metric with $\RR^3$ topology.
			
		\end{itemize}
	\end{cor}

	\subsubsection{FLRW spacetimes}
	In this section, we consider so-called collapsing FLRW metrics $g$ on $ [0,T_{S}) \times \RR^3$  of the following form  \begin{equation}\label{FRLW0}
		g= -dt^2 + a^2(t) \left( d\rho^2 + \rho^2 d\sigma_{\mathbb{S}^2}\right).
	\end{equation} 
	\begin{rmk}
		Note that \eqref{FRLW0} is not asymptotically flat (since, precisely, it is spatially-homogeneous with non-trivial $a(t)$). However, it will later be truncated to obtain an asymptotically flat end, while retaining the region near the center $\Gamma$, notably a subset of its spacelike singularity $\{t=T_S\} \times \RR^3$.
	\end{rmk}
	We look for solutions of \eqref{1.1}--\eqref{5.1}  with $F\equiv 0$, which translates into the following system of ODEs  \begin{equation}\label{FRLW1}
		\ddot{a}(t) + 3 (\dot{\phi})^2 a(t)=0,
	\end{equation}\begin{equation} \label{FRLW2}
		a^3(t) \dot{\phi}(t) = a_0^3 \dot{\phi_0},
	\end{equation}where $\dot{f}$ denotes $\frac{df(t)}{dt}$ and $f_0$ denotes $f(t=0)$.
	
	Viewing  $(g,\phi)$ as a spherically symmetric solution of \eqref{1.1}--\eqref{5.1}, we can define its area-radius function $r(t,\rho) = a(t) \rho$, and $\mathcal{A}$ the apparent horizon, $\mathcal{R}$ the regular region and $\mathcal{T}$ the trapped region as in Section~\ref{preliminary}.
	\begin{prop}\label{FRLW.prop}
		Let $(g,\phi)$ be a solution of  \eqref{1.1}--\eqref{5.1}  with $F\equiv 0$ on $ [0,T_{S}) \times \RR^3$   of the form \eqref{FRLW0} with initial data \begin{equation}\begin{split}
				a(t=0)= a_0>0,\ \dot{a}(t=0)<0,\  \dot{\phi_0} \neq 0.
			\end{split}
		\end{equation} Then, $0<T_S<\infty$, and $\mathcal{S}=\{T=T_S\}$ can be attached as a future spacetime boundary, a spacelike singularity at which $r=0$. Moreover, $(g,\phi)$ takes the spherically symmetric form \eqref{gdef}, where \begin{equation}\begin{split}
				&u(t,\rho) = A(t) - \rho,\ v(t,\rho) = A(t) + \rho,\ A(t) = \int_0^{t} a^{-1}(t') dt',\\ & \rd_u r = \nu_H(t,\rho):=\frac{a(t)}{2} [\dot{a}(t)\rho-1],\ \rd_v r = \lambda_H(t,\rho):=\frac{a(t)}{2} [\dot{a}(t)\rho+1],\\ & \Omega^2=\Omega^2_H(t,\rho):= 4a^2(t),\\ & r(t,\rho)= a(t)\rho. \end{split}
		\end{equation}  
		
		Finally, the following estimates hold  as $t\rightarrow T_S$, defining $a_{S}=  a_0  \cdot \left[\frac{3\sqrt{3}}{\sqrt{2}} |\dot{\phi}_0| \right]^{\frac{1}{3}}$,\ $\phi_S= \frac{-sign(\dot{\phi}_0)\sqrt{2}}{3\sqrt{3}}$: \begin{equation}\label{FRLW.est}
			a(t) \sim a_{S} \cdot  (T_S-t)^{\frac{1}{3}},\ \dot{\phi}(t) \sim  \frac{-\phi_S}{ T_S-t},\  \phi(t) \sim  \phi_{S}\ln(T_S-t).
		\end{equation}
		Denoting $\rho_{\mathcal{A}}(t)>0$ such that  $(\rho_{\mathcal{A}}(t),t) \in \mathcal{A}$, and $\rho_{v_0}(t)$ such that $v(t,\rho_{v_0}(t))=v_0$ also gives the following asymptotics as $t\rightarrow T_S$: \begin{equation}\label{FRLW.A.est}\begin{split}
				& 	\rho_{\mathcal{A}}(t) \sim   \frac{3}{a_S} (T_S-t)^{\frac{2}{3}} ,\\ &  	\rho_{v_0}(t)-v_0 \sim \frac{3}{2a_S} (T_S-t)^{\frac{2}{3}}. \end{split}.\end{equation}
	\end{prop}
	\begin{proof}

		We introduce 
		$A(t) = \int_0^{t} a^{-1}(t') dt'$ and the null coordinates \begin{equation*}
			u(t,\rho) = A(t)-\rho,\ v(t,\rho) = A(t)+\rho
		\end{equation*} noting that \begin{equation*}
			du = a^{-1}(t) dt - d\rho,\   dv = a^{-1}(t) dt + d\rho,\ \rd_u = \frac{a(t)}{2} \rd_t - \frac{1}{2} \rd_{\rho},\ \rd_v = \frac{a(t)}{2} \rd_t + \frac{1}{2} \rd_{\rho}.
		\end{equation*} and one can write the metric as \begin{equation*}
			g = a^2(t) \left(- du dv + \frac{[v-u]^2}{4}  d\sigma_{\mathbb{S}^2}\right)
		\end{equation*} so that the null lapse is $$\Omega^2= 4a^2(t)>0$$ and the area radius is $$ r=\frac{ a(t)  [v-u]}{2} = a(t) \rho$$
		
		\noindent	Note that \begin{equation*}
			\rd_u  r = \frac{a(t)}{2} [ \dot{a}(t) \rho-1]<0,\  	\rd_v  r = \frac{a(t)}{2} [ \dot{a}(t) \rho+1]
		\end{equation*}
		
		\noindent	So based on this, the apparent horizon $\mathcal{A}=\{\rd_v r= 0\}$ is necessarily spacelike, with \begin{equation*}
			\mathcal{A}=\{ \rho= - [\dot{a}(t)]^{-1}\}.
		\end{equation*}
		We also have  \begin{equation*}
			\Gamma=\{\rho=0\} \subset	\mathcal{R}=\{ \rho < - [\dot{a}(t)]^{-1}\},\ 	\mathcal{T}=\{ \rho> - [\dot{a}(t)]^{-1}\}.
		\end{equation*}
		Recall indeed that $\dot{a}(0)<0$, and by \eqref{FRLW1}, $\dot{a}(t)\blue{\leq \dot{a}_0}<0$ for all $t\geq 0$. It is thus clear that there exists $0<T_S< \frac{a_0}{-\dot{a_0}}$ such that \begin{equation}
			\lim_{t\rightarrow T_S} a(t)=0.
		\end{equation}
		
		\noindent	\eqref{FRLW1}, \eqref{FRLW2} can be cast in the following equivalent formulation: defining $c(t) = \frac{a^3(t)}{a_0^3}$, we have
		\begin{equation}\label{FRLW7}
			c(t)\ddot{c}(t) + 9  [ \dot{\phi_0}]^2  = \frac{2}{3} [\dot{c}(t)]^2
		\end{equation}
		It is useful to subsequently introduce the variable $v(t)= \dot{c}(t)$ and obtain the system  \begin{equation}\begin{split}
				&\dot{c}(t) = v(t),\\ & c(t) \dot{v}(t)= \frac{2 v^2}{3}- 9 [\dot{\phi}_0]^2.
			\end{split}
		\end{equation} which can be solved by separation of variables as \begin{equation}
			\frac{dc}{c} = \frac{v}{ \frac{2 v^2}{3} - 9 [\dot{\phi}_0]^2} dv,
		\end{equation} which can be solved as \begin{equation}\label{ODE.solved}
			c(t) = \frac{c_0 }{\bigl|  v^2_0- \frac{27}{2} [\dot{\phi}_0]^2\bigr|^{\frac{3}{4}}} \bigl|  v^2(t)- \frac{27}{2} [\dot{\phi}_0]^2\bigr|^{\frac{3}{4}}
		\end{equation}
		
		This means that $$ \lim_{t\rightarrow T_S} v(t) = - \frac{3\sqrt{3}}{\sqrt{2}} |\dot{\phi}_0|,$$ from which we deduce $$ \lim_{t\rightarrow T_S} \frac{c(t)}{T_S-t} = \frac{3\sqrt{3}}{\sqrt{2}}|\dot{\phi}_0|,$$ 
		from which \eqref{FRLW.est} and \eqref{FRLW.A.est} follow.

	\end{proof}
	
	\subsubsection{Uncharged spacelike gluing in spherical symmetry}
	
	Our objective  is to construct a class of  one-ended spherically symmetric solutions of \eqref{1.1}--\eqref{5.1} with $F\equiv 0$, which coincide in some spacetime regions with the FLRW solutions constructed in Proposition~\ref{FRLW.prop}. First, in this section, we turn to the spacelike gluing problem, which we will use a tool towards that goal.

	For this, we must solve the spacelike constraints equations  for \eqref{1.1}--\eqref{5.1} with $F\equiv 0$ in spherical symmetry. We start with a hypersurface $\magenta{\Sigma_G} =\{ v+u = 0\}$, which we parametrize by $\rho= v-u= 2v = -2u$. We prescribe \begin{equation}\begin{split}
			& \rd_v r_{|\magenta{\Sigma_G}}(\rho) = \lambda(\rho),\\ &  \rd_u r_{|\magenta{\Sigma_G}}(\rho) = \nu(\rho),\\ &  \rd_v \phi_{|\magenta{\Sigma_G}}(\rho) = T(\rho),\\ &  \rd_u \phi_{|\magenta{\Sigma_G}}(\rho) = X(\rho).
		\end{split}
	\end{equation} Note that  prescribing the function $\lambda(\rho)$ corresponds to fixing the $v$-gauge, while prescribing the function $\nu(\rho)$ corresponds to fixing the $u$-gauge, recall\blue{ing} the introduction of  gauge~\ref{gauge.spacelike}. Of course, since $\rho= v-u$, the parametrization of $\magenta{\Sigma_G}$ by $\rho$ is thus fixed by the choice of $\lambda(\rho)$ and $\nu(\rho)$. We then construct the area-radius $r(\rho)$, the Hawking mass $\varpi(\rho)$ and the scalar field $\phi(\rho)$ by solving the following constraint equations: \begin{equation}\label{r.constraint}
		\rd_{\rho} r(\rho)  = \lambda(\rho) -\nu(\rho),
	\end{equation}\begin{equation}\label{psi.constraint}
		\rd_{\rho} \phi(\rho)  =T(\rho) - X(\rho),
	\end{equation}  	\begin{equation}\label{varpi.constraint}
		\rd_{\rho} \varpi(\rho)  = \frac{1}{2} (1-\frac{2\varpi(\rho)}{r(\rho)})\left(r^2(\rho)\lambda^{-1}(\rho) T^2(\rho) +  r^2(\rho)|\nu|^{-1}(\rho) X^2(\rho)\right).
	\end{equation}
	
	\noindent We are now ready to state our main uncharged gluing result.  We refer the reader to Section~\ref{gluing.section} for precise definition of what it means to glue two spheres ``spatially''.
	
	\begin{thm}[Uncharged spacelike gluing]\label{uncharged.gluing.thm}
		Let $k\in \mathbb{N}$ and $\mathbb{S}_R$ a regular uncharged  $C^k$ data sphere of area-radius $R$ and Hawking mass $M>0$, with $R>2M$.   Let $R_A>R$ and $0<R_S^{\mathcal{T}}<R_A$, $M_S>\frac{R_S^{\mathcal{T}}}{2}$. Then, $\mathbb{S}_R$  can be glued spatially within the regular region to an apparent horizon uncharged data sphere $\mathbb{S}_A$ of area-radius $R_A$ and Hawking mass $\frac{R_A}{2}$. In turn, $\mathbb{S}_A$ can be glued characteristically to a Schwarzschild trapped sphere  $\mathbb{S}_S^{\mathcal{T}}$   of area-radius $R_S^{\mathcal{T}}$ and Hawking mass $M_S$.

	\end{thm}

	\begin{figure}[H]	\begin{center}
			\includegraphics[width=100 mm, height=50 mm]{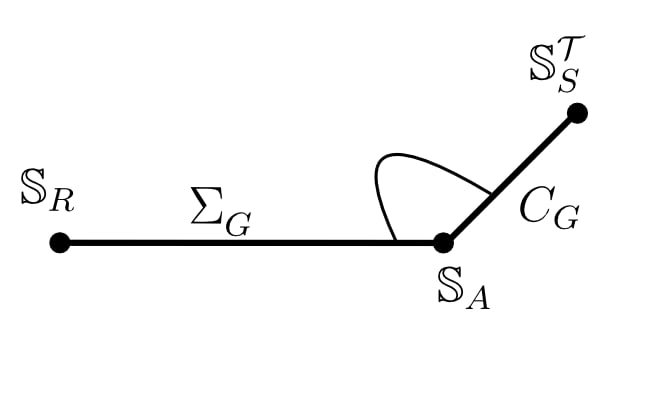
			}
		\end{center}
		\caption{The spacelike-characteristic gluing strategy in the proof of Theorem~\ref{uncharged.gluing.thm}. \blue{The bump represents the ingoing short pulse used for gluing, \teal{made quantitative} in Corollary~\ref{cor.pulse.1} and Corollary~\ref{cor.pulse.2} respectively.}}\label{fig:unchargedgluing}\end{figure}
	\magenta{The proof of Theorem~\ref{uncharged.gluing.thm} will be given below, after proving preliminary results: Proposition~\ref{spacelike.gluing.thm}, Corollary~\ref{cor.pulse.1}, Proposition~\ref{uncharged.null.gluing.prop} and Corollary~\ref{cor.pulse.2}.}
	In the proposition below, \blue{we emphasize again that} the $u$ and $v$ gauge will be determined by gauge~\ref{gauge.spacelike}, as we will see in the proof.
	\begin{prop}\label{spacelike.gluing.thm} Let $k \in \mathbb{N}$.
		Let $\mathbb{S}_1\in \R$ a regular $C^k$ sphere data with area-radius $r_1$ and mass $M_1>0$ such that \begin{equation}
			r_1> 2M_1.
		\end{equation} Let $M_2>\frac{r_1}{2}$, $\mathcal{N}_0<0$ and $(\varPhi,\varPhi_u^{1},...,\varPhi_u^{k},\varPhi_v^{1},...,\varPhi_v^{k})$ a list of $(2k+1)$ real numbers.  Denoting $\theta_2 = 2M_2 \varPhi_v^{1}$, assume that \begin{equation}\label{key.theta.bound}
			1<\theta_2^2< 1 + \frac{1}{k-1}.
		\end{equation}  
		Then, there exists  $\mathbb{S}_2\in \A$, an apparent horizon  $C^k$ sphere data  with area-radius $r_2=2M_2$, mass $M_2$ such that, in $C^k$ lapse-normalized gauge  \begin{equation} \label{CK.lapse.normalized}
			\Omega^2_{|\mathbb{S}_2} =1,\  \rd_u^{i} \Omega^2_{|\mathbb{S}_2}  = \rd_v^{i} \Omega^2_{|\mathbb{S}_2}=0 \text{ for all } 1\leq i \leq k.
		\end{equation}  the scalar field on $\mathbb{S}_2$ is given by  $(\varPhi,\varPhi_u^{1},...,\varPhi_u^{k},\varPhi_v^{1},...,\varPhi_v^{k})$, $\rd_u r_{|\mathbb{S}_2}= \mathcal{N}_0$ and  $\mathbb{S}_2$ can be spatially glued to  $\mathbb{S}_1$ in an uncharged way and  \blue{on a hypersurface $\magenta{\Sigma_G}$ within the regular region}.

	\end{prop}
	\begin{proof} 
		
		Let $(\varrho,\varrho_u^{1},...,\varrho_u^{k+1},\varrho_v^{1},...,\varrho_v^{k+1})$, $(\omega,\omega_u^{1},...,\omega_u^{k},\omega_v^{1},...,\omega_v^{k})$, $(\varphi,\varphi_u^{1},...,\varphi_u^{k},\varphi_v^{1},...,\varphi_v^{k})$ corresponding to the sphere data $\mathbb{S}_1$. We know that $\varrho>0$, $\omega>0$, and $\varrho_v^{1}>0$ by assumption. With no loss of generality,  assume that $\mathbb{S}_1$ is $C^k$ lapse-normalized (gauge choice), i.e.,  $(\omega,\omega_u^{1},...,\omega_u^{k},\omega_v^{1},...,\omega_v^{k})=(1,0,...0,0,...0)$. We also denote \begin{equation}\begin{split}
				&	\varrho:= r_1>0,\\ & 2M_1:=  r_1 [ 1+ \varrho_u^{1} \varrho_v^{1} ]>0,\\ & \phi_1:= \varphi \in \RR.
			\end{split}
		\end{equation}
		
		Start with $\magenta{\Sigma_G}$ a hypersurface parametrized by $\rho$, and isolate the ``initial sphere'' $\rho=\rho_1$ and the ``final sphere'' $\rho=\rho_2$ with $\rho_1<\rho_2$. 
		We will construct a triplet of functions $(r(\rho),\varpi(\rho),\phi(\rho))$ for $\rho_1 \leq \rho \leq \rho_2$,
		which we will use to obtain a solution of the constraint equations induced by \eqref{1.1}--\eqref{5.1} on $\magenta{\Sigma_G}$ in spherical symmetry.
		
		We then fix the following initial conditions: \begin{equation}\begin{split}\label{initial}
				&  \varpi(\rho_1) = M_1 >0\\  & r(\rho_1) = r_1>2M_1,\\ & \phi(\rho_1) = \varphi \in \RR.
			\end{split}
		\end{equation} Let $\lambda(\rho)$, $\nu(\rho)$,  $T(\rho)$, $X(\rho)$  smooth, and freely prescribed but obeying the initial compatibility conditions (recall that $\rd_v$ and $\rd_u$ derivatives can be obtained from the $\rd_{\rho}$ derivatives and the equations of Section~\ref{preliminary}), i.e., for all $1\leq j \leq k$ \begin{equation}\begin{split}
				& \lambda(\rho_1) = \varrho_v^{1},\ \rd_v^{j}\lambda(\rho_1)= \varrho_v^{j+1},\\ &
				\nu(\rho_1) = \varrho_u^{1},\  \rd_u^{j}\nu(\rho_1)= \varrho_u^{j+1},\\ &  \rd_v^{j-1}T(\rho_1) = \varphi_v^{j},\\   &  \rd_u^{j-1}X(\rho_1) = \varphi_u^{j}.
			\end{split}
		\end{equation} 
		Moreover, we \blue{impose} that $\lambda(\rho)$, $\nu(\rho)$ satisfy the following conditions:
		\begin{align}
			&\label{prescribed}  \lambda(\rho) >0 \text{ for all } \rho_1 \leq \rho< \rho_2,\\  & \lambda(\rho_2)=0,\\ &\label{prescribed2}  \nu(\rho) <0 \text{ for all } \rho_1 \leq \rho\leq \rho_2,\\  &\label{average.lambda-rho}\int_{\rho_1}^{\rho_2} [\lambda(\rho)-\nu(\rho) ]d\rho = 2M_2-r_1 >0.
		\end{align}  
		\blue{Imposing that $\lambda(\rho)$ and $\nu(\rho)$ are $C^k$ functions},	we   \blue{introduce} their Taylor expansions at order $k$ as $\rho \rightarrow \rho_2$   \begin{equation}\begin{split}
				& \lambda(\rho) = \sum_{q=1}^{k-1}\mathcal{L}_q (\rho_2-\rho)^{q}+ O((\rho_2-\rho)^{k}),\\ &
				\nu(\rho) = \sum_{q=0}^{k-1}\mathcal{N}_q (\rho_2-\rho)^{q}+ O((\rho_2-\rho)^{k}).
			\end{split}
		\end{equation}\blue{We obviously have $\mathcal{N}_0<0$ and subsequently,}  we \blue{also impose} that $$ \mathcal{L}_1>0.$$
		We also impose that $T(\rho)$, $X(\rho)$ satisfy the integral condition \begin{equation}
			\int_{\rho_1}^{\rho_2} [T(\rho)-X(\rho) ] d\rho =\Phi-\phi_1. 
		\end{equation} 
		
		\noindent	 In the sequel, we will denote $\theta:= rT$ and $\xi:= rX$.

		Then, we finally construct $r(\rho)$, $\varpi(\rho)$  and $\phi(\rho)$ by solving  the system of ODEs \eqref{r.constraint}, \eqref{psi.constraint}, \eqref{varpi.constraint} with initial conditions \eqref{initial}. 		By \eqref{prescribed}, \eqref{prescribed2}, $r$ is monotonically increasing, therefore $r(\rho )\geq r_1>0$. Therefore, by \eqref{varpi.constraint} (a linear equation in $\varpi$, for $r$, $\lambda$, $\nu$, $\theta$, $\xi$ given)  we have that $\varpi(\rho)$ is well-defined for all $\rho_1 \leq \rho \leq \rho_2$. 
		
		Now, let us show that $\varpi(\rho) > 0$ by showing a stronger result, i.e. that $\rd_{\rho} \varpi(\rho) > 0$ for all $\rho_1 \leq \rho < \rho_2$. If we show this, then $\varpi(\rho) > M_1>0$ so the mass will be positive.

		Suppose, by contradiction, that there exists $\rho_1 < \rho <  \rho_2$ such that $2\varpi(\rho)> r(\rho)$. Then, let $\rho^{*}= \inf\{\rho_1\leq\rho\leq \rho_2,\ \frac{2\varpi(\rho^{*})}{r(\rho^{*})}\geq 1\}$. By \eqref{initial}, $\rho^{*}>\rho_1$ and by continuity, $\frac{2\varpi(\rho^{*})}{r(\rho^{*})}= 1$.  By   definition of $\rho^{*}$, for all $\rho_1< \rho< \rho^{*}$: \begin{equation}\label{ineq}
			1-\frac{2\varpi(\rho)}{r(\rho)}>0.
		\end{equation} By \eqref{varpi.constraint}, we know that $\rd_{\rho}\varpi(\rho^{*})=0$. 
		Therefore, we have,  as $\rho \rightarrow \rho^{*}$ \begin{equation}\begin{split}\label{eq}
				&	\varpi(\rho) =  	\varpi(\rho^{*}) + O( [\rho - \rho^{*}]^2),\ r(\rho)= 2\varpi(\rho^{*}) +[\lambda(\rho^{*})-\nu(\rho^{*})] [\rho-\rho^{*}] +  O( [\rho - \rho^{*}]^2),\\ & 1-\frac{2\varpi(\rho)}{r(\rho)} = \frac{[\lambda(\rho^{*})-\nu(\rho^{*})] [\rho-\rho^{*}]}{2\varpi(\rho^{*})}+ O( [\rho - \rho^{*}]^2).\end{split}
		\end{equation} In particular, \eqref{eq} shows that $1-\frac{2\varpi(\rho)}{r(\rho)} <0$ for all $\rho< \rho^{*}$ sufficiently close to $\rho^{*}$, a contradiction with \eqref{ineq}.
		
		\noindent Conclusion:  for all $\rho_1 \leq \rho< \rho_2$, $2\varpi(\rho)> r(\rho)$. Thus,  $\varpi(\rho)>0$ for all $\rho_1 \leq \rho \leq \rho_2$.
		
		Now, we can solve the linear equation \eqref{varpi.constraint} on $\varpi$  using an integrating factor: we obtain for all $\rho_1 \leq \rho < \rho_2$:
		
		\begin{equation}\label{ODE.varpi.solved}
			\frac{r(\rho)}{2}- \varpi(\rho) = G(\rho) \left[\frac{r_1}{2} -M_1  + \int_{\rho_1}^{\rho} \frac{\rd_{\rho} r}{2G(\rho')}  d\rho' \right],
		\end{equation} where  \begin{equation}
			G(\rho) = \exp(- \int_{\rho_1}^{\rho} [\frac{\theta^2}{r\lambda}+ \frac{\xi^2}{r|\nu|}](\rho') d\rho' )
		\end{equation}
		Note that as $\rho \rightarrow \rho_2$, \begin{equation}
			\frac{\theta^2(\rho)}{r\lambda(\rho)}+ \frac{\xi^2}{r|\nu|} \sim \frac{\theta_2^2}{2M_2\mathcal{L}_1 } [\rho_2-\rho]^{-1},
		\end{equation}
		hence $G(\rho_2)=0$. In fact, there exists a constant $G_0>0$ such that as $\rho \rightarrow \rho_2$, \begin{equation}
			G(\rho) \sim G_0 [\rho_2-\rho]^{\frac{\theta_2^2}{2M_2\mathcal{L}_1 }}.
		\end{equation}
		We now \blue{impose} that \begin{equation}\label{alpha>1}
			\alpha:=\frac{\theta_2^2}{2M_2\mathcal{L}_1 }> 1,
		\end{equation} therefore $\frac{\rd_{\rho} r}{G(\rho)}$ is not integrable as $\rho \rightarrow \rho_2$, and it is easy to see that as $\rho \rightarrow \rho_2$: \begin{equation}
			G(\rho)\int_{\rho_1}^{\rho} \frac{\rd_{\rho} r}{2G(\rho')}  d\rho'  \sim \frac{|\mathcal{N}_0|}{2[\alpha-1]} [\rho_2-\rho],
		\end{equation} which gives, as $\rho \rightarrow \rho_2$: \begin{equation}\label{varpi.rho}
			1-\frac{2\varpi(\rho)}{r(\rho)} \sim  \frac{|\mathcal{N}_0|}{2M_2[\alpha-1]} [\rho_2-\rho].
		\end{equation}
		In fact, if we \blue{impose} the following  condition, which is stronger than \eqref{alpha>1}: \begin{equation}\label{alpha>k+1}
			\alpha:=\frac{\theta_2^2}{2M_2\mathcal{L}_1 }> k,
		\end{equation} we get a Taylor expansion of the following form: \begin{equation}
			G(\rho)\int_{\rho_1}^{\rho} \frac{\rd_{\rho} r}{2G(\rho')}  d\rho'  = \frac{|\mathcal{N}_0|}{2[\alpha-1]} [\rho_2-\rho] +\frac{1}{2} \sum_{q=1}^{k-1} \frac{\mathcal{L}_q+\mathcal{N}_q}{\alpha-1-q} [\rho-\rho_2]^{q+1}+O\left([\rho-\rho_2]^{\alpha}\right).
		\end{equation}
		
		Therefore by \eqref{ODE.varpi.solved} and \eqref{alpha>k+1}, we get as $\rho \rightarrow \rho_2$: \begin{equation}\label{varpi.rho2}\begin{split}
				&\varpi(\rho) = M_2 - \frac{\alpha |\mathcal{N}_0|}{ 2[\alpha-1]} [\rho_2-\rho]-\frac{1}{2} \sum_{q=1}^{k-1} [\frac{1}{\alpha-1-q}+\frac{1}{q}] [\mathcal{L}_q+\mathcal{N}_q] [\rho-\rho_2]^{q+1}+O\left([\rho-\rho_2]^{\alpha}\right),\\ & r(\rho)- 2\varpi(\rho) =  \frac{ |\mathcal{N}_0|}{ [\alpha-1]} [\rho_2-\rho]+  \sum_{q=1}^{k-1} \frac{\mathcal{L}_q+\mathcal{N}_q}{\alpha-1-q} [\rho-\rho_2]^{q+1}+O\left([\rho-\rho_2]^{\alpha}\right).\end{split}
		\end{equation} In particular, $\varpi(\rho)$ is $C^{k}$ on $[\rho_1,\rho_2]$.
		
		Let us define for all $\rho_1 \leq \rho < \rho_2$: \begin{equation}
			\Omega^2(\rho):= \frac{-\lambda(\rho) \nu(\rho)}{1-\frac{2\varpi(\rho)}{r(\rho)}}>0.
		\end{equation}
		By \eqref{varpi.rho}, $\Omega^2(\rho)$ extends continuously to $\rho=\rho_2$ and \begin{equation}
			\Omega^2(\rho_2) = \theta_2^2 - 2M_2 \mathcal{L}_1.
		\end{equation}
		
		Note that, by \eqref{varpi.rho2}, $\Omega^2(\rho)$ is $C^{k}$ on $[\rho_1,\rho_2]$. 	 Ultimately, we want to the final sphere $\mathbb{S}_2$ to be lapse-normalized,  in particular, we impose the gauge $\Omega^2(\rho_2)=1$, which fixes the value of $\mathcal{L}_1>0$ to be \begin{equation}\label{L.condition}
			2M_2 \mathcal{L}_1=\theta_2^2 - 1.
		\end{equation} and therefore requires, due to \eqref{alpha>k+1} that \begin{equation}\label{theta.cond}
			1<	 \theta_2^2< 1+\frac{1}{k-1}.
		\end{equation}
		
		Recall that the prescription of $\lambda(\rho)$ and $\nu(\rho)$ fixed the $v$-gauge choice and $u$-gauge choice respectively (in particular at the final sphere $\mathbb{S}_2$). However, it is well-known that any $C^k$ data sphere is $C^k$ gauge-equivalent to a lapse-normalized sphere \cite{KehleUnger}; we denote  	$\Phi_k(\rd_v \phi(\rho_2),\rd_v^2 \phi(\rho_2),...,\rd_v^{k} \phi(\rho_2), \rd_u \phi(\rho_2),\rd_u^2 \phi(\rho_2),...,\rd_u^{k} \phi(\rho_2)) $  to be the image of  $(\rd_v \phi(\rho_2),\rd_v^2 \phi(\rho_2),...,\rd_v^{k} \phi(\rho_2), \rd_u \phi(\rho_2),\rd_u^2 \phi(\rho_2),...,\rd_u^{k} \phi(\rho_2))$  under the gauge transform that makes $\mathbb{S}_2$ $C^k$ lapse-normalized. 
		
		Since we have already normalized $\Omega^2(\rho_2)=1$, $\rd_v \phi(\rho_2)$ is unchanged  and thus will still satisfy \eqref{theta.cond} (since $r(\rho_2)$ is gauge-invariant).  Since the gauge-transform  $\Phi_k$ is a bijection, we can choose  $(\rd_v \phi(\rho_2),\rd_v^2 \phi(\rho_2),...,\rd_v^{k} \phi(\rho_2))$ and $(\rd_u \phi(\rho_2),\rd_u^2 \phi(\rho_2),...,\rd_u^{k} \phi(\rho_2)) $ so that, as prescribed by the statement of the proposition: \begin{equation}
			\Phi_k(\rd_v \phi(\rho_2),\rd_v^2 \phi(\rho_2),...,\rd_v^{k} \phi(\rho_2), \rd_u \phi(\rho_2),\rd_u^2 \phi(\rho_2),...,\rd_u^{k} \phi(\rho_2)) = (\Phi_v^{1},...,\Phi_v^{k}, \Phi_u^{1}, ... \Phi_u^{k}).
		\end{equation}Fixing  $\rd_u r_{|\mathbb{S}_2}= \mathcal{N}_0$ can also be arranged following a similar logic. This concludes the proof of the proposition.

	\end{proof}
	
	\blue{We will then refine  the analysis of Proposition~\ref{spacelike.gluing.thm} to show that, if the left-most sphere comes from $C^1$ Cauchy data, the gluing procedure can be achieved by a \emph{short pulse}, i.e.,  a solution of the constraints whose largeness is localized near the right-most sphere. The subsequent result (Corollary~\ref{cor.pulse.1}) will not be used in the proof of Theorem~\ref{uncharged.gluing.thm} but instead in that of Theorem~\ref{uncharged.thm}, most notably to construct Cauchy data strictly to the past of the event horizon \magenta{(see also Corollary~\ref{cor.pulse.2} \blue{later} for the characteristic gluing analogue of Corollary~\ref{cor.pulse.1}).}

		\begin{cor}\label{cor.pulse.1} [Short pulse spacelike gluing].
			Let $k \in \mathbb{N}$ and  $\magenta{\Sigma'_L}\approx \RR^3$ a non-characteristic hypersurface  within the regular region on which we pose $C^k$ Cauchy data  satisfying the constraint equations.
			
			There exist $C>0$, $D_{\pm}>0$, such that for any sufficiently small $R_1>0$\teal{, $\ep\in (0,1]$} and numbers $M_2>\frac{R_1}{2}$, $\mathcal{N}_0<0$, $(\varPhi,\varPhi_u^{1},...,\varPhi_u^{k},\varPhi_v^{1},...,\varPhi_v^{k})$ \magenta{ and} if \eqref{key.theta.bound} from Proposition~\ref{uncharged.gluing.thm}  holds, together with the following extra conditions 
			\begin{equation}\label{large.2nd.deriv}
				D_- \leq |	\teal{\ep} R_1^2 \Phi^2_v - \magenta{\Phi^{1}_v}|\leq D_+,
			\end{equation}
			\begin{equation}\label{silly.3nd.deriv}
				\magenta{	|\varPhi_v^{3}| \leq C  \teal{\ep^{-2}} R_1^{-4},}
			\end{equation}
			\begin{equation}\label{silly.condition}
				|\varPhi_u^{1}| \leq \magenta{\frac{C}{2},}
			\end{equation}

			then, we can glue the $C^k$ sphere  $\mathbb{S}_1 \subset\magenta{\Sigma'_L}$ with area-radius $R_1$ to a $C^k$ sphere $\mathbb{S}_2 \in \A$ of area-radius $R_2=2M_2$, mass $M_2$ and scalar field given by $$(\phi,\rd_u \phi,...,\rd_u^{k} \phi,\rd_v \phi,..., \rd_v^{k} \phi)_{|\mathbb{S}_2}=(\varPhi,\varPhi_u^{1},...,\varPhi_u^{k},\varPhi_v^{1},...,\varPhi_v^{k})$$ and $$\rd_u r_{|\mathbb{S}_2}= \mathcal{N}_0 $$ in $C^k$ lapse-normalized gauge \eqref{CK.lapse.normalized} along a new hypersurface \magenta{$\Sigma_G'$} within the regular region and such that the following \magenta{short pulse} bounds hold, parametrizing $\magenta{\Sigma_G'}$ by $\rho=v-u\blue{\in [\rho_1,\rho_2]}$: for all $\rho_1\leq \rho \leq \rho_2$ \begin{equation}\label{shortA}
				|\rd_u \phi|_{|\magenta{\Sigma_G'}}(\rho),\  	r|\rd_v \phi|_{|\magenta{\Sigma_G'}}(\rho) \leq  C,
			\end{equation} and\teal{, if $\rho_2> \rho_1+ \ep R_1^2$, then} for all  $\rho_1\leq \rho \leq \rho_2- \teal{\ep} R_1^2$\begin{equation}\label{shortB}
				|\rd_v \phi|_{|\magenta{\Sigma_G'}}(\rho) \leq  C.
			\end{equation}

		\end{cor}
		\begin{proof}
			The proof is essentially the same as for Proposition~\ref{spacelike.gluing.thm}, though we need to show in addition that the short pulse conditions \eqref{shortA}, \eqref{shortB} are satisfied under the extra assumptions \eqref{large.2nd.deriv}, \eqref{silly.3nd.deriv}, \eqref{silly.condition}.

			\magenta{First, note that, since $\rd_{\rho} = \rd_v  - \rd_u$ and \eqref{Field} holds, we obtain  (in the notation of Proposition~\ref{spacelike.gluing.thm}: $T(\rho) = \rd_v \phi$, $X(\rho) = \rd_u \phi$):
				\begin{equation*}\begin{split}
						&	\rd_v^2\phi:=	\rd_{\rho} \rd_v \phi - \frac{\lambda}{r} \rd_u \phi  - \frac{\nu}{r} \rd_v \phi =  \rd_{\rho} T - \frac{\lambda}{r} X - \frac{\nu}{r} T, \\ &\rd_v^3\phi:=	\rd_{\rho} \rd_v^2 \phi -\rd_v( \frac{\lambda}{r} \rd_u \phi  - \frac{\nu}{r} \rd_v \phi),\\ & = 	\rd_{\rho} \rd_v^2 \phi +( \frac{\lambda[\lambda+\frac{\nu \mu}{1-\mu}]}{r^2 }- r^{-1} \rd_\rho \lambda) X\phi - \frac{(1-2\mu) \lambda \nu}{r^2(1-\mu)} T\phi+\frac{\nu}{r} \rd_v^2\phi.
					\end{split}
				\end{equation*}
				
				and $\rd_v^2 \phi(\rho=\rho_2) = \Phi^{2}_v$, $\rd_v^3\phi(\rho=\rho_2) = \Phi^{3}_v$. We start prescribing  $X(\rho)$ satisfying \eqref{shortA} for all $\rho_1 \leq \rho \leq \rho_2$. Now, we need that to prescribe  $T(\rho)$ such that $R_2|T(\rho_2)|\in[1,1+k^{-1}]$, and  \begin{equation}\begin{split} \label{imp.proof}
						T(\rho_2-\teal{\ep} R_1^2) = \blue{\Phi^1_v} - \int_{\rho_2-\teal{\ep} R_1^2}^{\rho_2} \rd_{\rho} T(\rho) d\rho \in [-C,C].\end{split}
				\end{equation}
				Now, note that since $R_1$ is small, and as a consequence of \eqref{large.2nd.deriv}, \eqref{silly.3nd.deriv},  we have $$\rd_{\rho} T = \rd_v^2 \phi + O (R_2^{-2})= \rd_v^2 \phi + O (R_1^{-2}),$$ $$ \rd_{\rho} \rd_v^2 \phi=\rd_v^3 \phi + O(R_2^{-3})+ O(R_2^{-1} \rd_v^2\phi)=\rd_v^3 \phi + O(R_1^{-3})+ O(R_1^{-1} \rd_v^2\phi).$$
				So, for $R_1$ sufficiently small, 	one can indeed prescribe $T(\rho)$ such that $\rho_2- \teal{\ep} R_1^2 \leq \rho \leq \rho_2$: \begin{equation}\label{silly1}
					\rd_\rho T(\rho) = \Phi^2_v + O(\teal{\ep^{-1}}R_1^{-2}) ,
				\end{equation}
				\begin{equation}\label{silly2}
					\rd_{\rho}	\rd_v^2 \phi(\rho) = O(\teal{\ep^{-2}}R_1^{-4}).
				\end{equation}    In turn, \eqref{silly1}, \eqref{silly2} result in  \begin{equation}\label{l}
					T(\rho_2- \teal{\ep}R_1^2) =\Phi^1_v-  \teal{\ep} R_1^2 \Phi^2_v + O(1)=O(1).
				\end{equation}Then,  one can then arrange $T(\rho) \blue{=O(1)}$ for $\rho_1 \leq \rho \leq \rho_2 -\teal{\ep}R_1^2$, which combined with  \eqref{l} shows that \eqref{shortB} holds and \blue{also} concludes the proof of \eqref{shortA}.
			}

		\end{proof}

	}
	
	\blue{We then turn to characteristic gluing, continuing towards the proof of Theorem~\ref{uncharged.thm} and Theorem~\ref{uncharged.gluing.thm}.} The next proposition offers a characteristic uncharged gluing result, allowing to glue an apparent horizon sphere $\mathbb{S}_A$ to a Schwarzschild trapped surface. We make a very strong assumption by not allowing traversal scalar field derivatives  at $\mathbb{S}_A$ to be prescribed, which is the reason for which  the argument is straightforward. Ultimately we will combine Proposition~\ref{spacelike.gluing.thm} and Proposition~\ref{uncharged.null.gluing.prop} to prove Theorem~\ref{uncharged.gluing.thm}: the reason why a weaker gluing result such as Proposition~\ref{uncharged.null.gluing.prop} is sufficient comes from the fact that Proposition~\ref{spacelike.gluing.thm} offers tremendous flexibility in prescribing the  traversal scalar field derivatives  at $\mathbb{S}_A$, which compensates for the rigidity of Proposition~\ref{uncharged.null.gluing.prop}.
	
	In the sequel,  the $v$-gauge will be fixed by gauge \eqref{gauge.unit.lapse}.
	\begin{prop} \label{uncharged.null.gluing.prop} Let $k\in \mathbb{N}$, $\ep_S>0$ and $r_S>0$. Let $\mathbb{S}_A\in \A$ an apparent horizon lapse-normalized $C^k$ sphere data with $(r,\varpi)= (2M_{A},M_{A})$ with $r_S<2M_{A}\magenta{<(1+\ep_S ) r_S}$: in particular, $\varrho_v^1 =0$ on $\mathbb{S}_A$. We  fix $k$ real-numbers $( \varphi^{1}_v,... \varphi^{k}_v)$ such that in the lapse-normalized gauge, \begin{equation}
			( \rd_v\phi_{|\mathbb{S}_A},..., \rd_v\phi_{|\mathbb{S}_A}^{k})=(  \varphi^{1}_v,... \varphi^{k}_v).
		\end{equation}

		\noindent There exists $k+1$ real numbers $(\varphi, \varphi^{1}_u,... \varphi^{k}_u)$  and $\mathcal{N}_0<0$ such that, assuming \begin{equation}
			(\phi_{|\mathbb{S}_A}, \rd_u\phi_{|\mathbb{S}_A},..., \rd_u\phi_{|\mathbb{S}_A}^{k},\rd_v\phi_{|\mathbb{S}_A},..., \rd_v\phi_{|\mathbb{S}_A}^{k})=(\varphi, \varphi^{1}_u,... \varphi^{k}_u, \varphi^{1}_v,... \varphi^{k}_v),
		\end{equation} \begin{equation}
			\varrho_u^{1} = \mathcal{N}_0,
		\end{equation} then, $\mathbb{S}_A$ can be characteristically glued \magenta{along an outgoing cone $C_G$} to $\mathbb{S}_S^{\T}$ a trapped Schwarzschild sphere with area-radius $r_S$ and Hawking mass $M_S>0$ such that \begin{equation}
			1-\frac{2M_S}{r_S} = -\ep_S.
		\end{equation}
		The associated characteristic data can be chosen to have no 
		anti-trapped spheres, namely $\rd_u r <0$.

	\end{prop}

	\begin{proof}
		
		We freely choose $\phi(v)$ for all $v_1 \leq v \leq v_2$, subject to the following conditions,  for all $1 \leq j \leq k$:\begin{equation}\begin{split}\label{phi.boundary.conditions}
				&\phi(v_1) = \varphi, \\ & \rd_v^{j}\phi (v_1) =\varphi_v^{j}, \\ & \phi(v_2)=0,\\ & 	\rd_v^{j}\phi (v_2) =0.
			\end{split}
		\end{equation} In the lapse normalization $\Omega^2(v)=1$, we solve $\lambda=\rd_v r$ as such (using \eqref{RaychV}) \begin{equation}\begin{split} & r(v_1)= 2M_A,\\ 
				&	\lambda(v_1) =0,\end{split}
		\end{equation}
		\begin{equation}\label{lambda.raych}
			\rd_v	\lambda(v) = - r(v) [\rd_v \phi]^2(v).
		\end{equation} 
		
		Then,	we pose the following ``final conditions'' for $\xi$: for all $0\leq j \leq k-1$ \begin{equation}
			\rd_u^{j} \xi(v_2) =0
		\end{equation} and subjected to the following recursive system of ODEs in $v$ obtained by formally differentiating in $u$  the equation \eqref{Field2}, i.e., $	\rd_v	 \xi = -\frac{\nu}{r} \theta$, and using the notation $\theta(v) = r(v) \rd_v\phi(v)$:  \begin{equation}\label{xi.evolution}
			\rd_v	\rd_u^{j} \xi = -	\rd_u^{j}[ \frac{\nu}{r} \theta]=  -	\rd_u^{j}[ \frac{\nu}{r}] \theta+   \sum_{q=1}^j\binom{j}{q}   \rd_u^{j-q}[ \frac{\nu}{r}]  \rd_u^{q-1} [-\frac{\lambda}{r} \xi],  	\end{equation} where we formally replace, in the above, $\rd_u r$ by $\nu$, $\rd_u [r\lambda]$ by $-1$, $\rd_u^{2}[r\lambda]$ by $0$ etc... consistently with \eqref{Radius} and the gauge choice $\Omega^2=1$. Note that the RHS of \eqref{xi.evolution} contains at most $j-1$ derivatives of $\xi$, and moreover, since $\theta$, $\lambda$, $\nu$ and $r$ are already fixed, it is a linear ODE in $\xi$ which is straightforward to solve.  In particular, the above procedure determines \begin{equation}\label{du.xi.v1}
			\rd_u^{j} \xi(v_1) 
		\end{equation}for all $0\leq j \leq k-1$. 	Then, we solve for $\nu$ and its $u$ derivatives  by initializing them at $v=v_1$:  for all $0 \leq j \leq k$:\begin{equation}\begin{split}
				&\rd_u^{j}\nu(v_1) = \mathcal{N}_j,
			\end{split}
		\end{equation} where $\mathcal{N}_0<0$ is imposed by assumption and $\mathcal{N}_j$ for $1 \leq j \leq k$ are determined by the following lapse-normalization $u$-condition for all  $1 \leq j \leq k$, originating from \eqref{RaychU}:\begin{equation}
			\rd_u^{j} \nu(v_1)=  -   	\rd_u^{j-1} [ r^{-1} \xi^2] (v_1),
		\end{equation} recalling that, by induction, the RHS is determined by the knowledge of \eqref{du.xi.v1}.
		
		Then, defining $n:= r\nu$ and
		\begin{equation}
			(n(v_1), \rd_u n(v_1) , ..., \rd_u^{k} n(v_1) ) = (n_0, n_1,..., n_{k}),
		\end{equation} where $n_0= 2M_A \mathcal{N}_0$, $n_1= 2M_A \rd_u \nu(v_1) + \mathcal{N}_0^2$ and similarly, the other $n_j$'s are uniquely determined from the  $\mathcal{N}_j$'s in this fashion, using the recurrence relation induced by $\rd_u n = r\rd_u \nu + \nu^2$.  Then we solve the  following ODE for all $1\leq j \leq k$: \begin{equation}
			\rd_v n  = -1,\  \rd_v \rd_u^{j} n =0,
		\end{equation} which is trivially solved by \begin{equation}
			n(v) = n_0 - [v-v_1],\   \rd_u^{j} n(v) = n_j,
		\end{equation} and again $\nu$ and its $\rd_u$ derivatives can be obtained straightforwardly using  $\rd_u n = r\rd_u \nu + \nu^2$. We note that $\nu(v)<0$ for all $v_1 \leq v \leq v_2$ (since $n_0<0$).  This gives, in particular, \begin{equation}\label{nu.v2}
			-\nu(v_2) = \frac{2M_A |\mathcal{N}_0| + [v_2-v_1]}{r(v_2)}.
		\end{equation}

		Now, denote $F(v)= r(v) [\rd_v \phi]^2(v)$: by \eqref{phi.boundary.conditions} and \eqref{lambda.raych}, the jets of order $k-1$ of $F(v)$ are fixed both at $v=v_1$ and $v=v_2$, and $F(v) \geq 0$ is non-identically zero.
		
		Note that,  \begin{equation}\begin{split}
				& -\lambda(v_2)  = \int_{v_1}^{v_2} F(v ) dv,\\ & \Delta r:= R_A-r(v_2)=   \int_{v_1}^{v_2}[\int^{v}_{v_1} F(v' ) dv'] dv,  \end{split}
		\end{equation}
		Then, since $2M(v_2)-r(v_2) = r(v_2)\nu(v_2)\lambda(v_2)$ \blue{as a consequence of $\Omega^2(v)=1$}, we have, also using \eqref{nu.v2}  and denoting $\Delta v= v_2-v_1$ and $\Delta r = 2M_A - r(v_2)>0$:
		\begin{equation}\label{M2.eq}
			M(v_2)= M_A  \left[1 + |\mathcal{N}_0| |\lambda|(v_2)\right] + |\lambda|(v_2) \Delta v- \Delta r 
		\end{equation} 
		
		Clearly, we have, since $F\geq 0$ \begin{equation}\label{inequality}
			\Delta r \leq   [\Delta v] |\lambda|(v_2),
		\end{equation} hence $M(v_2)>M_A$.  Let us fix $F$, respecting the above conditions. By choosing $\Delta v$ sufficiently small, we can arrange that $r_S<r(v_2) < 2M_A$, since $0<r_S<2M_A$. By \eqref{inequality}, one can arrange $F$  in a way that $$ |\lambda|(v_2) \Delta v- \Delta r < M_S-M_A.$$ Then, $|\lambda|(v_2)$ and $\Delta r$ being fixed in this way, and invoking \eqref{M2.eq}, we can choose $\mathcal{N}_0$  such that \begin{equation}
			M(v_2) = M_S.
		\end{equation}
		
		The sphere $\{v=v_2\}$ is $C^k$ Schwarzschild with area-radius $r_S<r(v_2)<2M_A$ and  Hawking mass $M_S$. This Schwarzschild trapped sphere can be continued to the future by an outgoing Schwarzschild cone with constant mass $M=M_S$ up to the sphere $r=r_S$, which concludes the proof of proposition.

	\end{proof}
	\blue{Then, similarly to the pair Proposition~\ref{spacelike.gluing.thm}/Corollary~\ref{cor.pulse.1}, we devise a refinement of Proposition~\ref{uncharged.null.gluing.prop} that shows that the short pulse structure  introduced  for spacelike gluing in  Proposition~\ref{spacelike.gluing.thm} is preserved in the characteristic gluing of Proposition~\ref{uncharged.null.gluing.prop}. Similarly to Corollary~\ref{cor.pulse.1}, Corollary~\ref{cor.pulse.2} below will not be used in the proof of Theorem~\ref{uncharged.gluing.thm} but in that of Theorem~\ref{uncharged.thm} instead.

		\begin{cor}\label{cor.pulse.2} Under the assumptions of Proposition~\ref{uncharged.null.gluing.prop}, assume that  $\varphi_v^{1}$, $\varphi_v^{2}$, $\varphi_v^{3}$  to satisfy \eqref{key.theta.bound}, \eqref{large.2nd.deriv}, \eqref{silly.3nd.deriv}. Then, for sufficiently small $M_A>0$, the conclusion of  Proposition~\ref{uncharged.null.gluing.prop} holds on an outgoing cone $C_G'$ and one can arrange that $\varphi_u^{1}$ satisfies \eqref{silly.condition} with $M_2=M_A$, where $C>0$ is  defined in the statement of Corollary~\ref{cor.pulse.1}.
			
			\magenta{Moreover, one can also arrange for the  short pulse bounds of Corollary~\ref{cor.pulse.1} to hold, namely there exists $C'>0$ independent of $M_A$ such that, in the $v$-gauge determined by \eqref{gauge.unit.lapse},
				for all $v_1\leq v \leq v_2$ \begin{equation}\label{shortA.null}
					|\rd_u \phi|_{|C_G'}(v),\  	r|\rd_v \phi|_{|C_G'}(v) \leq  C',
				\end{equation} and, additionally, if $v_2-v_1> \epsilon R_2^2$, then for all  $v_1+ 	\teal{\ep}R_2^2\leq v \leq v_2$\begin{equation}\label{shortB.null}
					|\rd_v \phi|_{|C_G'}(v) \leq  C'.
				\end{equation} 
			}
		\end{cor}
		
		\begin{proof}
			\magenta{The proof is essentially the same as that of Proposition~\ref{uncharged.null.gluing.prop}, but in addition we need to show that \eqref{shortA.null}, \eqref{shortB.null} hold. Recall that, in the proof of Proposition~\ref{uncharged.null.gluing.prop}, we fix $\xi(v=v_2)=0$, hence by \eqref{Field2} \begin{equation}\label{calc}
					\rd_u \phi(v) =\blue{-} r^{-1}(v)\int_{\blue{v}}^{\blue{v_2}} \nu \rd_v \phi(v') dv'= r^{-1}(v) \int_{\blue{v}}^{\blue{v_2}}\frac{2M_A |\mathcal{N}_0| + v'-v_1}{r(v')} \rd_v \phi(v') dv'.
				\end{equation}
				
				Now, note that, as in the proof of Corollary~\ref{cor.pulse.1}, we can arrange, as a consequence of \eqref{key.theta.bound}, \eqref{large.2nd.deriv}, \eqref{silly.3nd.deriv} that $\rd_v \phi$ is a short pulse, in the sense that for all $v_1 + 	\teal{\ep} R_2^2 \leq v\leq v_2$ \begin{equation*}
					|\rd_v \phi|(v) \leq C,
				\end{equation*} where $C>0$ is a $R_2$-independent constant,  while (assuming $R_2$ small enough), for all $v_1 \leq v \leq v_1 + 	\teal{\ep} R_2^2$  \begin{equation*}
					|\rd_v \phi|(v) \leq C R_2^{-1}.
				\end{equation*} This combined with \eqref{calc}  immediately gives that
				$	\rd_u \phi(v) = O(1)$ as desired, since $M_A=M_2=\frac{R_2}{2}$.

			}
		\end{proof}
	} 
	
	We finally turn to the proof of Theorem~\ref{uncharged.gluing.thm}. \blue{Note that that Corollary~\ref{cor.pulse.1} and Corollary~\ref{cor.pulse.2} will not be used in the proof of Theorem~\ref{uncharged.gluing.thm}, although they will be used in the proof of Theorem~\ref{uncharged.thm} to construct initial data strictly to the past of the event horizon}, \magenta{which requires tracking the short pulse quantitatively}.
	
	\begin{proof}
		
		\blue{We now prove Theorem~\ref{uncharged.gluing.thm}.} Let $R_A>R>2M$ and $0<R_S^{\mathcal{T}}<R_A$, $M_S>\frac{R_S^{\mathcal{T}}}{2}$, and the regular sphere $\mathbb{S}_R \in \R$. Let $\mathcal{N}_0<0$ and $(\Phi,\Phi_u^{1},...,\Phi_u^{k}, \Phi_v^{1},...,\Phi_v^{k})$ a list of  $(2k+1)$ real numbers such that $\Phi_v^{1}$ satisfies \begin{equation}\label{theta.bound.final}
			1<R_A^2[\Phi_v^{1}]^2<1+\frac{1}{k-1}.
		\end{equation}
		
		By applying Proposition~\ref{spacelike.gluing.thm} with $(\mathbb{S}_1,r_1,M_1)=(\mathbb{S}_R,R,M)$, $M_2=\frac{R_A}{2}$, one can glue $\mathbb{S}_R$ spatially to $\mathbb{S}_A\in \A$, an apparent horizon sphere  with area-radius $R_A$, Hawking mass $\frac{R_A}{2}$ and such that, in the lapse-normalized gauge $\Omega^2_{|\mathbb{S}_A}=1$,  $\rd_v^{j}\Omega^2_{|\mathbb{S}_A}=\rd_u^{j}\Omega^2_{|\mathbb{S}_A}=0$ for all $0\leq j \leq k$, the following holds:\begin{equation}
			(\phi_{|\mathbb{S}_A}, \rd_u\phi_{|\mathbb{S}_A},..., \rd_u\phi_{|\mathbb{S}_A}^{k},\rd_v\phi_{|\mathbb{S}_A},..., \rd_v\phi_{|\mathbb{S}_A}^{k})=(\Phi,  \Phi^{1}_u,... \Phi^{k}_u, \Phi^{1}_v,... \Phi^{k}_v).
		\end{equation}  
		
		Then, we can apply Proposition~\ref{uncharged.null.gluing.prop} with $M_A=\frac{R_A}{2}$, $r_S= R_S^{\mathcal{T}}$, $\ep_S = \frac{2M_S}{r_S}-1>0$ and  $( \varphi_v^{1},...,\varphi_v^{k})=( \Phi_v^{1},...,\Phi_v^{k})$. We can then choose $\mathcal{N}_0<0$ and  $( \Phi,\Phi_u^{1},...,\Phi_u^{k})$  so that the conclusion of Proposition~\ref{uncharged.null.gluing.prop} holds, and thus $\mathbb{S}_A$ can be spatially glued to $\mathbb{S}_S^{\T}$ a Schwarzschild trapped sphere of area-radius $ R_S^{\mathcal{T}}$ and Hawking mass $M_S$. This concludes the proof of Theorem~\ref{uncharged.gluing.thm}.

	\end{proof}

	\subsubsection{Global uncharged spacetime constructions}\label{global.uncharged.subsec}

	\blue{In this section, we complete the proof of Theorem~\ref{uncharged.thm} and Corollary~\ref{uncharged.cor} as an application of Theorem~\ref{uncharged.gluing.thm}. To this effect, we first prove a global existence result (Proposition~\ref{lemma.Cauchy.christo}) towards the past for initial data prescribed as in Corollary~\ref{cor.pulse.1}.  Global existence is obtained within the (past) domain of dependence exploiting the fact that the initial data consists of an \emph{ingoing pulse} (by which we mean it moves away from the center in the past direction, and towards the center in the future direction). While we formulate this as a global existence result towards the past to apply it directly to the initial data of Corollary~\ref{cor.pulse.1}, there is an analogous result towards the future if the data is  assumed to be outgoing. Before turning to global existence, we recall a standard quantitative local-existence result for \eqref{1.1}--\eqref{5.1} in spherical symmetry (towards both the past or future indiscriminately). Note that the initial data is set on a ball of finite but non-small area-radius $R_{\infty}>0$ to avoid the $r\rightarrow \infty$ limit.
		\begin{lemma}\label{LWP}
			Let $\Sigma'=\{u+v=0\}$ spherical hypersurface diffeomorphic to a ball in $\RR^3$ parametrized by $\rho \in [0,\rho_{\infty}]$,  we define the spherical eikonal functions $(u,v)$ by their Cauchy data $v_{|\Sigma'}(\rho) = \rho=-u_{|\Sigma'}(\rho) $, and we set initial data on $\Sigma'$ as $(\rd_u r,\rd_v r, \rd_u \phi, \rd_v \phi)$, and denote $r(\rho_{\infty})= R_{\infty}>0$.
			
			Assume that $N \leq N_{max}$, where $N$ is a $C^1$-initial data norm defined as \begin{equation}
				N= \max \{  \sup_{\Sigma'} |\rd_u r|,\ \sup_{\Sigma'} |\rd_v r|,\ \sup_{\Sigma'} |\rd_u \phi|,\ \sup_{\Sigma'} |\rd_v \phi|,\ R_{\infty}\}.
			\end{equation} Then, there exists $\tau(N_{max})>0$,  such that there exists a unique solution of \eqref{1.1}--\eqref{5.1} with $F\equiv 0$, and the above prescribed initial data in the spacetime region $\mathcal{D}_{N_{max}}:=\mathcal{D}_{N_{max}}^{F}\cup  \mathcal{D}_{N_{max}}^{P}$, where $$\mathcal{D}_{N_{max}}^{F}=\{ 0\leq u+v \leq \tau(N_{max}),\ v\leq \rho_{\infty}\},$$ $$\mathcal{D}_{N_{max}}^{P}=\{ -\tau(N_{max}) \leq u+v \leq 0,\ u \geq -\rho_{\infty}\}.$$ Moreover, if $\Sigma'$ does not contain any trapped or anti-trapped spheres, neither does $\mathcal{D}_{N_{max}}$.
		\end{lemma}
		\begin{proof}
			This is a standard local well-posedness argument in $C^1$-regularity: note that $\mathcal{D}_{N_{max}}^{F}$ is a subset of the future domain of dependence of $\Sigma'$, while  $\mathcal{D}_{N_{max}}^{P}$ is a subset of the past domain of dependence of $\Sigma'$.
	\end{proof}}
	\noindent \blue{Then, we turn to a global existence result in the past-direction for ingoing initial data. Note indeed that in Proposition~\ref{lemma.Cauchy.christo}, only $\rd_v \phi$ is allowed to be large, and its largeness is confined to a small region of initial data. Contrary to Lemma~\ref{LWP},  we will pose initial data on a ball of small area-radius $R>0$ in Proposition~\ref{lemma.Cauchy.christo} below.  
		\begin{prop}\label{lemma.Cauchy.christo} Let $C>0$. For all $R>0$,  we define $\Sigma'=\{u+v=0\}$ spherical hypersurface diffeomorphic to a ball in $\RR^3$ parametrized by $\rho \in [0,\rho_{R}]$, and we define the spherical eikonal functions $(u,v)$ by their Cauchy data $v_{|\Sigma'}(\rho) = \rho=-u_{|\Sigma'}(\rho) $, and we set initial data on $\Sigma'$ as $(\rd_u r=\nu(\rho),\rd_v r=\lambda(\rho), \rd_u \phi=X(\rho), \rd_v \phi=T(\rho))$, and assume that $r(\rho_{R})= R$ and that   $\lambda(\rho)$ and $\nu(\rho)$ are continuous functions on  $[0,\rho_R]$  such that\begin{equation}
				\label{lambda.hyp}	\lambda(\rho)>0, \text{ for all } \ \rho \in [0,\rho_R), \text{ and } (\lambda,\lambda')(\rho_R) \neq (0,0),
			\end{equation}
			\begin{equation}\label{nu.hyp}
				\nu(\rho)<0, \text{ for all } \ \rho \in [0,\rho_R].
			\end{equation}
			\noindent 
			\teal{Let $\ep\in (0,1]$.} We assume that following bounds for $T(\rho)$ and $X(\rho)$:  for $0\leq \rho \leq \rho_R- 			\teal{\ep}R^2$,
			\begin{equation}\begin{split} &\label{short1}
					|X|(\rho),\ |T|(\rho) \leq C\\ & \end{split}
			\end{equation} and for all $\rho_R- 			\teal{\ep}R^2\leq \rho \leq \rho_R$,  \begin{equation}
				\label{short2}|X|(\rho) \leq C,
			\end{equation}
			\begin{equation}
				\label{short3}	|T|(\rho) \leq C\cdot R^{-1}.
			\end{equation}

			\noindent Then, assuming that 
			$R>0$ is small enough,  there exists a unique solution of \eqref{1.1}--\eqref{5.1} with $F\equiv 0$, and the above prescribed initial data in $\mathcal{D}=\{ v_{\Gamma}(u)\leq v \leq -u,\ u \geq -\rho_R\}$,  the  past domain of dependence of $\Sigma'$.
			
			\noindent	  \magenta{We also have the analogous  statement for the characteristic initial value problem if $C_{out}$ is an outgoing cone emanating from the center $\Gamma$, and $\underline{C}_{in}$ whose future end-sphere has radius $R>0$ and coincides with $C_{out}$'s future end-sphere. We assume  $C_{out}=\{u=u_2,\ v_{\Gamma}(u_1) \leq v\leq v_2\}$ and  $\underline{C}_{in}=\{v=v_2,\ u_1 \leq u \leq u_2\}$ so that for all $u\in [u_1,u_2]$ \begin{equation}\label{data.charact}
					|\rd_u \phi|_{|\underline{C}_{in}}(u,v_2),\ |\nu|_{|\underline{C}_{in}}(u,v_2) \leq  C,
				\end{equation} and for  $v\in [v_{\Gamma}(u_1),v_2]$ 
				\begin{equation}\label{data.charact2}
					r|\rd_v\phi|_{|C_{out}}(u_2,v),\ 	|\log(\kappa)|_{|C_{out}}(u_2,v) \leq  C,
				\end{equation} and finally\teal{, if $v_2> v_{\Gamma}(u_1) + \ep R^2$, then} for  $v\in [v_{\Gamma}(u_1),v_2- 			\teal{\ep}R^2]$   \begin{equation}\label{data.charact3}
					|\rd_v\phi|_{C_{out}}(u_2,v) \leq  C.
				\end{equation} 
				Assume $|u_2-u_1|\leq \delta$, where $\delta>0$ is  a  sufficiently small constant (independent of $R$ and $\teal{\ep}$), and that $R$ is sufficiently small. Then, there exists  a unique solution of \eqref{1.1}--\eqref{5.1} with $F\equiv 0$, and the above prescribed initial data in $\mathcal{D}_{\mathcal{C}}=\{ v_{\Gamma}(u)\leq v \leq v_2,\  u_1\leq u \leq u_2\}$, which is the  past domain of dependence of $C_{out}\cup \underline{C}_{in} $. }
			
		\end{prop} 
		
		\begin{proof} Proceeding as in the proof of \magenta{Proposition~\ref{spacelike.gluing.thm}} and exploiting \eqref{lambda.hyp} and \eqref{nu.hyp} (this is the only place in the proof \magenta{where} the signs of $\lambda_{|\Sigma'}$ and $\nu_{|\Sigma'}$ are used), we find that there exists $D>1$ such that \begin{equation}\label{data.Omega}
				D^{-1}\leq |\nu|_{|\Sigma'},\ \Omega^2_{|\Sigma'}\leq D.
			\end{equation}
			
			\noindent We are interested in the past domain of dependence of $\Sigma'$, i.e., the region $\mathcal{D}=\{ v_{\Gamma}(u)\leq v \leq -u,\ u \geq -\rho_R\}$. Let us  isolate the sub-region $\mathcal{D}_{imp}=\{ \rho_R - 			\teal{\ep}R^2 \leq v  \leq -u,\ u \geq -\rho_R\}$, where we claim the ``impulsive behavior'' of the scalar field is localized. We make the following bootstrap assumptions in the region $\mathcal{D}_{imp}$:  \begin{equation}\label{B.imp1}
				|\rd_u \phi| \leq A,
			\end{equation}
			\magenta{	\begin{equation}\label{B.nu}
					[10 D]^{-1}	 \leq |\nu|\leq 10D
			\end{equation}}
			\begin{equation}\label{B.imp2}
				\frac{R}{2}\leq r \leq 2R,
			\end{equation}  for a large constant $A>0$  independent of $R$.  \noindent Then, integrating \eqref{RaychU} in $\mathcal{D}_{imp}$ using  \eqref{B.imp1},  \magenta{\eqref{B.nu}},  \eqref{B.imp2}  gives, for $R$ small enough \begin{equation}
				|\log|(\frac{\kappa(u,v)}{\kappa_{|\Sigma'}(v)})\ls R^2 A^2 \ll 1,
			\end{equation} so \begin{equation}\label{kappa.imp}
				1\ls	\kappa \ls 1.
			\end{equation}Then, integrating \eqref{Radius} under the form $\rd_v\log(-r\nu)= \kappa r^{-1}$ gives
			\begin{equation}
				|\log(-r\nu)(u,v)-\log(-r\nu)_{|\Sigma'}(-u)|\ls R^{-1} \Delta v \ls  R,
			\end{equation} from which we deduce the following estimate, using \eqref{data.Omega}, \eqref{kappa.imp}: for all $(u,v)\in \mathcal{D}_{imp}$:\begin{equation}\label{nu.Omega.imp}
				[4D]^{-1}\leq |\nu|(u,v),\ \Omega^2(u,v)\leq 4D,
			\end{equation}\magenta{which improves on \eqref{B.nu}.} Integrating $\nu= \rd_u r$ in $(u,v)\in \mathcal{D}_{imp}$ gives \begin{equation}
				|r(u,v) - r_{|\Sigma'}(v)| \ls  R^2,
			\end{equation} thus, for 
			$R$ small enough, \eqref{B.imp2} is improved. Similarly, integrating \eqref{Radius} gives:  for all $(u,v)\in \mathcal{D}_{imp}$
			\begin{equation}\label{lambda.est.LWP}
				|\lambda|(u,v) \ls 1.
			\end{equation} 
			
			\noindent	Recall, with the usual notation $\theta=r\rd_v \phi$, $\xi= r\rd_u \phi$, $\lambda=\rd_v r$, $\nu=\rd_u r$: \begin{equation}\begin{split}
					& \rd_u \theta = -\frac{\lambda}{r} \xi,\\ &  \rd_v \xi = -\frac{\nu}{r} \theta.
				\end{split}
			\end{equation} which, upon integrating in $u$ and using \eqref{lambda.est.LWP} gives  for all $(u,v)\in \mathcal{D}_{imp}$ \begin{equation}
				|\theta(u,v) - \theta_{|\Sigma'}(v)| \ls  A \cdot \Delta u \ls A  R^2, \text{ implying }  	|\theta|(u,v) \ls 1.
			\end{equation} Then, upon integrating in $v$
			gives  for all $(u,v)\in \mathcal{D}_{imp}$ \begin{equation}
				|\xi(u,v) - \xi_{|\Sigma'}(u)| \ls  r^{-1}\Delta v \ls   R, \text{ implying }  	|\rd_u \phi|(u,v) \ls |\rd_u \phi|_{|\Sigma'}(u)+ 1  \ls 1,
			\end{equation} which, for $A$ large enough, improves on bootstrap \eqref{B.imp1} and gives well-posedness in the region $\mathcal{D}_{imp}$. 
			
			It remains to show global existence within $\mathcal{D}-\mathcal{D}_{imp}$\magenta{, the region where the solution is ``not large''}. \blue{First}, to obtain global-wellposedness in the spacetime triangle $\{ v_{\Gamma}(u)\leq v\leq -u,\  u \geq -\rho_R + \teal{\ep} R^2\}$, we can invoke Lemma~\ref{LWP} as such: artificially extend the initial data on $\Sigma'\cap [\mathcal{D}-\mathcal{D}_{imp}]$ in a way that the $C^1$ norm of Lemma~\ref{LWP} is not too large, i.e., $N_{max}=O(1)$, then invoke  Lemma~\ref{LWP} to obtain a solution up to $u+v\geq -\tau_{N_{max}}$, where $\tau_{N_{max}}$ is independent of $R$. Then, by taking $R$ small enough if necessary, we see that $\{v_{\Gamma}(u)\leq v\leq-u,\ u \geq -\rho_R + \teal{\ep} R^2 \}\subset \{-\tau_N\leq u+v \leq 0\}$ and by the domain of dependence property, the solution in this region corresponds to that of  the original problem (before extending artificially the initial data). Then, one can obtain a solution in the \magenta{remaining} region $\{v_{\Gamma}(u) \leq v \leq \rho_R -\teal{\ep} R^2,\ -\rho_R \leq u\leq -\rho_R+ \teal{\ep}R^2 \}$ using a similar local well-posedness result, taking advantage of the smallness of the $u$-difference $O( R^2)$ and the fact that the $C^1$ norm of $\phi$ is $O(1)$ on the characteric hypersurfaces $\{u=-\rho_R+\teal{\ep} R^2\} \cup \{v=\rho_R-\teal{\ep} R^2\}$.
			
			Note that by Cauchy stability, there is no anti-trapped spheres in $\mathcal{D}$, and by monotonicity of $\rd_u(r\lambda)<0$, we also have $\mathcal{D}-\Sigma' \subset \mathcal{R}$. \magenta{This concludes the proof of the statement for spacelike initial data; for the  characteristic  initial data  analogue, the proof is identical.} 
	\end{proof}}

	\begin{rmk}
		\blue{	The global well-posedness argument in the proof of Proposition~\ref{lemma.Cauchy.christo} contains two different statements: first, global well-posedness within the region $\mathcal{D}-\mathcal{D}_{imp}$ can (also) be established as a consequence of Christodoulou's global existence result for small BV data \cite{Christo2}, provided we assume slightly more regularity on the initial data and choose $R$ small. However, the Christodoulou BV norm is not small in the sub-region $\mathcal{D}_{imp}$: here, we instead leverage the fact that the solution is ingoing (i.e., $\theta$ is of size $1$, but $\xi$ is $O(R)$ small) combined with the direction of evolution: we are solving towards the past. It is crucial to note that global existence towards the future is not always true for the initial data prescribed in Proposition~\ref{lemma.Cauchy.christo}. For example, in the context of Proposition~\ref{spacelike.gluing.thm}, the endsphere $\{\rho=\rho_R\}$ of the initial surface $\Sigma'$ is marginally trapped and it can be proven that a spacelike singularity $\mathcal{S}=\{r=0\}$ will form in its future domain of development \cite{Christo1,Christo2}.}
	\end{rmk}

	\blue{Finally}, we turn to the proof of Theorem~\ref{uncharged.thm} and Corollary~\ref{uncharged.cor}.
	\begin{proof}~\begin{enumerate}[Step 1.]
			\item \label{Step1}\blue{[First construction with event horizon intersecting the initial hypersurface]. We start with a  simpler construction where the event horizon $\HH$ intersects the initial data $\magenta{\Sigma_{0}}$}. Let $(\mathcal{M}_L,g_L,\phi_L)$, inducing initial data on $\Sigma_L$ and fix $\mathbb{S}_R$ to be a non-central sphere on $\Sigma_L$, with its induced $C^k$ sphere data, in particular its area-radius $R$ and Hawking mass $M$ with $R>2M$. By assumption, $\mathbb{S}_R \subset \mathcal{R}$. Then by Theorem~\ref{uncharged.gluing.thm}, we can glue $\mathbb{S}_R$ spatially to $\mathbb{S}_A$ \blue{along some hypersurface $\Sigma_G$}, which in turn can be glued characteristically to  a trapped Schwarzschild sphere  $\mathbb{S}_S^{\T}$, see Figure~\ref{fig:constbif1}.  
			\blue{We choose for some small $\eta>0$: \begin{equation}\begin{split}
						&R_S^{\mathcal{\tau}}= [1+\eta] R,\\ &2M_S = [1+\eta]R_S^{\mathcal{\tau}}= [1+\eta]^2 R,\\ & R_A = [1+\eta] [2M_S]=[1+\eta]^2R_S^{\mathcal{\tau}}=[1+\eta]^3 R.
					\end{split}
				\end{equation}	 Let $\Sigma_G$ the spacelike hypersurface connecting $\mathbb{S}_R$ to $\mathbb{S}_A$ and $C_G$ the outgoing light cone connecting  $\mathbb{S}_A$ to $\mathbb{S}_S^{\T}$, see Figure~\ref{fig:unchargedgluing} \magenta{and Figure~\ref{fig:constbif1}}. Note that $C_G$ is trapped, except at its past endpoint $\mathbb{S}_A$ \magenta{(i.e., $C_G-\mathbb{S}_A\subset \T$)}. Extend $\mathbb{S}_S^{\T}$  to a past-directed ingoing light cone $\underline{C}_S$ which is exactly Schwarzschild with Hawking mass $M_S$ and past endpoint $\mathbb{S}_S^{reg}$, \magenta{regular sphere} of area-radius $R_A= 2M_S(1+\eta)$.
				
				Choosing $\eta>0$ small shrinks the size of $\underline{C}_S$ to $0$. Let us set  bicharacteristic  initial data for \eqref{1.1}--\eqref{5.1} on the $C_G \cup \underline{C}_S$ and solve backwards. For   $\eta$ small enough, invoking local-wellposedness, we obtain a solution of  \eqref{1.1}--\eqref{5.1} in a whole open spacetime rectangle $\mathcal{D}_A$ bounded to the past by an ingoing cone $\underline{C}_A$ emanating from  $\mathbb{S}_A$ in the past-direction and an outgoing cone $C_{-\eta}$  emanating from  $\mathbb{S}_S^{reg}$  in the past-direction, see Figure~\ref{fig:constbif1}.
				
				\begin{figure}	\begin{center}
						\includegraphics[width=120mm, height=50 mm]{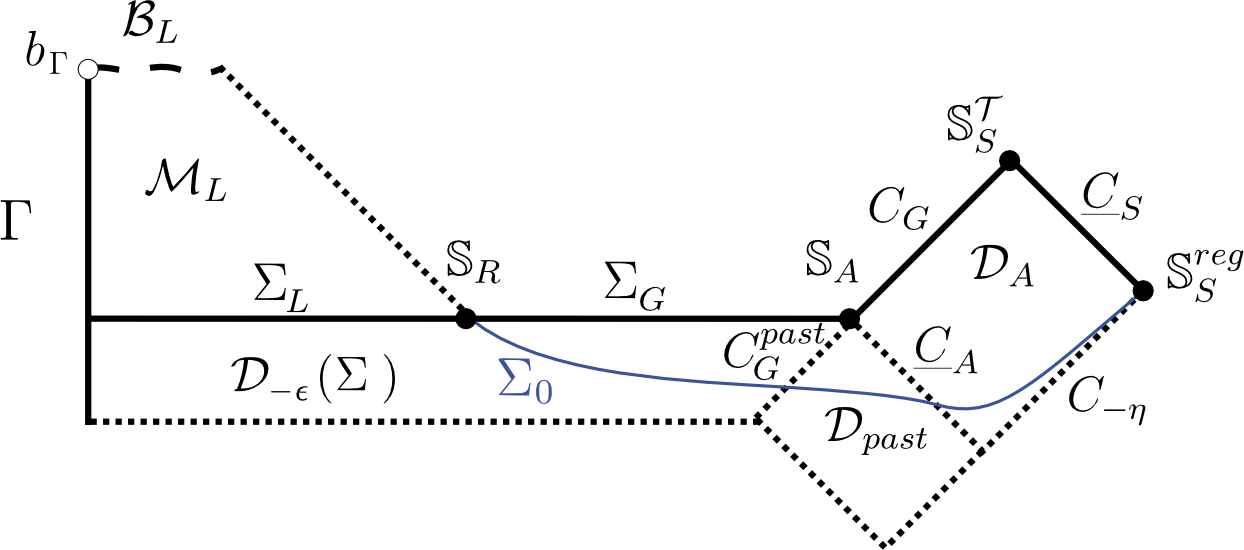}
					\end{center}
					\caption{The \magenta{first} construction strategy   in the proof of Theorem~\ref{uncharged.thm}. $\magenta{\Sigma_0}$ spacelike is constructed to be free of any trapped or anti-trapped spheres. \magenta{ Note, however, that the event horizon $\mathcal{H}^+$ will inevitably intersect $\Sigma_{0}$.}}\label{fig:constbif1}\end{figure}
				
				\noindent By Cauchy stability, $\rd_u r <0$ on $\mathcal{D}_A$, i.e., $\mathcal{D}_A$  does not contain any anti-trapped sphere.
				
				Moreover, by \eqref{RaychV}, $C_{-\eta} \subset \mathcal{R}$. Furthermore, by the monotonicity property $\rd_u \lambda <0$ from \eqref{Radius} and the fact that $\mathbb{S}_A \in \mathcal{A}$, we also have   $\underline{C}_A\magenta{-\mathbb{S}_A} \subset \mathcal{R}$ and thus $\mathcal{D}_A\magenta{-\mathbb{S}_A} \subset \R$.

				\noindent Now, we can solve  \eqref{1.1}--\eqref{5.1} backwards with \magenta{spacelike} initial data on $\magenta{\Sigma:=\Sigma_L}\cup\Sigma_G$: we obtain a solution of \eqref{1.1}--\eqref{5.1} on a small domain $\magenta{\mathcal{D}_{-\ep}(\Sigma)}$ emanating from $\Sigma$  in the past-direction \magenta{by Lemma~\ref{LWP}}. In particular, we denote $C_G^{past}\subset \magenta{\mathcal{D}_{-\ep}(\Sigma)}$ a small outgoing cone emanating from $\mathbb{S}_A$   in the past-direction (but not including its future endpoint $\mathbb{S}_A$). We note that, because of \eqref{RaychV} again, $C_G^{past}\subset \mathcal{R}$, and again by Cauchy stability,   $\rd_u r <0$ on $\magenta{\mathcal{D}_{-\ep}(\Sigma)}$, i.e., $\magenta{\mathcal{D}_{-\ep}(\Sigma)}$  does not contain any anti-trapped sphere.
				
				We can then pose characteristic initial data on  $C_G^{past}\cup \underline{C}_A$, and solve \eqref{1.1}--\eqref{5.1} backwards (see again Figure~\ref{fig:constbif1}): by local-wellposedness (\magenta{Lemma~\ref{LWP}}), taking $C_G^{past}$ shorter if necessary, we obtain a solution in a small spacetime rectangle $\mathcal{D}_{past}$ emanating from  $C_G^{past}\cup \underline{C}_A$  in the past-direction,  and yet again by Cauchy stability,   $\mathcal{D}_{past}$  does not contain any anti-trapped sphere. Moreover, by \eqref{RaychV} and the fact that $\underline{C}_A\subset \mathcal{R}$, we know that  $\mathcal{D}_{past} -  \mathbb{S}_A \subset \mathcal{R}$.

				By Cauchy stability, one can then  construct a  spacelike hypersurface $\magenta{\Sigma_{0}} \subset \mathcal{R}$ connecting $\mathbb{S}_R$ to   $\mathbb{S}_S^{reg}$, passing through $\magenta{\mathcal{D}_{-\ep}(\Sigma)}$, $\mathcal{D}_{past}$, $\mathcal{D}_A$ and , as depicted in Figure~\ref{fig:constbif1}.  We can then extend $\mathbb{S}_S^{reg}$ to its past as the Schwarzschild metric exactly, and extend $\Sigma_0$ into an asymptotic flat (in fact, exactly Schwarzschild) spacelike hypersurface}: \teal{then, the event horizon $\mathcal{H}^+$ of this spacetime will intersect $\Sigma_0$.} 
			
			\item  \blue{[Choosing a spacelike hypersurface to the future and gluing of spheres with small radii]. In the construction of Step~\ref{Step1}}, $\mathcal{H}^+$ intersects $\magenta{\Sigma_{0}}$, so to conclude the proof, we need to construct another spacelike hypersurface $\magenta{\Sigma_{0}}'$ strictly to the past of $\Sigma$ such that  $\mathcal{H}^+$ does not intersect $\magenta{\Sigma_{0}}'$\teal{, see Figure~\ref{fig:constbif2}}. 
			
			The key is that, with no loss of generality, one can replace $\Sigma_L$ by $\Sigma_L'$, another spacelike hypersurface in $(\mathcal{M}_L,g_L,\phi_L)$ to the future of $\Sigma_L$, and arbitrarily close to the terminal boundary of $\mathcal{M}_L$. Then, since $b_{\Gamma}$ is a first singularity by assumption,  we can integrate \eqref{Radius} (see \cite{Kommemi}) to show that \begin{equation}\label{r.bGamma.0}
				\lim_{p \rightarrow b_{\Gamma}} r(p) = 0.
			\end{equation} 
			We then revisit \blue{Step~\ref{Step1}} with $\Sigma_L'$ instead of $\Sigma_L$ and denote $\mathbb{S}_R'$ the regular \blue{(left-most)} uncharged sphere to be glued, with area-radius $R'>0$ \magenta{and mass $0<M'<\frac{R'}{2}$}. 
			By taking $\Sigma_L'$ sufficiently close to the terminal boundary of $\mathcal{M}_L$, one can always choose $R'$ to be arbitrarily small while ensuring that $\mathbb{S}_R'$ is in the strict future of $\underline{C}_{\Gamma}$ (the ingoing light cone emanating from $b_{\Gamma}$ towards its past), as a consequence of \eqref{r.bGamma.0} \magenta{(see Figure~\ref{fig:constbif3})}.  
			\blue{In what follows, our objective is to glue $\mathbb{S}_R'$ to $\mathbb{S}_A$, an apparent horizon sphere of area-radius $R_A> R'$ along the spacelike hypersurface $\magenta{\Sigma'_G}$ via a short pulse. \magenta{Defining $\Sigma'= \Sigma_L'\cup \Sigma_G'$,} the additional short pulse information will allow us to obtain a solution in $\mathcal{D}(\Sigma')$, the past domain of dependence\footnote{\magenta{Note that, in Step~\ref{Step1}, where we were not using any short pulse information, we could only obtain a small domain $\mathcal{D}_{-\ep}(\Sigma)$ by local existence and the construction leaves no space for a hypersurface $\Sigma_0'$ to the past of the event horizon, see Figure~\ref{fig:constbif1}.}} of $\Sigma'$,	 see Figure~\ref{fig:constbif3}}.
			\begin{figure}	\begin{center}
					\includegraphics[width=80 mm, height=70 mm]{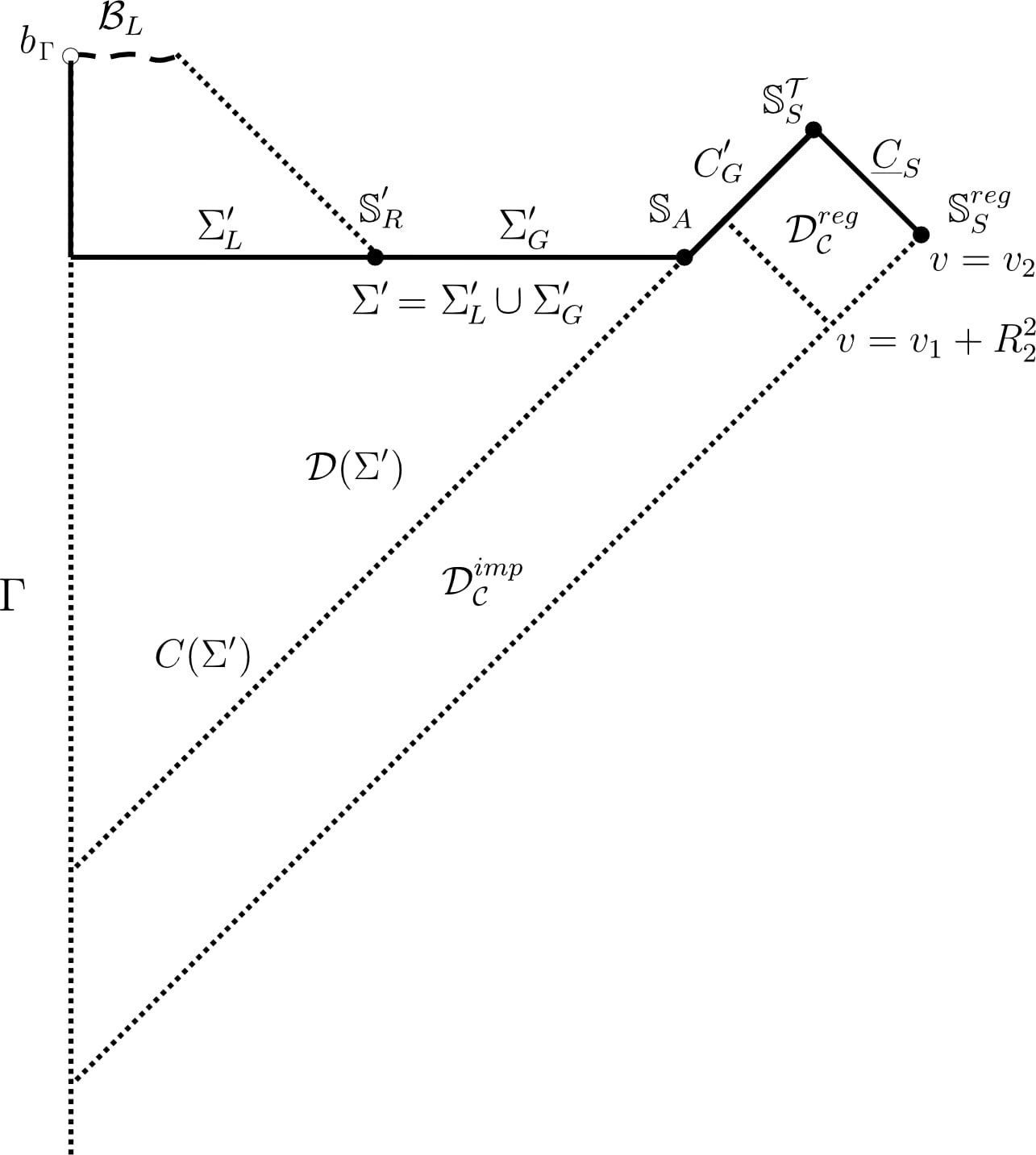}
				\end{center}
				\caption{The detailed implementation  of the proof of Theorem~\ref{uncharged.thm}. \magenta{With a quantitative ingoing short pulse near the event horizon sphere $\mathbb{S}_A$, we can solve backwards globally towards the past into $\mathcal{D}(\Sigma') \cup \mathcal{D}_{\mathcal{C}}^{imp}\cup \mathcal{D}_{\mathcal{C}}^{reg}$.}}\label{fig:constbif3}\end{figure}
			
			\item\blue{[Imposing short pulse initial data, invoking Corollary~\ref{cor.pulse.1} and Corollary~\ref{cor.pulse.2}].} \blue{We proceed as in Step~\ref{Step1} via spacelike-characteristic gluing, but instead of invoking \magenta{Theorem~\ref{uncharged.gluing.thm} (which relies on Proposition~\ref{spacelike.gluing.thm} and Proposition~\ref{uncharged.null.gluing.prop} for spacelike/characteristic gluing respectively), we   use a short pulse in the quantitative sense, invoking   Corollary~\ref{cor.pulse.1} for spacelike gluing and Corollary~\ref{cor.pulse.2} for characteristic gluing.}

				In more details, we choose, as in Step~\ref{Step1}, for some small $\eta>0$: \begin{equation}\begin{split}
						&R_S^{\mathcal{\tau}}= [1+\eta] R',\\ &2M_S = [1+\eta]R_S^{\mathcal{\tau}}= [1+\eta]^2  R',\\ & R_A = [1+\eta] 2M_S=[1+\eta]^2R_S^{\mathcal{\tau}}=[1+\eta]^3  R',
					\end{split}
				\end{equation}
				\blue{	and we now revisit the proof of Theorem~\ref{uncharged.gluing.thm}. Let $\mathcal{N}_0<0$ and $(\Phi,\Phi_u^{1},...,\Phi_u^{k}, \Phi_v^{1},...,\Phi_v^{k})$ a list of  $(2k+1)$ real numbers such that $\Phi_v^{1}$, $\Phi_v^{2}$, $\Phi_v^{3}$, $\Phi_u^{1}$ satisfy the short pulse assumptions \eqref{key.theta.bound}, \eqref{large.2nd.deriv}, \eqref{silly.3nd.deriv}, \eqref{silly.condition}.  
					
					By applying Corollary~\ref{cor.pulse.1} with $(\mathbb{S}_1,r_1,M_1)=(\mathbb{S}^{'}_R,R',M')$, $M_2=\frac{R_A}{2}$\magenta{, $R_A=R_2$,} \teal{ and $\ep=1$,} one can glue $\mathbb{S}_R'$ spatially to $\mathbb{S}_A\in \A$, an apparent horizon sphere  with area-radius $R_A$, Hawking mass $\frac{R_A}{2}$ \magenta{along an hypersurface $\Sigma'_G$} and such that, in the lapse-normalized gauge $\Omega^2_{|\mathbb{S}_A}=1$,  $\rd_v^{j}\Omega^2_{|\mathbb{S}_A}=\rd_u^{j}\Omega^2_{|\mathbb{S}_A}=0$ for all $0\leq j \leq k$, the following holds:\begin{equation}
						(\phi_{|\mathbb{S}_A}, \rd_u\phi_{|\mathbb{S}_A},..., \rd_u\phi_{|\mathbb{S}_A}^{k},\rd_v\phi_{|\mathbb{S}_A},..., \rd_v\phi_{|\mathbb{S}_A}^{k})=(\Phi,  \Phi^{1}_u,... \Phi^{k}_u, \Phi^{1}_v,... \Phi^{k}_v),
					\end{equation}   \magenta{and moreover the short pulse properties \eqref{shortA}, \eqref{shortB} hold throughout $\Sigma'=\Sigma_L'\cup \Sigma_G'$.}

					Then, we can apply Corollary~\ref{cor.pulse.2} with $M_A=\frac{R_A}{2}$, $r_S= R_S^{\mathcal{T}}$, $\ep_S = \frac{2M_S}{r_S}-1>0$ and  $( \varphi_v^{1},...,\varphi_v^{k})=( \Phi_v^{1},...,\Phi_v^{k})$. We can then choose $\mathcal{N}_0<0$ and  $( \Phi,\Phi_u^{1},...,\Phi_u^{k})$  so that the conclusion of Corollary~\ref{cor.pulse.1} holds, and thus $\mathbb{S}_A$ can be spatially glued to $\mathbb{S}_S^{\T}$ a Schwarzschild trapped sphere of area-radius $ R_S^{\mathcal{T}}$ and Hawking mass $M_S$ \magenta{along a null cone $C_G'
						=\{u_G'\} \times [v_1,v_2]$ and \eqref{shortA.null}, \eqref{shortB.null} hold  on $C_G'$.}
					
			}}

			\item \blue{[Solving backwards-in-time invoking Proposition~\ref{lemma.Cauchy.christo} and completion of the proof of Theorem~\ref{uncharged.thm}].}
			\magenta{Differently from Step~\ref{Step1}, we then invoke Proposition~\ref{lemma.Cauchy.christo} to
				solve backwards for \eqref{1.1}--\eqref{5.1} with induced spacelike initial data on $\Sigma':=\Sigma_L'\cup \Sigma_G'$, viewed as the interior of a ball of radius $R_A= [1+\eta^3] R'$, and since $R'$ can be chosen to be arbitrarily small, the assumptions of Proposition~\ref{lemma.Cauchy.christo} are indeed satisfied and we obtain}	
			existence within $\mathcal{D}(\Sigma')$ defined as the whole past domain of dependence 
			of $\Sigma'$, and, \blue{as a consequence of Proposition~\ref{lemma.Cauchy.christo}, $\mathcal{D}(\Sigma')$ does not contain any trapped or anti-trapped sphere. We denote $C(\Sigma')$, the outgoing cone which is the past boundary of $\mathcal{D}(\Sigma')$, see \magenta{Figure~\ref{fig:constbif3}}.

			}
			
			\blue{		Then, we imitate the procedure from Step~\ref{Step1} in attaching \teal{$\mathbb{S}_S^{\T}=(u_G',v_2)$ to a regular Schwarzschild sphere $\mathbb{S}_S^{reg}=(u_S^{\R},v_2)$ through} a Schwarzschild cone $\underline{C}_S=[u_S^{\R},u_G']\times \{v_2\}$ of small length \magenta{$|u_G'-u_S^{\R}|\leq \delta$  (independent of $R'$)}. 		\magenta{Let us first solve a characteristic initial value problem ``away from the short pulse'':  we 	 solve backwards for \eqref{1.1}--\eqref{5.1} with  data on $  \big(C_G' \cap \{v_1 + R_2^2 \leq v \leq v_2\}\big) \bigcup \underline{C}_S$.}
				\teal{Since the data is $O(1)$  in $C^1$ norm,	b}y standard local well-posedness as in Lemma~\ref{LWP}, if $\delta$ is sufficiently small,	we obtain a causal region  $\mathcal{D}^{reg}_{\mathcal{C}}$  (which is free of anti-trapped surfaces by Cauchy stability), as  depicted in  Figure~\ref{fig:constbif3}. \magenta{Then, invoking again Proposition~\ref{lemma.Cauchy.christo} with characteristic initial data $C_{out} = C(\Sigma') \cup \big(C_G' \cap \{v_1  \leq v \leq v_1+ R_2^2\}\big)$, $\underline{C}_{in}= [u_S^{\R},u_G']\times \{v_1+R^2_2\}$ and (note that \eqref{data.charact}, \eqref{data.charact2}, \eqref{data.charact3} are satisfied as a direct consequence of the construction)}  provides a unique solution of \eqref{1.1}--\eqref{5.1} in $\mathcal{D}_{\mathcal{C}}^{imp}$, the whole past domain of dependence of $C_{out} \cup \underline{C}_{in}$. \magenta{Note that $\mathcal{D}_{\mathcal{C}}=\mathcal{D}_{\mathcal{C}}^{reg}\cup \mathcal{D}_{\mathcal{C}}^{imp}$ is the whole past domain of dependence of $C(\Sigma') \cup C_G' \cup \underline{C}_S$, as desired. We can then extend $\mathbb{S}^{reg}_S$ into a Schwarzschild domain up to null infinity as in Step~\ref{Step1}. After doing so, we see that the event horizon $\mathcal{H}^+$   necessarily intersects $\mathcal{D}_{\mathcal{C}}$, since  $\mathbb{S}^{\T}_S\in \T$ and $\mathbb{S}^{reg}_S\in\R$ \teal{is extended  to the future as a Schwarzschild outgoing cone in the black hole exterior}, see \magenta{Figure~\ref{fig:constbif2}}. \\Finally, for $C^{k}$ initial data, the solution remains $C^{k-1}$ (one derivative is lost due to the presence of the center, see e.g. \cite{KehleUnger}) in its past domain of dependence \teal{$\mathcal{D}(\Sigma') \cup \mathcal{D}_{\mathcal{C}}$}.}

			}
			
			\magenta{As a conclusion, one} can then construct  $\Sigma_0'$ spacelike \magenta{with $C^{k-1}$ initial data}, and non-intersecting with $\mathcal{H}^+$  as in \magenta{Figure~\ref{fig:constbif2}}. 	 This concludes the proof of Theorem~\ref{uncharged.thm}.

			\begin{figure}[]	\begin{center}
					\includegraphics[width=82 mm, height=70 mm]{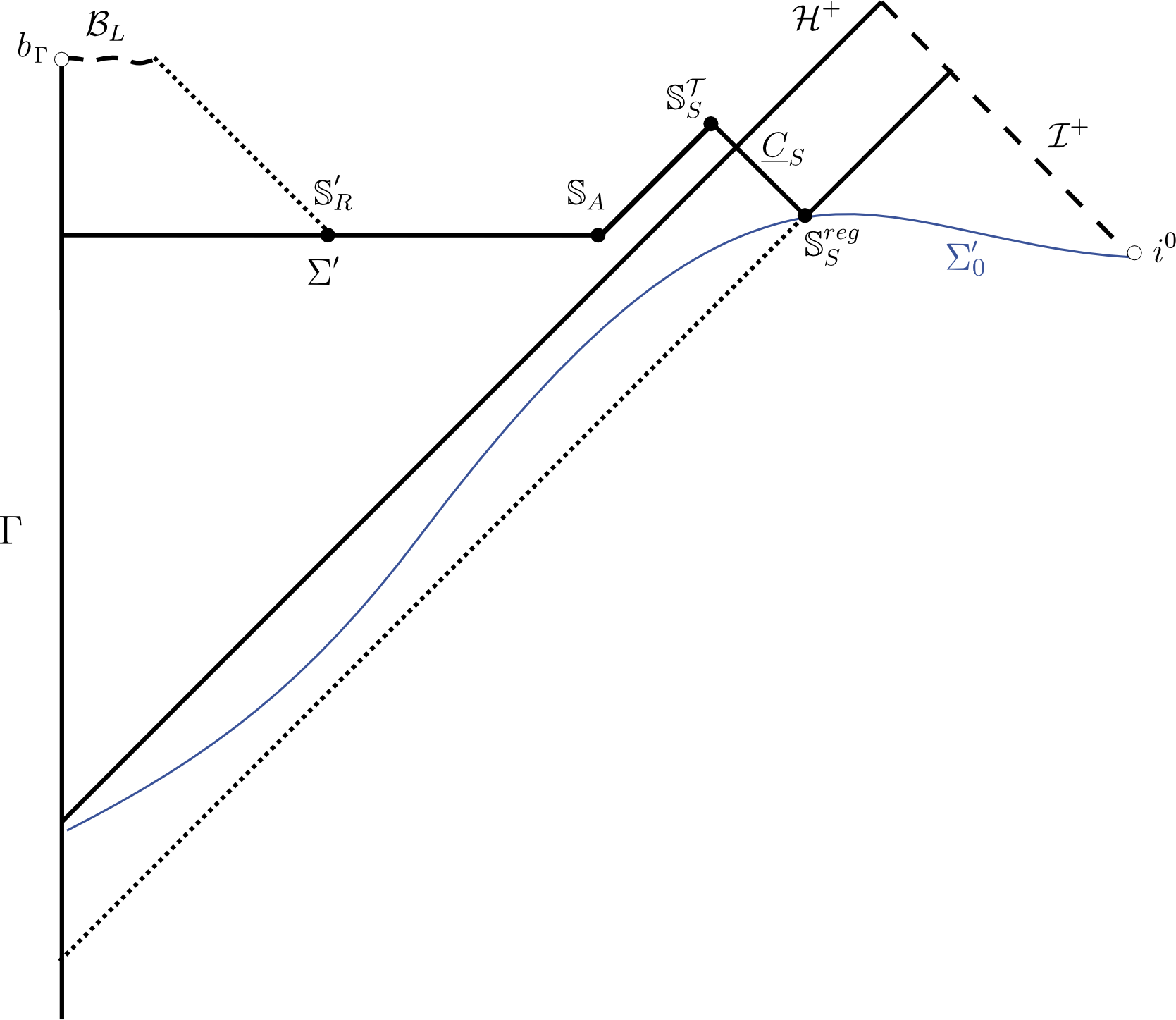
					}
				\end{center}
				\caption{\magenta{The final step in the proof of Theorem~\ref{uncharged.thm}: $\Sigma_0'$ is spacelike, strictly to the past of the event horizon $\mathcal{H}^+$, passing by $\mathbb{S}_S^{reg}$ and isometric to  Schwarschild to the right of $\mathbb{S}_S^{reg}$} \blue{(in particular, $\Sigma_0'$ is asymptotically flat)}.}\label{fig:constbif2}\end{figure}
			
			\item \blue{[Initiating the proof of Corollary~\ref{uncharged.cor}].} First, consider the spatially-homogeneous spacetime of Proposition~\ref{FRLW.prop} and let $\underline{C}_{\Gamma}$ the ingoing cone emanating from  $b_{\Gamma}$--the endpoint of the center $b_{\Gamma}=(T=T_S, \rho=0)$--in the past-direction. Note that, by \eqref{FRLW.A.est}, the apparent horizon $\mathcal{A}$ is spacelike for $t$ sufficiently close to $T_S$, therefore $\underline{C}_{\Gamma} \cap \{t\geq T_S-\ep\} \subset \mathcal{R}$ for $\ep>0$ small enough.

			So, we can create a spacelike hypersurface $\Sigma_L\subset \mathcal{R}\cap \{t\geq T_S-\ep\}  $ which intersects the center $\Gamma$ and whose endpoint $\mathbb{S}_L$ is strictly to the future of $\underline{C}_{\Gamma}$, since  $\mathcal{R}$ is open. 
			
			One can then  extend $\Sigma_L$ into an asymptotically flat spacelike hypersurface by solving the following  ODEs \eqref{r.constraint}--\eqref{varpi.constraint} as in Proposition~\ref{spacelike.gluing.thm}, imposing $\lambda(\rho)>0$, $\nu(\rho)<0$, $\theta(\rho)$, $\xi(\rho)$ smooth to match with $\mathbb{S}_L$ and assuming $\theta(\rho)$, $\xi(\rho)$ compactly supported. One then obtains the MGHD  $(\mathcal{M}_L,g_L,\phi_L)$ of the induced initial data  on $\Sigma_L$, free of anti-trapped or trapped  spheres, and such that $b_{\Gamma}$ is a first singularity, so  one can then apply Theorem~\ref{uncharged.thm} to conclude. 
			
			\item \blue{[Completion of the proof of Corollary~\ref{uncharged.cor} and  Theorem~\ref{OS.thm.intro} for $q=0$].} Note that the statement regarding the terminal boundary of the MGHD consisting of $\mathcal{S}=\{r=0\}$, a $C^1$-spacelike singularity follows from Theorem~\ref{Kommemi.thm} together with the monotonicity properties specific to the $F\equiv 0$ case, as first established by Christodoulou in \cite{Christo1}. In more details, we note that, in addition to $\mathcal{S}$, the only other possible boundary components are $\mathcal{CH}_{\Gamma}$, a Cauchy horizon emanating from the center and $\mathcal{CH}_{i^+}$, a Cauchy horizon emanating from $i^+$. In this construction, the metric is exactly Schwarzschild near $i^+$, thus\footnote{Note that, more generally, we always have $\mathcal{CH}_{i^+}=\emptyset$ if $F\equiv 0$ as proven in \cite{Christo1}.} $\mathcal{CH}_{i^+}=\emptyset$. Moreover, the fact that the spacetime is exactly FLRW to the past of a cone $C_H$ shows the terminal boundary is an (isotropic) spacelike singularity  in this region, and therefore $\mathcal{CH}_{\Gamma}=\emptyset$ too. This concludes the proof of Corollary~\ref{uncharged.cor} and of Theorem~\ref{OS.thm.intro} in the case $q=0$.
		\end{enumerate}

	\end{proof}

	\subsection{Charging of the uncharged spacetime via trapped spheres gluing}\label{charging.section}
	
	In this section, we start from the uncharged spacetime constructed in Theorem~\ref{uncharged.thm} and we ``charge'' it via characteristic gluing, with a strategy similar to that of \cite{KehleUnger}. The main difference is that, here, we must stay away from the extremal case, the main object of study in \cite{KehleUnger}, since we glue an uncharged trapped surface to a charged one. Therefore, the argument does not  strictly follow from \cite{KehleUnger}, although the methods are similar.
	
	As a result, we will be able to show a charged analogue of Theorem~\ref{uncharged.thm} and Corollary~\ref{uncharged.cor}, as seen below. The strategy builds upon the methods previously developed in the proofs of Theorem~\ref{uncharged.thm} and Corollary~\ref{uncharged.cor}. It is also in this section that we will prove Theorem~\ref{main.gluing.thm}.

	\begin{thm}\label{charging.thm} Let $k \in \mathbb{N}$, $k\geq 2$ and  $(\mathcal{M}_L,g_L,\phi_L)$, a subset of the MGHD of \blue{$C^{k}$}
		spherically symmetric asymptotically flat initial data on a hypersurface $\Sigma_L$ with one end for \eqref{1.1}--\eqref{5.1}  containing no anti-trapped spheres and no trapped spheres and such that $b_{\Gamma}$ is a first singularity.
		
		Then, there exists $q_L \in (0,1)$ such that for	 all  $q \in (0,q_L)$ and $\varsigma =\pm 1$, there exist $C^k$ solutions $(\mathcal{M},g,F,\phi)$ of \eqref{1.1}--\eqref{5.1} with $F\neq  0$ with the following properties: \begin{itemize}
			\item $(\mathcal{M},g,F,\phi)$ is the MGHD of 
			spherically symmetric asymptotically flat initial data on a spacelike hypersurface $\Sigma$ with one end for \eqref{1.1}--\eqref{5.1}  containing no anti-trapped spheres and no trapped spheres.
			\item The black hole region of $(\mathcal{M},g,F,\phi)$  is non-empty with an event horizon $\mathcal{H}^+$ and, moreover, $\mathcal{H}^+$ does not intersect $\Sigma$, i.e., it is located in the strict causal future of $\Sigma$.

			\item  There exist $M_f>0$ and an incoming null  cone $\underline{C}_{RN}$ and an outgoing null  $C_{RN}$ intersecting at a trapped sphere $\mathbb{S}_{RN}^{\T}$\blue{--which is  the future endpoint of  $\underline{C}_{RN}$ and the past endpoint of $C_{RN}$--}
			such that $\mathcal{M} \cap J^{+}(\underline{C}_{RN}) \cap J^{-}(C_{RN}) $ is isometric to a  metric with mass $M_f$ and charge $\varsigma q M_f$. In particular,   $\mathbb{S}_{RN}^{\T}$ is a Reissner--Nordstr\"{o}m trapped sphere and $\mathcal{H}^+ \cap J^{+}(\underline{C}_{RN})$ is coincides with a future affine complete portion of  a Reissner--Nordstr\"{o}m event horizon. 
			\item There exists an incoming null  cone $\underline{C}_{v_L}$ such that $\mathcal{M} \cap J^{-}(\underline{C}_{v_L})$ coincides with $\mathcal{M}_L \cap J^{-}(\underline{C}_{v_L})$. Moreover,  $\underline{C}_{v_L}$ can be chosen to be in the complement of the causal past of $b_{\Gamma}$.
		\end{itemize}
	\end{thm}
	
	Next, similarly to  Corollary~\ref{uncharged.cor}, we provide an application of Theorem~\ref{charging.thm} in the case where $\mathcal{M}_L$ is a FLRW spacetime, which will also immediately give the proof of Theorem~\ref{OS.thm.intro}.

	\begin{cor}\label{charging.cor} Let $k \in \mathbb{N}$, $k\geq 2$. Then, there exists $q_L \in(0,1)$ such that  	for all $q \in (0,q_L)$, $\varsigma =\pm 1$, there exist $C^k$ solutions $(\mathcal{M},g,\phi)$  of \eqref{1.1}--\eqref{5.1} with $F\neq  0$ with the following properties: \begin{itemize}
			\item $(\mathcal{M},g,F,\phi)$ is the MGHD of 
			spherically symmetric asymptotically flat initial data on a spacelike hypersurface $\Sigma$ with one end for \eqref{1.1}--\eqref{5.1}  containing no anti-trapped spheres and no trapped spheres.

			\item The black hole region of $(\mathcal{M},g,F,\phi)$  is non-empty with an event horizon $\mathcal{H}^+$ and, moreover, $\mathcal{H}^+$ does not intersect $\Sigma$, i.e., it is located in the strict causal future of $\Sigma$.
			\item  There exist $M_f$ and an incoming null  cone $\underline{C}_{RN}$ and an outgoing null  $C_{RN}$ intersecting at a trapped sphere $\mathbb{S}_{RN}^{\T}$\blue{--which is  the future endpoint of  $\underline{C}_{RN}$ and the past endpoint of $C_{RN}$--}such that $\mathcal{M} \cap J^{+}(\underline{C}_{RN}) \cap J^{-}(C_{RN}) $ is isometric to a  metric with mass $M_f$ and charge $\varsigma q M_f$. In particular,   $\mathbb{S}_{RN}^{\T}$ is a Reissner--Nordstr\"{o}m trapped sphere and $\mathcal{H}^+ \cap J^{+}(\underline{C}_{RN})$ is coincides with a future affine complete portion of  a Reissner--Nordstr\"{o}m event horizon.
			
			\item The MGHD terminal boundary of   $(\mathcal{M},g,F,\phi)$ is \begin{equation}
				\CH\cup	\mathcal{S},
			\end{equation} where $\CH\neq \emptyset$ (the Cauchy horizon) is a null boundary emanating from $i^+$ on which $r$ extends to a non-zero function, which is constant near $i^+$ (Reissner--Nordstr\"{o}m Cauchy horizon near $i^+$), and $r$ extends to $0$ on $\mathcal{S}$, which is a curvature singularity. 
			
			\item There exists an incoming null  cone $\underline{C}_{v_L}$  in the complement of the causal past of $b_{\Gamma}$ such that $\mathcal{M} \cap J^{-}(\underline{C}_{v_L})$ is spatially homogeneous and moreover $F\equiv 0$ in  $\mathcal{M} \cap J^{-}(\underline{C}_{v_L})$. Moreover, $\mathcal{S}_L:=  \mathcal{S}\cap J^{-}(\underline{C}_{v_L})$ is spacelike and coincides with the singularity of a FLRW metric with $\RR^3$ topology.

		\end{itemize}
	\end{cor}

	\begin{figure}	\begin{center}
			\includegraphics[width=100 mm, height=50 mm]{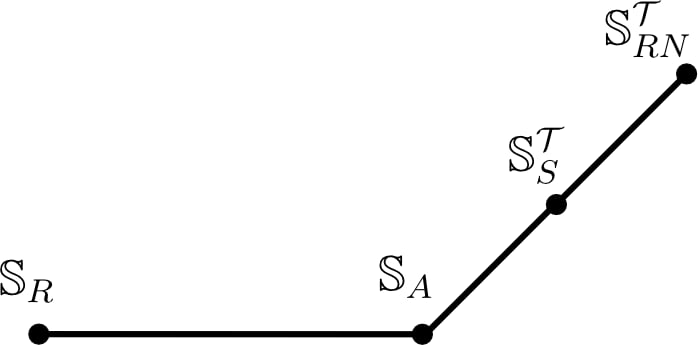}
		\end{center}
		\caption{The spacelike-characteristic gluing strategy \magenta{of Theorem~\ref{main.gluing.thm}, also}  used in the proof of Theorem~\ref{charging.thm}. \magenta{Its proof  combines} Theorem~\ref{uncharged.gluing.thm} and the new charged characteristic result Theorem~\ref{charged.gluing.thm}.}\label{fig:chargedgluing}\end{figure}

	\subsubsection{Stationary black hole outgoing: the trapped sphere case}\label{trapped.gluing.section}
	
	Inspired by the gluing techniques via a charged scalar field of \cite{KehleUnger}, we now prove that  a  Schwarzschild trapped (or event horizon) sphere can be glued to a Reissner--Nordstr\"{o}m trapped sphere. Recall that to glue a  regular Schwarzschild sphere to a Reissner--Nordstr\"{o}m event horizon  sphere in \cite{KehleUnger}, the authors require a large mass $M_f$ or a small charge ratio $q$, \blue{assuming a condition of the following form holds}:  \begin{equation}\label{large.condition.ext}
		\frac{|q_0| M_f }{q} \gg 1,
	\end{equation}  where $(e,M_f)$ are the charges and mass of the Reissner--Nordstr\"{o}m event horizon  sphere and $q= \frac{|e|}{M_f} \in (0,1]$. Obviously, if we require the final  Reissner--Nordstr\"{o}m sphere to be trapped instead of marginally trapped, we now must restrict   $q= \frac{|e|}{M_f} \in (0,1)$. Moreover, we replace \eqref{large.condition.ext} by the more demanding condition  \begin{equation}\label{large.condition2}
		\frac{|q_0| M_f[1-q] }{q} \gg 1,
	\end{equation}  which now penalizes Reissner--Nordstr\"{o}m trapped spheres that are too close to extremality. The following theorem below provides a precise statement of this new gluing result based on an adaption of techniques of \cite{KehleUnger}. 
	
	In this section, the $v$-gauge will be fixed by gauge~\eqref{gauge.unit.lapse}.\begin{thm} \label{charged.gluing.thm}[Characteristic charged trapped surfaces gluing]. Let $k\in \mathbb{N}$ and  $\delta_0\in (0,1)$ be a small, but fixed constant. Let $\mathbb{S}_i$ be a   trapped or marginally trapped  Schwarzschild $C^k$  data sphere of mass $M_i$ and radius $R_i$, with $R_i\leq 2M_i$, and   $\mathbb{S}_f$ be a Reissner--Nordstr\"{o}m   $C^k$  data  sphere of mass $M_f$, charge $\magenta{e}$ and radius $R_f$. We assume the sub-extremality condition $q= \frac{|\magenta{e}|}{M_f}\in(0,1)$, and that $\mathbb{S}_f$ is a trapped Reissner--Nordstr\"{o}m sphere, namely:  $R_f \in (r_-(M_f,\magenta{e}), r_+(M_f,\magenta{e}))$, where $r_{\pm}(M_f,\magenta{e}) = M_f \pm \sqrt{M_f^2-\magenta{e}^2}= M_f (1\pm \sqrt{1-q^2})$.  We also make the following largeness assumption:  \begin{equation}\label{large.condition}
			\frac{|q_0| M_f [1-q]}{q} \geq   \teal{100 \delta_0^{-2}}.
		\end{equation} 
		Moreover, we choose $0<R_f<R_i$ such that the trapped Reissner--Nordstr\"{o}m sphere of mass $M_f$, charge $e$ and radius $R_f$ is away from the event horizon and Cauchy horizon, more precisely:
		\begin{equation}\label{Mi.choice}
			1 - (1-\delta_0) \sqrt{1-q^2}<\frac{R_f}{M_f} +\delta_0 \sqrt{1-q^2}< \frac{R_i}{M_f} < 1 + (1-\delta_0) \sqrt{1-q^2}.
		\end{equation}
		Then, the $C^k$ Schwarzschild (event horizon or trapped) sphere $\mathbb{S}_i$ can be characteristically glued to the \blue{$C^k$} Reissner--Nordstr\"{o}m trapped sphere  $\mathbb{S}_f$. The associated characteristic data can be chosen to have no 
		anti-trapped spheres, namely $\rd_u r <0$.
	\end{thm}
	\begin{proof}
		The proof  is inspired \blue{from} Theorem B in \cite{KehleUnger}, with some   technical modifications. Like in \cite{KehleUnger}, \teal{the goal is to glue to a Reissner--Nordstr\"{o}m trapped sphere of radius $r(v_f)> R_f$  and} we  impose the following ansatz  for the scalar field: for some  $\alpha \in \mathbb{C}^N$, we define $\phi(v)$ for $v \in [v_i,v_f]$ as \begin{equation}\label{phi.gluing}\phi_{\alpha}(v) = [\sum_{j=1}^{k+1} \alpha_j \chi_j(v) ]e^{-i\frac{ v}{\Delta v}},
		\end{equation} where the $\chi_j(v)$'s are smooth compactly supported cut-offs with disjoint support, and $\chi_j(v)$ has support in $[v_{j-1},v_{j}]$ where $\{v_i,v_0,..., v_{k+1},v_f\}$ is an equipartition of $[v_i,v_f]$ with $v_0>v_i$  and $v_{k+1}< v_f$.  Note that $\phi(v,\alpha)\equiv 0$ for $v\in [v_i,v_0] \cup [v_{k+1},v_f]$.
		We will also define \begin{equation}
			q:= \frac{|e|}{M_f} \in (0,1),
		\end{equation}
		\begin{equation}
			\Delta v = v_f - v_i >0,
		\end{equation}  and we will take \begin{equation}\label{alpha.size}
			\frac{|e|}{  |q_0| R_i^2}\ls 	|\alpha|^2 \ls \frac{|e|}{ |q_0| R_i^2}.
		\end{equation} 
		We also initialize various quantities (note, however, that the value $\rd_u r$ is fixed at $v=v_f$) as such 
		\begin{equation*}
			Q(v_i,\alpha) =0,
		\end{equation*}
		\begin{equation*}
			r(v_i,\alpha) = R_i,\
		\end{equation*} \begin{equation*}
			\rd_v r (v_i,\alpha) = 1-\frac{2M_i}{R_i} := - \ep_i \leq 0,
		\end{equation*}
		\begin{equation*}
			-\rd_u r (v_f,\alpha) = \frac{1-\frac{2M_f}{r(v_f,\alpha)} + \frac{e^2}{r^2(v_f,\alpha)}}{\rd_v r(v_f,\alpha)},
		\end{equation*}
		
		We will moreover, as in \cite{KehleUnger}, set the gauge to be $\Omega^2=1$, $A_v \equiv 0$, and solve the following set of ODEs with the above initial conditions \begin{equation}\label{ODE.gluing}\begin{split}
				& \rd_{v}^2 r = - r |\rd_v \phi|^2,\\ &  \rd_v Q = q_0 r^2 \Im(\bar{\phi} \rd_v \phi),\\ & \rd_v(-r\rd_u r ) = 1-\frac{Q^2}{r^2}.
			\end{split}
		\end{equation}

		\noindent First, by monotonicity and \eqref{Mi.choice}, we have \begin{equation}
			r(v,\alpha) \leq R_i \leq  2M_f.
		\end{equation} Then, note that differentiating \eqref{phi.gluing} gives \begin{equation}
			\rd_v \phi \approx \frac{\alpha}{\Delta v},
		\end{equation} and   integrating  \eqref{ODE.gluing} thus gives  \begin{equation}\label{lambda.est.gluing}
			|\lambda|(v) \lesssim \ep_i +  \frac{R_i}{\Delta v} |\alpha|^2; 
		\end{equation} and thus, integrating in $v$ again gives
		\begin{equation}\label{delta.r}
			\underset{ v_i \leq v \leq v_f}{\sup} |r(v,\alpha) - R_i| \magenta{\ls  \ep_i \Delta v + |\alpha|^2 R_i \ls |\alpha|^2 R_i,}
		\end{equation} 
		where the last  inequality follows choosing $\Delta v$ small enough. 
		\magenta{By \eqref{large.condition} and \eqref{alpha.size}, $|\alpha|^2$ is sufficiently small so that $\text{(RHS of } \eqref{delta.r})\leq \frac{R_i}{3}$, thus}	 $r(v,\alpha)$ is always comparable to $R_i$ \teal{for all $v_i \leq v\leq v_f$}. 
			\teal{Denoting $$  \Delta=\frac{|q_0| M_f [1-q] \delta_0^2}{q}\geq 100,$$ note further that by \eqref{Mi.choice}, we have $R_i \geq \delta_0 M_f \sqrt{1-q^2}$, hence \begin{equation}\label{lower.r0}
					\frac{	|\alpha|^2 R_i}{M_f} = \frac{q }{|q_0| M_f (1-q)} \frac{(1-q) M_f}{R_i} \leq \frac{\delta_0^2}{\Delta} \frac{\delta_0^{-1} (1-q)}{\sqrt{1-q^2}} \leq \frac{\delta_0}{\Delta} \sqrt{1-q^2}\leq \frac{\delta_0}{100} \sqrt{1-q^2}.
				\end{equation} Note that \eqref{lower.r0} combined with by \eqref{delta.r} ensures in particular that $r(v_f,\alpha) > R_f$. Moreover, by \eqref{delta.r} again, we find that \begin{equation}\label{lower.r}
					\frac{ r(v_f,\alpha)}{M_f} \geq 1 - (1-\frac{\delta_0}{2}) \sqrt{1-q^2}.
			\end{equation} } Then, integrating \eqref{ODE.gluing} gives \begin{equation}
				|Q|(v_f,\alpha) \approx q_0 R_i^2 \alpha^2 \approx |e|,
			\end{equation} thus, as in \cite{KehleUnger}, we can choose $\alpha$ so that \begin{equation}
				Q(v_f,\alpha) = e,
			\end{equation} exactly. Now, note that\blue{, since $\Omega^2 \equiv 1$ by gauge choice,} $r\rd_u r$ at $v=v_f$ is \blue{given by} \begin{equation}\label{dur.init}
				-r\rd_u r(v_f,\alpha) =  r (v_f,\alpha)\frac{1-\frac{2M_f}{r(v_f,\alpha)}+ \frac{e^2}{r^2(v_f,\alpha)}}{\rd_v r (v_f,\alpha)}=M_f^2  \frac{[1+\sqrt{1-q^2}-\frac{r(v_f,\alpha)}{M_f}][\frac{r(v_f,\alpha)}{M_f}-(1-\sqrt{1-q^2})]}{-r\rd_v r (v_f,\alpha)}\end{equation} As a consequence of \eqref{lower.r}, we have \begin{equation}
				-r\rd_u r(v_f,\alpha) \gtrsim\frac{\teal{ \delta_0} M_f^2 [1-q] }{R_i [-\rd_v r]}.
			\end{equation}
			
			Thus, combining \eqref{dur.init} with \eqref{lambda.est.gluing} gives, by \eqref{large.condition} and choosing  $\Delta v$ small:
			\begin{equation}
				-r\rd_u r(v_f,\alpha)   \gtrsim \Delta v \frac{|q_0|\teal{\delta_0} M_f [1-q]}{q}\gtrsim \teal{\delta_0^{-1}} \Delta v.
			\end{equation}
			
			Finally, note that for all $v \in [v_i,v_f]$, \begin{equation}
				-r\rd_u r(v,\alpha) = 	-r\rd_u r(v_f,\alpha)  -  \int_v^{v_f} [1-\frac{Q^2(v',\alpha)}{r^2(v',\alpha)} ] dv' \geq 	-r\rd_u r(v_f,\alpha) - (v_f-v) \gtrsim \teal{\delta_0^{-1}}\Delta v>0.
			\end{equation}
			
			To conclude and obtain zero ingoing scalar field derivatives at $v=v_f$, we use the argument of \cite{KehleUnger}: for each $A>0$, $F_{A}: \alpha \in \{\alpha \in \RR^{k+1},\ |\alpha|=A \} \rightarrow (\rd_u \phi(v_f,\alpha),..., \rd_u^{k} \phi(v_f,\alpha))$ is continuous and odd, so by the Borsuk--Ulam theorem,  there exists $\alpha^{*}$ such that  $|\alpha^{*}|=A$ and $(\rd_u \phi(v_f,\alpha^{*}),..., \rd_u^{k} \phi(v_f,\alpha^{*}))=(0,...,0)$.

			Thus, we have successfully  glued the  Schwarzschild event horizon $C^k$ data sphere $\mathbb{S}_i$   to the Reissner--Nordstr\"{o}m trapped $C^k$ data sphere  of radius $r(v_f,\alpha)$, mass $M_f$ and charge $e$  with no anti-trapped surface, i.e., $\rd_u r<0$. Here, $r(v_f,\alpha)>R_f$ so to conclude the proof, one trivially glues the above Reissner--Nordstr\"{o}m trapped sphere  to $\mathbb{S}_f$  by a Reissner--Nordstr\"{o}m outgoing cone composed of trapped spheres.
			
		\end{proof}
		
		\teal{\begin{rmk}\label{charged.gluing.rmk}
				The proof of Theorem~\ref{charging.thm} shows that $\phi$ \magenta{and $\Delta v$} can be chosen so that for all $v_i\leq v\leq v_f$ \begin{equation}
					\label{Delta.v.choice}
					M_f	|\rd_v\phi|(v) \lesssim \ep_i \sqrt{ \frac{|q_0| M_f}{q}},\ 	\magenta{	\Delta v \ls \ep_i^{-1}  \frac{M_f}{R_i}\frac{q}{|q_0|}}.
				\end{equation} 	Thus, for $q\in(0,\frac{1}{2})$ and assuming that \eqref{large.condition} is sharp in the sense that 
				\begin{equation}\label{large.condition.sharp}
					\frac{|q_0| M_f}{q} \approx \delta_0^{-2},
				\end{equation}
				we get that  $\phi$  can be chosen so that for all $v_i\leq v\leq v_f$ \begin{equation}\label{charged.gluing.rmk.est2}
					M_f	|\rd_v\phi|(v) \lesssim \ep_i \delta_0^{-1}.
				\end{equation}
				
		\end{rmk}}
		\subsubsection{Global charged spacetime constructions}

		We are now ready to prove Theorem~\ref{main.gluing.thm}, Theorem~\ref{charging.thm} and Corollary~\ref{charging.cor}. \begin{proof}
			We will prove Theorem~\ref{main.gluing.thm} and Theorem~\ref{charging.thm} together. Starting with Theorem~\ref{main.gluing.thm}, we set $R>0$, $q\in [0,1)$, $M_f>0$ satisfying \eqref{small.R.condition} and $\delta_0 \in (0,1)$ small such that \eqref{large.condition} is satisfied.
			
			Similarly to  the proof of Theorem~\ref{uncharged.thm} \magenta{(Step~\ref{Step1})},  we choose for some small $\eta>0$: \begin{equation}\label{choice}\begin{split}
					&R_S^{\mathcal{\tau}}= [1+\eta] R,\\ &2M_S = [1+\eta]R_S^{\mathcal{\tau}}= [1+\eta]^2 R,\\ & R_A = [1+\eta] 2M_S=[1+\eta]^2R_S^{\mathcal{\tau}}=[1+\eta]^3 R,
				\end{split}
			\end{equation} then, we apply Theorem~\ref{uncharged.gluing.thm}
			with these values, obtaining the spacelike gluing of $\mathbb{S}_R$ to the apparent horizon sphere $\mathbb{S}_A$ on some hypersurface $\magenta{\Sigma_G}$, and then the characteristic gluing of $\mathbb{S}_A$ to the trapped Schwarzschild sphere $\mathbb{S}_S^{\T}$ on some null cone $\magenta{C_G}$, in the notations of Theorem~\ref{uncharged.thm} and Figure~\ref{fig:constbif1}.  Note that, by \eqref{small.R.condition}, we have \begin{equation}\label{small.R.condition.2}
				1-\sqrt{1-q^2}< \frac{R_S^{\mathcal{\tau}}}{M_f [1+\eta]}< 1+\sqrt{1-q^2}.
			\end{equation}
			By choosing $\eta = O(\delta_0 \sqrt{1-q^2})$, one can then arrange (while keeping $\eta$ small) that \begin{equation}
				1-(1-\delta_0)\sqrt{1-q^2}< \frac{R_S^{\mathcal{\tau}}}{M_f}.
			\end{equation}
			
			If  $\frac{R_S^{\mathcal{\tau}}}{M_f}<1+(1-\delta_0) \sqrt{1-q^2}$, we choose $R_i= R_S^{\mathcal{\tau}}$. If not, we can find $R_S^{',\mathcal{\tau}}< R_S^{\mathcal{\tau}}$ such that  \begin{equation}		1-(1-\delta_0)\sqrt{1-q^2}< \frac{R_S^{',\mathcal{\tau}}}{M_f}<1+(1-\delta_0) \sqrt{1-q^2},
			\end{equation} and glue trivially $\mathbb{S}_S^{\T}$ to the trapped Schwarzschild sphere $\mathbb{S}_S^{',\T}$ of mass $M_S$ and area-radius $R_S^{',\mathcal{\tau}}$; after which we choose $R_i = R_S^{',\mathcal{\tau}}$.   Then, we set $R_f>0$ such that \begin{equation}\label{Rf.choice}
				1-(1-\delta_0) \sqrt{1-q^2} < \frac{R_f}{M_f} + \delta_0 \sqrt{1-q^2} < \frac{R_i}{M_f}.	\end{equation} 
			
			Then, we apply Theorem~\ref{charged.gluing.thm} with $M_i= M_S$, and $R_i$, $R_f$ chosen as such; note that \eqref{Mi.choice} is then satisfied. Thus, we glue characteristically $\mathbb{S}_S^{\T}$ (or $\mathbb{S}_S^{',\T}$ ) to the trapped Reissner--Nordstr\"{o}m sphere $\mathbb{S}_{RN}^{\mathcal{T}}$ of radius $R_f$, Hawking mass $M_f$ and charge $e:= \pm q M_f$. We then extend the outgoing cone $\magenta{C_G}$ to the future up to  $\mathbb{S}_{RN}^{\mathcal{T}}$.  
			
			\blue{ To prove Theorem~\ref{main.gluing.thm}, we then invoke local-existence: solving forward from induced spacelike-characteristic initial data on $\Sigma_G' \cup C_G'$ provides a spacelike hypersurface $\Sigma'$ connecting $\mathbb{S}_R$ to $\mathbb{S}_{RN}^{\T}$,  which is free of anti-trapped surfaces by Cauchy stability; this concludes the proof of Theorem~\ref{main.gluing.thm}.}

			Next, we turn to the proof of Theorem~\ref{charging.thm}: \teal{first, in accordance with \eqref{small.R.condition.2}, we will choose $q \in(0,\frac{1}{2})$, $M_f \sim R$, $R_i = R_{S}^{\T}$, $R_f= R_i - \delta_0 M_f \sqrt{1-q^2}$, and the other constants according to \eqref{choice}}. \magenta{Then, we proceed as in the proof of Theorem~\ref{uncharged.thm} and choose $\Sigma_L'$ sufficiently close to $b_{\Gamma}$.}  	\magenta{By Corollary~\ref{cor.pulse.1}, we glue $\mathbb{S}_R'$ spatially to $\mathbb{S}_A$ with a short pulse through the spacelike  hypersurface $\Sigma_G'$. Then, by   Corollary~\ref{cor.pulse.2}, we glue $\mathbb{S}_A$ to $\mathbb{S}_S^{\T}$  with a short pulse through the outgoing  null cone $C_G'$. We will denote $\Sigma'=\Sigma_L' \cup \Sigma_G'$.}

			\teal{Then,  choosing $\delta_0$ small enough if necessary, we will glue $\mathbb{S}_S^{\T}$ to $\mathbb{S}_{RN}^{\T}$ with the help of Theorem~\ref{charged.gluing.thm}.} 	\teal{Recalling Remark~\ref{charged.gluing.rmk},  we will choose $0<q<q_L$ small\footnote{\magenta{Note that $\eta =O(\delta_0)$, and $\delta_0 =O(R)$, so that $R_S^{\T} = R +O(R^2)$, and $R$ is small, consistently with the scaling of Proposition~\ref{lemma.Cauchy.christo}. \teal{This choice, when combined with \eqref{large.condition.sharp}, in turn gives $|q| =O(R^3)$ for $0<R\ll1$ sufficiently small.}}} so that  \eqref{large.condition.sharp} is satisfied. \teal{Moreover, since $M_i=M_S=(1+\eta)^2 R$, and further imposing $\eta = O(\delta_0)$, we have \begin{equation}
						|1-\frac{2M_S}{R_S^{\T}}|=		\frac{\eta}{1+\eta}=	|\ep_i| = O(\delta_0),
					\end{equation} hence, by Theorem~\ref{charged.gluing.thm} (and Remark~\ref{charged.gluing.rmk}), one can glue $\mathbb{S}_S^{\T}$ to $\mathbb{S}_{RN}^{\T}$ (with a series of short pulses) through the outgoing null cone $C_G^{charged}$, and by  \eqref{charged.gluing.rmk.est2}, \magenta{\eqref{Delta.v.choice},} there exists a constant $C$ independent of $\delta_0$ and $R$ such that \begin{equation}\label{charged.theta.bound.gluing}
						|\theta|_{|C_G^{charged}} \leq C,\ \text{ and } v_f - v_i = O( \delta_0^{-1} q).
				\end{equation} }

				We can then trivially glue $\mathbb{S}_{RN}^{\mathcal{T}}$ along a Reissner--Nordstr\"{o}m ingoing cone $\underline{C}_{RN}$ towards its past, up to a regular sphere  $\mathbb{S}_{RN}^{\mathcal{R}}$ of area-radius $$R_{RN}^{\R}:= M_f \left[ 1 +  (1+\delta_0)\sqrt{1-q^2}\right],$$ Hawking mass $M_f$ and charge $\magenta{e}:= \pm q M_f$. We will take $\delta_0$ small enough as needed, which shrinks the size of $\underline{C}_{RN}$.} \teal{Adapting the techniques of Proposition~\ref{lemma.Cauchy.christo}, in view of the short pulse condition \eqref{charged.theta.bound.gluing}, we can solve backwards for \eqref{1.1}--\eqref{5.1} with initial data on $\underline{C}_{RN} \cup C_G^{charged}$ within  its whole past domain of dependence $\mathcal{D}_{charged}$, upon taking $\delta_0=O(R)$ \magenta{and $\delta_0^{-1}|q|=O(R^2)$}. Denoting $\underline{C}_{charged}$ the ingoing component of the past boundary of $\mathcal{D}_{charged}$, one can then proceed as in the proof of Theorem~\ref{uncharged.thm}, and solve backwards with initial data on $\Sigma' \cup C_G'\cup \underline{C}_{charged}$} within  its whole past domain of dependence $\mathcal{D}(\Sigma') \cup \mathcal{D}_{\mathcal{C}}^{imp} \cup \mathcal{D}_{\mathcal{C}}^{reg}$ free of anti-trapped surfaces by construction as in the proof of Theorem~\ref{uncharged.thm}.
			Finally, we proceed as in the proof of Theorem~\ref{uncharged.thm} to construct another spacelike hypersurface $\magenta{\Sigma_0'}$  strictly to the past of $\Sigma'$ with no anti-trapped or trapped spheres  and such that $\mathcal{H}^+$ does not intersect $\magenta{\Sigma_0}'$, analogously to the situation depicted on 	\magenta{Figure~\ref{fig:constbif2}}. Note that this step requires to take $R$ small, and thus $M_f$ is small as well.   This concludes the proof of Theorem~\ref{charging.thm}.

			Once Theorem~\ref{charging.thm} is proved, the proof of Corollary~\ref{charging.cor} follows completely analogously from that of Corollary~\ref{uncharged.cor}. The proof of Theorem~\ref{OS.thm.intro} is also obtained as an immediate consequence of that of  Corollary~\ref{uncharged.cor}.
		\end{proof}
		\subsection{Construction of the asymptotically flat end for dynamical horizons}\label{right.section}  The previous gluing approaches, materialized by Theorem~\ref{uncharged.thm} and Theorem~\ref{charging.thm}, allowed gluing of any regular (uncharged) sphere to an exact Schwarzschild or sub-extremal Reissner--Nordstr\"{o}m trapped sphere, which was then extended as a larger   Schwarzschild/Reissner--Nordstr\"{o}m region (namely, an (electro)-vacuum solution of \eqref{1.1}--\eqref{5.1}), including both an open neighborhood of timelike infinity $i^+$ and an asymptotically flat end (neighborhood of spacelike infinity $i^0$), which led to the proof of Theorem~\ref{OS.thm.intro}.

		However, this is not sufficient to carry out the unconditional construction of Theorem~\ref{main.thm.global.ii}, namely construct a black hole interior  exhibiting a spacelike-null singularity, for the following reasons: \begin{itemize}
			\item The uncharged spacetime of Theorem~\ref{uncharged.thm} has a spacelike singularity $\mathcal{S}=\{r=0\}$, but no Cauchy horizon from infinity, i.e., $\CH =\emptyset$. This is due to the absence of charge ($F\equiv 0$).
			
			\item The charged spacetime of Theorem~\ref{charging.thm} has both a spacelike singularity  $\mathcal{S}=\{r=0\}$ and a null Cauchy horizon from infinity $\CH\neq \emptyset$, but a subset of this null Cauchy horizon $\CH$ is exactly  isometric to Reissner--Nordstr\"{o}m's: in particular, this subset is not weakly singular and  the scalar field is zero. A crucial aspect of Theorem~\ref{main.thm}, however,  is \emph{mass inflation} at $\CH$, and the fact that the scalar field drives the dynamics and thus must be non-zero in the vicinity of $\CH$, so one requires a different construction.
		\end{itemize}
		
		To construct a global one-ended asymptotically flat spacetime obeying the assumptions of Theorem~\ref{main.thm}, we must ensure that $\phi$ is non-zero near $\CH$, more precisely that \eqref{hyp4}--\eqref{hyp3} are satisfied.  Our approach is as follows: \begin{enumerate}
			\item \label{step1.EH} We start with a charged spacetime as in Theorem~\ref{charging.thm}, which we truncate at the regular Reissner--Nordstr\"{o}m sphere $\mathbb{S}_{RN}^{\mathcal{R}}$. Then, instead of extending the spacetime to exactly Reissner--Nordstr\"{o}m to its future as in Theorem~\ref{charging.thm}, we extend it to a dynamical horizon $\mathcal{H}^+$ such that \begin{equation}\label{EH.cond}
				\phi_{|\mathcal{H}^+}(v) = \Phi_H(v),
			\end{equation} where $\Phi_H(v)$ is an arbitrary\footnote{We moreover need to cut-off $\Phi_H$ near the Reissner--Nordstr\"{o}m junction to ensure a smooth transition.}  profile satisfying \eqref{decay.s}. This is the object of Section~\ref{EH.section} below.
			
			\item  \label{step2.EH} Keeping $\Phi_H$ arbitrary, but still assuming it satisfies \eqref{decay.s} with  $s>\frac{3}{2}$, we prove that one can construct an asymptotically flat black hole, with a (transversally) regular event horizon $\mathcal{H}^+$  on which  \eqref{EH.cond} is satisfied\blue{, at least for $v$ sufficiently large}. This step uses spherical symmetry to solve ``sideways'' and requires the smallness of the black hole charge. This is the object of Sections~\ref{redshift.section}--\ref{completion.section} below.
			
			\item  \label{step3.EH} We show that for an adequate choice of $\Phi_H(v)$ satisfying \eqref{decay.s} with  $s>\frac{3}{2}$, the assumptions \eqref{hyp4}--\eqref{hyp3} are satisfied at the Cauchy horizon $\CH$: more precisely, we require $\Phi_H$ to satisfy \eqref{EH.bound2}. We relegate the proof of this step, which relies on a refinement of the scattering strategy   of \cite{MoiChristoph} to the  next section---Section~\ref{scattering.section}.

		\end{enumerate}
		The completion of Step~\ref{step1.EH} and Step~\ref{step2.EH} in this section will result in the proof of Theorem~\ref{EH.AF.thm} below. As in Theorem~\ref{uncharged.thm}/Theorem~\ref{charging.thm}, we do not have the freedom to fix the final black hole mass $M$ arbitrarily, due to our requirement that the event horizon does not intersect the initial hypersurface (recall Remark~\ref{mass.prescribed.rmk}).  Note, in addition, that unlike Theorem~\ref{charging.thm}, the final charge ratio $q$ now cannot be chosen a priori; instead, it is only approximately prescribed, up to an arbitrarily small degree of precision $\delta$. 
		\begin{thm}\label{EH.AF.thm}
			Let $k \in \mathbb{N}$, $k\geq 2$ and  $(\mathcal{M}_L,g_L,\phi_L)$, a subset of the MGHD of \blue{$C^{k}$}
			spherically symmetric asymptotically flat initial data on a hypersurface $\Sigma_L$ with one end for \eqref{1.1}--\eqref{5.1}  containing no anti-trapped spheres  and no trapped spheres and such that $b_{\Gamma}$ is a first singularity.  
			
			Then there exists $q_L \in(0,1)$ such that for all   $q \in (0,q_L)$,  $\varsigma =\pm 1$ and $\Phi_H(v)$, a $C^{k}$ function on $[1,+\infty)$ satisfying \eqref{decay.s} for some $s>\frac{3}{2}$,  $\delta \in (0,1)$, there exists $v_0>1$ sufficiently large, $M>0$, $q' \in (0,1)$ such that \begin{equation}\begin{split}
					|q-q'|< \delta q,
				\end{split}
			\end{equation}  and $C^k$ solutions $(\mathcal{M},g,\phi)$ of \eqref{1.1}--\eqref{5.1} with $F\neq  0$ with the following properties:
			\begin{itemize}
				\item $(\mathcal{M},g,F,\phi)$ is the MGHD of 
				spherically symmetric asymptotically flat initial data on a spacelike hypersurface $\Sigma$ with one end for \eqref{1.1}--\eqref{5.1}  containing no anti-trapped spheres and no trapped spheres.
				\item The black hole region of $(\mathcal{M},g,F,\phi)$  is non-empty with an event horizon $\mathcal{H}^+$ and $\mathcal{H}^+$ does not intersect $\Sigma$, i.e., it is located in the strict causal future of $\Sigma$. Moreover, for $v_0$ large enough, \begin{equation}\label{=,p}
					\phi_{|\mathcal{H}^+}(v) = \Phi_H(v),
				\end{equation}\blue{for all $v\geq v_0$,} where $v$ in \eqref{=,p} corresponds to the Eddington--Finkelstein gauge \eqref{gauge.EH.v}.

				\item The black hole region converges to a sub-extremal Reissner--Nordstr\"{o}m black hole of mass $M$ and charge $\varsigma q' M_f$ towards $i^+$, in the sense that \begin{equation}\begin{split}
						&  \lim_{p \rightarrow i^+} \varpi(p) =  M,\\ &\lim_{p \rightarrow i^+} Q(p) = \varsigma q' M.
					\end{split}
				\end{equation}

				\item There exists an incoming null  cone $\underline{C}_{v_L}$  such that $\mathcal{M} \cap J^{-}(\underline{C}_{v_L})$ coincides with $\mathcal{M}_L \cap J^{-}(\underline{C}_{v_L})$.  Moreover, $\underline{C}_{v_L}$ can be chosen to be in the complement of the causal past of $b_{\Gamma}$. 
			\end{itemize}
		\end{thm}
		
		Similarly to the previous couples Theorem~\ref{uncharged.thm}/Corollary~\ref{uncharged.cor} and Theorem~\ref{charging.thm}/Corollary~\ref{charging.cor}, we formulate a corollary to Theorem~\ref{EH.AF.thm} applying the result where $\mathcal{M}_L$ corresponds to the FLRW spacetime of Proposition~\ref{FRLW.prop}.
		
		\begin{cor}\label{EH.AF.cor}
			Let $k \in \mathbb{N}$, $k\geq 2$. There exists $q_L \in(0,1)$ such that for all  $q \in (0,q_L)$,  $\varsigma =\pm 1$ and $\Phi_H(v)$ a $C^{k+1}$ function on $[1,+\infty)$ satisfying \eqref{EH.bound2} for some $s>\frac{3}{2}$, $\delta \in (0,1)$, there exists $v_0>1$ sufficiently large, $M>0$, $q' \in (0,1)$ such that \begin{equation}\begin{split}
					|q-q'|< \delta q,
				\end{split}
			\end{equation}  and $C^k$ solutions $(\mathcal{M},g,F,\phi)$ of \eqref{1.1}--\eqref{5.1} with $F\neq  0$ with the following properties:
			\begin{itemize}
				\item $(\mathcal{M},g,F,\phi)$ is the MGHD of 
				spherically symmetric asymptotically flat initial data on a spacelike hypersurface $\Sigma$ with one end for \eqref{1.1}--\eqref{5.1}  containing no anti-trapped spheres  and no trapped spheres.
				\item The black hole region of $(\mathcal{M},g,F,\phi)$  is non-empty with an event horizon $\mathcal{H}^+$ and $\mathcal{H}^+$ does not intersect $\Sigma$, i.e., it is located in the strict causal future of $\Sigma$. Moreover, for $v_0$ large enough, \begin{equation}\label{=,}
					\phi_{|\mathcal{H}^+}(v) = \Phi_H(v),
				\end{equation}\blue{for all $v\geq v_0$,}  where $v$ in \eqref{=,} corresponds to the Eddington--Finkelstein gauge \eqref{gauge.EH.v}.

				\item The black hole region converges to a sub-extremal Reissner--Nordstr\"{o}m black hole of mass $M$ and charge $\varsigma q' M_f$ towards $i^+$, in the sense that \begin{equation}\begin{split}
						&  \lim_{p \rightarrow i^+} \varpi(p) =  M,\\ &\lim_{p \rightarrow i^+} Q(p) = \varsigma q' M.
					\end{split}
				\end{equation}

				\item There exists an incoming null  cone $\underline{C}_{v_L}$ to the future of $b_{\Gamma}$  such that $\mathcal{M} \cap J^{-}(\underline{C}_{v_L})$ is spatially homogeneous and moreover $F\equiv 0$ in  $\mathcal{M} \cap J^{-}(\underline{C}_{v_L})$. Moreover, $\mathcal{S}_L:=  \mathcal{S}\cap J^{-}(\underline{C}_{v_L})$ is spacelike and coincides with the singularity of a FLRW metric with $\RR^3$ topology.

				\item The MGHD terminal boundary of   $(\mathcal{M},g,F,\phi)$ is \begin{equation}
					\CH\cup	\mathcal{S},
				\end{equation} where $\CH\neq \emptyset$ (the Cauchy horizon) is a null boundary emanating from $i^+$ on which $r$ extends to a non-zero function, and $r$ extends to $0$ on $\mathcal{S}$, which is a curvature singularity. Moreover, near $\Gamma$, $\mathcal{S}$ is spacelike and spatially-homogeneous.

			\end{itemize}
		\end{cor}
		\subsubsection{Construction of spherically symmetric event horizons}\label{EH.section}
		
		We start by constructing the event horizon with ``initial conditions'' at timelike infinity $i^+$, solving the null constraints (system of ODEs).
		\begin{prop}\label{EH.prop} Let $(M,e)$ such that $0<|e|<M$,  $C^1$ functions $\phiH(v)$ satisfying \eqref{EH.bound2}  for some $s>1$ and $\phi_{in}(U)$ satisfying \eqref{EH.boundU}, with $\phi_{in}(0)=\phiH(v_0)$. Then, for sufficiently large $v_0$ and sufficiently small $U_S>0$ and imposing the gauge condition  \eqref{gauge.EH.v} on $\HH=\{U=0\} \times [v_0,+\infty)$ and the gauge \eqref{gauge.EH.U} on $[0,U_S]\times\{v_0\}$, there exist unique solutions $(r,\Omega,Q)$  of the ODE system in $v$  consisting of    \eqref{Radius}, \eqref{ChargeVEinstein}, \eqref{RaychV} on $\HH=\{U=0\} \times [v_0,+\infty)$   and of the ODE system consisting of   \eqref{chargeUEinstein}, \eqref{RaychU}  on $[0,U_S]\times\{v_0\}$ such that  
			\begin{align}
				& \lim_{v\to +\infty} r(0,v) = r_+(M,e)=M+\sqrt{M^2+e^2},\\
				& \lim_{v\to +\infty} Q(0,v) = e.
			\end{align}
			Moreover, $\HH$ is in the regular region, namely $\rd_v r(0,v) \geq 0$ for all $v\geq v_0$, and in fact $\rd_v r(0,v)>0$ if $\Phi_H$ is not identically zero.
		\end{prop}
		\begin{proof} We already prescribed $\phi_{|\HH}(v) = \phi_H(v)$, where $\phiH(v)$ satisfies \eqref{decay.s}  for some $s>1$. 	On $\HH$, we work with the variables $(r,\lambda, \Omega^2,Q)$ subjected  to the following system of ODEs  \eqref{Radius}, \eqref{ChargeVEinstein}, \eqref{RaychV} which we re-write, under the gauge condition \eqref{gauge.EH.v}, \eqref{A.gauge} as: \begin{align} & \rd_v r = \lambda,\\ 
				& \rd_v( \frac{\lambda}{\Omega^2} ) = -\frac{r |\rd_v \phiH|^2}{\Omega^2},\label{ODE.lambda}\\ & \rd_v(r \Omega^2)= \Omega^2[1-\frac{Q^2}{r^2}],\label{Omega.EH}\\ & \rd_v Q = - q_0r^2 \Im(\overline{\phiH}\rd_v \phiH)\label{Q.EH}.
			\end{align}
			Let $r_+>|e|> 0$.	Then, we impose initial conditions at $v=+\infty$ for this system of ODEs, as such:
			\begin{align}  	&  \lim_{v \rightarrow +\infty}r(0,v)= r_+,\\ &   \lim_{v \rightarrow +\infty}\frac{\lambda}{\Omega^2}(0,v)= 0,\\ & \lim_{v \rightarrow +\infty}Q(0,v)=  e,\end{align}
			together with an initial condition at $v=v_0$: \begin{equation}
				\label{Omega.data} \Omega^2(v_0) = e^{2K_+ v_0},
			\end{equation}
			where \eqref{Omega.data} is imposed to respect the gauge compatibility of gauge~\eqref{gauge.EH.v} and gauge~\eqref{gauge.EH.U} at the sphere $(U,v)=(0,v_0)$.	Note that, by \eqref{Omega.EH},  the above conditions implies that \begin{equation}
				\lim_{v\rightarrow+\infty} \rd_v \log(r\Omega^2)(v) = \lim_{v\rightarrow+\infty} \rd_v \log(\Omega^2)(v) = \frac{1-\frac{e^2}{r_+^2}}{r_+}:=2K_+>0,
			\end{equation}  
			Then, we solve backwards the ODE system in $v$ with unknown $(r,\log(\Omega^2),Q)$ and final condition for $\log(\Omega^2)$ determined by \eqref{Omega.data} \begin{equation}
				\lim_{v\rightarrow+\infty} \log(\Omega^2) - 2K_+ v= \log(A_0),
			\end{equation} where $A_0>0$ is precisely chosen so that \eqref{Omega.data} holds (note indeed that \eqref{ODE.lambda}, \eqref{Omega.EH} are invariant by rescaling $\Omega^2$ by a constant): this procedure produces  a unique solution to the system of ODE with above conditions for $v \in [v_0,+\infty)$, assuming $v_0$ is large enough.
			Defining $\varpi(v)$ through the explicit formula: \begin{equation}
				1-\frac{2\varpi(v)}{r(v)} + \frac{Q^2(v)}{r^2(v)} = \lambda(v),
			\end{equation} we find that there exists $M>0$ such that \begin{equation}\label{M.formula}
				\lim_{v\rightarrow+\infty} \varpi(v) = M=\blue{\frac{r_+}{2} [1+\frac{e^2}{r_+^2}]}
			\end{equation}\blue{such that} $$ r_+ = M +\sqrt{M^2-e^2},\ 2K_+ = \frac{2}{r_+^2} [M-\frac{e^2}{r_+}].$$ Note that the formula \eqref{M.formula} allows to fix $r_+$ so that the couple $(M,e)$ with $0\leq |e|<M$ equates its prescribed values.
			
			The analogous construction on $[0,U_s] \times \{v_0\}$ is   straightforward, for $U_s$ small enough.
		\end{proof}

		\begin{cor}\label{EH.cor}
			On the event horizon $\HH= \{U=0,\ v\geq v_0\}$, the following estimates hold true: 
			\begin{equation}	\label{lambda.EH}
				0 \leq \lambda(0,v) \lesssim v^{-2s},\end{equation}
			\begin{equation}	\label{r.EH}
				0 \leq r_+ -r(0,v) \lesssim v^{1-2s},	\end{equation}
			\begin{equation}	\label{dvlogOmega.EH}
				|\rd_v \log(\Omega^2)(0,v) - 2K_+|\lesssim v^{1-2s},\end{equation}
			\begin{equation} \label{dUlogOmega.EH} 
				|\rd_U \log(\Omega^2)|(0,v) \lesssim e^{2K_+v},\end{equation}
			\begin{equation} \label{dUphi.EH} 
				|D_U \phi|(0,v)\lesssim e^{2K_+ v} v^{-s},\end{equation}
			\begin{equation}\label{EH.Q.M}
				|Q(0,v) - e|,\ |\varpi(0,v) - M| \lesssim v^{1-2s}.
			\end{equation}
			
			\noindent	Moreover, for any sufficiently small, but fixed $\delta_0 \in (0,1)$ (independent of $v_0$), we can choose $U_S$ so that $$ r(U_S,v_0) = r_+[1-\delta_0],$$ and the sphere   $(U_S,v_0)$ is trapped.
		\end{cor}
		\begin{proof}
			This is the content of Proposition 4.2 and Proposition 4.4 in \cite{Moi} but we provide a brief sketch for the reader. Using the fact that $r(0,v)$  is bounded and integrating  \eqref{ChargeVEinstein}, \eqref{massVEinstein}  gives \eqref{EH.Q.M}.  Then, integrating \eqref{Omega.EH} and using \eqref{EH.Q.M} gives \eqref{dvlogOmega.EH}. 
			From \eqref{RaychV} and $\underset{v\rightarrow+\infty}{\lim}\lambda(0,v) =0$, we know that $\lambda(0,v) \geq 0$, and using \eqref{dvlogOmega.EH} gives \eqref{lambda.EH}. \eqref{dUlogOmega.EH},\eqref{dUphi.EH} then follow integrating \eqref{Omega} and \eqref{Field} in $v$, respectively. 
			
			Regarding the claims on the sphere $(U_S,v_0)$, this follows from the fact that $U_A(v) e^{2K_+v} \ls \lambda(0,v) \ls v^{-2s}$, where $U_A(v)$ is defined so that $(U_A(v),v) \in \mathcal{A}$ (see again  Proposition 4.2 and Proposition 4.4 in \cite{Moi}).
		\end{proof}
		
		Now, we address the first step of Theorem~\ref{EH.AF.thm} in constructing the requested spacetime up to its event horizon $\mathcal{H}^+$. The corresponding black hole exterior region (including an affine complete null infinity to the future, and an asymptotically flat end) will be constructed in Sections~\ref{redshift.section}--\ref{completion.section} below. 
		\begin{prop}\label{EH.AF.1st.prop} [Construction of the black hole interior region].
			Let $k \in \mathbb{N}$, $k\geq 2$ and  $(\mathcal{M}_L,g_L,\phi_L)$, a subset of the MGHD of \blue{$C^k$}
			spherically symmetric asymptotically flat initial data on a hypersurface $\Sigma_L$ with one end for \eqref{1.1}--\eqref{5.1}  containing no anti-trapped spheres and no trapped spheres.

			Let  $q \in (0,1)$. There exists \blue{$M_f>0$} 
			such that for all $\varsigma =\pm 1$ 
			and $\Phi_H(v)$, a $C^k$ function on $[1,+\infty)$ satisfying \eqref{decay.s} for some $s>1$ and $\delta\in(0,1)$ a  sufficiently small number, there exist $v_0>1$ sufficiently large,  $(M,e)$ real numbers such that $0<|e|<M$ and \begin{equation}\begin{split}
					&	|M-M_f| < \delta M_f,\\ &  	|e-\varsigma q M_f| < \delta q M_f,
				\end{split}
			\end{equation} and $C^k$ solutions $(\mathcal{M},g,F,\phi)$ of \eqref{1.1}--\eqref{5.1} with $F\neq  0$ with the following properties:

			\begin{itemize}
				\item $(\mathcal{M},g,F,\phi)$  is a  spherically symmetric solution of \eqref{1.1}--\eqref{5.1}. Its past boundary  $ \mathcal{H}^+$ 
				\blue{intersects the} center $\Gamma$ \blue{in its past} and $\mathcal{H}^+$ is a null affine-complete outgoing cone on which $r$ is bounded and strictly increasing towards the future. 	Moreover, $\mathcal{H}^+$	 contains no anti-trapped spheres  and no trapped or marginally-trapped spheres.
				We denote $i^+$ the future endpoint of the hypersurface  $\mathcal{H}^+$.

				\item $ \mathcal{H}^+$ converges to a sub-extremal Reissner--Nordstr\"{o}m black hole of mass $M_f$ and charge $\varsigma q M_f$ towards $i^+$, in the sense that \begin{equation}\begin{split}
						&  \lim_{p \rightarrow i^+} \varpi_{|\mathcal{H}^+}(p) =  M_f,\\ &\lim_{p \rightarrow i^+} Q_{|\mathcal{H}^+}(p) = \varsigma q M_f.
					\end{split}
				\end{equation}

				\item There exists an incoming null  cone $\underline{C}_{v_L}$ such that $\mathcal{M} \cap J^{-}(\underline{C}_{v_L})$ coincides with $\mathcal{M}_L \cap J^{-}(\underline{C}_{v_L})$.  
			\end{itemize}
		\end{prop}
		\begin{proof}
			Let $(\mathcal{M},g,F,\phi)$ the spacetime constructed as an application of Theorem~\ref{charging.thm}, whose assumptions are satisfied. We denote $M_{f}$ the mass of the Reissner--Nordstr\"{o}m black hole, and $e_f$ its charge, with $$ e_f= \varsigma q M_f,$$ and, as in  Theorem~\ref{charging.thm}, \blue{recall the presence of the Reissner--Nordstr\"{o}m  regular sphere $\mathbb{S}_{RN}^{\R}$  with radius $$ R_{RN}^{\R} = M_f [1 + (1+\delta_0) \sqrt{1-q^2}],$$ where $\delta\in (0,1)$ is a small number.}

			Let $M$ sufficiently close to $M_f$, $e$ sufficiently close to $e_f$, to be chosen later. We choose the $v$-coordinate according to the gauge choice \eqref{gauge.EH.v} on the event horizon $\mathcal{H}^+$ and since the metric is Reissner--Nordstr\"{o}m in a broader region, we have\begin{equation}\label{gauge.RN}
				\frac{-4\rd_u r}{\Omega^2}(u,v) =1, \text{ which is equivalent to } \rd_v r= 1-\frac{2M_f}{r} + \frac{e_f^2}{r^2}
			\end{equation} for all $u\leq u_{RN}^{\blue{\R}}$ and $v\geq v_{RN}^{\blue{\R}}$, where  $\mathbb{S}_{RN}^{\blue{\R}}=(u_{RN}^{\blue{\R}},v_{RN}^{\blue{\R}})$. 	Let $v_0>1$ and $U_S>0$ such that the conclusion of Proposition~\ref{EH.prop} applies. With no loss of generality and by taking $v_0$ larger if necessary, we can assume that $v_0 > v_{RN}^{\blue{\R}}$. \magenta{We  then construct a dynamical event horizon: first, extend $\mathbb{S}_{RN}^{\R}$ into the Reissner--Nordstr\"{o}m regular cone $\{u_{RN}^{\R}\} \times [v_{RN}^{\R},v_0]$.} \blue{Note  that, defining the tortoise coordinate 
				$r_*(u,v,M,e) = r(u,v)+ \frac{1}{2K_-(M,e)}\log(r(u,v)-r_-(M,e))  + \frac{1}{2K_+(M,e)}\log(r_+(M,e)-r(u,v))$, then, by \eqref{gauge.RN}, we have the following identity: \begin{equation}\label{r.v.eq}
					v_0 - v_{RN}^{\R} = r^{*}(u_{RN}^{\R},v_0,M_f,e_f) -  r^{*}(u_{RN}^{\R}, v_{RN}^{\R} ,M_f,e_f).
				\end{equation}

			}

			\noindent		Then, \magenta{in the $U$-gauge \eqref{gauge.EH.U}, we set $U(u_{RN}^\R)=0$, and we collate  bicharacteristic initial data on $\left([0,U_S] \times \{v_0\}\right) \cup \left(\{U=0\} \times [v_0,+\infty)\right)$ to $\{U=0\} \times [v_{RN}^{\R} , v_0]$. Then,} we apply Proposition~\ref{EH.prop} with \blue{the} choice $(M,e)$ to the bicharacteristic initial data on $\left([0,U_S] \times \{v_0\}\right) \cup \left(\{U=0\} \times [v_0,+\infty)\right)$ as follows:  \begin{equation}\begin{split}
					&	\phi(U,v_0) \equiv 0,\ 	\varpi(U,v_0) \equiv M_f,\ Q(U,v_0) \equiv  \varsigma q M_f, \\ &  \phi(0,v) = \left(1-\chi(v-v_0)\right) \Phi_H(v),
				\end{split}
			\end{equation} where $\chi$ a cut-off function such that $\chi(x)=1$ when $0\leq x\leq 1$, and $\chi(x)=0$ when $x\geq 2$. Note that \begin{equation}\label{M.Q.R.est}
				|\varpi(0,v_0) - M| \ls v_0^{1-2s},\  |Q(0,v_0) - e| \ls v_0^{1-2s},\ \blue{ |r(0,v_0)- r_+(M,e)|,\   |r(0,v_0)- r_+(\varpi(0,v_0),Q(0,v_0))|\ls v_0^{1-2s},}
			\end{equation} so we can choose $(M,e)$ to be $(M_f,e_f )+O( v_0^{1-2s})$, so that 
			$\varpi(0,v_0)=M_f,\ Q(0,v_0)= e_f$.
			\blue{By \eqref{r.v.eq} and \eqref{M.Q.R.est}, we also have $r^{*}(u_{RN}^{\R}, v_{RN}^{\R} ,M_f,e_f)= -v_0 + v_{RN}^{\R} +O(\log(v_0))$, implying  that \begin{equation}
					r(u_{RN}^{\R}, v_{RN}^{\R})= R_{RN}^{\R}= M_f [1+(1+\delta_0 )\sqrt{1-q^2}] = r_+(M_f, e_f) +O(e^{-v_0}),
				\end{equation}  thus   $\delta_0$ can be made arbitrarily small for $v_0$ sufficiently large and, also choosing $M_f$ small, we can repeat the proof of Theorem~\ref{charging.thm} to construct the spacetime region $\{U \geq 0,\ v_{\Gamma}(U)\leq v\leq v_{RN}^{\R} \}$, which we collate to the outgoing cone $\{U=0  \}\times [v_{RN}^{\R},+\infty)$: note that $\HH:=\{U=0\} \times  [v_{\Gamma}(U=0),+\infty) \subset \R$,  and $r_{|\HH}$ is bounded.}

			\blue{Finally,} we apply Theorem~\ref{CH.thm.SS} to obtain  a solution of \eqref{1.1}--\eqref{5.1} in the spacetime rectangle $(U,v) \in [0,U_S] \times [v_0,+\infty)$, which is exactly  Reissner--Nordström on $[0,U_S] \times \blue{[v_{RN}^{\R},v_0]}$\blue{: this completes the proof}.

		\end{proof}
		\subsubsection{Backwards propagation from the event horizon to a constant-$r$ curve}\label{redshift.section}
		
		We now consider the spacetime of Proposition~\ref{EH.AF.1st.prop}, which is the causal future  $\Sigma \cup \mathcal{H}^+$, a spherically symmetric solution of \eqref{1.1}-\eqref{5.1} with $F\neq 0$. We additionally assume that \eqref{EH.bound2} holds on $\mathcal{H}^+$ for some $s>\frac{3}{2}$ (note that this is more demanding that the assumption $s>1$ from Proposition~\ref{EH.AF.1st.prop}). With no loss of generality, we  assume that $\frac{3}{2}<s<2$ for convenience of notation. We denote \begin{equation}
			\lim_{v\rightarrow+\infty} \varpi_{|\mathcal{H}^+}(v)=M>0,\ 	\lim_{v\rightarrow+\infty} Q_{|\mathcal{H}^+}(v)=e\neq 0,
		\end{equation} where we choose  $|e|$ sufficiently small, possibly depending on $q_0$, $M$ and $s$: in particular, $|e|<M$. In this section, we begin the construction of an asymptotically flat end that we will ``glue'' to the past of this spacetime. 
		Let $\ep>0$ chosen so that $\ep e^{2K_+ v_0}$ is sufficiently small. We impose initial data on the ingoing cone $[-\ep,0] \times \{v_0\}$ satisfying \eqref{EH.boundU}. Our first objective is to solve ``rightward'' in the spacetime rectangle $[-\ep,0] \times \{v_0\}\cup \{0\} \times [v_0,+\infty)$, for $\ep$ small enough and $v_0$ large enough. However, in this section, we will only solve up to a timelike curve $\gamma_{R_0}$ on which $r$ converges to a large constant $R_0>r_+$. The construction of the solution to the future of $\gamma_{R_0}$ will be completed in the next Section~\ref{timelike.section} and Section~\ref{nullinf.section}.

		We recall the gauge condition \eqref{gauge.EH.v}, which we supplement with the following gauge choice for $U$: \begin{equation}\label{U.gauge.past}
			\rd_U r(U,v_0) =  -e^{2K_+(M,e) v_0},
		\end{equation}   corresponding to gauge \eqref{gauge.EH.U}. We will later switch to the different $u$-gauge \eqref{gauge.U.past} (see Section~\ref{nullinf.section}), although we first need to show it is well-defined, since it is a teleological gauge.
		
		Then, we can apply Corollary~\ref{EH.cor} to obtain estimates on the event horizon $\mathcal{H}^+ = \{0\}\times [v_0,+\infty)$. We will establish estimates in the following red-shift region defined as \begin{equation}
			\mathcal{R}= \{ -\ep \leq U \leq 0,\ v_0 \leq v \leq v_{	\mathcal{R}}(U)\},\  v_{	\mathcal{R}}(U)= v_0+ [2K_+]^{-1}\ln(\frac{\ep}{|U|}) = v_0  +u(U)- [2K_+]^{-1} \log(\ep^{-1}),
		\end{equation}  where $u(U)$ is defined so that $-U = [2K_+]^{-1}e^{-2K_+ u}$ and $\Delta = [2K_+]^{-1} \ln(\ep^{-1})$.
		\begin{prop}\label{RS.prop}
			For all $(U,v) \in 	\mathcal{R}$, the following estimates hold: 
			\begin{equation}\label{SF.R}
				|\phi|(U,v) + |\rd_v \phi|(U,v) \ls v^{-s},
			\end{equation}
			\begin{equation}\label{SF.U.R}
				|D_U \phi|(U,v) \ls e^{2K_+ v} v^{-s},
			\end{equation}
			\begin{equation}\label{kappa.R}
				|\kappa(U,v) -1| \ls \ep  v^{-2s},
			\end{equation}
			\begin{equation}\label{Q.M.R}
				|\varpi(U,v) -M|,\ |Q(U,v) -e|  \ls   v^{1-2s},
			\end{equation}
			\begin{equation}\label{Omega.R}
				\bigl|\log(\Omega^2)(U,v)- 2K_+(u+v) \bigr| \ls |U| e^{2K_+v},
			\end{equation}
			\begin{equation}\label{duOmega.R}
				|\rd_U\log	(\Omega^2)(U,v)|,\ |A_U|(U,v) \ls e^{2K_+v},
			\end{equation}\begin{equation}\label{dvOmega.R}
				|\rd_v\log	(\Omega^2)(U,v) - 2K(u,v)|\ls v^{-2s}.
			\end{equation}
			
			\noindent	On the curve $\gamma_{	\mathcal{R}}= \{(u,v_{	\mathcal{R}}(u)),\ u \geq \Delta\}$  \begin{equation}\label{iota.R}
				\bigl|	\iota(u,v_{	\mathcal{R}}(u))- 1\bigr| \ls u^{-s},
			\end{equation}
			\begin{equation}\label{dulogOmega.R}
				\bigl|\rd_u \log(\Omega^2)(u,v_{	\mathcal{R}}(u))+2K(u,v_{	\mathcal{R}}(u))\bigr|\ls u^{-s}.
			\end{equation}
		\end{prop}
		
		\begin{proof}
			The proof follows entirely from that of Proposition 4.5 in \cite{Moi}, but we give a brief sketch of the proof for the reader. We make the following bootstrap assumptions:
			\begin{equation} \label{B2}
				|\phi|+|\partial_v \phi| \leq 4C v^{-s},
			\end{equation}
			\begin{equation} \label{Bnew}
				|D_U \phi|	 \leq D e^{2K_+ v} v^{-s},
			\end{equation}
			\begin{equation} \label{B3}
				|\rd_U r|	 \leq D e^{2K_+ v},
			\end{equation}
			\begin{equation} \label{B4}
				\frac{1}{2} \leq \kappa \leq 1,
			\end{equation}
			\begin{equation} \label{B5}
				|Q-e| \leq 4 |e|,
			\end{equation} where $C>0$ is chosen so that  \begin{equation} 
				|\phi|(0,v)+|\partial_v \phi|(0,v) \leq C v^{-s}.
			\end{equation}
			First, we integrate \eqref{B3} in $U$ to obtain for $v_0$ large enough and using  Corollary~\ref{EH.cor} \begin{equation}
				|r(U,v) - r_+ | \lesssim  |r(0,v) - r_+| + |U| e^{2K_+ v} \lesssim \frac{r_+}{10},
			\end{equation} where in the last line we have used the fact that  $|U| e^{2K_+ v}\leq \ep e^{2K_+ v_0}$ in this region. Similarly integrating \eqref{massUEinstein} and \eqref{chargeUEinstein} gives
			\begin{equation} 
				| \varpi(U,v) - M|,\ 	| Q(U,v) -e| \lesssim v^{1-2s}.
			\end{equation}	which improves on \eqref{B5} for $v_0$ large enough. Then integrating \eqref{dv.Au} give
			$$ |A_U| \lesssim  e^{2K_+v},$$ hence
			$$ |\partial_U \phi | \lesssim D  e^{2K_+ v} v^{-s}.$$
			We can then  integrate  to get : 
			\begin{equation} \label{phinew}
				| \phi(U,v)-\phi(0,v)| \lesssim  \ep e^{2K_+ v_0} v^{-s},
			\end{equation} thus \eqref{B2} is improved for $\ep$ small enough.	Let $a>0$. We can rewrite \eqref{Field} together with \eqref{Radius} as: 
			\begin{equation} \label{Fieldshift}
				\partial_v ( e^{av} r\frac{D_U \phi}{\nu_H}) =  \left( a- \kappa( 2K- rm^2 |\phi|^2) \right) e^{av} r\frac{D_U \phi}{\nu_H} - e^{av}\partial_v \phi +\kappa e^{av} r m^2 \phi .
			\end{equation}By the earlier estimates, we get
			$$ |K(U,v)-K_+| \ls \ep e^{2K_+ v_0} ,$$ and therefore $K(U,v)$ is lower bounded by $\frac{3K_+}{4}$ for small enough $\ep$.	Choosing say $0<a < \frac{K_{+}}{4}$ and integrating \eqref{Fieldshift}  then improves on  bootstrap \eqref{Bnew}.
			
			Integrating \eqref{RaychU} in $U$ using the above estimates then also improves  on \eqref{B4} for $\ep$ small enough. 	The estimates on $\Omega$ then easily follow from the integration of \eqref{Omega}; the reader can consult \cite{Moi} for details.
			
			\noindent	\eqref{iota.R} and \eqref{dulogOmega.R} are a bit trickier to prove but the argument is identical to that of Proposition 4.6 in \cite{Moi}.

		\end{proof}

		Next, we turn to the  region where $r$ is bounded and away from the event horizon $\mathcal{H}^+$, i.e., the region $\{  v_{	\mathcal{R}}(U)\leq v\leq v_{R_0}(U),\ u \geq \Delta \}$, where $v_{R_0}(u)=R_0+u$, where $R_0> r_+$ is a large constant to be determined later. This region is  analogous to the no-shift region $\mathcal{N}$ in \cite{Moi} and the proof of the estimates is very similar, using a Gr\"{o}nwall-like iteration. 
		
		\begin{prop} \label{noshift.prop}
			
			\blue{For all $(u,v)\in\{  v_{	\mathcal{R}}(U)\leq v\leq v_{R_0}(U),\ u \geq \Delta \}$,}	the following estimates hold:
			\begin{equation} \label{phivNSprop}
				|\phi|(u,v)+|\partial_v \phi|(u,v) \lesssim_{R_0} v^{-s}
			\end{equation}	\begin{equation} \label{phiUNSprop}
				|D_u \phi|(u,v)  \lesssim_{R_0}  v^{-s} ,
			\end{equation}
			
			\begin{equation} \label{OmegaPropNSprop}
				|\log(\frac{\Omega^2(u,v)}{4})-\log(1-\frac{2M}{r(u,v)}+\frac{e^2}{r^2(u,v)})|\lesssim_{R_0}  v^{1-2s} ,
			\end{equation}  
			\begin{equation} \label{kappaNStpropprop}
				0 \leq 1-\kappa(u,v) \lesssim_{R_0} v^{-2s},
			\end{equation}						\begin{equation} \label{iotaNStpropprop}
				|1-\iota(u,v)| \lesssim_{R_0}  v^{-s},
			\end{equation}
			\begin{equation} \label{partialuNSOmegaprop2}
				|\partial_u \log(\Omega^2)(u,v)+2\frac{M-\frac{e^2}{r(u,v)}}{r^2(u,v)}  |  \lesssim_{R_0}   v^{-s},
			\end{equation}
			\begin{equation} \label{partialvNSOmegaprop2}
				|\partial_v \log(\Omega^2)(u,v)-2K(u,v) |  \lesssim_{R_0}   v^{-2s},
			\end{equation}

			\begin{equation} \label{QNS2} |Q(u,v)-e| \lesssim_{R_0}  v^{1-2s}, 
			\end{equation}
			\begin{equation} \label{MNS2} |\varpi(u,v)-M| \lesssim_{R_0}  v^{1-2s} .
			\end{equation}

		\end{prop}
		\begin{proof}
			For the proof, we cut $\mathcal{N}=\{  v_{	\mathcal{R}}(U)\leq v\leq v_{R_0}(U),\ u \geq \Delta \}$ into small regions \begin{equation*}
				\mathcal{N}=\overset{N}{\underset{i=1}{\cup}} \mathcal{N}_i,
			\end{equation*} where  $\mathcal{N}_i= \{ v_{i-1}(u)
			\leq v \leq  v_{i}(u),\ u\geq \Delta\}$ and $v_0(u)= v_{\mathcal{R}}(u)$, $v_N(u) = v_{R_0}(u)$ and $v_{i}(u) - v_{i}(u) =\eta_0$, a small number. We then prove the result by induction on $i$, making use of a standard bootstrap method and obtaining exponential growth in $N$. The reader interested in the details can consult \cite{Moi}, Appendix B which employs the same proof.
		\end{proof}

		\subsubsection{Backwards propagation from a constant-$r$ curve to a timelike geodesic}\label{timelike.section} We define the timelike curve  $\gamma= \{ (u,v_{\gamma}(u)),\ v_{\gamma}(u)  = \frac{ 3u}{2},\  u \geq \Delta \}$; note that $\gamma$ is not technically a timelike geodesic but it approaches one as $u\rightarrow+\infty$, corresponding to hyperbolic motion on the Reisser--Nordström \blue{black hole}. We will prove estimates in the following  spacetime region \begin{equation}
			\{ u \geq \Delta,\ v_R(u) \leq v \leq v_{\gamma}(u)\}.
		\end{equation} The proof relies on the use of the $r^p$-method and takes advantage of the smallness of $e$, similarly to the approach employed in \cite{Moi2}. However, a major difference is that the well-controlled boundary terms are located on a (timelike) curve $\gamma_{R_0}$ to the past of the region of integration, and therefore one can only apply for $r^p$ method for $p<0$, trying to keep $|p|$ as small as possible.
		\begin{prop}\label{timelikecurve.prop}\blue{Let $\eta > \frac{-1+\sqrt{1+4q_0^2 e^2}}{2}$. Assuming $|q_0 e|$ is sufficiently small, the following estimates hold for all $(u,v)\in \{ u \geq \Delta,\ v_R(u) \leq v \leq v_{\gamma}(u)\}$}
			\begin{equation}\label{phi.gamma}
				|\phi|(u,v) \lesssim u^{-s+\frac{1}{2}+\eta},
			\end{equation}
			\begin{equation}\label{duphi.gamma}
				r|D_u \phi|(u,v),\	|D_u \psi|(u,v) \lesssim u^{-s+\frac{1}{2}+\eta},
			\end{equation}
			\begin{equation} \label{dvphi.gamma}
				r|\rd_v \phi|(u,v),\ |\rd_v \psi|(u,v) \lesssim u^{-s+\frac{1}{2}+\eta},
			\end{equation}
			\begin{equation} \label{M.Q.gamma}
				|\varpi(u,v) -M |,\	|Q(u,v) -e | \ls  r^{1+\eta}u^{-2s}\ls  u^{-2s+1+\eta},
			\end{equation}\begin{equation} \label{iota.gamma}
				|\iota(u,v ) -1|\ls  u^{-s},
			\end{equation}
			\begin{equation} \label{kappa.gamma}
				|\kappa(u,v ) -1|\ls \log(u) u^{-2s+1+2\eta},
			\end{equation}
			\begin{equation}\label{Omega.gamma}
				\bigl|	\frac{\Omega^2(u,v)}{4}  - (1-\frac{2M}{r(u,v)} + \frac{e^2}{r^2(u,v)})\bigr|\ls  u^{-s},
			\end{equation}
			\begin{equation}
				\bigr|	\rd_v \log(\Omega^2) - \frac{2}{r^2} [M-\frac{e^2}{r}]\bigl|\ls u^{-2s+1+2\eta}.
			\end{equation}
		\end{prop}

		\begin{proof}

			We make the following bootstrap asumptions:\phantom{\eqref{B2.gamma},\eqref{B3.gamma}} \begin{equation}\label{B1.gamma}
				|Q|(u,v) \leq 2 |e|,
			\end{equation}\begin{equation}\label{B2.gamma}
				|	\frac{\Omega^2(u,v)}{4}-1| \leq \frac{10M}{r},
			\end{equation}\begin{equation}\label{B3.gamma}
				|\rd_u r+1|(u,v) \leq \frac{10M}{r},
			\end{equation}\begin{equation}\label{B4.gamma}
				|\rd_v r-1|(u,v) \leq \frac{10M}{r}.\end{equation}
			
			Let $p>0$.	Then, we take advantage of \eqref{Field.psi} which we multiply by 
			$ r^{-p} \overline{\rd_v\psi} $ and we take the real-part. Note the identity \begin{equation}
				r^{-p} \Re(\overline{\rd_v \psi} D_u \rd_v \psi )=   r^{-p}  \rd_u( \frac{|\rd_v \psi |^2}{2}) =  \rd_u (r^{p}   \frac{|\rd_v \psi |^2}{2}) - p r^{-p-1} [-\nu] |\rd_v \psi |^2,
			\end{equation}  from which we get, under the bootstrap assumptions   \eqref{B1.gamma}--\eqref{B4.gamma}, and taking $R_0$ large enough  \begin{equation}\label{ID.gamma}
				\bigl|	-\rd_u( r^{-p}   \frac{|\rd_v \psi |^2}{2}) + pu r^{-p-1}|\nu|  |\rd_v \psi |^2\bigr| \leq   2|q_0 e| r^{-2-p} |\psi| |\rd_v\psi|.
			\end{equation}
			To control the RHS, we will first prove a  Hardy inequality. Let $ v_{R_0}(u) \leq v \leq v_{\gamma}(u)$ and $P>-2$: \begin{equation}\begin{split}
					&\int_{v_{R_0}(u)}^{v}  r^{-P-3} |\psi|^2 dv' \ls	\int_{v_{R_0}(u)}^{v} \lambda r^{-P-3} |\psi|^2 dv' = 	-[2+P]^{-1}\int_{v_{R_0}(u)}^{v}\rd_v (r^{-P-2}) |\psi|^2 dv' \\ &\leq  [2+P]^{-1} R_0^{-P} |\phi|^2(u,v_{R_0}(u))+ \frac{2}{2+P} \int_{v_{R_0}(u)}^{v}r^{-P-2} |\psi \rd_v \psi|dv \ls u^{-2s} +  	\int_{v_{R_0}(u)}^{v} r^{-P-1} |\rd_v\psi|^2 dv',\end{split}
			\end{equation} where we have invoked Proposition~\ref{noshift.prop} and the bootstrap assumption \eqref{B4.gamma}. In fact, we even proved the following stronger  and more precise estimates\blue{: there exists a constant $C_P>0$ such that} \begin{equation}\label{Hardy.gamma}\begin{split}
					& \frac{r^{-P}}{2+P} |\phi|^2(u,v)+	\int_{v_{R_0}(u)}^{v} r^{-P-3} |\psi|^2 dv' \leq \frac{4}{(P+2)^2}	\int_{v_{R}(u)}^{v} r^{-P-1} |\rd_v\psi|^2 dv'+ C_{P}  u^{-2s},\end{split}
			\end{equation} and \begin{equation}\label{Hardy.gamma2}
				\int_{v_{R_0}(u)}^{v} r^{-P-2} |\psi| |\rd_v \psi| dv'\ls [\int_{v_{R_0}(u)}^{v} r^{-P-1}  |\rd_v \psi|^2 dv']^{\frac{1}{2}} [\int_{v_{R_0}(u)}^{v} r^{-P-3}  | \psi|^2 dv']^{\frac{1}{2}} \leq  \frac{2}{P+2}	\int_{v_{R_0}(u)}^{v} r^{-P-1} |\rd_v\psi|^2 dv'+ C_P u^{-2s}.
			\end{equation}
			We combine \eqref{Hardy.gamma2} with $P=p$ with  \eqref{ID.gamma}, which gives a coercive estimate as long as $\frac{4 |q_0 e|}{2+p}<p$, i.e., $p>p_{min}(q_0e)=-1+\sqrt{1+4|q_0e|^2}$, resulting in the following estimate, for such $p$ \begin{equation}\label{rp.est1}
				E_{-p}[\psi](u)= \int_{v_{R_0}(u)}^{+\infty} r^{-p-2} |\psi||\rd_v \psi|(u,v') dv'+  \int_{v_{R_0}(u)}^{+\infty} r^{-p-3} | \psi|^2(u,v') dv'+ \int_{v_{R_0}(u)}^{+\infty} r^{-p-1} |\rd_v \psi|^2(u,v') dv' \ls u^{1-2s},
			\end{equation} where we invoked Proposition~\ref{noshift.prop} to control the boundary term in the past. Let $\eta>0$ small: one can choose $|e|$ small enough so that $-1+\sqrt{1+4|q_0e|^2}< 2\eta$, hence we will take $p=2\eta$ in what follows. Thus, \eqref{rp.est1}, by \eqref{Hardy.gamma},  gives in turn \begin{equation}\label{phi.proof}
				|\phi|(u,v) \ls u^{-s+\frac{1}{2}} r^{\eta} \ls  u^{-s+\frac{1}{2}+\eta} .
			\end{equation}
			Now, combining \eqref{phi.proof} with \eqref{Field.psi}, we get \begin{equation}
				|D_u \rd_v \psi|(u,v) \ls r^{-1+\eta} u^{-s+\frac{1}{2}},
			\end{equation} which gives the following estimate upon integration in $u$ or $v$ respectively, and invoking the estimates of Proposition~\ref{noshift.prop} (recall that $v$ and $u$ are comparable in this region):
			\begin{equation}
				|D_u\psi|(u,v),\ 	|\rd_v \psi|(u,v) \ls u^{-s+\frac{1}{2}} r^{\eta}  \ls  u^{-s+\frac{1}{2}+\eta},
			\end{equation}  and by \eqref{phi.proof}, we also have
			\begin{equation}
				r|D_u\phi|(u,v),\ r	|\rd_v \phi|(u,v) \ls u^{-s+\frac{1}{2}} r^{\eta}  \ls  u^{-s+\frac{1}{2}+\eta}.
			\end{equation}

			It is then easy to integrate \eqref{Omega}, \eqref{Radius}, \eqref{chargeUEinstein} using the above estimates to improve on \eqref{B1.gamma}--\eqref{B4.gamma} and prove all the remaining estimates (note that to prove \eqref{Omega.gamma}, we use \eqref{murelation} together with \eqref{M.Q.gamma}, \eqref{iota.gamma}, \eqref{kappa.gamma}).
			
		\end{proof}

		\subsubsection{Backwards propagation from a timelike geodesic to null infinity} \label{nullinf.section}
		In this section, we will prove estimates in  the  spacetime region  to the future on $\gamma$ and up to null infinity, namely\begin{equation}
			\{ u \geq \Delta,\  v\geq v_{\gamma}(u)\}.
		\end{equation}

		\begin{prop}\label{nullinf.prop}   The following estimates hold for all  $u \geq \Delta,\  v\geq v_{\gamma}(u)$:
			\begin{equation}\label{phi.I}
				|\psi|(u,v) \lesssim u^{\frac{3}{2}-s+\eta},
			\end{equation}
			\begin{equation}\label{duphi.I}
				r|D_u \phi|(u,v),\		|D_u \psi|(u,v) \lesssim u^{\frac{1}{2}-s+\eta}
			\end{equation}
			\begin{equation} \label{dvphi.I}
				|\rd_v \psi|(u,v) \lesssim v^{-s+\frac{1}{2}+\eta},\ r|\rd_v \phi|(u,v) \ls v^{-1}u^{\frac{3}{2}-s+\eta},
			\end{equation}
			\begin{equation} \label{M.Q.I}
				|\varpi(u,v) -M |,\	|Q(u,v) -e | \ls u^{3-2s  + 2\eta},
			\end{equation}\begin{equation} \label{iota.I}
				|\iota(u,v ) -1|\ls  u^{-s},
			\end{equation}
			\begin{equation} \label{kappa.I}
				|\kappa(u,v ) -1|\ls  v^{-1}u^{2-2s+2\eta},
			\end{equation}
			\begin{equation}\label{Omega.I}
				\bigl|	\frac{\Omega^2(u,v)}{4}  - (1-\frac{2M}{r(u,v)} + \frac{e^2}{r^2(u,v)})\bigr|\ls  u^{-2s+1+2\eta}[1+u^{2}v^{-1}],
			\end{equation}
			\begin{equation}\label{dvlogOmega.I}
				|\rd_v \log(\Omega^2)|(u,v) \ls v^{-2}.
			\end{equation}
		\end{prop}

		\begin{proof}

			As in the proof of Proposition~\ref{timelikecurve.prop}, we make the following bootstrap assumptions: \begin{equation}\label{B1.I}
				|Q|(u,v) \leq 2 |e|,
			\end{equation}\begin{equation}\label{B2.I}
				|\frac{\Omega^2(u,v)}{4}-1|\leq \frac{10M}{r},
			\end{equation}\begin{equation}\label{B3.I}
				|\rd_u r+1|(u,v) \leq \frac{10M}{r},
			\end{equation}
			\begin{equation}\label{B4.I}
				|\rd_v r-1|(u,v) \leq \frac{10M}{r}.\end{equation}
			Additionally, we make the bootstrap assumption \begin{equation}\label{r.lower}
				r(u,v) \geq \frac{v}{10}.
			\end{equation}
			
			\noindent	Then, choosing $|q_0e|$ small enough, we have    \begin{equation}
				|D_u \rd_v \psi|(u,v) \leq 400 |q_0 e|   v^{-2} |\psi|(u,v)\leq \eta v^{-2}|\psi|(u,v).
			\end{equation} from which we deduce, invoking Proposition~\ref{timelikecurve.prop}, there exists $D>0$ such that (recall $s<2<\frac{5}{2}$) \begin{equation}
				| \rd_v \psi|(u,v) \ls v^{-s+ \frac{1}{2}+\eta} +\eta   v^{-s+\frac{1}{2}+\eta} \sup_{u\leq  u ' \leq u_\gamma(v)} [u']^{s-\frac{3}{2}-\eta}|\psi|(u',v).
			\end{equation} Integrating now in $v$ and invoking Proposition~\ref{timelikecurve.prop}  again gives (recall $s>\frac{3}{2}+\eta$) 
			\begin{equation}
				|\psi|(u,v) \ls u^{-s+ \frac{3}{2}+\eta} [1+\eta \sup_{u\leq  u ' \leq u_\gamma(v)} [u']^{s-\frac{3}{2}-\eta}|\psi|(u',v)],
			\end{equation} so since $\eta$ is small, we close the estimate and improve the bootstrap assumptions, following the same strategy  as in Proposition~\ref{timelikecurve.prop}.
			
		\end{proof}

		Before turning to the construction of spatial infinity, we will need to redefine the gauge choice $u$ (initially chosen to be \eqref{U.gauge.past}) according to gauge~\ref{gauge.U.past}, i.e., for all $u\geq \Delta$: \begin{equation}\label{U.new.gauge}
			\lim_{v\rightarrow+\infty} \iota(u,v) = 1.
		\end{equation}  \eqref{iota.I} shows that the new coordinate $u'$ thus defined is comparable to $u$, and in fact  obeys the following estimate as $u\rightarrow+\infty$ \begin{equation}\label{u'}
			u'= u (1+O(u^{-s})).
		\end{equation}  We will abuse notation and still use $\iota(u,v)$ for $\frac{\rd_v r}{\Omega^2}$ in the new $(u',v)$ coordinate system and also keep using the letter $u$ instead of $u'$.
		\begin{cor}\label{nullinf.cor}
			Under the new gauge choice \eqref{U.new.gauge}, the following estimates hold for $u\geq \Delta$, $v \geq v_{\gamma}(u)$:
			\begin{equation} \label{iota.II}
				|\iota(u,v ) -1|\ls v^{-2} u^{-2s+3+2\eta},
			\end{equation} \begin{equation}\label{nu.II}
				\bigl| -\nu(u,v)- (1-\frac{2\varpi(u,v)}{r(u,v)}+ \frac{e^2}{r^2(u,v)}) \bigr| \ls r^{-2} u^{-2s+3+2\eta},
			\end{equation}\begin{equation}\label{lambda.II}
				|\lambda(u,v)-1|\ls v^{-1},
			\end{equation} \begin{equation}\label{Omega.II}
				\bigl|	\frac{\Omega^2(u,v)}{4}  - (1-\frac{2M}{r(u,v)} + \frac{e^2}{r^2(u,v)})\bigr|\ls  u^{-2s+3+2\eta}v^{-1},
			\end{equation} \begin{equation}\label{dulogOmega.II}
				|	\rd_u \log(\Omega^2)|\ls r^{-2}.
			\end{equation} Moreover, \eqref{duphi.I}, \eqref{dvphi.I}, \eqref{M.Q.I}, \eqref{kappa.I}, \eqref{dvlogOmega.I} still hold.
		\end{cor}
		
		\begin{proof}

			Redoing the proof of Proposition~\ref{nullinf.prop} under the new gauge \eqref{U.new.gauge}  gives \eqref{iota.II}. By \eqref{murelation}, this gives \eqref{nu.II}. Then, using \eqref{Radius3}, we get \begin{equation}
				\rd_u\log \lambda= \frac{-\nu Q^2}{r^3} - \frac{-\iota^{-1}+1}{r} [1-\frac{Q^2}{r^2}]- \frac{\nu+1}{r}[1-\frac{Q^2}{r^2}]=O(r^{-2}),
			\end{equation} which, after integrating in $u$, gives \eqref{lambda.II}. Then, by \eqref{murelation} combined with \eqref{M.Q.I}, \eqref{kappa.I}, \eqref{iota.II},  we obtain \eqref{Omega.II}. Integrating \eqref{Omega}, we also obtain \eqref{dulogOmega.II}.
		\end{proof}
		\subsubsection{Backwards propagation from null infinity to spatial infinity}\label{spatial.inf.section}
		
		From Proposition~\ref{timelikecurve.prop} and Proposition~\ref{nullinf.prop}, we have obtained a solution of \eqref{1.1}--\eqref{5.1} in the spacetime rectangle $[\Delta,+\infty] \times [v_0,+\infty)$. To construct the asymptotically flat end, we impose initial data on  $ \mathcal{I}^+ \cap (-\infty,\Delta]$, where  $\mathcal{I}^+=\{v=+\infty\}$. We  still impose the $u$-gauge \eqref{U.new.gauge} as in the previous section.

		We will  also assume decay towards spatial initial $i^0=\{u=-\infty,v=+\infty\}$ in  that there exists a constant $D>0$,  such that  for $u\leq -\Delta$, the following estimate holds:\begin{equation}\label{global.data.bound}
			|\psi_{|\mathcal{I}^+}|(u)\leq D |u|^{-s+\frac{3}{2} +\eta},\ |D_u\psi_{|\mathcal{I}^+}|(u)  \leq D |u|^{-s+\frac{1}{2} +\eta}.
		\end{equation}

		First, we fix $\Delta'>0$ a large constant and construct the solution in the spacetime rectangle $(u,v) \in [-\Delta',\Delta] \times [v_0,+\infty)$. We denote $D'(\Delta')$ such that for all $-\Delta'\leq u\leq \Delta$: \begin{equation}\label{local.data.bound}
			|\psi_{|\mathcal{I}^+}|(u),\ 	|D_u\psi_{|\mathcal{I}^+}|(u)  \leq D'. 
		\end{equation}
		\begin{lemma}\label{i0.Delta.lemma}
			Assuming that $v_0$ is large enough, depending on $\Delta$ and $\Delta'$, the following estimates hold for all $-\Delta'\leq u \leq \Delta$, $v\geq v_0$:
			\begin{equation}
				|\frac{\Omega^2(u,v)}{4}-1|,\ |	\rd_v r (u,v)-1|,\ |	\rd_u r(u,v)+1| \ls v^{-1}.
			\end{equation}
			\begin{equation}
				|\psi|(u,v),\ |D_u \psi|(u,v),\ r|D_u \phi|(u,v),\ |Q|(u,v),\ |\varpi|(u,v) \leq C.
			\end{equation}
			\begin{equation}
				|\rd_v \psi|(u,v) \ls v^{-s+\frac{1}{2}+ \eta},\ r|\rd_v \phi|(u,v) \ls v^{-1},
			\end{equation}\begin{equation}
				|\rd_u \log(\Omega^2)|(u,v),\ |\rd_v \log(\Omega^2)|(u,v) \ls v^{-2},
			\end{equation}
			\begin{equation}\label{kappa.loc}
				|\kappa(u,v)-1|\ls v^{-1},
			\end{equation}
			\begin{equation}\label{iota.loc}
				|\iota(u,v)-1|\ls v^{-2}.
			\end{equation}
		\end{lemma}
		
		\begin{proof}
			We make the bootstrap assumptions: \begin{equation}\label{B1.loc}
				\frac{1}{2}\leq \frac{\Omega^2(u,v)}{4},\ \rd_v r(u,v),\ -\rd_u r(u,v) \leq 2,
			\end{equation}
			\begin{equation}\label{B2.loc}
				|\psi|(u,v) \leq 2D',
			\end{equation}
			\begin{equation}\label{B3.loc}
				|Q|(u,v) \leq C(D',e,\Delta),
			\end{equation}
			\begin{equation}\label{B4.loc}
				r(u,v) \geq \frac{v}{10}.
			\end{equation}

			\noindent	Then, \eqref{Field.psi} with \eqref{B1.loc}, \eqref{B2.loc}, \eqref{B3.loc} gives \begin{equation}
				|D_u \rd_v \psi|\ls v^{-2},
			\end{equation} which upon integrating in $u$ and invoking Corollary~\ref{nullinf.cor}, gives \begin{equation}
				|\rd_v \psi(u,v) | \ls v^{-s+\frac{1}{2}+\eta} + v^{-2}\ls v^{-s+\frac{1}{2}+\eta},
			\end{equation} hence, integrating in $v$ gives   \begin{equation}
				| \psi(u,v) -\psi_{|\mathcal{I}^+}(u)| \ls v^{-s+\frac{3}{2}+\eta},
			\end{equation} which is enough to improve \eqref{B2.loc} for $v_0$ large enough. It is then easy to derive the rest of the estimates and improve the bootstrap assumptions. 
		\end{proof}

		Finally, we   turn to the propagation towards spacelike infinity $i^0=\{u=-\infty,\ v=+\infty\}$. 
		
		\begin{prop}\label{i0.prop}
			There exists a large constant $D>0$ such that	for all $u\leq -\Delta'$, $v\geq D[-\Delta'-u]+v_0$
			\begin{equation}\label{Omega.lambda.nu.i0}
				|	\frac{\Omega^2}{4}(u,v)-1|,\ |	\rd_v r (u,v)-1|,\ |	\rd_u r(u,v)+1| \ls v^{-1},
			\end{equation}
			\begin{equation}\label{Q.M.i0}|Q|(u,v),\ |\varpi|(u,v) \leq C,
			\end{equation}
			\begin{equation}\label{psi.i0}
				|\psi|(u,v)  \ls |u|^{-s+\frac{3}{2}+ \eta},
			\end{equation}
			\begin{equation}\label{dupsi.i0}
				|D_u \psi|(u,v),\ 	r|D_u \phi|(u,v)  \ls u^{-s+\frac{1}{2}+ \eta},
			\end{equation}
			\begin{equation}\label{dvpsi.i0}
				|\rd_v \psi|(u,v) \ls v^{-s+\frac{1}{2}+ \eta},\ r|\rd_v \phi|(u,v) \ls v^{-1}|u|^{-s+\frac{3}{2}+ \eta},
			\end{equation}\begin{equation}\label{dlogOmega.i0}
				|\rd_u \log(\Omega^2)|(u,v),\ |\rd_v \log(\Omega^2)|(u,v) \ls v^{-2},
			\end{equation}
			\begin{equation}\label{kappa.i0}
				|\kappa(u,v)-1|\ls v^{-1},
			\end{equation}
			\begin{equation}\label{iota.i0}
				|\iota(u,v)-1|\ls v^{-2},
			\end{equation}
			\begin{equation}\label{dlambda.i0}
				|\rd_v \lambda|(u,v),\ |\rd_u \lambda|(u,v),\ 	|\rd_v \nu|(u,v)\ls v^{-2},
			\end{equation}
			\begin{equation}\label{dnu.i0}
				|\rd_u \nu|(u,v)\ls \blue{v^{-1}u^{-1}}.
			\end{equation}
		\end{prop}	\begin{proof}
			We make the following bootstrap assumptions		\begin{equation}\label{B1.i0}
				\frac{1}{2}\leq \frac{\Omega^2(u,v)}{4},\ \rd_v r(u,v),\ -\rd_u r(u,v) \leq 2,
			\end{equation}
			\begin{equation}\label{B2.i0}
				|\psi|(u,v) \leq   4 \sup_{-\Delta' \leq u'\leq u}|\psi_{|\mathcal{I}^+}|(u'),
			\end{equation}
			\begin{equation}\label{B3.i0}
				|Q|(u,v) \leq C,
			\end{equation}
			\begin{equation}\label{B4.i0}
				r(u,v) \geq \frac{v}{10}.
			\end{equation}

			We get, by \eqref{Field.psi} \begin{equation}
				|D_u \rd_v \psi| \leq C v^{-2}  \sup_{-\Delta' \leq u'\leq u}|\psi_{|\mathcal{I}^+}|(u'),
			\end{equation} hence, integrating in $u$, invoking Lemma~\ref{i0.Delta.lemma} \begin{equation}
				| \rd_v \psi| \leq D  [v^{-s+\frac{1}{2}+\eta}+   v^{-2}[-\Delta'-u] \sup_{-\Delta' \leq u'\leq u}|\psi_{|\mathcal{I}^+}|(u')].
			\end{equation} Thus, now integrating in $v$:  \begin{equation}
				|  \psi(u,v) - \psi_{|\mathcal{I}^+}(u)| \leq D' v^{-s+\frac{3}{2}+\eta} + D v^{-1}[-\Delta'-u]   \sup_{-\Delta' \leq u'\leq u}|\psi_{|\mathcal{I}^+}|(u').
			\end{equation} To improve on \eqref{B2.i0}, we just need that $v_0$ large enough and moreover \begin{equation}
				v \geq  D [-\Delta'-u]+ v_0.
			\end{equation} The other bootstrap assumptions improvements and other claimed estimates  are then straightforward to derive, except perhaps \eqref{dlambda.i0}, \eqref{dnu.i0}, which require a bit more explanation. Invoking \eqref{RaychV} under the form $\rd_v \log(\lambda)-\rd_v \log(\Omega^2) = \frac{-r|\rd_v\phi|^2}{\lambda}$, together with \eqref{dvpsi.i0}, \eqref{dlogOmega.i0}  shows that \begin{equation}
				|\rd_v \log(\lambda)|\ls v^{-2}.
			\end{equation} Similarly,  invoking \eqref{RaychU} together with  \eqref{dupsi.i0}, \eqref{dlogOmega.i0} shows that \begin{equation}
				|\rd_u\log(\nu)|\ls \blue{v^{-1}u^{-1}}.
			\end{equation}
			
			\noindent	Now, invoking \eqref{Radius} combined with \eqref{Omega.lambda.nu.i0} immediately gives \begin{equation}
				|\rd_u \lambda|(u,v),\ |\rd_v \nu|(u,v) \ls v^{-2},
			\end{equation} which completes the proof of \eqref{dlambda.i0}, \eqref{dnu.i0}.
		\end{proof}

		\subsubsection{Construction of the charged one-ended black hole}\label{completion.section}
		
		We are now ready to complete the proof of Theorem~\ref{EH.AF.thm} and Corollary~\ref{EH.AF.cor}.
		
		\begin{proof} We start with $(\mathcal{M},g,F,\phi)$ solution of \eqref{1.1}--\eqref{5.1} with $F \neq 0$,  the spacetime of Proposition~\ref{EH.AF.1st.prop} \blue{to the future of $\mathcal{H}^+$, outgoing cone connecting  $\Gamma$ to $i^+$, which will later act as the  event horizon}. 	The gauge is chosen to be gauge \eqref{gauge.EH.v} and gauge~\ref{gauge.U.past}  (which we recall comes from \eqref{U.new.gauge}, which  is a teleological choice).
			
			\begin{enumerate}[A.]

				\item \label{stepA} [Existence up to a constant-$r$ timelike curve near  $\mathcal{H}^+$]. We pose consider the characteristic initial value problem for \eqref{1.1}--\eqref{5.1} with data on $[-\ep,0]\times \{v_{0}\} \cup \{0\}\times [v_0,+\infty)$ recalling that $v_0 >1$ is sufficiently large. Then, we invoke Proposition~\ref{RS.prop} and Proposition~\ref{noshift.prop} with characteristic initial data on $[-\ep,0]\times  \{v_0\} \cup \{0\}\times [v_0,+\infty)$ to solve \eqref{1.1}--\eqref{5.1} in the region $\{-\ep \leq U \leq 0,\ v_0 \leq v\leq v_{R_0}(u)\}$, where we recall $ v_{R_0}(u)$ is defined so that $ v_{R_0}(u)+u = R_0$, where $R_0$ is a large constant.
				
				\item \label{stepB} [Existence up to a far-away timelike curve $\gamma$ near  $\mathcal{H}^+$]. We invoke Proposition~\ref{timelikecurve.prop} to solve \eqref{1.1}--\eqref{5.1} in the region $\{u \geq \Delta,\ v_{R_0}(u)\leq v\leq v_{\gamma}(u) \}$,  where we recall $ v_{\gamma}(u)=\frac{3u}{2}$,  and $u= [2 K_+]^{-1} \ln(|U|^{-1})$, with $\Delta= [2 K_+]^{-1} \ln(\ep^{-1})$.
				
				\item  \label{stepC} [Existence up to null infinity $\mathcal{I}^+$ near  $\mathcal{H}^+$]. We invoke Proposition~\ref{nullinf.prop} and Corollary~\ref{nullinf.cor} to solve \eqref{1.1}--\eqref{5.1} in the region $\{u \geq \Delta,\  v\geq v_{\gamma}(u) \}$ ($\mathcal{R}_{\mathcal{H}}$ in Figure~\ref{fig:ext_bif}).
				
				\item \label{stepD} [Existence up to spacelike infinity $i^0$]. We invoke Lemma~\ref{i0.Delta.lemma}   to solve \eqref{1.1}--\eqref{5.1} in the region $\{-\Delta' \leq u \leq \Delta,\  v\geq v_0  \}$ ($\mathcal{R}_{\mathcal{I}}$ in Figure~\ref{fig:ext_bif}).,  and subsequently  Proposition~\ref{i0.prop} to solve \eqref{1.1}--\eqref{5.1} in the region $\{u \leq -\Delta',\  v\geq v_0 + D[-\Delta'-u] \}$, for some (sufficiently large) constant $D>0$ ($\mathcal{D}_{i^{0}}$ in Figure~\ref{fig:ext_bif}).
				
			\end{enumerate}
			
			Combining Steps~\ref{stepA}--\ref{stepD} gives a solution of \eqref{1.1}--\eqref{5.1} in  the following spacetime region (Figure~\ref{fig:ext_bif})\begin{equation}\begin{split} & \mathcal{M}':= R_{\mathcal{H}} \cup R_{\mathcal{I}}\cup \mathcal{D}_{i^0}, \\ &
					R_{\mathcal{H}} := \{\blue{u \geq \Delta},\  v\geq v_{0}\},\\ & R_{\mathcal{I}} := \{ -\Delta'\leq u \leq \Delta,\  v\geq v_0\},\\ & \mathcal{D}_{i^0}  :=\{  u \leq -\Delta',\  v\geq v_0+ D[-\Delta'-u]\}.\end{split}
			\end{equation}
			
			\begin{figure}	\begin{center}
					\includegraphics[width=100 mm, height=50 mm]{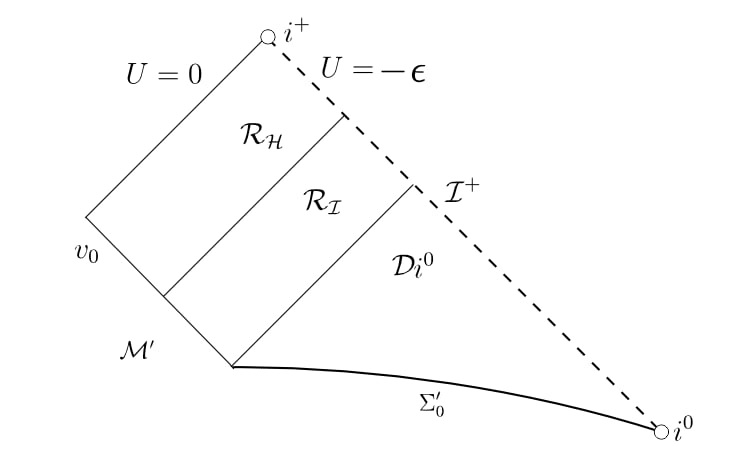}
				\end{center}
				\caption{Construction of an asymptotically flat end corresponding to $\mathcal{M}'$, the right of the event horizon in the Penrose diagram, here labeled as $\{U=0\}$ as part of the proof of Theorem~\ref{EH.AF.thm}.}\label{fig:ext_bif}\end{figure}
			
			Then, one can glue $\mathcal{M}$ to $\mathcal{M}'$ to obtain a global spacetime. Finally, \blue{by Cauchy stability, one can solve \eqref{1.1}-\eqref{5.1} into a small neighborhood towards the past of $\mathcal{H}^+ \cup \mathcal{R}_{\mathcal{H}}\cup \mathcal{R}_{\mathcal{I}}$, and this region is free of trapped or anti-trapped surfaces. Subsequently, we can} construct a spacelike hypersurface $\Sigma_0' \subset \mathcal{M}'$ with the following properties (see Figure~\ref{fig:final_bif}): \begin{enumerate}
				\item $\Sigma_0'$ does not intersect with the event horizon $\mathcal{H}^+$.
				\item $\Sigma_0'$ terminates to the left at the center $\Gamma$. 
				\item In a neighborhood of $\Gamma$, the induced data on $\Sigma_0'$ coincides with that of $\blue{\Sigma_L}$ from Proposition~\ref{EH.AF.1st.prop}.
				\item $\Sigma_0'$ does not have any (marginally) trapped sphere, i.e., $\lambda_{|\Sigma_0'}>0$.
				\item $\Sigma_0'$ does not have any (marginally) antitrapped sphere, i.e., $\nu_{|\Sigma_0'}<0$.
				\item For some $\Delta''>\Delta'$, $\Sigma_0' \cap \{u\leq -\Delta''\}$ coincides with $\{v= D[-\Delta'+u]+v_0,\ u\leq -\Delta''\} $.
			\end{enumerate}	\begin{figure}[H]	\begin{center}
					\includegraphics[width=100 mm, height=50 mm]{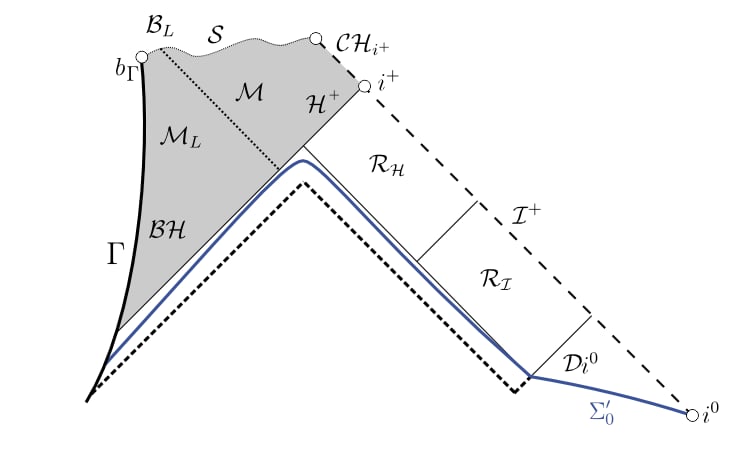}
				\end{center}
				\caption{Asymptotically flat one-ended black hole construction in the proof of Theorem~\ref{EH.AF.thm}.}\label{fig:final_bif}\end{figure} Moreover, by Theorem~\ref{CH.thm.SS}, the MGHD terminal boundary admits a non-empty null component $\CH$ on which $r$ extends to a non-zero function.
			Finally, we need to check that the spacetime is asymptotically flat in the sense of Definition~\ref{AF.def}. By definition, $\mathcal{M}'\cap J^{+}(\Sigma_0')$ is the MGHD of initial data $(r,f,h,l,\phi,\dot{\phi})$ corresponding to \eqref{hat.g.def}, \eqref{hat.k.def}, \eqref{dot.phi.def}, where $\vec{n}$ is the  future-directed unit normal to $\Sigma_0'$ defined as such,
			\begin{equation}\begin{split}
					& f(\rho):= \Omega(\rho),\ l(\rho):= \frac{r(\rho)}{\Omega(\rho)}[\lambda(\rho) + \nu(\rho)],\\ &    h(\rho) = \frac{r(\rho) l^2(\rho)}{\lambda(\rho) - \nu(\rho)} \left[\frac{\rd_{\rho} l(\rho)}{r^2(\rho)} - \frac{ l(\rho)[\lambda(\rho)-\nu(\rho)]}{r^3(\rho)}+\frac{\theta^2(\rho)- \xi^2(\rho)}{f(\rho) r(\rho)}\right].
				\end{split}
			\end{equation} where $\rho=u^{\Gamma}_{|\Sigma_0'}-u_{|\Sigma_0'}$. Then, by Proposition~\ref{i0.prop}, the following estimates hold as $\rho \rightarrow +\infty$: \begin{equation}\begin{split} &| r(\rho) - \rho|\ls\log(\rho),\ |\rd_{\rho}r(\rho)|\ls \rho^{-1},\ |\rd_{\rho}^2 r(\rho)|\ls \rho^{-2},  \\ 
					&|f(\rho)-1|\ls \rho^{-1},\  |\rd_{\rho}f(\rho)|\ls \rho^{-2},\\ & |l(\rho)|\ls 1,\  |\rd_{\rho}l(\rho)|\ls \rho^{-1},\\ & |h(\rho)|\ls \rho^{-2},\\ & |\phi(\rho)| \ls \rho^{-s+\frac{1}{2}+\eta},\  |\dot{\phi}(\rho)| \ls \rho^{-s-\frac{1}{2}+\eta},
				\end{split}
			\end{equation}  therefore the spacetime is indeed asymptotically flat and Theorem~\ref{EH.AF.thm} is proved. Corollary~\ref{EH.AF.cor} is then obtained as an immediate application of Theorem~\ref{EH.AF.thm}, the same way we proved Corollary~\ref{uncharged.cor}.

		\end{proof}

		\subsection{Nonlinear scattering theory in the black hole interior}\label{scattering.section} In this section, we turn to the proof of finer asymptotics inside the black hole, which are necessary to apply Theorem~\ref{main.thm}. Ultimately, we will  apply the new scattering result to a profile $\phi_{|\mathcal{H}^+}$ obeying the asymptotics \eqref{EH.bound2} for $s>\frac{3}{2}$, and in the weakly-charged case $|e|\ll M$ of Theorem~\ref{charging.thm}. However, for now we try to remain as general as possible. Let $0<|e|<M$ and $\Phi(v)$ a $C^1$ function.  We assume
		\begin{equation}\label{data.EH}
			\phi_{|\HH}(v) = \Phi_H(v) e^{-i q_0 \omer v},
		\end{equation}\begin{equation}\label{data.EH2}
			\lim_{v\rightarrow +\infty} \varpi(0,v) = M,\  \lim_{v\rightarrow +\infty} Q(0,v) = e,\ \lim_{v\rightarrow +\infty} r(0,v) = r_+(M,e) := M + \sqrt{M^2-e^2}.
		\end{equation}
		\begin{equation}\label{data.ing}
			|\phi|(U,v_0), |D_U \phi|(U,v_0) \leq D.
		\end{equation}
		Our nonlinear scattering theorem, instrumental to the proof of Theorem~\ref{main.thm.global.ii}, is stated as follows.
		\begin{thm}\label{nonlinearscat.thm}
			Let $\Phi(v)$ be a $C^1$ function and $s>1$. Assume that \begin{equation}
				\sup_{v\in \RR}\ (1+|v|)^{s}  |\Phi_H|(v)<+\infty,\ 	\sup_{v\in \RR}\ (1+|v|)^{s}  |\rd_v \Phi_H|(v)<+\infty,\  \rd_v^2 \Phi_H\in L^2(\RR),\ \rd_v^3 \Phi_H\in L^2(\RR).
			\end{equation}  We consider the solution  of \eqref{1.1}--\eqref{5.1} on the rectangle $[0,U_s] \times [v_0,+\infty)$.	Then, there exists $\Delta$ large enough so that for all $u \leq -\Delta$, and $v$ sufficiently large depending on $u$, there exists $\theta(u)$ such that  for all $N \in  \mathbb{N}$:	\begin{equation}\label{dv.psi.nonlinear.est}
				\bigl|  e^{-i q_0 \int_{v}^{+\infty}  A_v(u,v') dv'+ i\theta(u)}D_v \psi(u,v) + i \frac{r_+   }{ \sqrt{2\pi} }  \mathfrak t(\omer) \Phi_{H}(v) \bigr| \lesssim|\rd_{v}\Phi_{H}(v)|+ (1+|v|)^{-N}  \left(1+ \| (1+|v|)^{N} \rd_{v}^2 \Phi_{H}\|_{L^{\infty}}\right)+ v^{1-2s},
			\end{equation}
			\begin{equation}\label{dvv.psi.nonlinear.est}
				\bigl| |D_v^2 \psi|(u,v)  -  \frac{r_+ |\mathfrak t(\omer)|  } { \sqrt{2\pi} }   |\rd_v \Phi_{H}|(v)\bigr|\lesssim |\rd_{v}^2\Phi_{H}(v)|+ (1+|v|)^{-N}  \left(1+ \| (1+|v|)^{N} \rd_{v}^3 \Phi_{H}\|_{L^{\infty}}\right)+ v^{-2s},
			\end{equation}
			\begin{equation}\label{du.psi.nonlinear.est}
				\bigl|  |D_u \psi|(u,v) - \frac{r_+   }{ \sqrt{2\pi} } | \mathfrak r|(\omer) |\Phi_{H}|(|u|) \bigr| \lesssim|\rd_{v}\Phi_{H}(|u|)|+ (1+|u|)^{-N} \left(1+ \| (1+|v|)^{N} \rd_{v}^2 \Phi_{H}\|_{L^{\infty}}\right)+ |u|^{-2s},
			\end{equation} where $\mathfrak t(\omer) \in \mathbb{C}-\{0\}$,  $\mathfrak r(\omer) \in \mathbb{C}-\{0\}$  are constants only depending on $M$, $e$ and $q_0$, and $|\mathfrak t|(\omer) = |\mathfrak r|(\omer)$. 
		\end{thm}
		Then, we apply Theorem~\ref{nonlinearscat.thm} to profiles satisfying \eqref{EH.bound2}, as explained in the beginning of the section.
		
		\begin{cor}\label{nonlinearscat.cor} Let us \blue{de}note $t(\omer) =|t|(\omer) e^{ i \theta(\omer)}$. 
			Under the assumptions of Theorem~\ref{nonlinearscat.thm}, assume that there exists $\delta>0$ such that for all integers $0\leq k\leq 3$:\begin{equation}\label{EH.profile.hyp} \begin{split}
					&	|\rd_v^{k}\Phi_H|(v)\lesssim v^{-s-k},\\ & |\Phi_H|(v)\gtrsim v^{-s},\\  & \bigl| \Im( i e^{i\theta(\omer)} \Phi_H(v)) \bigr| \ls v^{-s-\delta}. \end{split}
			\end{equation}
			Then, there exists a large constant $\Delta>0$ such that for all $u\leq -\Delta$, and $v$ sufficiently large (depending on $u$)\begin{equation}\begin{split}\label{Dv.phi.final.est}
					& v^{-s} \lesssim |D_v \phi|(u,v) \lesssim v^{-s},\\ &   |D_v^2 \phi|(u,v) \lesssim v^{-s-1},
				\end{split}
			\end{equation} \begin{equation}\begin{split}\label{ImDv.phi.final.est}
					& \bigl| \Im( e^{-i q_0 \int_{v}^{+\infty}  A_v(u,v') dv'+ i \theta(u)}D_v \psi(u,v)) \bigr| \ls v^{-s-\delta'},
				\end{split}
			\end{equation} 
			\begin{equation}\label{dv.r.final}
				v^{-2s} \ls -\rd_v r\ls v^{-2s},
			\end{equation}
			\begin{equation}\label{du.phi.final}
				\lim_{v\rightarrow+\infty}|D_u \phi|(u,v) \gtrsim |u|^{-s},
			\end{equation}
			\begin{equation}\label{du.r.final}
				\lim_{v\rightarrow +\infty} |\rd_u r|(u,v) >0,
			\end{equation} where $\delta'= \min\{\delta,s-1\}$. In particular, there exists $u_0 \in \RR$ such that  \eqref{hyp4}--\eqref{hyp3} are satisfied.
		\end{cor}

		\subsubsection{Recalling previous linear scattering results in the black hole interior from \cite{MoiChristoph}}\label{previous.paper.withCK.section}
		
		We place ourselves in the framework of Kehle's and the author's previous work \cite{MoiChristoph}, Section 5, which concerns the linear theory of charged scalar fields on a \emph{fixed} subextremal Reissner--Nordström interior \eqref{RN} with the subextremal parameters $0<|e|<M$.    We will use the following electromagnetic gauge choice	\begin{align}\label{eq:gauge_lienartheory}
			A'_{RN} = \left( \frac{e}{r} - \frac{e}{r_+} \right)dt =  \frac 12\left( \frac{e}{r} - \frac{e}{r_+} \right) dv  - \frac 12  \left( \frac{e}{r} - \frac{e}{r_+} \right)  du,
		\end{align}  recalling $r_{\pm}(M,e) = M\pm \sqrt{M^2-e^2}$, and which satisfies  $F_{RN}=\d A'_{RN}$ for
		\begin{align}\label{eq:linearversionofF}
			F_{RN} = \frac{e}{2 r^2} \OmegaRN du \wedge dv. 
		\end{align}

		We now consider solutions $\phil$ of the charged scalar field equation \eqref{5.1} with $g=g_{RN}$
		\begin{equation}
			\label{eq:chargedKGlinear}
			(\nabla_\mu +i q_0 (A'_{RN})_\mu)(\nabla^\mu + iq_0 (A'_{RN})^\mu) \phil   =0,
		\end{equation}
		
		We also define  \begin{align} \omega_r  = \frac{q_0 e}{r},\ \omega_+  = \frac{q_0 e }{r_+},\  \omega_-  = \frac{q_0 e }{r_-},\ \omer  = \omega_- - \omega_+.\end{align}
		
		Then, we define the $t$-Fourier  transform of $\phil$ as
		\begin{align}
			\mathcal{F}[ \phil](r,\omega) = \hat{\phil}  =  \frac{1}{\sqrt{2\pi}} \int_{\mathbb R} \phil (r,t) e^{i \omega t } dt,\ 	 u (\omega,r^\ast) := r(r^\ast)\mathcal{F}[ \phil ](r(r^\ast),\omega)
		\end{align}
		and since $\phil$ solves \eqref{eq:chargedKGlinear}, $u$ solves the following ODE: 
		\begin{equation}
			-u'' - (\omega - ( \omega_r - \omega_+) )^2u + V u =0,\label{eq:radialode}
		\end{equation}
		where  
		\begin{align}\label{eq:potential}
			V = - \OmegaRN (r)\left( \frac{2M}{r^3} - \frac{2e^2}{r^4}  \right).
		\end{align} We also define \begin{equation}\label{r*.def}
			r^{*}(u,v) = \frac{v+u}{2}.
		\end{equation}
		The radial ODE \eqref{eq:radialode} admits the following two fundamental pairs of solution associated to the event horizon ($r^\ast \to -\infty$) and the Cauchy horizon ($r^\ast \to +\infty$). 
		\begin{defn}
			\label{defn:ua}
			Let $\uhr$, $\uhl$, $\uchr$ and $\uchl$ be the unique smooth solutions to \eqref{eq:radialode} satisfying 
			\begin{align}
				&\uhr (r^\ast) = e^{-i\omega r^\ast}  + O(\OmegaRN) \text{ as } r^\ast \to -\infty \\
				&\uhl (r^\ast) = e^{i \omega r^\ast} + O(\OmegaRN)  \text{ as } r^\ast \to -\infty\\
				&\uchr (r^\ast) = e^{i(\omega -\omer ) r^\ast}  + O(\OmegaRN) \text{ as } r^\ast \to +\infty \\
				&\uchl (r^\ast) = e^{- i (\omega - \omer  )  r^\ast} + O(\OmegaRN)  \text{ as } r^\ast \to +\infty
			\end{align}
			for $ \omega \in \mathbb R$. The pairs $(\uhr, \uhl)$ and ($\uchr, \uchl$)  span the solution space of \eqref{eq:radialode} for $\omega \in \mathbb R - \{ 0 \}$ and $\omega \in \mathbb R - \{ \omer \}$, respectively.
		\end{defn}
		\noindent Using the fact that the Wronskian
		\begin{align}
			\mathfrak{W} (f,g) := f g' - f' g
		\end{align}
		of two solution of \eqref{eq:radialode} is independent of $r^\ast$, we define  the transmission and reflection coefficients $\mathfrak{T} (\omega)$ and $\mathfrak{R} (\omega)$ as follows.
		\begin{defn}\label{defn:trans}
			For $\omega \in \mathbb R - \{ \omer \} $, we define the transmission and reflection coefficients $\mathfrak T$ and $\mathfrak R$ as 
			\begin{align}
				&\mathfrak{T}(\omega) := \frac{\mathfrak{W}(\uhr, \uchr)}{\mathfrak{W}(\uchl, \uchr)} =  \frac{\mathfrak{W}(\uhr, \uchr)}{2i (\omega - \omer )}\\
				&\mathfrak{R}(\omega) := \frac{\mathfrak{W}(\uhr, \uchl)}{\mathfrak W (\uchr, \uchl)} =  \frac{\mathfrak W(\uhr, \uchl )}{-2i (\omega - \omer )} ,
			\end{align}
			where $\uhr$, $\uhl$, $\uchr$ and $\uchl$ are defined in \ref{defn:ua}. Indeed, this allows us to write 
			\begin{align}
				\uhr = \mathfrak T \uchl + \mathfrak R  \uchr
			\end{align} 
			for $\omega \in \mathbb R - \{ \omer \} $.
			Moreover, we define the  normalized transmission and reflection coefficients as
			\begin{align}
				\label{def.t}	& \mathfrak t(\omega)  = ( \omega - \omer)  \mathfrak T(\omega) = \frac{\mathfrak W(\uhr, \uchr )}{2i},\\
				&   \mathfrak r(\omega) = ( \omega - \omer) \mathfrak R (\omega)=  \frac{\mathfrak{W}(\uhr, \uchl)}{-2i},
			\end{align}
			which manifestly satisfy
			\begin{align}\label{eq:tomegromegr}
				\mathfrak t(\omer) = - \mathfrak r(\omer)\neq 0.
			\end{align}
		\end{defn}
		We will also define the re-normalized functions
		\begin{defn}\label{defn:tildeuhretc}
			We define
			\begin{align}
				&\tilde \uhr (r^\ast,\omega)  :=e^{i\omega r^\ast} \uhr (r^\ast,\omega)  ,  \\
				&\tilde \uhl (r^\ast,\omega)  :=e^{-i\omega r^\ast} \uhl (r^\ast,\omega)  , \\
				&\tilde \uchr  (r^\ast,\omega) :=e^{-i (\omega - \omer ) r^\ast} \uchr (r^\ast,\omega) ,   \\
				&\tilde \uchl (r^\ast,\omega)  :=e^{i (\omega - \omer )  r^\ast} \uchl (r^\ast,\omega)    .
			\end{align}
		\end{defn}
		
		The following result is extracted from \cite{MoiChristoph} and will be used in the subsequent sections.
		
		\begin{prop}[\cite{MoiChristoph}, Theorem V]\label{previousthm.withCK}
			\begin{equation}\label{dvphi.rep.formuala}
				\partial_v \left(  r e^{i \frac{\omer  v}{2} }	\phil(u,v) \right)  =   -i \frac{r_+  e^{i\frac{\omer u}{2}} }{ \sqrt{2\pi} } \int_{\mathbb R}  \mathcal F[\phiHL   e^{i \omer \cdot} ] (\omega)    \mathfrak t(\omega+\omer)     e^{-i\omega v}   d\omega + \Phi_{error}(u,v),
			\end{equation}  where $\Phi_{error}(u,v)$ obeys the following estimates for any fixed $0<\alpha <2$ and every $(u,v)$ such that $v+u \geq 2$ \begin{equation}\label{error.lin}
				|\Phi_{error}|(u,v),\ |\rd_v \Phi_{error|}(u,v)\lesssim   \left(\int_{\RR}[ |\rd_v \phiHL|^2(v')+|\phiHL|^2(v') ]dv\right) \Omega_{RN}^{2-\alpha}(u,v)
			\end{equation} and  where  $\omega\rightarrow \mathfrak{t}(\omega)$ is analytic and  we have \begin{equation}\label{t.easy}
				\mathfrak{t}(\omer) \neq 0,\ |\mathfrak{t}|(\omega),\ |\rd_{\omega}\mathfrak{t}|(\omega) \lesssim 1+ |\omega|.
			\end{equation}
		\end{prop} 
		\noindent We then record a representation formula for the ingoing derivative which is not strictly speaking stated in \cite{MoiChristoph} but immediately follows from the techniques there. The reflection coefficient estimates are in \cite{MoiChristoph}, Lemma 5.9.
		
		\begin{prop}\label{u.prop} The following estimates hold:
			\begin{equation}\label{duphi.rep.formuala}
				\partial_u \left(  r e^{-i \frac{\omer  u}{2} }	\phil(u,v) \right)  =   i \frac{r_+  e^{-i\frac{\omer v}{2}} }{ \sqrt{2\pi} } \int_{\mathbb R}  \mathcal F[\phiHL   e^{i \omer \cdot} ] (\omega)    \mathfrak r(\omega+\omer)     e^{i\omega u}   d\omega + \tilde{\Phi}_{error}(u,v),
			\end{equation}  where $\tilde{\Phi}_{error}(u,v)$ obeys the following estimates for any fixed $0<\alpha <2$ and every $(u,v)$ such that $v+u \geq 2$ \begin{equation}\label{error.lin.u}
				|\tilde{\Phi}_{error}|(u,v) \lesssim   \left(\int_{\RR}[ |\rd_v \phiHL|^2(v')+|\phiHL|^2(v') ]dv\right) \Omega_{RN}^{2-\alpha}(u,v)
			\end{equation} and  where  $\omega\rightarrow \mathfrak{r}(\omega)$ is analytic and  we have \begin{equation}\label{r.easy}
				\mathfrak{r}(\omer) \neq 0,\ |\mathfrak{r}|(\omega),\ |\rd_{\omega}\mathfrak{r}|(\omega) \lesssim 1+ |\omega|.
			\end{equation}
		\end{prop}

		\subsubsection{New refined linear scattering results}\label{novel.paper.withCK.section}
		
		In this section, we remain in the framework of Section~\ref{previous.paper.withCK.section} but we provide more refined estimates on the transmission coefficient $ \mathfrak{t}$ than what is already available in \cite{MoiChristoph}. These refinements will end up being necessary to prove Theorem~\ref{nonlinearscat.thm}.
		
		\begin{lemma}\label{lemma.t}
			For all $N \in \mathbb{N}$, 
			
			\begin{equation}\label{t.est.N}
				|\rd_{\omega}^N \mathfrak{t}|(\omega) \lesssim_{N} 1+ |\omega|.
			\end{equation}
		\end{lemma}
		\begin{proof}
			This lemma can be viewed as a refinement of Lemma 5.9 in \cite{MoiChristoph}, which itself is based  on Lemma 5.7 in \cite{MoiChristoph}. Recall, for instance, that $\tilde \uchr$ solves the Volterra integral equation (equation (5.77) in \cite{MoiChristoph})
			\begin{align}  
				\tilde \uchr (r^\ast, \omega)=  1 +  \int_{r^\ast}^{+\infty} & \frac{\sin[ (\omega - \omer)(r^\ast - y) ]}{\omega -\omer} e^{- i (\omega - \omer) (r^\ast - y)} \\
				& \left[V(y) - (\omega_- - \omega_{r(y)})(2 \omega + 2\omega_+ - \omega_- - \omega_{r(y)} ) \right] \tilde \uchr (\omega, y) dy.
			\end{align} Then, taking $\rd_{\omega}^N$ derivatives and using standard Volterra estimates as in \cite{MoiChristoph} generalizes the result into  \begin{equation}\begin{split}
					&|\rd_{\omega}^{N}	\tilde \uchr (r^\ast, \omega)|\ls \Omega^2_{RN}(r^{*}), \\  &	|\rd_{\omega}^{N} \rd_{r^{*}}\tilde \uchr (r^\ast, \omega)| \lesssim \Omega^2_{RN}(r^{*}) (1+|\omega|),
				\end{split}
			\end{equation} for $r^{*} \geq 0$, and the analogous estimate replacing  $	\tilde \uchr$ by $	\tilde \uchl$. Moreover, we also get, for $r^{*}\leq 0$: \begin{equation}\begin{split}
					&|\rd_{\omega}^{N}	\tilde \uhr (r^\ast, \omega)|\ls \Omega^2_{RN}(r^{*}), \\  &	|\rd_{\omega}^{N} \rd_{r^{*}}\tilde \uhr (r^\ast, \omega)| \lesssim \Omega^2_{RN}(r^{*}) (1+|\omega|).
				\end{split}
			\end{equation} Now, as in \cite{MoiChristoph}, we argue evaluating the Wronskian at $r^{*}=0$, for all $N\in \mathbb{N}$, \begin{align} & |\partial_\omega^{N} \mathfrak t |\lesssim |\partial_\omega^{N}  \mathfrak W(\uhr, \uchr)|   \lesssim \sum_{k=0}^{N} | \mathfrak W(\partial_\omega^{k} \uhr,  \partial_\omega^{N-k}\uchr) | (r^\ast =0)  \lesssim 1+|\omega|, \end{align} which concludes the proof.
		\end{proof}
		Then, in the next step, we take advantage of Lemma~\ref{lemma.t} to prove a refinement of Proposition~\ref{previousthm.withCK}:
		\begin{prop}\label{new.scat.prop} The following estimates hold true: for any $N \in \mathbb{N}$:
			\begin{equation}\label{dv.psi.linear.est}
				\bigl|	\partial_v \left(  r e^{i \frac{\omer  v}{2} }	\phil(u,v) \right)  + i \frac{r_+  e^{i\frac{\omer u}{2}} }{ \sqrt{2\pi} } \left( \mathfrak t(\omer) \Phi_{H}(v)  + \mathfrak t'(\omer) \rd_{v}\Phi_{H}(v)\right) \bigr| \lesssim (1+|v|)^{-N}  \left(\| (1+|v|)^{N} \rd_{v}^2 \Phi_{H}\|_{L^{\infty}}+ \|  \rd_{v}^2 \Phi_{H}\|_{L^{2}}\| \right),
			\end{equation}
			\begin{equation}\label{dvv.psi.linear.est}
				\bigl|	\partial_v^2 \left(  r e^{i \frac{\omer  v}{2} }	\phil(u,v) \right)  -\frac{r_+  e^{i\frac{\omer u}{2}} }{ \sqrt{2\pi} } \left( \mathfrak t(\omer) \rd_v \Phi_{H}(v)  + \mathfrak t'(\omer) \rd_{v}^2\Phi_{H}(v)\right) \bigl|\lesssim (1+|v|)^{-N}  \left(\| (1+|v|)^{N} \rd_{v}^3 \Phi_{H}\|_{L^{\infty}}+ \|  \rd_{v}^3 \Phi_{H}\|_{L^{2}}\| \right),
			\end{equation} 
			\begin{equation}\label{du.psi.linear.est}\begin{split}
					&	\bigl|	\partial_u \left(  r e^{-i \frac{\omer  u}{2} }	\phil(u,v) \right)  -i \frac{r_+  e^{-i\frac{\omer v}{2}} }{ \sqrt{2\pi} } \left( \mathfrak r(\omer) \Phi_{H}(v)  + \mathfrak r'(\omer) \rd_{v}\Phi_{H}(v)\right) \bigr| \\&\lesssim (1+|v|)^{-N}  \left(\| (1+|v|)^{N} \rd_{v}^2 \Phi_{H}\|_{L^{\infty}}+ \|  \rd_{v}^2 \Phi_{H}\|_{L^{2}}\| \right).\end{split}
			\end{equation} 
		\end{prop} 
		\begin{proof}
			This proposition builds up on the claims of Theorem V  in \cite{MoiChristoph} and essentially follows from its proof. We nonetheless will give some rapid details on how to obtain the claimed refinements on the estimates of Theorem V in \cite{MoiChristoph}. First, Statement C in Theorem V  in \cite{MoiChristoph} gives
			
			\begin{equation}
				\partial_v \left(  r e^{i \frac{\omer  v}{2} }	\phil(u,v) \right)  =   -i \frac{r_+  e^{i\frac{\omer u}{2}} }{ \sqrt{2\pi} } \int_{\mathbb R}  \mathcal F[\phiHL   e^{i \omer \cdot} ] (\omega)    \mathfrak t(\omega+\omer)     e^{-i\omega v}   d\omega + \Phi_{error}(u,v),
			\end{equation}  where $$|\Phi_{error}|(u,v) \lesssim \Omega_{RN}(u,v).$$
			
			Then, we can write, using Lemma~\ref{lemma.t} to obtain  an improvement of \cite{MoiChristoph}, equation (5.182): for any $N \in \mathbb{N}$, \begin{equation}\begin{split}\label{t.est}
					& \mathfrak{t}(\omega+\omer) = \mathfrak t(\omer) + \mathfrak t'(\omer) \omega +  \omega^2 \tilde{\mathfrak{t}}_{\omer}(\omega),  \\  & |\tilde{\mathfrak{t}}_{\omer}|(\omega),\ |\rd_{\omega}^{N}\tilde{\mathfrak{t}}_{\omer}|(\omega)\ls (1+|\omega|)^{-1},
				\end{split}
			\end{equation} from which we deduce, by Plancherel Theorem, that for any $N \in \mathbb{N}$: \begin{equation}\label{t.est2}
				(1+|v|)^{N} \mathcal{F}[\tilde{\mathfrak{t}}_{\omer}] \in L^2(\RR_v),
			\end{equation} hence \begin{equation}\label{t.est3}
				\mathcal{F}[\tilde{\mathfrak{t}}_{\omer}] \in L^1(\RR_v).
			\end{equation} The conclusion of this becomes, defining $\Phi_{H}(v):= \phiHL(v)   e^{i \omer v}$: \begin{equation}\begin{split}
					&\partial_v \left(  r e^{i \frac{\omer  v}{2} }	\phil(u,v) \right)  =   -i \frac{r_+  e^{i\frac{\omer u}{2}} }{ \sqrt{2\pi} } \int_{\mathbb R}  \mathcal F[\Phi_{H}] (\omega)    \mathfrak t(\omega+\omer)     e^{-i\omega v}   d\omega + \Phi_{error}(u,v)\\ &= -i \frac{r_+  e^{i\frac{\omer u}{2}} }{ \sqrt{2\pi} } \left( \mathfrak t(\omer) \Phi_{H}(v)  + \mathfrak t'(\omer) \rd_{v}\Phi_{H}(v)+\int_{\mathbb R}  \mathcal F[\rd_{v}^2 \Phi_{H} ] (\omega)    \tilde{\mathfrak t}_{\omer}(\omega)     e^{-i\omega v}   d\omega\right) + \Phi_{error}(u,v).\end{split}
			\end{equation}  	Then, we can use  \eqref{t.est}, \eqref{t.est2}, \eqref{t.est3} to estimate the expression as such to write  (see equation (5.191) in \cite{MoiChristoph}, where the same argument is used): \begin{equation}\begin{split}\label{fourier.est}
					&\bigl|\int_{\mathbb R}  \mathcal F[\rd_{v}^2 \Phi_{H} ] (\omega)    \tilde{\mathfrak t}_{\omer}(\omega)     e^{-i\omega v}   d\omega\bigl|\\  \lesssim& (1+|v|)^{-N} \left(\| (1+|v|)^{N} \rd_{v}^2 \Phi_{H}\|_{L^{\infty}}\| \mathcal{F}[\tilde{\mathfrak t}_{\omer}]\|_{L^{1}}+ \|  \rd_{v}^2 \Phi_{H}\|_{L^{2}}\| (1+|v|)^{N}\mathcal{F}[\tilde{\mathfrak t}_{\omer}]\|_{L^{2}}\right)\\ \ls & (1+|v|)^{-N} \left(\| (1+|v|)^{N} \rd_{v}^2 \Phi_{H}\|_{L^{\infty}}+\|  \rd_{v}^2 \Phi_{H}\|_{L^{2}}\right).
				\end{split}
			\end{equation}
			\noindent	Then,	equation (5.178) in \cite{MoiChristoph} gives
			\begin{align}\nonumber
				\partial_v^2(e^{i \omer r^\ast} r	\phil(u,v) )  = 	&\frac{r_+  e^{i\omer u} }{ \sqrt{2\pi} }  \textup{p.v.}\int_{\mathbb R}  \Big[ \mathcal F[\phiHL \chi_{\leq v_1} e^{i \omer \cdot} ] (\omega) \\ &   \cdot \frac{ \mathfrak r(\omega+\omer) \partial_{v}^2  \tilde \uchr(\omega + \omer , r^{*}) e^{i \omega u} + \mathfrak t(\omega+\omer)  \partial_{v}^2 ( \tilde \uchl(\omega+\omer,r^{*}) e^{-i\omega v} )}{\omega } \Big] d \omega.\label{eq:termsofpartialvphi}
			\end{align}
			Recall from the proof of Lemma~\ref{lemma.t} that 	$\frac{d \uchl}{dr^{*}} =   O(\Omega^2_{RN}(u,v))$ and 	$\frac{d \uchr}{dr^{*}} =   O(\Omega^2_{RN}(u,v))$.	It is not difficult to also show that \begin{equation}
				\frac{d^2 \uchl}{d(r^{*})^2} =  \left([\omega-(\omega_r-\omega_+)]^2 - (\omega-\omer)^2 \right)\uchl+ O(\Omega^2_{RN}(u,v))= O(\Omega^2_{RN}(u,v)),
			\end{equation}\begin{equation}
				\frac{d^2 \uchr}{d(r^{*})^2} =  \left([\omega-(\omega_r-\omega_+)]^2 - (\omega-\omer)^2 \right)\uchr+ O(\Omega^2_{RN}(u,v))= O(\Omega^2_{RN}(u,v)).
			\end{equation}
			Therefore, we end up with   \begin{equation}\begin{split}
					&	\partial_v^2 \left(  r e^{i \frac{\omer  v}{2} }	\phil(u,v) \right)  =   - \frac{r_+  e^{i\frac{\omer u}{2}} }{ \sqrt{2\pi} } \int_{\mathbb R}  \mathcal F[\phiHL   e^{i \omer \cdot} ] (\omega) \omega   \mathfrak t(\omega+\omer)     e^{-i\omega v}   d\omega + \Phi_{error}'(u,v)\\ &  = - \frac{r_+  e^{i\frac{\omer u}{2}} }{ \sqrt{2\pi} } \int_{\mathbb R}  \mathcal F[\rd_v(\phiHL   e^{i \omer \cdot}) ] (\omega)    \mathfrak t(\omega+\omer)     e^{-i\omega v}   d\omega + \Phi_{error}'(u,v),\end{split}
			\end{equation} where $$ |\Phi_{error}'|(u,v)\lesssim \Omega_{RN}(u,v).$$Then, using again the notation  $\Phi_{H}(v):= \phiHL(v)   e^{i \omer v}$ and \eqref{t.est}, we obtain \begin{equation}\begin{split}
					&	\partial_v^2 \left(  r e^{i \frac{\omer  v}{2} }	\phil(u,v) \right)  =   - \frac{r_+  e^{i\frac{\omer u}{2}} }{ \sqrt{2\pi} } \int_{\mathbb R}  \mathcal F[\phiHL   e^{i \omer \cdot} ] (\omega) \omega   \mathfrak t(\omega+\omer)     e^{-i\omega v}   d\omega + \Phi_{error}'(u,v)\\ &  = - \frac{r_+  e^{i\frac{\omer u}{2}} }{ \sqrt{2\pi} } \left( \mathfrak t(\omer) \rd_v \Phi_{H}(v)  + \mathfrak t'(\omer) \rd_{v}^2\Phi_{H}(v)+\int_{\mathbb R}  \mathcal F[\rd_{v}^3 \Phi_{H} ] (\omega)    \tilde{\mathfrak t}_{\omer}(\omega)     e^{-i\omega v}   d\omega\right) + \Phi_{error}'(u,v) .\end{split}
			\end{equation}
			Finally, similarly to \eqref{fourier.est}, we obtain \begin{equation}\begin{split}
					&	\bigl|\int_{\mathbb R}  \mathcal F[\rd_{v}^3 \Phi_{H} ] (\omega)    \tilde{\mathfrak t}_{\omer}(\omega)     e^{-i\omega v}   d\omega\bigl| \\ \lesssim &(1+|v|)^{-N} \left(\| (1+|v|)^{N} \rd_{v}^3 \Phi_{H}\|_{L^{\infty}}\| \mathcal{F}[\tilde{\mathfrak t}_{\omer}]\|_{L^{1}}+ \|  \rd_{v}^3 \Phi_{H}\|_{L^{2}}\| (1+|v|)^{N}\mathcal{F}[\tilde{\mathfrak t}_{\omer}]\|_{L^{2}}\right)\\ \lesssim &(1+|v|)^{-N} \left(\| (1+|v|)^{N} \rd_{v}^3 \Phi_{H}\|_{L^{\infty}}+ \|  \rd_{v}^3 \Phi_{H}\|_{L^{2}}\right),
				\end{split}
			\end{equation}
			which concludes the proof of \eqref{dv.psi.linear.est}, \eqref{dvv.psi.linear.est}.  \eqref{du.psi.linear.est} is obtained similarly.
			
		\end{proof}

		\subsubsection{Recalling previous nonlinear estimates in the black hole interior from \cite{MoiChristoph}}
		We now retrieve estimates previously proven in \cite{MoiChristoph}, Proposition 6.16. Strictly speaking, the estimates of \cite{MoiChristoph} are stated assuming $s<1$, whereas here we assume to the contrary $s>1$, so rates of the form $O(v^{1-2s})$ or $O(v^{1-3s})$ are replaced by  the weaker rates $v^{-s}$ and $v^{-2s}$ in the $s>1$ case, respectively. We introduce the curve $\gamma$ from \cite{Moi,MoiChristoph} as $\gamma=\{(u,v),\ u+v =-\Delta'+ \frac{2s}{2|K_-|}\log(v),\ u \leq -\Delta\}$ for some constant $\Delta'>0$, and the notations $u_{\gamma}(v)$ and $v_{\gamma}(u)$ designed so that $(u_{\gamma}(v),v)\in \gamma $ and  $(u,v_{\gamma}(u))\in \gamma $.
		\begin{prop}[\cite{MoiChristoph}, Proposition 6.16.]\label{MoiChristoph.nonlinear.prop} Let $u\leq -\Delta$. Then, for $v \geq v_{\gamma}(u)$
			\begin{equation}\label{dvlogomega.est}
				|\rd_v \log(\Omega^2)(u,v)-2K_-| \lesssim v^{-s},
			\end{equation}
			\begin{equation}\label{lambda.est}
				|\rd_v r|(u,v) \lesssim v^{-2s},
			\end{equation}
			\begin{equation}\label{dvphi.est}
				|\rd_v \phi|(u,v) \lesssim v^{-s},
			\end{equation}
			\begin{equation}\label{dvphi.diff.est}
				\bigl|		|D_v \psi|(u,v)- 	|D_v \psi_{\mathcal{L}}|(u,v) \bigr|\lesssim v^{-2s}.
			\end{equation}
			\begin{equation}\label{duphi.diff.est}
				\bigl|		|D_u \psi|(u,v)- 	|D_u \psi_{\mathcal{L}}|(u,v) \bigr|\lesssim |u|^{-2s}\log(u).
			\end{equation}  Moreover, there exists $\theta(u)$ such that  in the gauge choice \eqref{A.gauge}
			\begin{equation}\label{dvphi.diff.est2}
				\bigl|		\rd_v \psi(u,v)-  e^{i \theta(u)}	\rd_v \psi_{\mathcal{L}}(u,v) \bigr|\lesssim v^{1-2s}.
			\end{equation} 
			
			Finally, denoting $\delta g=\{ r-r_{RN},\ \log(\Omega^2)- \log(\Omega^2_{RN})\}$, $\delta A_u = \{ A_u- A_u^{RN}\}$ and $\delta \phi=\phi- \phil$, we have the following difference estimates (away from the Cauchy horizon $\CH$): for all  $u\leq - \Delta$, $v\leq v_{\gamma}(u)$:\begin{equation}\label{delta.g}
				|\delta g|(u,v)+ |\rd_u \delta g|(u,v)+ |\rd_v \delta g|(u,v)+ |\delta A_u|(u,v) \ls v^{-s},
			\end{equation}
			\begin{equation}\label{delta.phi}
				|\delta \phi|(u,v)+  	|\rd_u \delta \phi|(u,v) + 	|\rd_v \delta \phi|(u,v) \ls v^{-2s}.
			\end{equation}
		\end{prop} 
		
		Then, we collect some extra estimates which follow  from the analysis in \cite{MoiChristoph}. These estimates are sub-optimal, but nonetheless sufficient for our purpose.
		
		\begin{prop}\label{MoiChristoph.nonlinear.prop.new} Let $u\leq -\Delta$. Then, for $v\geq v_{\gamma}(u)$:
			\begin{equation}\label{dv.lambda}
				|\rd_v^{2} r|(u,v) \lesssim v^{-2s},
			\end{equation}
			\begin{equation}\label{phi.est}
				|\phi|(u,v) \lesssim |u|^{1-s}.
			\end{equation}
			\begin{equation}\label{dvdvphi.diff.est}
				\bigl|		|D_v^2 \psi|(u,v)- 	|D_v^2 \psi_{\mathcal{L}}|(u,v) \bigr|\lesssim v^{-2s}.
			\end{equation}
		\end{prop}
		\begin{proof}
			\eqref{phi.est} is easily obtained integrating \eqref{dvphi.est} in $v$. For \eqref{dv.lambda}, we invoke \eqref{RaychV} which we write as \begin{equation}
				-	\rd_v \lambda + \lambda \rd_v \log(\Omega^2) = -r|D_v \phi|^2,
			\end{equation} and \eqref{dv.lambda} then follows from \eqref{dvlogomega.est},  \eqref{lambda.est}, \eqref{dvphi.est}. 
			
			Proving \eqref{dvdvphi.diff.est} is slightly more involved but follows from the same strategy as \cite{MoiChristoph}, Section 6: we control the difference between $\rd_v^2\psi$ and  $\rd_v^2\psi_{\mathcal{L}}$ with the help of \eqref{Field.dvpsi}. First, we prove \eqref{dvdvphi.diff.est} for $v\leq v_{\gamma}(u)$ using the difference estimates \eqref{delta.g} and \eqref{delta.phi} of Proposition~\ref{MoiChristoph.nonlinear.prop}. Then, in the region  $v\geq v_{\gamma}(u)$, one can integrate \eqref{Field.dvpsi} using the rest of the estimates (including \eqref{dv.lambda}) to estimate its RHS, and obtain \begin{equation}
				\bigl|	|D_v^2 \Psi|(u,v)  - |D_v^2 \Psi|(u_{\gamma}(v),v) \bigr|\ls v^{-2s},
			\end{equation} and, of course, a similar estimates holds  for $\bigl|	|D_v^2 \Psi_{\mathcal{L}}|(u,v)  - |D_v^2 \Psi_{\mathcal{L}}|(u_{\gamma}(v),v) \bigr|$. Combining the two estimates then gives  \eqref{dvdvphi.diff.est} in the region  $ v\geq v_{\gamma}(u),\ u\leq -\Delta$.
		\end{proof}
		
		\subsubsection{Completing the proof of Theorem~\ref{nonlinearscat.thm} and Corollary~\ref{nonlinearscat.cor}}
		
		Theorem~\ref{nonlinearscat.thm} follows immediately  from  Proposition~\ref{new.scat.prop} combined 
		with Proposition~\ref{MoiChristoph.nonlinear.prop} and Proposition~\ref{MoiChristoph.nonlinear.prop.new}, notably \eqref{dvphi.diff.est}, \eqref{dvphi.diff.est2} and \eqref{dvdvphi.diff.est} under the gauge \eqref{A.gauge}. We then turn to the proof of Corollary~\ref{nonlinearscat.cor}. \begin{proof}
			The following estimate: \begin{equation}\begin{split}
					& v^{-s} \lesssim |D_v \phi|(u,v) \lesssim v^{-s},\\ &   |D_v^2 \psi|(u,v) \lesssim v^{-s-1}.
				\end{split}
			\end{equation} follows immediately from Theorem~\ref{nonlinearscat.thm} as a consequence of  \eqref{dv.psi.nonlinear.est} and \eqref{dvv.psi.nonlinear.est}  combined with \eqref{EH.profile.hyp}  given that $s>1$ (note that we can choose $N=\lfloor s \rfloor+2$ in \eqref{dv.psi.nonlinear.est}  and  $N=\lfloor s \rfloor+3$ in \eqref{dvv.psi.nonlinear.est}). Then, using the formula together with \begin{equation}\begin{split}
					& D_v \psi = r D_v \phi  + [\rd_v r] \phi, \\ &  D_v^2 \psi = r D_v^2 \phi  + [\rd_v^2 r] \phi+2[\rd_v r] D_v\phi,
				\end{split}
			\end{equation} and \eqref{lambda.est}, \eqref{dv.lambda}, \eqref{dvphi.est}, \eqref{phi.est} concludes the proof of \eqref{Dv.phi.final.est}. To obtain \eqref{dv.r.final}, it is then enough to integrate \eqref{RaychV}, taking advantage of \eqref{Dv.phi.final.est}. \eqref{ImDv.phi.final.est} and \eqref{du.phi.final} are obtained similarly.

			Then, to prove \eqref{du.r.final}, note that \eqref{du.phi.final}, valid for all $u\leq -\Delta$ by Corollary~\ref{nonlinearscat.cor}  shows by \eqref{RaychU} that for all $u\leq -\Delta$ \begin{equation}
				|\rd_u r|_{|\CH}(u) \gtrsim |u|^{-2s},
			\end{equation}  and thus, 	$|\rd_u r|_{|\CH}(u)>0$. This concludes the proof of Corollary~\ref{nonlinearscat.cor}.
		\end{proof}

		\subsection{Piecing everything together and application of Theorem~\ref{main.thm}}\label{final.section}

		We are now finally ready to prove Theorem~\ref{main.thm.global.ii} as advertised at the beginning of this section.
		
		\begin{proof}
			Let $\Phi_H$ satisfying \eqref{EH.bound2} for some $s>\frac{3}{2}$ and $q \in(0,1)$ sufficiently small. We apply Corollary~\ref{EH.AF.cor} to construct an asymptotically flat  black hole spacetime with initial data on $\Sigma$ with no trapped or anti trapped spheres, with regular event horizon  $\mathcal{H}^+$ strictly to the future of $\Sigma$, a Cauchy horizon $\CH$, and $\mathcal{S}=\{r=0\}$ which is FLRW near $\Gamma$. Then, we impose that $\Phi_H$ satisfies the assumptions of Corollary~\ref{nonlinearscat.cor} and thus, by Corollary~\ref{nonlinearscat.cor}, the assumptions  \eqref{hyp1}, \eqref{hyp2} and \eqref{hyp4}--\eqref{hyp3} of Theorem~\ref{main.thm} are satisfied under the Cauchy horizon $\CH$ for $u\leq -\Delta$. To show that $\CH$ is weakly singular, we apply Theorem~\ref{main.thm.global.i}, Statement~\ref{E.}.	Note that, because the spacetime is FLRW near $\Gamma$, $\mathcal{CH}_{\Gamma} =\emptyset$ in the language of Theorem~\ref{Kommemi.thm} (in other words, there are no locally naked singularities). Thus, we can apply  Theorem~\ref{main.thm.global.i} to obtain the desired black hole construction. This concludes the proof of  Theorem~\ref{main.thm.global.ii}.
		\end{proof}
		Note that the second statement of Theorem~\ref{inext.thm} in the one-ended case immediately follows from Theorem~\ref{main.thm.global.ii}.
		\section{Unconditional constructions of  two-ended spacetimes}\label{section.global.uncond2}
		
		In this section, we address  the proof of the unconditional results in the two-ended case, i.e.,  Theorem~\ref{main.thm.2end.ii}.

		In the two-ended case, we are allowed to work with $q_0=0$, where the results of Luk--Oh \cite{JonathanStabExt,twotails} and Gautam \cite{Gautam} are available. In fact, the class of spacetime we construct is rather large (contrary to the one-ended case), since it applies to any generic solution which happens to have a non-empty crushing singularity $\mathcal{S}=\{r=0\}$. We note that this set of solutions we construct is open in some topology, since the property $\mathcal{S}\neq \emptyset$ is stable to perturbations (by a soft Cauchy stability argument). We show additionally that the class of generic solutions such that $\mathcal{S} \neq \emptyset$ is non-empty by a new argument relying of Theorem~\ref{breakdown.thm.new}.

		In this subsection, we assume that $q_0 =0$ and that the scalar field $\phi$ is real-valued and we consider a two-ended spacetime MGHD arising from admissible initial data $(\mathcal{M},g_0,K_0,\phi_0,\phi_1,e) \in \mathcal{G}$, where $\mathcal{G}$ is the generic set (in  smooth topology) on which $L^2$-averaged lower bounds hold on the event horizon, as constructed in \cite{JonathanStab,JonathanStabExt}.
		
		Before turning to the proof of Theorem~\ref{main.thm.2end.ii}, we prove a lemma regarding the propagation of scalar field upper and lower bounds. Since these estimates follow from known techniques \cite{MihalisPHD,JonathanStab,massinflationYakov,Moi}, we will simply sketch the proof. We also note that the proof of the following lemma only applies to the uncharged scalar field case, i.e, $q_0=0$ in \eqref{1.1}--\eqref{5.1}. The charged scalar field case $q_0 \neq 0$ is more complicated and requires the use of Fourier transform, see Section~\ref{scattering.section}.
		
		\begin{lemma}\label{propagation.lemma}
			We consider   $C^1$ initial data on $\mathcal{H}^+ \cup \Cin$ in the framework of Theorem~\ref{CH.thm.SS}. Let $s>1$.
			
			\begin{itemize}
				\item Suppose that as $v\rightarrow+\infty$ \begin{equation}\label{upper.1.hyp}
					|\rd_v \phi|_{|\mathcal{H}^+}(v) \ls v^{-s}.
				\end{equation} Then, there exists $u_s<0$ sufficiently negative so that for all $u\leq u_s$, $v$ sufficiently large \begin{equation}\label{upper.1}
					|\rd_v \phi|\blue{(u,v)}\ls v^{-s}.
				\end{equation}
				\item Assume that \eqref{upper.1.hyp} holds and moreover that as $v\rightarrow+\infty$ \begin{equation}\label{lower.1.hyp}
					\rd_v \phi_{|\mathcal{H}^+}(v)  \gtrsim v^{-s}.
				\end{equation} Then, for all $u\leq u_s$, $v$ sufficiently large \begin{equation}\label{lower.1}
					\rd_v \phi\blue{(u,v)} \gtrsim v^{-s}.
				\end{equation}
				\item Assume that \eqref{upper.1.hyp} holds and moreover that as $v\rightarrow+\infty$ \begin{equation}\label{upper.2.hyp}
					|	\rd_v^2 \phi|_{|\mathcal{H}^+}(v) \ls v^{-s-1}.
				\end{equation} Then, for all $u\leq u_s$, $v$ sufficiently large \begin{equation}\label{upper.2}
					|	\rd_v^2 \phi|\blue{(u,v)} \ls v^{-s-1}.
				\end{equation}
			\end{itemize}
			
		\end{lemma}
		
		\begin{proof}
			The proof of \eqref{upper.1} is well-known, see e.g., \cite{JonathanStab,Moi}; to give a brief sketch, we split the spacetime rectangle $[-\infty,u_s] \times [v_0,+\infty)$ into several regions where different estimates are proved: \begin{itemize}
				\item the red-shift region $\blue{\mathfrak{R}}=\{u+v \leq -\Delta\}$ for $\Delta>0$, where red-shift estimates are exploited. Note that $\lambda<0$ to the future of $\blue{\mathfrak{R}}$, and inside $\blue{\mathfrak{R}}$, $\lambda(u,v) \ls v^{-2s}$.
				\item the no-shift region $\mathcal{N}=\{-\Delta \leq u+v \leq \Delta_N\}$ for $\Delta_N>0$ where a Gr\"{o}nwall argument is used.
				\item the early blue-shift region  $\mathcal{EB}=\{\Delta_N \leq u+v \leq -\Delta'+ \frac{2s}{2|K_-|(M,e)}\log(v)\}$ for $\Delta'>0$, where the logarithmic-size of the region is exploited.
				\item the late blue-shift region  $\mathcal{LB}=\{ u+v \geq -\Delta'+ \frac{2s}{2|K_-|(M,e)}\log(v)\}$, where the blue-shift is exploited, specifically the estimate $\Omega^2 \ls v^{-2s}$ in this region. Note  that $|\lambda|(u,v) \ls v^{-2s}$ for $(u,v) \in \mathcal{LB}$.
			\end{itemize}
			For \eqref{lower.1}, we use a special monotonicity for \eqref{Field} in the case $q_0=0$, first exploited by Dafermos \cite{MihalisPHD,Mihalis1}; we quickly sketch the main argument here. Integrating \eqref{Field2} in $u$, we get\blue{, in the notations $\theta=r\rd_v \phi$ and $\xi=r\rd_u \phi$} \begin{equation}
				\theta(u,v)  =  	r\rd_v \phi_{|\mathcal{H}^+}(v) + \int^{u}_{-\infty} \frac{-\lambda }{r} \xi(u',v) du',
			\end{equation} then integrating \eqref{Field3} in $v$ gives  \begin{equation}
				\theta(u,v)  =  	r\rd_v \phi_{|\mathcal{H}^+}(v) + \int^{u}_{-\infty} \frac{-\lambda }{r}(u',v) \xi(u',v_0) du'+ \int^{u}_{-\infty} \frac{-\lambda }{r}(u',v)[ \int^{v}_{v_0} \frac{-\nu }{r}(u',v')\theta(u',v') dv'] du'.
			\end{equation} Note that $\int^{u}_{-\infty} \frac{-\lambda }{r}(u',v) \xi(u',v_0) du'=O(e^{-2K_+v})$ if $(u,v) \in \blue{\mathfrak{R}}\cup \mathcal{N} \cup \mathcal{EB}$ (red-shift estimate on the initial data term $\xi(u',v_0)$) and  $\int^{u}_{-\infty} \frac{-\lambda }{r}(u',v) \xi(u',v_0) du'=O(v^{-2s})$ if $(u,v ) \in  \mathcal{LB}$. Moreover, $\frac{-\lambda }{r}(u',v)>0$, except if $(u',v) \in \blue{\mathfrak{R}}$, in which case $\lambda(u',v)=O(v^{-2s})$. All in all, after bootstrapping that $\theta$ remains positive, we get   \begin{equation}
				\theta(u,v)  \geq  	r\rd_v \phi_{|\mathcal{H}^+}(v) +O(v^{-2s}),
			\end{equation}
			which then gives \eqref{lower.1} as a consequence of \eqref{lower.1.hyp}.
			
			\noindent \eqref{upper.2} is slightly more subtle to prove. First, consider $\phi_{\mathcal{L}}$ a solution of the linear wave equation on \eqref{RN2} \begin{equation}
				\Box_{g_{RN}} \phi_{\mathcal{L}}=0,
			\end{equation} assuming the initial data on $\mathcal{H}^+$ satisfy \eqref{upper.1.hyp}, \eqref{upper.2.hyp}. Since the Reissner--Nordstr\"{o}m metric is stationary, one can commute $\phi_{\mathcal{L}}$ with the Killing vector field $T= \rd_v - \rd_u$, and then \eqref{upper.1.hyp} is satisfied for $T \phi_{\mathcal{L}}$ at the level $s+1$ (since $[\rd_u\phi_{\mathcal{L}}]_{|\mathcal{H}^+}=0$ \blue{and $[\rd_u ]_{|\mathcal{H}^+}=0$}, due to red-shift), so by \eqref{upper.1}, we have for all $(u,v) \in \mathcal{LB}$:\begin{equation}
				|\rd_v T\phi_{\mathcal{L}}|(u,v) \ls v^{-s-1},
			\end{equation} and by \eqref{Field} (applied on the Reissner--Nordstr\"{o}m metric), we show that $\rd_v \rd_v \phi_{\mathcal{L}}$ decays faster and thus  we have \begin{equation}\label{dvv.phi.l}
				|\rd_v^2 \phi_{\mathcal{L}}|(u,v)\blue{,\ 	|\rd_v^2 \psi_{\mathcal{L}}|(u,v)} \ls v^{-s-1},
			\end{equation}\blue{for all $(u,v) \in \mathcal{LB}$.} Then, we argue as in Section~\ref{scattering.section} and compare the linear and nonlinear scalar field estimates: from Proposition~\ref{MoiChristoph.nonlinear.prop.new}, we know that for $(u,v) \in \mathcal{LB}$, \begin{equation}
				|\rd_v^2 \psi(u,v) - \rd_v^2 \psi_{\mathcal{L}}(u,v)| \ls v^{-2s}.
			\end{equation} Since $s>1$, \eqref{upper.2} is proved as a consequence of $|\lambda|(u,v) \ls v^{-2s}$ \blue{and $|\rd_v\lambda|(u,v) \ls v^{-2s}$} for \blue{all} $(u,v) \in \mathcal{LB}$ \blue{(invoking again Proposition~\ref{MoiChristoph.nonlinear.prop.new} for this estimate)} and \eqref{dvv.phi.l}.
		\end{proof}
		Finally, we turn to the proof of Theorem~\ref{main.thm.2end.ii}.
		\begin{proof}   We start with a two-ended spacetime MGHD arising from admissible initial data $(\mathcal{M},g_0,K_0,\phi_0,\phi_1,e) \in \mathcal{G}$. We will focus on the right-most side of the two-ended black hole, where we denote $\mathcal{H}_R$ and $\mathcal{CH}_R$ the event and Cauchy horizons, respectively; the left-side is treated analogously.  First, \eqref{hyp5} is trivially satisfied since $\phi$ is real-valued.    To show that \eqref{hyp1},  \eqref{hyp2}, \eqref{hyp4} and \eqref{hyp3} are satisfied for all $\uH < u_0 < \uch$, we first appeal to \cite{Gautam}, Theorem 1.15 (which itself relies on an application of  Main Theorem 2 in \cite{twotails}, upon showing its assumptions are satisfied) showing that on the event horizon $\mathcal{H}_R$: \begin{equation}\label{EHbound}
				\rd_v \phi_{|\mathcal{H}_R}(v) = D v^{-4} + o(v^{-4}),\ \rd^2_{v v} \phi_{|\mathcal{H}_R}(v)= O(v^{-5}) \text{ as } v\rightarrow+\infty,
			\end{equation} for some $D\in \mathbb{R}$ and $D\neq 0$  from the fact that $(\mathcal{M},g_0,K_0,\phi_0,\phi_1,e) \in \mathcal{G}$, see \cite{JonathanStabExt}. As a consequence of \eqref{EHbound}, we can apply Theorem~\ref{CH.thm.SS} to show that $\mathcal{CH}_{R} \neq \emptyset$, and symmetrically  $\mathcal{CH}_{L} \neq \emptyset$.  Finally, by \cite{Gautam}, Corollary 1.1.4 (which relies \blue{too} on \cite{massinflationYakov} ), we know that mass inflation occurs, i.e., \begin{equation}\label{mass.inf.2end}
				\varpi_{|\mathcal{CH}_{L}}=+\infty,\  \varpi_{|\mathcal{CH}_{R}}=+\infty.
			\end{equation} \noindent	Then, by Lemma~\ref{propagation.lemma}, for all $u_{\mathcal{H}^+}<u<u_s$: \begin{equation}\label{insidebound}
				\frac{D}{2 } v^{-4} \leq \rd_v \phi(u,v) \lesssim v^{-4},\ |\rd^2_{v v} \phi(u,v)|\ls v^{-5} \text{ as } v\rightarrow+\infty.
			\end{equation} Therefore,  \eqref{hyp1} (note that the lower bound on $|\rd_u r|$ in \eqref{hyp1} is a consequence of \eqref{mass.inf.2end}, see \cite{Moi4}), \eqref{hyp2},  \eqref{hyp4} and \eqref{hyp3} hold for $s=4$ as a consequence of \eqref{insidebound}. 
			
			Then, by Statement~\ref{Atwo} of Theorem~\ref{main.thm.2end.i}, we have $\mathcal{S}_{i^+} =\emptyset$ already.  Moreover, define $\mathcal{G}' \subset \mathcal{G}$ as the subset of spacetimes such that $\mathcal{S}\neq \emptyset$ (we will show below that $\mathcal{G}'\neq \emptyset$). Theorem~\ref{main.thm} also shows that there exists  $u_{T} <\uch$, $v_{T} \in \RR$ such that  $(u_T,\uch]\times[v_T,+\infty) \subset \T$. Since $\T$ is open, there exists $\ep>0$ such that   $(u_T,\uch+\ep)\times\{v_T\} \subset \T$ (note that since $\mathcal{S}\neq \emptyset$, $(u_T,\uch+\ep)\times\{v_T\} \subset \mathcal{Q}^+$ for $\ep>0$ small enough). By the monotonicity of \eqref{RaychV},  $\blue{\mathcal{Q}^+\cap\{  u_T<u<\uch+\ep,\ v\geq v_T\}}\subset \T$. Therefore, there exists $u_F\in \RR$ such that $\mathcal{Q}^+\cap\left([u_T,u_F)\times \{v_T\} \right)\subset \T$, and $\lim_{u \rightarrow u_F} r(u,v_T) =0$. Thus, we can apply Theorem~\ref{main.thm} to the region $\mathcal{Q}^+\cap \{ u \geq u_T, v\geq v_T\}$, therefore there exists $\delta>0$ such that $\mathcal{S}\cap \{u \leq \uch+\delta\}$ is spacelike and the Kasner asymptotics \eqref{K1}-\eqref{K5} hold.\\

			To show that $\mathcal{G}'\neq \emptyset$, we will modify $\mathcal{M}$ (of course, it could be the case that $\mathcal{S}\neq \emptyset$ in $\mathcal{M}$ already, but since we do not know that, we will modify $\mathcal{M}$ regardless if whether  $\mathcal{S}= \emptyset$ or not) to obtain another admissible spacetime with $\mathcal{S}\neq \emptyset$. The strategy relies on the argument leading to the proof Theorem~\ref{breakdown.thm.new}, embodied by Proposition~\ref{apriori.prop1}. Let us fix $u_{\mathcal{H}_R} < u_0^{R} \leq u_{\mathcal{CH}_R}$. By Proposition~\ref{apriori.prop1}, we know that there exist  constants $C_R^{+}>0$, $C_R^{-}>0$, $v_0^{R}$ (only depending on $u_0^{R}$) such that  for all $u_{0}^{R}\leq u \leq u_{\mathcal{CH}_R}$, $v\geq v_0^{R}$ \begin{equation}\label{last.crucial}
				C_R^{-}\ v^{1-2s}		\leq	|r^2(u,v)- r^2_{|\mathcal{CH}_R}(u)| \leq C_R^{+}\ v^{1-2s}.
			\end{equation} In particular, for all $u_{0}^{R}\leq u \leq u_{\mathcal{CH}_R}$, $v\geq v_0^{R}$ \begin{equation}\label{last.crucial'}
				C_R^{-}\ v^{1-2s}		\leq	r^2(u,v).
			\end{equation} 
			
			Let $V_R\geq v_0^{R}$. There is, of course, the analogous notion on the left,  with evident notations $v_0^{L}$, $u_0^{L}$, $U_L$.\\		 We construct a new two-ended black hole manifold $\mathcal{M}'$ as such: \begin{enumerate}
				\item \label{last0} $\mathcal{M}' $ is the MGHD of admissible two-ended asymptotically flat initial data on a  spacelike hypersurface $\Sigma'$.
				\item \label{last1} $\mathcal{M}' \cap \{ u \leq u_0^{R}, v \geq V_R\}= \mathcal{M} \cap \{ u \leq u_0^{R}, v \geq V_R\}$ and    $\mathcal{M}' \cap \{ v \leq v_0^{L}, u \geq U_L\}= \mathcal{M} \cap \{v \leq v_0^{L}, u \geq U_L\}$. In particular, $\mathcal{M}$ and $\mathcal{M'}$ have the same asymptotically flat ends and thus $\mathcal{M}'\in \mathcal{G}$.
				\item \label{last2} As a consequence of the above,  $\mathcal{CH}_R\neq \emptyset$ and $\mathcal{CH}_L\neq \emptyset$ in $\mathcal{M}'$ and $\mathcal{CH}_R$, $\mathcal{CH}_L$ are weakly singular in the sense of mass inflation, i.e.,  \eqref{mass.inf.2end} holds.
				
				\item  \label{last3} $\mathcal{S}\neq \emptyset$ in $\mathcal{M}'$.

			\end{enumerate}
			
			Note that for \eqref{last2}, the Hawking mass blow up on $\CH$ is clear in the regions   $\mathcal{M}' \cap \{ u \leq u_0^{R}, v \geq V_R\}$  and $\mathcal{M}' \cap \{ v \leq v_0^{L}, u \geq U_L\}$, since they coincide with the corresponding regions in $\mathcal{M}$ where mass inflation occurs on $\CH$ by \eqref{mass.inf.2end},  and \eqref{last2} then follows from the propagation of the Hawking mass blow-up towards the future. We now describe the construction and how to arrange for \eqref{last0} and \eqref{last3} to be satisfied. On the right side, let us pose \blue{new} initial data on the ingoing cone \begin{equation}
				\underline{C}_{V_R} = [u_0^{R}, u_F^{R}] \times \{V=V_R\},
			\end{equation} such that at the future end-point of $	\underline{C}_{V_R}$  \begin{equation}\label{r.small2}
				0<	r^2(u_F^{R},V_R) < \frac{C_R^{-}}{4} V_R^{1-2s}.
			\end{equation}  Then, we consider the solution of \eqref{1.1}--\eqref{5.1} with $q_0=0$ with bicharacteristic data on 	$\underline{C}_{V_R} \cup \left( \{u=u_0^{R}\}\times[V_R,+\infty)\right)$ and note that \eqref{last.crucial} is still satisfied since $C_R^{-}$ only depends on the estimates \eqref{hyp1} on the cone $\{u=u_0^{R}\}
			\times[V_R,+\infty)$, which are unchanged. We claim that for the new solution in the spacetime region $\mathcal{M}'\cap\{u_0^{R}\leq u \leq u_F^{R},\ v\geq V_R\}$, the right Cauchy horizon breaks down, namely the following is true \begin{equation}\label{last.crucial2}
				u_{\mathcal{CH}_R}< u_F^{R}.
			\end{equation} 
			Indeed, if \eqref{last.crucial2} did not hold, \eqref{r.small2}  would contradict \eqref{last.crucial'}. So \eqref{last.crucial2} indeed holds, and by Proposition~\ref{apriori.prop2}, there exists $\mathcal{S}=\{r=0\} \neq \emptyset$, a future-spacetime boundary that can be attached to the spacetime region $\mathcal{M}'\cap\{u_0^{R}\leq u \leq u_F^{R},\ v\geq V_R\}$,  which shows \eqref{last3} is true.

			To conclude, we still need to show \eqref{last1} holds. For this, note that we have constructed two solutions of \eqref{1.1}--\eqref{5.1}: on the right, $\mathcal{M}'_R$ in the spacetime region $\{u \leq u_F^{R},\ v\geq V_R\}$ and \blue{on the left,} $\mathcal{M}'_L$ in the spacetime region  $\{v \leq v_F^{L},\ u\geq U_L\}$. To connect these two, we can choose $(u_F^{R},V_R) = (U_L,v_F^{L})$ and connect them at the sphere $(U_L,V_R)$. Then, we invoke local-well posedness for \eqref{1.1}--\eqref{5.1} with bicharacteristic initial data   on \begin{equation}
				\left(	\{U_L\} \times [V_R-\ep,V_R] \right) \cup  	\left(	 [U_L-\ep,U_L]\times \{V_R\}  \right),
			\end{equation} and choosing $\ep>0$ small enough, we obtain a solution of \eqref{1.1}--\eqref{5.1} in the spacetime rectangle $[U_L-\ep,U_L] \times [V_R-\ep,V_R]$. Moreover, noting that $(U_L,V_R)$ is trapped, i.e.,  $\rd_u r(U_L,V_R)<0$, and $\rd_v r(U_L,V_R)<0$, by Cauchy stability, we have that \begin{equation}
				\rd_u r(u,v)<0,\ \rd_v r(u,v)<0 \text{ for all } (u,v) \in [U_L-\ep,U_L] \times [V_R-\ep,V_R].
			\end{equation}
			Thus, it is easy to construct a two-ended asymptotically spacelike hypersurface $\Sigma'$ which is admissible in the sense of \cite{JonathanStabExt}, i.e., that there exists $u_L^{ad}$, $v_L^{ad}$ such that \begin{equation}
				\rd_u r>0 \text{ on } \Sigma' \cap \{ u < u_L^{ad}\},\ 	\rd_u r<0 \text{ on } \Sigma' \cap \{ u > u_L^{ad}\},
			\end{equation}
			\begin{equation}
				\rd_v r>0 \text{ on } \Sigma' \cap \{ v < v_R^{ad}\},\ 	\rd_v r<0 \text{ on } \Sigma' \cap \{ v> v_R^{ad}\}.
			\end{equation} Moreover, we can arrange that $\Sigma'$ in $\mathcal{M}'$ coincides with $\Sigma$ in  $\mathcal{M}$ when the area-radius $r$ is sufficiently large.
			
		\end{proof}
		\begin{rmk}
			Note that ensuring \eqref{r.small2} requires a large scalar field on the ingoing cone $\underline{C}_{V_R}$ (this is a consequence of the Raychaudhuri equation \eqref{RaychU}, see also \cite{bif}, Section 8 where this idea is discussed). Therefore, the construction of Theorem~\ref{main.thm.2end.ii} is indeed a large-data one, which we know is necessary to obtain $\mathcal{S}\neq \emptyset$, since by Theorem~\ref{Dafermos.thm}, small scalar field data leads, on the contrary,  to   $\mathcal{S}= \emptyset$.
		\end{rmk}

		\small
		\bibliographystyle{plain}
		\bibliography{bibliography.bib}

	\end{document}